\documentclass[a4paper,11pt]{article}

\usepackage{jheppub}

\usepackage[numbers,sort&compress]{natbib} 
\usepackage{url}
\usepackage{amsmath}
\usepackage{rotating}

\usepackage{multirow}
\usepackage{amsmath}

\usepackage[T1]{fontenc} 

\usepackage{float, array, xspace, amscd, amsthm}
\usepackage{fancyhdr}
\usepackage{longtable}

\theoremstyle{theorem}
\newtheorem{theorem}{Theorem}
\newtheorem{lemma}{Lemma}
\newtheorem{proposition}{Proposition}

\theoremstyle{definition}

\newtheorem{definition}{Definition}

\theoremstyle{remark}
\newtheorem{remark}{Remark}[section]

\allowdisplaybreaks

\usepackage[usenames,dvipsnames,svgnames,table,x11names]{xcolor}
\definecolor{MagentaXD}{RGB}{204, 48, 152}
\definecolor{MagentaXDdetail}{RGB}{150, 79, 126}
\definecolor{GreenMAF}{RGB}{28, 112, 46}
\definecolor{GreenMAFdetail}{RGB}{80, 117, 88}
\definecolor{detail}{RGB}{110,110,110}
\definecolor{quantumviolet}{HTML}{53257F} 
\definecolor{quantumgray}{HTML}{555555} 
\definecolor{quantumgreen}{HTML}{007474} 
\definecolor{quantumblue}{HTML}{002366} 
\definecolor{quantumpurple}{HTML}{66023C} 
\definecolor{quantumdarkviolet}{HTML}{5D3954} 

\newsavebox{\ns}
\newsavebox{\dbrane}
\newsavebox{\dbshort}

\def\be{\begin{equation}}
\def\ee{\end{equation}}
\def\bea{\begin{eqnarray}}
\def\eea{\end{eqnarray}}

\newcommand{\Vcal}{\mathcal{V}}

\newlength{\sswidth}

\numberwithin{equation}{section}

\usepackage{tikz}
\usetikzlibrary{positioning,intersections}
\usetikzlibrary{calc}
\usetikzlibrary{shapes.geometric}
\usepackage{braids}
\usepackage{tikz-cd}

\newif\ifcomments
\commentstrue

\newif\ifdetails
\detailstrue

\usepackage{upgreek}
\usepackage[normalem]{ulem}
\usepackage{mathtools}
\usepackage{datetime}
\usepackage{cancel}

\usepackage[all,cmtip,knot]{xy}
\usepackage{xypic}

\newcommand{\orcid}[1]{\href{https://orcid.org/#1}{\includegraphics[width=8pt]{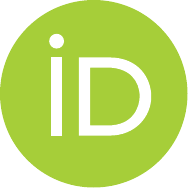}}}

\usepackage{amsmath,amsfonts,amsthm}

\usepackage{bm,latexsym,galois,euscript,dsfont}
\newcommand\EA{\EuScript{A}}
\newcommand\EB{\EuScript{B}}
\newcommand\EC{\EuScript{C}}
\newcommand\ED{\EuScript{D}}

\newcommand\EM{\EuScript{M}}
\newcommand\EN{\EuScript{N}}
\newcommand\EP{\EuScript{P}}

\newcommand\Fun{\mathsf{Fun}}

\newcommand\Rep {\mathsf{Rep}}
\newcommand{\Irr}{\operatorname{Irr}}

\newcommand\id {\mathrm{id}}
\newcommand\Hom {\mathrm{Hom}}
\newcommand{\one}{\mathds{1}}
\newcommand\Mod{\textsf{Mod}}

\newcommand{\FPdim}{\operatorname{FPdim}}

\newcommand{\VV}{\mathcal{V}}
\newcommand{\AMod}{{_{\mathfrak{A}}}\mathsf{Mod}}
\newcommand{\BMod}{{_{\mathfrak{B}}}\mathsf{Mod}}

\newcommand{\What}{\hat{W}}
\newcommand{\Vhat}{\hat{V}}

\newcommand{\FA}{\mathfrak{A}}
\newcommand{\FB}{\mathfrak{B}}
\newcommand{\FM}{\mathfrak{M}}
\newcommand{\FN}{\mathfrak{N}}

\usepackage{mathtools,amssymb,scalerel}

\newcommand\biopencrossl{%
	\mathrel{\scalerel*{>\kern-.4\LMpt\joinrel\blacktriangleleft}{x}}}
\newcommand\biopencrossr{%
	\mathrel{\scalerel*{\blacktriangleright\joinrel\kern-.4\LMpt<}{x}}}

\settimeformat{ampmtime}

\title{\color{black} On weak Hopf symmetry and weak Hopf quantum double model
}

\author[a,b]{Zhian Jia\orcid{0000-0001-8588-173X},}
\author[c]{Sheng Tan,}
\author[a,b]{Dagomir Kaszlikowski,}
\author[d]{and Liang Chang}
\affiliation[a]{Centre for Quantum Technologies, National University of Singapore, SG 117543, Singapore}
\affiliation[b]{Department of Physics, National University of Singapore, SG 117543, Singapore}
\affiliation[c]{Department of Mathematics, Purdue University, West Lafayette, IN 47907, USA}
\affiliation[d]{School of Mathematical Sciences and LPMC, Nankai University, Tianjin 300071, China}

\emailAdd{giannjia@foxmail.com}
\emailAdd{tan296@purdue.edu}
\emailAdd{phykd@nus.edu.sg}
\emailAdd{changliang996@nankai.edu.cn}

\abstract{
Symmetry is a central concept for classical and quantum field theory, usually, symmetry is described by a finite group or Lie group.
In this work, we introduce the weak Hopf algebra extension of symmetry, which arises naturally in anyonic quantum systems; and we establish weak Hopf symmetry breaking theory based on the fusion closed set of anyons.
As a concrete example, we implement a thorough investigation of the quantum double model based on a given weak Hopf algebra and show that the vacuum sector of the model has weak Hopf symmetry.
The topological excitations and ribbon operators are discussed in detail.
The gapped boundary and domain wall theories are also established. We show that the gapped boundary is algebraically determined by a comodule algebra, or equivalently, a module algebra; and the gapped domain wall is determined by the bicomodule algebra, or equivalently, a bimodule algebra.
The microscopic lattice constructions of the gapped boundary and domain wall are discussed in detail.
We also introduce the weak Hopf tensor network states, via which we solve the weak Hopf quantum double lattice models on closed and open surfaces.
The duality of the quantum double phases is discussed in the last part.
}

\begin{document}
	
\maketitle
\flushbottom

\section{Introduction}
\label{sec:intro}

The study of topologically ordered phases of matter has attracted a considerable amount of attention during the past several decades \cite{Wen2004,Levin200photons,Nayak2008,Chiu2016classification,Witten2016fermion,Zhou2017quantum,Wen2017zoo}.
These topological quantum phases of matter usually have symmetries that are beyond the Landau-Ginzburg symmetry-breaking paradigm.  To understand these exotic quantum phases of matter, several kinds of generalized concepts of symmetry have been proposed: (i) higher form symmetry, for which the charged particles have a dimension $p>0$, such as string, membrane, volume, etc. \cite{gaiotto2015generalized};
(ii) non-invertible symmetry, for which the symmetry operators are in general non-invertible and form a fusion algebra \cite{bhardwaj2018finite,bhardwaj2022universal,bartsch2022non} rather than a group;
(iii) categorical symmetry, where some monoidal category with extra structures characterizes the symmetry, like the unitary fusion category (UFC) and unitary modular tensor category (UMTC), and even higher category \cite{Levin2005,kong2014braided,johnson2022classification}.
These concepts of symmetry are not independent, they are closely related to each other, and many efforts have been made to subsume these symmetries in a unified framework.

Categorical symmetry is a promising candidate framework for characterizing the generalized phase transition where there is no symmetry breaking in the usual sense, like fractional quantum Hall system \cite{Tsui1982, Wen2004}. There are several different constructions of UFC symmetry in $(1+1)D$, like rational CFT \cite{FROHLICH2007duality}, TQFT \cite{thorngren2019fusion,thorngren2021fusion,inamura2021topological,inamura2022lattice}, anyonic chain \cite{Feiguin2007interacting,Gils2013anyonic,buican2017anyonic}, and so on.
The UMTC symmetry in $(2+1)D$ is also extensively investigated, like Levin-Wen string-net model \cite{Levin2005}, Kitaev quantum double model \cite{Kitaev2003}, etc.
It is well known that the representation category of a Hopf algebra is a UFC. And when the Hopf algebra is equipped with a quasi-triangular structure, its representation category is a UMTC.
This quantum group symmetry has been investigated for a long time \cite{majid2000foundations}, and it sits in between the usual group symmetry and categorical symmetry.
A comprehensive understanding of Hopf algebra symmetry may shed new light on the understanding of phase, phase transition, and categorical symmetry.

A typical Hopf algebra symmetry is the so-called quantum double symmetry.
Consider the $(2+1)D$ discrete gauge theory, like Gauge-Higgs theory for which Higgs field has broken the continuous gauge group $U$ down to a finite group $G$ \cite{Propitius1995,bais1992quantum}. The theory becomes topological in the low energy limit or long-distance limit. 
The topological excitations carry the fusion and braiding data. These data are mathematically captured by $\mathsf{Rep}(D(G))$, the representation category of quantum double for some finite group $G$.
The Hopf algebraic generalization of the formalism is discussed in \cite{bais2003hopf}, where the unbroken symmetry is given by a Hopf algebra $H$, for which group algebra is a particular case.
It turns out that this Hopf algebraic symmetry has broad applications in quantum gravity \cite{Bais2002quantum,majid2000foundations,delcamp2017fusion}, conformal field theory (CFT) \cite{fuchs1995affine}, topological quantum field theory (TQFT) \cite{meusburger2021hopf,meusburger2017kitaev}, quantum Hall effect \cite{Slingerland2001quantum}, topological phase \cite{Buerschaper2013a,buerschaper2013electric,meusburger2017kitaev,meusburger2021hopf,koppen2020defects,girelli2021semidual,voss2021defects,chen2021ribbon,jia2022boundary}, and so on.

From a lattice model perspective, Kitaev proposed the so-called quantum double model based on a finite group \cite{Kitaev2003}, which turned out to be a powerful lattice model that has potential applications in many fields, including topological quantum computation \cite{Nayak2008}, and topological quantum error correcting code \cite{Terhal2015quantum}.
The topological orders of these quantum double phases have been extensively investigated from many aspects, e.g., topological ground state degeneracy \cite{Kitaev2003}, anyons and their mutual statistics \cite{Kitaev2003}, topological entanglement entropy \cite{Kitaev2005topological}, boundary and domain wall \cite{bravyi1998quantum,Bombin2008family,Beigi2011the,Kitaev2012a,Cong2017}, anyon condensation \cite{Kong2014}, and electric-magnetic (EM) duality \cite{Buerschaper2009mapping,buerschaper2013electric,wang2020electric,hu2020electric,Jia2022electric}.
Its twisted generalization \cite{Bullivant2017twisted} and higher dimensional (twisted) generalization \cite{Moradi2015universal,Wan2015twisted,wang2018gapped,Levin2005,Hamma2005string,kong2020defects,delcamp2021tensornet,delcamp2022tensor} are also discussed.
However, in a general (weak) Hopf algebra setting, such a characterization of topological properties of the quantum double model has been proved far more elusive and they are largely unexplored \cite{Buerschaper2013a,meusburger2017kitaev,meusburger2021hopf,koppen2020defects,girelli2021semidual,voss2021defects,chen2021ribbon,jia2022boundary},
even less research has been done on the weak Hopf quantum dual model  \cite{chang2014kitaev}.

On the other hand, the Hopf and weak Hopf algebras provide us with the general framework to understand symmetry and quantized space-time.
The notion of a weak Hopf algebra is proposed in \cite{BOHM1998weak}. As a generalization of Hopf algebra, it is co-associative with weaker unit and counit.
They are designed in such a way that their representation category has a monoidal structure.
In fact, it is shown that any finite monoidal category can be realized as a representation category for some weak Hopf algebra \cite{ostrik2003module}. 
Thus a comprehensive investigation of weak Hopf symmetry will be helpful for us to understand categorical symmetry.
Compared with categorical symmetry, weak Hopf symmetry is easier to control due to its intimate relation with the group description of symmetry.
This is one motivation for us to carry out the current work.
From the lattice model perspective, the weak Hopf quantum double lattice models are of their own interests.
Since both group algebra and Hopf algebra are special cases of weak Hopf algebras, establishing the theories of topological excitation, gapped boundary, gapped domain wall, and duality of weak Hopf quantum double model can solve the corresponding problems for quantum double phase once and for all.

In this work, we take the first step to understanding the weak Hopf symmetry and weak Hopf quantum double model.
We establish the theory of weak Hopf symmetry and its breaking. 
As a concrete example, we carefully investigate the weak Hopf quantum double model, including its lattice construction, topological excitation and ribbon operators, its boundary and domain wall theory, and the duality between different weak Hopf quantum double models.

Our first aim of the paper is to summarize the basic facts about weak Hopf algebras and establish the theory of weak Hopf symmetry and weak Hopf quantum double model in a way hopefully accessible to readers from high energy, condensed matter, and quantum information communities.
Thus in the first part (Sec.~\ref{sec:pre}), we collect some necessary definitions, formulae, and theorems that will be used later. We have tried to give more details whenever we think it is not too lengthy to explain. 

Sec.~\ref{sec:WHAsymmetry} establishes the theory of weak Hopf symmetry and its breaking.
Our discussion of symmetry focuses on the symmetry of the vacuum sector. 
Using the concept of fusion closed set of anyons \cite{bais2003hopf} and injective and surjective morphisms of weak Hopf algebras, we obtain our main result in this part, Theorems~\ref{thm:SymBreak1} and \ref{thm:SymBreak2}.

Sec.~\ref{sec:Kitaev} is devoted to the discussion of the weak Hopf quantum double lattice model.
The construction has been briefly discussed by one of the authors in \cite{chang2014kitaev}.
Here we give a more detailed discussion, including the local stabilizer, the ribbon operators, and the topological excitations.

Sec.~\ref{sec:BdDomain} and \ref{sec:bdHam} generalize our algebraic theory and lattice construction of gapped boundary and domain wall \cite{jia2022boundary} of Hopf quantum double model to the weak Hopf quantum double lattice.
We show that most results remain similar to the Hopf case.

Sec.~\ref{sec:HopfTN} solves the weak Hopf quantum double model using the weak Hopf tensor network representation of states.

Sec.~\ref{sec:duality} discusses the duality of weak Hopf quantum double models.
The topological phases related by a duality can be regarded as the boundary theories of a one-dimension higher topological order.
For the Hopf algebra case, the duality is fully characterized by the twist deformation.

In Sec.~\ref{sec:conclusion}, we conclude and discuss future questions. In the appendix, we collect some detailed calculations about weak Hopf quantum double and ribbon operators.


\section{Preliminary}
\label{sec:pre}

To start with, let us briefly review some basic concepts and mathematical results concerning weak Hopf algebra (a.k.a. quantum groupoid) and meanwhile fix our notations.
Hereinafter, we assume all algebras are finite-dimensional unless otherwise specified.
The reader familiar with the notion of weak Hopf algebra might skip this part and get back if necessary.

\begin{definition}[Weak bialgebra \cite{bohm1996coassociative,BOHM1998weak,nill1998axioms}]
\label{def:weakbialgebra} 
	
	A complex weak bialgebra $(W,\mu,\eta,\Delta,\varepsilon)$ is  a $\mathbb{C}$-vector space $W$ equipped with the following structures:
	\begin{itemize}
		\item An algebra structure $(W,\mu,\eta)$, where $\mu: W\otimes W \to W$ and  $\eta: \mathbb{C} \to W$  are linear morphisms called multiplication and unit morphisms respectively.  Diagrammatically,
		\begin{equation}
			\mu=\begin{aligned}
			\begin{tikzpicture}
				 \draw[black, line width=1.0pt]  (-0.5, 0) .. controls (-0.4, 1) and (0.4, 1) .. (0.5, 0);
				 \draw[black, line width=1.0pt]  (0,0.75)--(0,1.15);
				\end{tikzpicture}
			\end{aligned},\quad 
		\eta= \begin{aligned}
			\begin{tikzpicture}
				\draw[black, line width=1.0pt]  (0,0.75)--(0,1.5);
				\draw [fill = white](0, 0.73) circle (2pt);
			\end{tikzpicture}
		\end{aligned}\,\,\, .
		\end{equation}
		They satisfy 
		\begin{equation}
			\mu \comp (\mu \otimes \id) =\mu \comp (\id \otimes \mu), \quad  \mu \comp (\eta \otimes \id)=\id =\mu \comp (\id \otimes \eta).
		\end{equation}

	\begin{equation}
	\begin{aligned}
			\begin{tikzpicture}
				\draw[black, line width=1.0pt]  (-0.5, 0) .. controls (-0.4, 0.8) and (0.4, 0.8) .. (0.5, -0.5);
				\draw[black, line width=1.0pt]  (-0.08,0.53)--(-0.08,0.93);
				\draw[black, line width=1.0pt]  (-0.8, -0.5) .. controls (-0.7, 0.19) and (-0.3, 0.19) .. (-0.2, -0.5);
			\end{tikzpicture}
		\end{aligned}
=	
	\begin{aligned}
	\begin{tikzpicture}
		\draw[black, line width=1.0pt]  (-0.5, -0.5) .. controls (-0.4, 0.8) and (0.4, 0.8) .. (0.5, 0);
		\draw[black, line width=1.0pt]  (0.08,0.53)--(0.08,0.93);
		\draw[black, line width=1.0pt]  (0.2, -0.5) .. controls (0.3, 0.19) and (0.7, 0.19) .. (0.8, -0.5);
	\end{tikzpicture}
\end{aligned},
	\quad 
			\id\,\,\,
	\begin{aligned}
		\begin{tikzpicture}
			\draw[black, line width=1.0pt] (0,-0.5)--(0,0.7);
		\end{tikzpicture}
	\end{aligned}
= \begin{aligned}
			\begin{tikzpicture}
				\draw[black, line width=1.0pt]   (0,0.3) arc (90:180:0.5);
				\draw [fill = white](-0.5, -0.2) circle (2pt);
				\draw[black, line width=1.0pt] (0,-0.5)--(0,0.7);
			\end{tikzpicture}
		\end{aligned}
	=
			 \begin{aligned}
		\begin{tikzpicture}
			\draw[black, line width=1.0pt]   (0.5,-0.2) arc (0:90:0.5);
			\draw [fill = white](0.5, -0.2) circle (2pt);
			\draw[black, line width=1.0pt] (0,-0.5)--(0,0.7);
		\end{tikzpicture}
	\end{aligned}\,\,\, .
	\end{equation}
	
		We can abbreviate $\mu(a,b)=a\cdot b$, 	and the unit element is denoted as $1_W=\eta(1)$.
		
		\item A coalgebra structure $(W,\Delta,\varepsilon)$, where $\Delta: W\to W\otimes W$ and $\varepsilon: W\to \mathbb{C}$ are  linear morphisms called comultiplication and counit morphisms respectively. Diagrammatically,
				\begin{equation}
			\Delta=\begin{aligned}
				\begin{tikzpicture}
					\draw[black, line width=1.0pt]  (-0.5, 1) .. controls (-0.4, 0) and (0.4, 0) .. (0.5, 1);
					\draw[black, line width=1.0pt]  (0,0.23)--(0,-0.23);
				\end{tikzpicture}
			\end{aligned},\quad 
			\varepsilon= \begin{aligned}
				\begin{tikzpicture}
					\draw[black, line width=1.0pt]  (0,0.75)--(0,1.5);
					\draw [fill = white](0, 1.5) circle (2pt);
				\end{tikzpicture}
			\end{aligned}\,\,\, .
		\end{equation}
		They satisfy
		\begin{equation}
			(\Delta \otimes \id)\comp \Delta = (\id \otimes \Delta )\comp \Delta,\quad  (\varepsilon \otimes \id) \comp \Delta = \id = (\id \otimes \varepsilon)\comp \Delta.
		\end{equation}
		
		\begin{equation}
	\begin{aligned}
			\begin{tikzpicture}
				\draw[black, line width=1.0pt]  (-0.5, 0.8) .. controls (-0.4, 0) and (0.4, 0) .. (0.5, 1.3);
				\draw[black, line width=1.0pt]  (-0.1,-0.16)--(-0.1,0.24);
				\draw[black, line width=1.0pt]  (-0.8, 1.3) .. controls (-0.7, 0.62) and (-0.3, 0.62) .. (-0.2, 1.3);
			\end{tikzpicture}
		\end{aligned}
	=
		\begin{aligned}
		\begin{tikzpicture}
			\draw[black, line width=1.0pt]  (-0.5, 1.3) .. controls (-0.4, 0) and (0.4, 0) .. (0.5, 0.8);
			\draw[black, line width=1.0pt]  (0.1,-0.16)--(0.1,0.24);
			\draw[black, line width=1.0pt]  (0.2, 1.3) .. controls (0.3, 0.62) and (0.7, 0.62) .. (0.8, 1.3);
		\end{tikzpicture}
	\end{aligned},
	\quad 
		\id\,\,\,
			\begin{aligned}
			\begin{tikzpicture}
				\draw[black, line width=1.0pt] (0,-0.5)--(0,0.7);
			\end{tikzpicture}
		\end{aligned}
		=\begin{aligned}
			\begin{tikzpicture}
				\draw[black, line width=1.0pt]   (-0.5,0.5) arc (180:270:0.5);
				\draw [fill = white](-0.5, 0.5) circle (2pt);
				\draw[black, line width=1.0pt] (0,-0.5)--(0,0.7);
			\end{tikzpicture}
		\end{aligned}
		=
		\begin{aligned}
			\begin{tikzpicture}
				\draw[black, line width=1.0pt]   (0,0) arc (270:360:0.5);
				\draw [fill = white](0.5, 0.5) circle (2pt);
				\draw[black, line width=1.0pt] (0,-0.5)--(0,0.7);
			\end{tikzpicture}
		\end{aligned}.
	\end{equation}

		We will adopt the Sweedler's notation $\Delta(u)=\sum_{(u)}u^{(1)}\otimes u^{(2)}:=\sum_{i}u_i^{(1)}\otimes u_i^{(2)}$.
		The comultiplication law ensures that $(\Delta\otimes \id)\comp \Delta(u) = (\id \otimes \Delta)\comp \Delta (u)=\sum_{(u)} u^{(1)}\otimes u^{(2)}\otimes u^{(3)}$. In general, we define $\Delta_1=\Delta$ and $\Delta_n=(\id \otimes \cdots \otimes \id \otimes \Delta)\comp \Delta_{n-1}$, then $\Delta_n(u)=\sum_{(u)} u^{(1)} \otimes \cdots \otimes u^{(n+1)}$.
	\end{itemize}
	To form a weak bialgebra, the above algebra and coalgebra structures must satisfy the following consistency conditions:
	\begin{enumerate}
		\item[(1)] The comultiplication preserves multiplication, i.e., $\Delta\comp\mu = \mu_{W\otimes W}\comp (\Delta\otimes \Delta)= (\mu\otimes \mu)\comp (\id\otimes \sigma_{W,W} \otimes \id)\comp(\Delta\otimes \Delta)$,  
			\begin{equation}
			\begin{aligned}
			\begin{tikzpicture}
				\draw[black, line width=1.0pt]   (-0.5,0)..   controls (-0.4,0.8) and (0.4,0.8)             ..(0.5,0);
				\draw[black, line width=1.0pt] (0,0.6)--(0,1);
					\draw[black, line width=1.0pt]   (-0.5,1.6)..   controls (-0.4,0.8) and (0.4,0.8)             ..(0.5,1.6);
			\end{tikzpicture}
		\end{aligned}
	=
		\begin{aligned}
		\begin{tikzpicture}
		\draw[black, line width=1.0pt] (-0.1,-1) arc (180:360:0.3);
		\draw[black, line width=1.0pt] (0.5,0) arc (0:180:0.3);
		 \draw[black, line width=1.0pt] (-0.1,0)--(-0.1,-1);
		 	\draw[black, line width=1.0pt] (1.6,0) arc (0:180:0.3);
		 	\draw[black, line width=1.0pt] (1,-1) arc (180:360:0.3);
		 \draw[black, line width=1.0pt] (1.6,0)--(1.6,-1);
		 	 \draw[black, line width=1.0pt] (0.2,-1.3)--(0.2,-1.6);
		 	  \draw[black, line width=1.0pt] (0.2,0.3)--(0.2,0.6);
		 	  \draw[black, line width=1.0pt] (1.3,0.3)--(1.3,0.6);
		 	  \draw[black, line width=1.0pt] (1.3,-1.3)--(1.3,-1.6);
			\braid[
			width=0.5cm,
			height=0.5cm,
			line width=1.0pt,
			style strands={1}{black},
			style strands={2}{black}] (Kevin)
			s_1^{-1} ;
		\end{tikzpicture}
	\end{aligned}\,\,\,. 
	\end{equation}
		Here $\sigma_{W,W}$ is the swap map $\sigma_{W,W}(u\otimes v)=v\otimes u$, diagrammatically denoted as $\sigma_{W,W}=\begin{aligned}\begin{tikzpicture}\braid[
			width=0.4cm,
			height=0.1cm,
			line width=0.3pt,
			style strands={1}{black},
			style strands={2}{black}] (Kevin)
			s_1^{-1} ;	\end{tikzpicture}\end{aligned}\,$,
		and we will also adopt the notations $\mu^{\rm op}=\mu \comp \sigma_{W,W}$ and $\Delta^{\rm op}=\sigma_{W,W}\comp \Delta$.
		Written in elements, we have
		\begin{equation}\label{}
			\Delta(uv)=\Delta(u)\cdot \Delta(v), \quad \forall\, u,v\in W.
		\end{equation}
		\item[(2)]Compatibility of comultiplication and unit  $\Delta_2 \comp \eta=(\id \otimes \mu \otimes \id )\comp  (\Delta \otimes \Delta)  \comp(\eta \otimes \eta) =  (\id \otimes \mu \otimes \id )\comp (\id \otimes \sigma_{W,W}\otimes \id)\comp (\Delta \otimes \Delta)\comp (\eta \otimes \eta)$, 
			\begin{equation}\label{eq:unitG}
				\begin{aligned}
			\begin{tikzpicture}
				\draw[black, line width=1.0pt] (-0.1,-1) arc (180:360:0.3);
				\draw[black, line width=1.0pt] (-0.1,-0.8)--(-0.1,-1);
				\draw[black, line width=1.0pt] (0.5,-1).. controls (0.6,-0.6)and (0.7,-0.6) ..(0.8,-0.5);
				\draw[black, line width=1.0pt] (-0.4,-0.5) arc (180:360:0.3);
				\draw[black, line width=1.0pt] (0.2,-1.3)--(0.2,-1.6);
				\draw [fill = white] (0.2,-1.6) circle (2pt);
			\end{tikzpicture}		\end{aligned}
=			
	\begin{aligned}	\begin{tikzpicture}
		\draw[black, line width=1.0pt] (-0.1,-0.8) arc (180:360:0.3);
		\draw[black, line width=1.0pt] (1,-0.8) arc (180:360:0.3);
		\draw[black, line width=1.0pt] (1,0) arc (0:180:0.25);
		\draw[black, line width=1.0pt] (-0.1,0.5)--(-0.1,-0.8);
        \draw[black, line width=1.0pt] (0.5,0)--(0.5,-0.8);
          \draw[black, line width=1.0pt] (1,0)--(1,-0.8);
		\draw[black, line width=1.0pt] (0.75,0.25)--(0.75,0.5);
		\draw[black, line width=1.0pt] (1.6,0.5)--(1.6,-0.8);
		\draw[black, line width=1.0pt] (1.3,-1.1)--(1.3,-1.35);
		\draw[black, line width=1.0pt] (0.2,-1.1)--(0.2,-1.35);
		\draw [fill = white](0.2,-1.35) circle (2pt);
		\draw [fill = white] (1.3,-1.35) circle (2pt);
	\end{tikzpicture}
\end{aligned}
=
			\begin{aligned}	\begin{tikzpicture}
				\draw[black, line width=1.0pt] (-0.1,-0.8) arc (180:360:0.3);
					\draw[black, line width=1.0pt] (1,-0.8) arc (180:360:0.3);
			\draw[black, line width=1.0pt] (1,0) arc (0:180:0.25);
				\draw[black, line width=1.0pt] (-0.1,0.5)--(-0.1,-0.8);
					\braid[
				width=0.5cm,
				height=0.3cm,
				line width=1.0pt,
				style strands={1}{black},
				style strands={2}{black}] (Kevin)
				s_1^{-1} ;
				\draw[black, line width=1.0pt] (0.75,0.25)--(0.75,0.5);
				\draw[black, line width=1.0pt] (1.6,0.5)--(1.6,-0.8);
	    	\draw[black, line width=1.0pt] (1.3,-1.1)--(1.3,-1.35);
	    	\draw[black, line width=1.0pt] (0.2,-1.1)--(0.2,-1.35);
	    		\draw [fill = white](0.2,-1.35) circle (2pt);
	    		\draw [fill = white] (1.3,-1.35) circle (2pt);
			\end{tikzpicture}
		\end{aligned}\,\,\,,
	\end{equation}
		\begin{equation}\label{}
			\begin{aligned}
			(\Delta\otimes \id)\comp \Delta(1_W)=&(\Delta(1_W)\otimes 1_W)(1_W\otimes \Delta(1_W))\\
			=&(1_W\otimes \Delta(1_W))(\Delta(1_W)\otimes 1_W).
	\end{aligned}
	\end{equation}
		\item[(3)]Compatibility of multiplication and counit $\varepsilon \comp \mu \comp (\id \otimes \mu) = (\varepsilon \otimes \varepsilon)\comp (\mu \otimes \mu)\comp (\id \otimes  \Delta \otimes\id)=(\varepsilon \otimes \varepsilon)\comp (\mu \otimes \mu)\comp (\id \otimes  \Delta^{op} \otimes\id)$, 
					\begin{equation}
			\begin{aligned}
				\begin{tikzpicture}
					\draw[black, line width=1.0pt] (0.5,-0.8) arc (0:180:0.3);
					\draw[black, line width=1.0pt] (-0.1,-0.8)--(-0.1,-1);
					\draw[black, line width=1.0pt] (0.2,-0.5)--(0.2,-0.25);
				    \draw[black, line width=1.0pt]  (0.5,-0.8).. controls (0.6,-1.18)and (0.7,-1.25) ..(0.8,-1.3);
					\draw[black, line width=1.0pt] (0.2,-1.3) arc (0:180:0.3);
							\draw [fill = white] (0.2,-0.25) circle (2pt);
			\end{tikzpicture}		\end{aligned}
			=			
\begin{aligned}	\begin{tikzpicture}
		\draw[black, line width=1.0pt] (0.5,0) arc (0:180:0.3);
		\draw[black, line width=1.0pt] (1.6,0) arc (0:180:0.3);
		\draw[black, line width=1.0pt] (0.5,-0.8) arc (180:360:0.25);
		\draw[black, line width=1.0pt] (0.5,0)--(0.5,-0.8);
			\draw[black, line width=1.0pt] (1,0)--(1,-0.8);
		\draw[black, line width=1.0pt] (-0.1,0)--(-0.1,-1.3);
		\draw[black, line width=1.0pt] (1.6,0)--(1.6,-1.3);
		\draw[black, line width=1.0pt] (0.75,-1.05)--(0.75,-1.3);
		\draw[black, line width=1.0pt] (1.3,0.3)--(1.3,0.55);
		\draw[black, line width=1.0pt] (0.2,0.3)--(0.2,0.55);
		\draw [fill = white]  (0.2,0.55) circle (2pt);
		\draw [fill = white] (1.3,0.55) circle (2pt);
	\end{tikzpicture}
\end{aligned}
			=
			\begin{aligned}	\begin{tikzpicture}
					\draw[black, line width=1.0pt] (0.5,0) arc (0:180:0.3);
               	\draw[black, line width=1.0pt] (1.6,0) arc (0:180:0.3);
               	 	\draw[black, line width=1.0pt] (0.5,-0.8) arc (180:360:0.25);
					\braid[
					width=0.5cm,
					height=0.3cm,
					line width=1.0pt,
					style strands={1}{black},
					style strands={2}{black}] (Kevin)
					s_1^{-1} ;
					\draw[black, line width=1.0pt] (-0.1,0)--(-0.1,-1.3);
					\draw[black, line width=1.0pt] (1.6,0)--(1.6,-1.3);
					\draw[black, line width=1.0pt] (0.75,-1.05)--(0.75,-1.3);
					\draw[black, line width=1.0pt] (1.3,0.3)--(1.3,0.55);
					\draw[black, line width=1.0pt] (0.2,0.3)--(0.2,0.55);
					\draw [fill = white]  (0.2,0.55) circle (2pt);
					\draw [fill = white] (1.3,0.55) circle (2pt);
				\end{tikzpicture}
			\end{aligned}\,\,\,,
		\end{equation}
		\begin{equation}\label{eq:counit}
\varepsilon(uvw)=\sum_{(v)}\varepsilon(uv^{(1)})\varepsilon(v^{(2)}w)=\sum_{(v)}\varepsilon(uv^{(2)})\varepsilon(v^{(1)}w), \quad \forall\, u,v,w\in W.
		\end{equation}
	\end{enumerate}
\end{definition}

\begin{definition}[Weak Hopf algebra \cite{bohm1996coassociative,BOHM1998weak,nill1998axioms}]
\label{def:antipode}
	A complex weak Hopf algebra $(W,\mu,\eta,\Delta,\varepsilon,S)$ is a weak bialgebra  $(W,\mu,\eta,\Delta,\varepsilon)$ equipped with a linear morphism $S: W\to W$ called antipode which satisfies the following three conditions:
	\begin{enumerate}
		\item[(1)] $\mu \comp (\mathrm{id}_W \otimes S) \comp \Delta  =  (\varepsilon \otimes \mathrm{id}_W) \comp (\mu \otimes \mathrm{id}_W) \comp (\mathrm{id}_W \otimes \sigma_{W, W}) \comp (\Delta \otimes \mathrm{id}_W) \comp (\eta \otimes \mathrm{id}_W)$, 
					\begin{equation}
			\begin{aligned}	
				\begin{tikzpicture}
					\draw[black, line width=1.0pt] (0.5,0) arc (0:180:0.3);
					\draw[black, line width=1.0pt] (-0.1,-0.8) arc (180:360:0.3);
					\draw[black, line width=1.0pt] (-0.1,0)--(-0.1,-0.8);
					\draw[black, line width=1.0pt] (0.5,-0.8)--(0.5,-0.65);
					\draw[black, line width=1.0pt] (0.5,0)--(0.5,-0.3);
					\draw[black, line width=1.0pt] (0.2,-1.1)--(0.2,-1.35);
					\draw[black, line width=1.0pt] (0.2,0.3)--(0.2,0.55);
					\draw (0.3,-0.3) rectangle (0.7,-0.65);
					\node (start) [at=(Kevin-1-s),yshift=-0.47cm] {$S$};
				\end{tikzpicture}
			\end{aligned}
			=
			\begin{aligned}	
				\begin{tikzpicture}
					\draw[black, line width=1.0pt] (0.5,0) arc (0:180:0.3);
					\draw[black, line width=1.0pt] (-0.1,-0.8) arc (180:360:0.3);
					\braid[
					width=0.5cm,
					height=0.3cm,
					line width=1.0pt,
					style strands={1}{black},
					style strands={2}{black}] (Kevin)
					s_1^{-1} ;
					\draw[black, line width=1.0pt] (-0.1,0)--(-0.1,-0.8);
					\draw[black, line width=1.0pt] (1,-0.8)--(1,-1.35);
					\draw[black, line width=1.0pt] (1,0)--(1,0.55);
					\draw[black, line width=1.0pt] (0.2,-1.1)--(0.2,-1.35);
					\draw[black, line width=1.0pt] (0.2,0.3)--(0.2,0.55);
					\draw [fill = white]  (0.2,0.55) circle (2pt);
					\draw [fill = white] (0.2,-1.35) circle (2pt);
				\end{tikzpicture}
			\end{aligned}\,\,\,.
		\end{equation}
        The map on the right hand side will be denoted as $\varepsilon_L: W\to W$,  $\varepsilon_L(h)=(\varepsilon\otimes\id)(\Delta(1_W)(h\otimes 1_W))=\sum_{(1_W)} \varepsilon(1_W^{(1)}h)1^{(2)}_W$ and we denote $W_L=\varepsilon_L(W)$.
		Equivalently, we have
	\begin{equation}\label{eq:SL}
		\mu\comp(\id\otimes S)\comp\Delta(h)=\varepsilon_L(h),\quad  \forall\, h\in W.
	\end{equation}

		\item[(2)] $\mu \comp (S \otimes \mathrm{id}_W) \comp \Delta = (\mathrm{id}_W \otimes \varepsilon) \comp (\mathrm{id}_W \otimes \mu) \comp (\sigma_{W, W} \otimes \mathrm{id}_W) \comp (\mathrm{id}_W \otimes \Delta) \comp (\mathrm{id}_W \otimes \eta)$, 	
							\begin{equation}
			\begin{aligned}	
				\begin{tikzpicture}
					\draw[black, line width=1.0pt] (0.5,0) arc (0:180:0.3);
					\draw[black, line width=1.0pt] (-0.1,-0.8) arc (180:360:0.3);
					\draw[black, line width=1.0pt] (-0.1,-0.65)--(-0.1,-0.8);
					\draw[black, line width=1.0pt] (-0.1,0)--(-0.1,-0.3);
					\draw[black, line width=1.0pt] (0.5,-0.8)--(0.5,0);
					\draw[black, line width=1.0pt] (0.2,-1.1)--(0.2,-1.35);
					\draw[black, line width=1.0pt] (0.2,0.3)--(0.2,0.55);
					\draw (-0.3,-0.3) rectangle (0.1,-0.65);
					\node (start) [at=(Kevin-1-s),xshift=-0.58cm,yshift=-0.47cm] {$S$};
				\end{tikzpicture}
			\end{aligned}
			=
			\begin{aligned}	
				\begin{tikzpicture}
					\draw[black, line width=1.0pt] (1.6,0) arc (0:180:0.3);
					\draw[black, line width=1.0pt] (1.3,0.3)--(1.3,0.55);
					\draw[black, line width=1.0pt] (1.0,-0.8) arc (180:360:0.3);
					\draw[black, line width=1.0pt] (1.6,0)--(1.6,-0.8);
					\draw[black, line width=1.0pt] (1.3,-1.1)--(1.3,-1.35);
					\braid[
					width=0.5cm,
					height=0.3cm,
					line width=1.0pt,
					style strands={1}{black},
					style strands={2}{black}] (Kevin)
					s_1^{-1} ;
					\draw[black, line width=1.0pt] (0.5,0)--(0.5,0.55);
					\draw[black, line width=1.0pt] (0.5,-0.6)--(0.5,-1.35);
					\draw [fill = white] (1.3,0.55) circle (2pt);
					\draw [fill = white] (1.3,-1.35) circle (2pt);
				\end{tikzpicture}
			\end{aligned}\,\,\,.
		\end{equation}
			The map on the right hand side will be denoted as $\varepsilon_R: W\to W$,  $\varepsilon_R(h)=(\id\otimes \varepsilon)((1_W\otimes h)\Delta(1_W))=\sum_{(1_W)} 1_W^{(1)}\varepsilon(h1_W^{(2)})$ and we denote $W_R=\varepsilon_R(W)$. 
		Equivalently, we have
	\begin{equation}\label{eq:SR}
		\mu\comp(S\otimes \id)\comp\Delta(h)=\varepsilon_R(h),\quad \forall\, h\in W.
	\end{equation}

	   \item[(3)] $S  = \mu \comp (\mu \otimes \mathrm{id}_W) \comp (S \otimes \mathrm{id}_W \otimes S) \comp (\Delta \otimes \mathrm{id}_W) \comp \Delta$, 
	   \begin{equation}
	   		\begin{aligned}	
	   		\begin{tikzpicture}
	   			\draw[black, line width=1.0pt] (-0.1,-0.65)--(-0.1,-1.05);
	   			\draw[black, line width=1.0pt] (-0.1,0.1)--(-0.1,-0.3);
	   			\draw (-0.3,-0.3) rectangle (0.1,-0.65);
	   			\node (start) [at=(Kevin-1-s),xshift=-0.58cm,yshift=-0.47cm] {$S$};
	   		\end{tikzpicture}
	   	\end{aligned}
   	=
	   	\begin{aligned}	
	   	\begin{tikzpicture}
	   		\draw[black, line width=1.0pt] (0.5,0) arc (0:180:0.3);
	   		\draw[black, line width=1.0pt] (-0.1,-0.4) arc (180:360:0.3);
	   		\draw[black, line width=1.0pt] (0.5,-0.4)--(0.5,0);
	   		\draw[black, line width=1.0pt] (0.2,-0.7) arc (180:360:0.3);
	   		\draw[black, line width=1.0pt] (0.8,0.3) arc (0:180:0.3);
	   		\draw[black, line width=1.0pt] (0.8,0.3)--(0.8,0);
	   		\draw[black, line width=1.0pt] (0.8,-0.7)--(0.8,-0.4);
	   		\draw[black, line width=1.0pt] (0.5,0.6)--(0.5,0.85);
	   		\draw[black, line width=1.0pt] (0.5,-1.25)--(0.5,-1);
	   		\draw (-0.3,0) rectangle (0.1,-0.4);
	   		\node (start) [at=(Kevin-1-s),xshift=-0.58cm,yshift=-0.2cm] {$S$};
	   		\draw (0.6,-0.4) rectangle (1,0);
	   		\node (start) [at=(Kevin-1-s),xshift=0.35cm,yshift=-0.2cm] {$S$};
	   	\end{tikzpicture}
	   \end{aligned}\,\,\,.
	   \end{equation}
	   Equivalently, we have
		\begin{equation}\label{}
			S(h)=\sum_{(h)} S(h^{(1)})h^{(2)} S(h^{(3)}),\quad \forall\, h \in W.
		\end{equation}
	\end{enumerate}
\end{definition}

A Hopf algebra is a particular case of weak Hopf algebra.
If $\Delta(1_W)=1_W\otimes 1_W$, then $\Delta$ is an algebra homomorphism. From Eq.~\eqref{eq:counit} by taking $v=1_W$, we see that $\varepsilon$ is also an algebra homomorphism. Eqs.~\eqref{eq:SL} and \eqref{eq:SR} imply the consistency condition of the antipode for Hopf algebra. In this case, $W$ becomes a Hopf algebra.
Also, if $\varepsilon$ is a homomorphism of algebras, acting on the middle top of Eq.~(\ref{eq:unitG})
with $\varepsilon$, from graphical calculus, it is easily checked that $\Delta(1_W)=1_W\otimes 1_W$, so $W$ becomes a Hopf algebra. If $S$ is a Hopf antipode, then $W$ also becomes a Hopf algebra \cite{nill1998axioms,BOHM1998weak}.

Notice that for a weak Hopf algebra, the antipode is anti-multiplicative and anti-comultiplicative
\begin{equation}
	S(xy)=S(y)S(x), \quad \sum_{(S(x))}S(x)^{(1)}\otimes S(x)^{(2)}=\sum_{(x)}S(x^{(2)})\otimes S(x^{(1)}).
\end{equation}
The antipode preserves the unit and the counit
\begin{equation}
	S(1_W)=1_W,\quad \varepsilon\comp S=\varepsilon.
\end{equation}
And the antipode $S$ is invertible. See \cite[Theorem 2.10]{BOHM1998weak} for proofs of the above properties.
Note that in contrast to finite-dimensional semisimple Hopf algebras where $S^2=\id$, the antipode in a finite-dimensional weak Hopf algebra is not necessarily involutive \cite{nikshych2003invariants}.
If $W$ is a weak Hopf algebra such that $S^2=\id$, it is called a weak Kac algebra. 

The $W_L=\varepsilon_L(W),W_R=\varepsilon_R(W)$ in Definition~\ref{def:antipode} are two separable subalgebras of $W$\,\footnote{When $W$ is a Hopf algebra, $\varepsilon_L=\varepsilon_R: x\mapsto \varepsilon(x)1_W$.}. They are called left and right counital subalgebras and play crucial roles in studying weak Hopf symmetry.
We have $\varepsilon_L\comp S=\varepsilon_L\comp\varepsilon_R=S\comp \varepsilon_R$ and $\varepsilon_R\comp S=\varepsilon_R\comp\varepsilon_L=S\comp \varepsilon_L$. It holds that $\Delta(1_W)\in W_R\otimes W_L$.

A complex weak Hopf algebra $W$ is called simple (or indecomposable) if its underlying algebra $(W,\mu,\eta)$ has no nontrivial subalgebras, and is called semisimple if its underlying algebra can be written as a direct sum of simple algebras. 
A $*$-weak Hopf algebra $(W,*)$ is a weak Hopf algebra $W$ equipped with a $C^*$-structure  $*:W\to W$ such that $\Delta$ is a $*$-homomorphism. That is
\begin{equation}
	(x^*)^*=x,\; (x+y)^*=x^*+y^*,\; (xy)^*=y^*x^*,\; (c x)^*=\bar{c} x^*,\;
	\Delta(x)^*=\Delta (x^*),
\end{equation}
for all $x,y\in W$ and $c\in \mathbb{C}$. We also have $S(x^*)=S^{-1}(x)^*$. 
$(W,*)$ is called a $C^*$-weak Hopf algebra if there exists a fully faithful $*$-representation $\rho: W\to \mathbf{B}(\mathcal{H})$ for some operator space over some Hilbert space $\mathcal{H}$.

Let $W$ be a weak Hopf algebra. A left (resp. right) integral of $W$ is an element $l$ (resp. $r$) satisfying $xl=\varepsilon_L(x) l$ (resp. $rx=r\varepsilon_R(x)$). 
A left (resp. right) integral $l$ (resp. $r$) is called left (resp. right) normalized if $\varepsilon_L(l)=1_W$ (resp. $\varepsilon_R(r)=1_W$).
If $h$ is both a left and right integral, it is called a two-side integral.
A Haar integral in $W$ (or Haar measure on $\hat{W}$) is a two-side normalized two-side integral.

Notice that Haar integral in a weak Hopf algebra, if exists,  must be unique. To see this, suppose that $h,h'$ are two Haar integrals, then $h'=\varepsilon_{L} (h)h' = hh'=h\varepsilon_R(h')=h$. Since $S(h)$ is a Haar integral if $h$ is, then from uniqueness we see that Haar integral is $S$-invariant, $S(h)=h$. It is also clear that $h^2=\varepsilon_L(h)h=h$.
A  $C^*$ weak Hopf algebra always has a unique Haar integral $h$ satisfying $h^*=h$ \cite{BOHM1998weak}. 
An element $h\in W$ is called cocommutative if $\Delta^{\rm op}(h)=\Delta(h)$; the set of all cocommutative elements in $W$ is denoted as $\operatorname{Cocom}(W)$. We will only consider weak Hopf algebras which are $C^*$ and whose Haar integral is cocommutative.

A crucial tool that we will use is the pairing between two weak Hopf algebras.
     A pairing $\lambda=\langle \bullet, \bullet \rangle: J\otimes K\to \mathbb{C}$ between two weak Hopf algebras $J,K$ is a bilinear map  satisfying
     \begin{align}
       &  \langle hg,a\rangle = \sum_{(a)} \langle h,a^{(1)}\rangle \langle g,a^{(2)}\rangle,\\
    &     \langle h,ab\rangle =\sum_{(h)} \langle h^{(1)} ,a\rangle \langle h^{(2)},b\rangle,\\
     &    \langle 1_J,a\rangle =\varepsilon_K(a),\quad
         \langle h,1_K\rangle =\varepsilon_J(h).
     \end{align}
     A pairing between $J^{\rm cop}$ (the weak Hopf algebra with coproduct $\Delta^{\rm op}$, see below) and $K$ is called a skew-pairing between $J$ and $K$.

For a $C^*$ weak Hopf algebra $W$, its dual space $\hat{W}:=\Hom (W,\mathbb{C})=W^{\vee}$ has a canonical $C^*$ weak Hopf algebra structure induced by the canonical pairing $\langle\bullet, \bullet \rangle: W^{\vee}\times W\to \mathbb{C}$, $\langle \varphi, h\rangle:=\varphi(h)$. More precisely, $\hat{\mu}=\Delta^{\vee}$, $\hat{\eta}=\varepsilon^{\vee}$, $\hat{\Delta}=\mu^{\vee}$, $\hat{\varepsilon}=\eta^{\vee}$ and $\hat{S}=S^{\vee}$, \emph{viz.},
\begin{align}
&	\langle \hat{\mu}(\varphi\otimes \psi),x\rangle=\langle\varphi\otimes \psi, \Delta(x)\rangle,\\
&	\langle \hat{\eta} (1),x\rangle= \varepsilon(x),   \; \text{i.e.,}\; \hat{1}=\varepsilon,\\
&   \langle \hat{\Delta}(\varphi), x\otimes y \rangle=\langle \varphi, \mu(x\otimes y)\rangle,\\
&    \hat{\varepsilon}(\varphi)=\langle \varphi, \eta(1)\rangle,\\
&    \langle \hat{S}(\varphi),x\rangle =\langle \varphi, S(x)\rangle.
\end{align}
The star operation on $\hat{W}$ is defined as 
\begin{equation}
	\langle \varphi^*, x\rangle=\overline{ \langle \varphi, S(x)^*\rangle}. 
\end{equation}
The opposite weak Hopf algebra $W^{\rm op}$ is defined as $(W, \mu^{\rm op},\eta,\Delta,\varepsilon, S^{-1})$, and the co-opposite weak Hopf algebra $W^{\rm cop}$ is defined as $(W, \mu,\eta, \Delta^{\rm op},\varepsilon,S^{-1})$. It is easily checked that $(W^{\rm op})^{\vee}\simeq (W^{\vee})^{\rm cop}$ and  $(W^{\rm cop})^{\vee}\simeq (W^{\vee})^{\rm op}$ as weak Hopf algebras.

When considering the action of $W$ on $\hat{W}$, the Sweedler's arrow notation will be useful,
\begin{equation} \label{eq:poon}
	x\rightharpoonup \varphi :=\sum_{(\varphi)} \varphi^{(1)} \langle \varphi^{(2)}, x\rangle, \quad  \varphi \leftharpoonup x :=\sum_{(\varphi)}\langle \varphi^{(1)},x\rangle \varphi^{(2)}.
\end{equation}
Since $W$ is the dual weak Hopf algebra of $\hat{W}$, we also have
\begin{equation}
\varphi \rightharpoonup x:=\sum_{(x)}x^{(1)} \langle \varphi, x^{(2)}\rangle, \quad  x \leftharpoonup \varphi:=\sum_{(x)}   \langle \varphi,    x^{(1)}    \rangle  x^{(2)}.
\end{equation}
The left (resp. right) action of $W$ on $W$ will be denoted as $h\triangleright x:=hx$ (resp. $x\triangleleft h = xh$), and similarly for $\hat{W}$.

The following formulae are useful for computation. 
    For $x_L=\varepsilon_L(x)\in W_L,~y_R=\varepsilon_R(y)\in W_R$ and $\varphi\in\hat{W}$, the following identities hold: 
    \begin{align}
        & \varphi \leftharpoonup y_R = \varphi(\varepsilon \leftharpoonup y_R), \quad y_R \rightharpoonup \varphi = \varphi(y_R\rightharpoonup \varepsilon),\\
        &\varepsilon \leftharpoonup S(x_L) = \varepsilon \leftharpoonup x_L,  \quad   S(y_R)\rightharpoonup \varepsilon = y_R\rightharpoonup \varepsilon, \label{eq:S-and-original} \\
       \varepsilon & \leftharpoonup S^{-1}(y_R) = \varepsilon \leftharpoonup y_R, \quad S^{-1}(x_L) \rightharpoonup \varepsilon = x_L \rightharpoonup \varepsilon.  \label{eq:S-inverse-and-original}
    \end{align}
    The first line is from \cite[Scholium 2.7]{BOHM1998weak}. To prove the second line, for any $y\in W$, 
    \begin{align*}
        \langle \varepsilon \leftharpoonup S(x_L),y \rangle & = \varepsilon(S(x_L)y) = \varepsilon(S(\varepsilon_L(x))y) \\
        & = \varepsilon(\varepsilon_R(S(x))y) = \varepsilon(\varepsilon_R( \varepsilon_L(x))y) \\
        &= \varepsilon(x_Ly) = \langle  \varepsilon \leftharpoonup x_L, y\rangle.
    \end{align*}
    Since this is true for any $y\in W$, it follows that $\varepsilon \leftharpoonup S(x_L) = \varepsilon \leftharpoonup x_L$. The other identity in the second line is proved similar. The third lines follow from the second lines by taking $x'_L = S^{-1}(y_R)$ and $y'_R = S^{-1}(x_L)$. 

The Haar integral of $\hat{W}$ is called the Haar measure of $W$. It is proved in \cite{BOHM1998weak} that the Haar measure induces an inner product structure over $W$ by
\begin{equation}
    \langle x, y\rangle =\varphi_{\What}(x^*y).
\end{equation}
Hereinafter, we will denote the Haar integrals of $W$ and $\What$ as $h_W$ and $\varphi_{\What}$ respectively.

In this work, we will mainly consider the lattice model whose weak Hopf symmetry is given by the quantum double of some weak Hopf algebra $W$. We shall call such a kind of lattice model a weak Hopf quantum double model.
By definition, the quantum double of $W$ is a quotient algebra of $\hat{W}^{\rm cop} \otimes W$ equipped with a weak Hopf algebra structure; we denote it as $D(W)$, and the elements are equivalence classes $[\varphi\otimes h]$.\footnote{There are several different constructions of quantum double \cite{majid2000foundations}. To keep the construction consistent with that for the Hopf algebra case  discussed in our previous work \cite{jia2022boundary}, we choose the construction here. The generalization of results in this work to all other constructions is straightforward.}
Both $\hat{W}^{\rm cop}$ and $W$ can be embedded into $D(W)$ as weak Hopf subalgebras. 

For a weak Hopf algebra $W$, $\hat{W}^{\rm cop} \otimes W$ is an algebra with the multiplication \cite{drinfel1988quantum,majid1990physics,majid1994some}
\begin{equation}
	(\varphi \otimes h) (\psi\otimes g)= \sum_{(h)} \sum_{(\psi)}  \varphi \psi^{(2)} \otimes  h^{(2)} g  \langle \psi^{(1)}, S^{-1}(h^{(3)}) \rangle \langle \psi^{(3)}, h^{(1)}\rangle,
\end{equation}
and the unit $\varepsilon\otimes 1_W$. Here the comultiplication of $\psi$ is taken in $\hat{W}$. The linear span $J$ of the elements 
\begin{align}
	\varphi \otimes xh -  \varphi (x\rightharpoonup \varepsilon)\otimes h, \quad x\in W_L,\\
	\varphi \otimes yh -  \varphi ( \varepsilon \leftharpoonup y) \otimes h, \quad y\in W_R,
\end{align}
is a two-sided ideal of $\hat{W}^{\rm cop} \otimes W$ (see Appendix \ref{app:QD}). We denote the quotient algebra $(\hat{W}^{\rm cop} \otimes W)/J$ as $D(W)$ and equivalent classes in $D(W)$ as $[\varphi \otimes h]$ for $\varphi\otimes h\in \hat{W}^{\rm cop}\otimes W$.

\begin{definition}[Quantum double] \label{def:quantum-double}
The quantum double of $W$ is $D(W)$ equipped with the following weak Hopf algebra structure:
\begin{enumerate}
	\item[(1)] The multiplication $[\varphi \otimes h] [\psi \otimes g]=\sum_{(\psi),(h)} [\varphi \psi^{(2)} \otimes  h^{(2)} g]  \langle \psi^{(1)}, S^{-1}(h^{(3)}) \rangle \langle \psi^{(3)}, h^{(1)}\rangle$.
	\item[(2)] The unit $[\varepsilon \otimes 1_W]$.
	\item[(3)] The comultiplication $\Delta([\varphi \otimes  h])= \sum_{(\varphi), (h)} [\varphi^{(2)} \otimes h^{(1)}] \otimes [\varphi^{(1)} \otimes h^{(2)}] $.
	\item[(4)] The counit $\varepsilon ([\varphi \otimes h])= \langle \varphi, \varepsilon_R(S^{-1}(h))\rangle$.
	\item[(5)] The antipode $S([\varphi \otimes h])
	= \sum_{(\varphi), (h)} [\hat{S}^{-1}(\varphi^{(2)}) \otimes S(h^{(2)})]
	\langle  \varphi^{(1)},h^{(3)}\rangle
	\langle  \varphi^{(3)}, S^{-1}(h^{(1)})\rangle$.
\end{enumerate}
It can be verified $D(W)$ is indeed a weak Hopf algebra (see Appendix \ref{app:QD}); see also \cite{nikshych2003invariants} for other construction. We will also use the notation $D(W)=\What \Join W$, and  ``$\Join$'' is usually called the bicrossed product.
\end{definition}

Both $\hat{W}^{\rm cop}$ and $W$ are weak Hopf subalgebras of $D(W)$ with inclusion maps
\begin{align}
&	i_{\hat{W}^{\rm cop}} :\hat{W}^{\rm cop} \hookrightarrow D(W), \quad \varphi \mapsto [\varphi\otimes 1_W],\\
&	i_{W}: W\hookrightarrow D(W), \quad h \mapsto [\varepsilon \otimes h].
\end{align}
This is the origin of the name ``quantum double''. Also, these inclusions give us the identifications $W\simeq (\varepsilon\otimes W)/J$ and $\hat{W}^{\rm cop}\simeq (\hat{W}^{\rm cop}\otimes 1_W)/J$. The multiplication in $D(W)$ implies that 
\begin{align}
    i_{\hat{W}^{\rm cop}}(\varphi)i_{W}(h) & = \sum_{(1_W),(\varepsilon)}[\varphi\varepsilon^{(2)}\otimes 1_W^{(2)}h]\langle \varepsilon^{(1)},S^{-1}(1_W^{(3)})\rangle \langle \varepsilon^{(3)},1_W^{(1)} \rangle \nonumber \\
    & = \sum_{(1_W),(\varepsilon)}\sum_{(1_W')}[\varphi\varepsilon^{(2)}\otimes \langle \varepsilon^{(3)},1_W^{(1)} \rangle 1_W^{(2)}1_W'^{(1)}\langle \varepsilon^{(1)},S^{-1}(1_W'^{(2)})\rangle h] \nonumber \\
    & = \sum_{(1_W),(\varepsilon)}\sum_{(1_W')}[\varphi\varepsilon^{(2)}(1_W^{(2)}\rightharpoonup \varepsilon)\langle \varepsilon^{(3)},1_W^{(1)} \rangle (\varepsilon\leftharpoonup 1_W'^{(1)})\langle \varepsilon^{(1)},S^{-1}(1_W'^{(2)})\rangle\otimes h] \nonumber \\
    & = \sum_{(\varepsilon)}[\varphi\varepsilon^{(2)}\varepsilon_R(\varepsilon^{(3)})S^{-1}(\varepsilon_R(\varepsilon^{(1)}))\otimes h] = [\varphi\otimes h], \label{eq:prod-with-unity}
\end{align}
for all $\varphi\in \hat{W}^{\rm cop}$ and $h\in W$.\footnote{In general, $[\varphi\otimes 1_W][\psi\otimes h] = [\varphi\psi\otimes h]$ for any $\varphi,\psi\in\hat{W}$ and $h\in W$. This can be proved exactly as Eq.~\eqref{eq:prod-with-unity}.} Write $\varphi h$ instead of $i_{\hat{W}^{\rm cop}}(\varphi)i_{W}(h)$ for simplicity. Then the multiplication in $D(W)$ is determined by the following ``straightening formula''
\begin{equation}
    h\varphi = \sum_{(h)}\varphi(S^{-1}(h^{(3)})\bullet h^{(1)})h^{(2)},\quad \forall\, \varphi\in\hat{W}^{\rm cop},\, \forall\, h\in W,
\end{equation}
where ``$\bullet$'' is the argument of the function. For cocommutative $h\in \operatorname{Cocom}(W)$ and $\varphi \in \operatorname{Cocom}(\hat{W}^{\rm cop})$, we have
\begin{equation}
	[\varphi \otimes 1] [\varepsilon \otimes h] = [\varphi \otimes h] = [\varepsilon \otimes h]  [\varphi \otimes 1].
\end{equation}

\begin{definition}[Yetter-Drinfeld module]
    A left-right Yetter-Drinfeld module over $W$ is a complex vector space $M$ such that:
    \begin{itemize}
        \item[(1)] $M$ is a left $W$-module with the action $W\otimes M\to M$, $h\otimes m\mapsto h\triangleright m$;
        \item[(2)] $M$ is a right $W$-comodule with the coaction $\rho:M\to M\otimes W$, $\rho(m) = \sum_{(m)}m^{[0]}\otimes m^{[1]}$;
        \item[(3)] The action and coaction are compatible: for all $h\in W$ and $m\in M$, 
        \begin{align*}
            \sum_{(h),(m)}(h^{(1)}\triangleright m^{[0]})\otimes h^{(2)}m^{[1]} &= \sum_{(h)}(h^{(2)}\triangleright m)^{[0]}\otimes (h^{(2)}\triangleright m)^{[1]}h^{(1)}, \\
            \sum_{(1_W),(m)}(1_W^{(1)}\triangleright m^{[0]})\otimes & 1_W^{(2)}m^{[1]} =\sum_{(m)}m^{[0]}\otimes m^{[1]}. 
        \end{align*}
    \end{itemize}
\end{definition}

Let ${_W}\mathsf{YD}^W$ be the category of left-right Yetter-Drinfeld modules over $W$ whose morphisms are $W$-linear and $W$-colinear. It is proved in \cite{nenciu2002center} that ${_W}\mathsf{YD}^W$ can be identified with $\mathsf{Rep}(D(W))$. In fact, for $V\in\mathsf{Rep}(D(W))$, it is naturally a left $W$-module and a left $\hat{W}$-module. Define the coaction $\rho:V\to V\otimes W$ by $\rho(v) = \sum_k (\hat{k}\triangleright v)\otimes k$ where $\{k\}$ and $\{\hat{k}\}$ are dual bases in $W$ and $\hat{W}$ respectively. One can verify that $V$ is a left-right Yetter-Drinfeld module. Conversely, for $M\in {_W}\mathsf{YD}^W$, the action $[\varphi\otimes h]\triangleright m = \sum_{(m)}\varphi(m^{[1]})h\triangleright m^{[0]}$ makes $M$ a left $D(W)$-module. By \cite[Theorem 4.5]{nenciu2002center}, $\mathcal{Z}(\mathsf{Rep}(W))\simeq \mathsf{Rep}(D(W))$ as braided tensor categories. Its proof makes use of the aforementioned identification through ${_W}\mathsf{YD}^W$.


\section{Weak Hopf symmetry}
\label{sec:WHAsymmetry}  
\subsection{Weak Hopf symmetry of vacuum sector}

Consider a quantum system $(H,\mathcal{V}_{\rm GS})$ with Hamiltonian $H$ and ground state space $\mathcal{V}_{\rm GS}$, the Hamiltonian symmetry $G_{H}$ is a group such that $U_gH=HU_g$ for all $g\in G_{H}$, and the ground state symmetry $G_{\rm GS}$ is a group such that $U_g$ stabilizes $\mathcal{V}_{GS}$ for all $g\in G_{\rm GS}$.
It is clear that $G_{H}\subseteq G_{\rm GS}$.
When dealing with quantum group symmetry, the requirement that the vacuum sector is invariant under the quantum group action is too strong, so we will need to modify the definition of symmetry \cite{bais2003hopf}.
In this section, we will establish the theory of weak Hopf symmetry.

For a quantum system, we are mainly concerned about the symmetry of the vacuum sector $\mathcal{V}_{\rm GS}$.
Recall that for usual group symmetry, to discuss the symmetry a phase $\mathcal{P}=\{H_k,\mathcal{V}_{\rm GS}^k\}$ (Here, the parameter $k$ serves as a label for the Hamiltonian in the phase. For example, in the case of a 1d Ising model, $k$ can be chosen to represent the length of the spin chain), we need to fix a general symmetry group, usually chosen as a unitary group\,\footnote{When considering time-reversal symmetries, we also need to consider the antiunitary operators. And for the open system, the Hamiltonian is replaced with a superoperator called Lindbladian, and the symmetry operators need to be generalized to linear and antilinear completely positive trace-preserving (CPTP) maps \cite{wei2022antilinear}.} $U(\mathcal{V}_{\rm GS}^k)$ for each $k$.
A $G$ symmetry of the system is a group homomorphism $G\to U(\mathcal{V}_{\rm GS}^k)$ such that $U_g$ stabilizes the 
$\mathcal{V}_{\rm GS}$ for all $g\in G$. However, when dealing with Hopf algebra, this requirement is too strong.
Bais et al. \cite{bais2003hopf} proposed the notion of Hopf symmetry, where a state $\psi$ is called Hopf invariant if $g\triangleright \psi= \varepsilon(g)\psi$.
For the weak Hopf algebra case, we cannot straightforwardly generalize the above definition.
To see this, consider the vacuum sector $\mathcal{V}_{\rm GS}$ which is a $W$-module such that $g\triangleright \phi=\varepsilon_L(g)\triangleright \phi$ for all $\phi \in \mathcal{V}_{\rm GS}$.
A general state should transform in a similar way as that of the vacuum state, this inspires us to introduce the following definition of weak Hopf symmetry:

\begin{definition}
    Let $W$ be a weak Hopf algebra and $\psi \in \mathcal{V}$ be a state in some $W$-module $\mathcal{V}$.
    We say that $\psi$ is left invariant under the action of $g\in W$ if $g\triangleright \psi =\varepsilon_L(g)\triangleright \psi$; similarly, $\psi$ is right invariant under the action of $g\in W$ if $\psi \triangleleft g=\psi \triangleleft \varepsilon_R(g)$.
\end{definition}

If $W$ is a Hopf algebra, the above definition reduces to the definition given in \cite{bais2003hopf}.
If $W$ is the group algebra, the group symmetry definition is also recovered.
The weak Hopf symmetry of a state $\psi$ can be naturally defined as the maximal weak Hopf subalgebra that leaves $\psi$ invariant in the sense of the above definition.

\begin{definition}
    For a system with ground state $\psi$, the weak  Hopf symmetry of the system is the maximal weak Hopf subalgebra in a weak Hopf algebra $W$ that leaves $\psi$ invariant.
\end{definition}

We will denote the weak Hopf symmetry of $\psi$ as $\operatorname{Stab}_W(\psi)$. 
More precisely, we define $\operatorname{Stab}_W(\psi)$ in such a way that, if weak Hopf subalgebra $K$ leaves $\psi$ invariant, then $K\subseteq \operatorname{Stab}_W(\psi)$.
Notice that it is unique, because if $K,G$ are two weak Hopf symmetries of $\psi$, then they can generate a larger weak Hopf symmetry that contains both of them, contradicting to the maximality of $\operatorname{Stab}_W(\psi)$.

The basic physical data of the system are particle types, conjugation of particles, and the fusion rule between these particles.
These physical data are captured by the category of unitary left $W$-modules $\Rep(W)$ for the weak Hopf symmetry $W$.
\begin{enumerate}
    \item Particle types (a.k.a., topological charges, sectors) are given by the equivalent classes of irreducible representations $a=(\rho,M_{\rho}) \in \operatorname{Irr}(W)$;
    \item Antiparticles are given by the dual (conjugation) of representations $\bar{a}=(\bar{\rho},M_{\bar{\rho}})$;
    \item Fusion rule is given by the fusion rule of irreducible representations $a\times b =\sum_{c}N_{ab}^c c$.
\end{enumerate}
Thus to understand these physical data, we need to explore the (unitary) representation category of $W$.

A unitary module $M$ is a $W$-module with an inner product such that $\langle x,h\triangleright y\rangle=\langle h^*\triangleright x,y\rangle$, for all $h\in W$, $x,y\in M$ ($h^*$ is the $C^*$ involution of $h$).
It is proved that $\Rep(W)$ is a unitary multi-fusion category \cite{nill1998axioms,nikshych2003invariants,bohm2011weak}.

The vacuum sector is given by the Gelfand–Naimark–Segal (GNS) representation associated with the counit $\varepsilon$.
More specifically, the module of vacuum sector is $\mathds{1}=W_L$, whose $W$-module structure is given by $h\triangleright x =\varepsilon_L(hx)$, for $h\in W$ and $x\in W_L$. The inner product is chosen as $\langle x,y\rangle := \varepsilon(x^*y)$.
Notice that  there are several equivalent choices of vacuum modules, like $\hat{W}_L,\hat{W}_R$, and $W_R$ with their corresponding module structures \cite[Lemma 2.12]{BOHM1998weak}.
The vacuum sector is in general not simple, so it can be decomposed into the direct sum of some particle sectors $\one =a_1\oplus \cdots \oplus a_n$.
The fusion (tensor product) of two $W$-modules $M,\,N$ are defined by $M\otimes N:=\{x\in M\otimes_{\mathbb
C} N\,|\,x=\Delta(1)\,x\}$\,\footnote{It is defined in this way to ensure that the identity $1$ acts as the identity of the tensor product of modules.}, and its $W$-module structure is given by $h\triangleright (x\otimes y)=\sum_{(h)}(h^{(1)}\triangleright x)
\otimes (h^{(2)}\triangleright y)$.
The tensor product of module morphisms $f_1,f_2$ is defined as the restriction of $f_1\otimes_{\mathbb{C}}f_2$ to the subspace $M\otimes N$.

Since $\Rep(W)$ is a rigid category, the antiparticle of $a=(\rho,M)$ is given by the dual $\bar{a}=(\bar{\rho},\bar{M})$. Here $\bar{M}=\hat{M}=\Hom(M,\mathbb{C})$, where the left $W$-module structure is given by $x\triangleright \phi=\phi(S(x)\bullet)$.
Since $W$ is a $C^*$ weak Hopf algebra, the left dual and right dual of a $W$-module are isomorphic ($M^{\vee}\cong {^{\vee}}M=:\bar{M}$).
The inner product of $\bar{M}$ is canonically induced by inner product of $M$, $\langle \bar{x},\bar{y}\rangle :=\langle x, y\rangle$.

\begin{definition}
\label{def:fusionclosed}
    For a fusion anyon model $\EC=\Rep(W)$, a subset of simple anyons (irreps) $\EA\subset \operatorname{Irr}(\EC)= \operatorname{Irr}(W)$ is called fusion closed if it is (i) conjugation closed: $a\in \EA \Rightarrow \bar{a}\in \EA$; (ii) tensor product closed: $a,b\in \EA$, and $a\times b =\sum_{c}N_{ab}^c c$, then $c \in \EA$ for all $N_{ab}^c\neq 0$.
\end{definition}

For a $(2+1)D$ phase, besides the charge type, antiparticle, and fusion rule, there is some extra data called braiding that characterizes the mutual statistics of these charges.
The weak Hopf algebra in this case needs to possess a quasitriangular structure.

A quasitriangular weak Hopf algebra $(W,R)$ is a weak Hopf algebra $W$ equipped with an element  $R=\sum_{i}a_i\otimes b_i \in \Delta^{\rm op}(1)(W\otimes_{\mathbb{C}} W) \Delta(1)$ which satisfies
    \begin{enumerate}
        \item $(\id \otimes \Delta)(R)=R_{13}R_{12}$,
        \item $(\Delta \otimes \id)(R)=R_{13}R_{23}$,
        \item $\Delta^{\rm op}(h)R=R\Delta(h),\forall \,h\in W$,
        \item there exists $\Tilde{R}\in \Delta(1)(W\otimes_{\mathbb{C}} W) \Delta^{\rm op}(1)$ such that $\tilde{R}R=\Delta(1)$ and $R\tilde{R}=\Delta^{\rm op}(1)$,
    \end{enumerate}
    where we have used the notation $R_{kl}=\sum_i1\otimes \cdots 
 \otimes 1\otimes a_i \otimes 1\otimes \cdots \otimes 1\otimes b_i\otimes 1 \otimes \cdots \otimes 1$ with $a_i$ and $b_i$ appearing in the $k$-th and $l$-th places respectively, whose length depends on the context.
For a weak Hopf algebra $W$, its quantum double $D(W)$ has a quasitriangular structure \cite{nikshych2003invariants}. We also include detailed proof of this fact in Appendix \ref{app:QD} for completeness.

\subsection{Weak Hopf symmetry breaking}

In this part, we generalize the results of Hopf symmetry breaking \cite{bais2003hopf} to the weak Hopf symmetry case.
Consider a system with vacuum sector $\mathcal{V}_{\one}$ that has a weak Hopf symmetry $W$. After condensation, the vacuum of the condensed phase is the sea of vacuum charges. The $W$ symmetry is broken and the vacuum sector becomes $\mathcal{V}_{\one'}$. The residual symmetry of the condensate is a weak Hopf subalgebra $V=\operatorname{Stab}_W(\mathcal{V}_{\one'})\subset W$.
The excitation of the effective theory must carry an irreducible representation of $\operatorname{Stab}_W(\mathcal{V}_{\one'})$.
Some charges of the original phase get confined, and others are free (deconfined).
These deconfined particles are fusion closed, and they are irreducible representations of a new weak Hopf symmetry $U$, which is the quotient of $V$. In this part we shall prove the following theorem:

\begin{theorem}
\label{thm:SymBreak1}
Consider a quantum phase $\EC =\Rep(W)$ with weak Hopf symmetry $W$, after the formation of condensate, the symmetry is broken into a weak Hopf subalgebra $V\hookrightarrow W$.
The new phase is given by $\ED=\Rep(V)$. In this new phase, the particles  that have nontrivial monodromy with condensate are confined, and the particles that have trivial monodromy with condensate are deconfined.
The set $\EA_{\rm dc}$ of deconfined particles are fusion closed.
These deconfined particles are irreducible representations of a new weak Hopf symmetry $U$, which is a weak Hopf quotient of $V$.
\begin{equation}
    \text{Symmetry breaking:}
    \begin{tikzcd}
 \Rep(W) \arrow[r, "F_1"] \arrow[d, leftrightarrow]
& \Rep(V)  \arrow[r, "F_2"] \arrow[d, leftrightarrow] & \Rep(U)\arrow[d, leftrightarrow]\\
W 
& V \arrow[l, hook,"\iota"] \arrow[r, two heads,"\pi"'] & U
\end{tikzcd}
\end{equation}
where $F_1$ and $F_2$ are monoidal functors, $\iota$ is an injective weak Hopf map (embedding), and $\pi$ is a surjective weak Hopf map (quotient).
\end{theorem}

To understand the weak Hopf symmetry breaking, it is crucial to understand the structure of weak Hopf subalgebras and quotients. By definition,
a weak Hopf morphism $f:W \to V$ is a weak bialgebra morphism such that $S_V\comp f=f\comp S_W$.
The weak Hopf subalgebra and quotient are naturally defined as follows:

\begin{definition}
    Let $W,V$ be weak Hopf algebras. (i) $V$ is called a weak Hopf quotient of $W$ if there is a surjective weak Hopf morphism $f:W \twoheadrightarrow
 V$; (ii) $V$ is called a weak Hopf subalgebra if there is an injective weak Hopf morphism $f:V \hookrightarrow W$.
\end{definition}

It is clear that the image $\operatorname{Im} f$ of a weak Hopf morphism $f$ is a weak Hopf subalgebra of the target weak Hopf algebra.
As we have pointed out before, the weak Hopf algebra is self-dual, and the dual algebra of $W$ is also a weak Hopf algebra.
The subalgebra and quotient of $W$ and $\What$ are also mutually related.

\begin{lemma}
\label{lemma:dual}
    If $f:W\to V$ is a weak Hopf morphism, then the dual map $\hat{f}:\hat{V} \to \hat{W}$ is also a weak Hopf morphism.
    Moreover, if $f$ is injective then $\hat{f}$ is surjective; if $f$ is surjective, then $\hat{f}$ is injective.
    This further implies that 
    \begin{enumerate}
        \item If $V$ is a weak Hopf subalgebra of $W$, then $\hat{V}$ is a weak Hopf quotient of $\hat{W}$.
        \item If $V$ is a weak Hopf quotient of $W$, then $\hat{V}$ is a weak Hopf subalgebra of $\hat{W}$.
    \end{enumerate}
\end{lemma}

\begin{proof}
    Clearly, $\hat{f}$ is a weak bialgebra morphism. For any $\varphi\in\hat{V}$ and $x\in W$,
    \begin{equation}
        \langle (\hat{f}\comp \hat{S}_V)(\varphi), x \rangle = \langle \varphi, (S_V\comp f)(x) \rangle = \langle \varphi, (f\comp S_W)(x) \rangle = \langle (\hat{S}_W\comp \hat{f})(\varphi),x\rangle. 
    \end{equation}
   This means $\hat{f}$ is a weak Hopf morphism. 
   Since  the dual of the injective map is surjective and the dual of the surjective map is injective, we complete the proofs of 1 and 2.
\end{proof}

The above result means that we can equivalently describe the weak Hopf symmetry breaking using weak Hopf subalgebra $f: V\hookrightarrow W$ or using quotient $\hat{f}:\hat{W}\twoheadrightarrow \hat{V}$.
For a weak Hopf symmetry, the physical data, like particle type, antiparticle, and fusion rule, of the system with a weak Hopf quotient  symmetry has a one-to-one correspondence with that of the original system.
This correspondence is characterized by the representation factors over the quotient map.
For a weak Hopf quotient map $f: W\to V$, a representation of $W$ which factors over $f$ is defined as $\rho=\sigma \comp f$, where $\sigma$ is a representation of $V$.

\begin{lemma}
\label{lemma:FusionClose}
    Let $f: W\to V$ be a quotient map of weak Hopf algebras, then the physical data of $V$-symmetry model are in one-to-one correspondence with the physical data of $W$-symmetry that factors over $f$. 
    \begin{enumerate}
        \item  For a charge type $a=(\sigma,M_{\sigma})$ of $V$-symmetry model, there is a corresponding charge type $a_f=(\rho=\sigma\comp f,M_{\sigma})$.
        If $\rho$ is an irreducible representation of $V$, it is also an irreducible representation of $W$.
        \item The correspondence preserves conjugation of particles $\overline{a_f}=(\bar{\rho}=\bar{\sigma}\comp f,M_{\bar{\sigma}})=\bar{a}_f$.
        \item The correspondence preserves fusion rules.
        \item If $W$ is semisimple, then $V$ is also semisimple.
    \end{enumerate}
\end{lemma}

\begin{proof}
1. If $\sigma: V \to \operatorname{End}(M_{\sigma})$ is a representation of $V$, then $\rho=\sigma\comp f$ is a representation of $W$ factors over $f$; and vice versa.
If $\sigma$ is irreducible then there is no nontrivial submodule that is invariant under the action of $\sigma$. Since $f$ is surjective, this also implies that there is no nontrivial submodule that is invariant under the action of $\rho$. If $\rho$ is irreducible, a similar argument shows that $\sigma$ is irreducible.

2. Since the matrix of conjugation of a representation 
 $\zeta$ satisfies $\bar{\zeta}(g)_{ij}=\zeta(S(g))_{ji}=\zeta_{ji}\comp S(g)$ for all $g$, we see that $\bar{\sigma}_{ij}\comp f=\sigma_{ji}\comp S_V\comp f=\sigma_{ji}\comp f \comp S_W=\rho_{ji}\comp S_W=\bar{\rho}_{ij}$. Notice that we have used the fact that $f$ is a weak Hopf morphism.

 3. For fusion of particles, suppose that $\rho_1=\sigma_1 \comp f$ and $\rho_2=\sigma_2\comp f$. By definition of tensor product of two representations, we have $(\rho_1\otimes \rho_2)\comp \Delta_W=(\sigma_1 \otimes \sigma_2)\comp (f\otimes f)\comp \Delta_W = (\sigma_1 \otimes \sigma_2)\comp \Delta_V \comp f$ by the fact that $f$ is a weak bialgebra morphism. Then using the fact that $f$ is surjective, we see that the fusion rule remains unchanged under the correspondence.

 4. It is a direct corollary of 1.
\end{proof}

\begin{lemma}
    For a weak Hopf algebra $W$ and a fusion closed set $\EA \subset \operatorname{Irr}(W)$, the linear space $\mathcal{V}_{\EA}$ generated by the matrix elements of representations in $\EA$ is a weak Hopf semi-subalgebra of $\hat{W}$, \emph{viz.}, $\mathcal{V}_{\EA}$ is closed under multiplication, coclosed under comultiplication, and closed under antipode.
\end{lemma}

\begin{proof}
    It is sufficient to prove that $\mathcal{V}_{\EA}$ satisfy the following conditions:

    (i) $\mathcal{V}_{\EA}$  is closed under multiplication $\hat{\mu}$ of $\hat{W}$. This is clear from $\hat{\mu}(\rho_{ij}\sigma_{kl})=((\rho_{ij}\otimes \sigma_{kl})\comp \Delta_{W})$. The product of matrix elements of two representations is the matrix elements of the tensor product of those two representations, and since $\EA$ is fusion closed, it is also in $\mathcal{V}_{\EA}$.
    
    (iii) $\mathcal{V}_{\EA}$  is coclosed under coproduct of $\hat{W}$. From pairing $\langle \hat{\Delta} (\rho_{ij}), x\otimes y\rangle = \rho_{ij}(xy)=\sum_k \rho_{ik}(x)\rho_{kj}(y)$, we see that $\hat{\Delta}(\rho_{ij})=\sum_k \rho_{ik}\otimes \rho_{kj}$. Thus $\hat{\Delta}(\mathcal{V}_{\EA}) \subset \mathcal{V}_{\EA}\otimes \mathcal{V}_{\EA}$.

    (iv) $\hat{S}(\mathcal{V}_{\EA})\subset \mathcal{V}_{\EA}$. This is clear from $\hat{S}(\rho_{ij})=\rho_{ij}\comp S=\bar{\rho}_{ji}$.
\end{proof}

It is natural to introduce the weak Hopf subalgebra generated by fusion closed set $\EA$ as the minimal weak Hopf subalgebra of $\What$ that contains $\VV_{\EA}$. We denote it as $\tilde{\VV}_{\EA}$.
Since the double-dual $\hat{\hat{W}}$ can be canonically identified with $W$, using the above-established results, we will see that each weak Hopf subalgebra $V\hookrightarrow W$ is determined by a fusion closed set of anyons.
For a weak Hopf quotient algebra, a similar result holds.

\begin{theorem}
\label{thm:SymBreak2}
        Consider a finite-dimensional $C^*$ weak Hopf algebra $W$.
    \begin{enumerate}
        \item  Let $V \hookrightarrow W$ be a weak Hopf subalgebra, then $V$ is determined by a fusion closed subset $\EA \subset \Irr(\hat{W})$. More precisely, $V$ is
        isomorphic to the weak Hopf subalgebra $\tilde{\VV}_{\EA}$ generated by a fusion closed set $\EA$.
        \item Let $W \twoheadrightarrow V$ be a weak Hopf quotient, then $V$ is determined by a fusion closed subset $\EA \subset \Irr (W)$. More precisely, $V$ is isomorphic to the dual weak Hopf algebra $\tilde{\VV}_{\EA}^{\vee}$.
    \end{enumerate}
    
\end{theorem}

\begin{proof}
    1. From Lemma \ref{lemma:dual} we see that $\Vhat$ is a quotient of $\What$.
    It is clear that $\EA=\Irr(\hat{V})$ of $\hat{V}$ is a fusion closed set of $\Vhat$. By Lemma \ref{lemma:FusionClose}, it is also a fusion closed set of $W$. Since weak Hopf algebra $\tilde{\VV}_{\EA}$ generated by $\EA$ is isomorphic to $V\simeq\hat{\Vhat}$, we complete the proof.

    2. From Lemma \ref{lemma:dual}, $\hat{V}$ is a weak Hopf subalgebra of $\What$. Then using 1, we obtain the assertion.
\end{proof}

The above result provides an operational method to determine the subalgebra and/or quotient algebra of a given weak Hopf algebra. For a given weak Hopf symmetry, if we have the information of particles, then we can determine all the sub and quotient algebras.

Now let us turn back to the weak Hopf symmetry breaking.
The particles of a quantum phase with weak Hopf symmetry $W$ are irreducible representations.
The formation of condensate particles in a new vacuum sector $\VV_{\one'}$ breaks weak Hopf symmetry $W$ to weak Hopf stabilizer of the new vacuum sector. The residual weak Hopf symmetry is thus a weak Hopf subalgebra $W\hookleftarrow V=\operatorname{Stab}_W(\VV_{\one'})$.
The particles in this new phase are irreducible representations of $V$.
In this new phase, not all particles are free, the particles that have nontrivial monodromy with condensate particles are confined.
The one with trivial monodromy with condensate particles is deconfined.
Two natural requirements of deconfined particles are that their antiparticles are also deconfined, and the tensor product of any two deconfined particles is deconfined.
This means that the set $\EA_{\rm dc}\subset \Irr(V)$ of deconfined particles in the new phase is fusion closed.
This $\EA_{\rm dc}$ will generate a weak Hopf subalgebra $\hat{U}\hookrightarrow \hat{V}$.
And $U\cong \hat{\hat{U}}\twoheadleftarrow V$ is a quotient of $V$. 
The deconfined particles are representations of $U$. This completes the proof of Theorem \ref{thm:SymBreak1}.


\section{Two-dimensional lattice model}
\label{sec:Kitaev}

In this section, we present the lattice gauge theory whose symmetry is given by a finite-dimensional weak Hopf algebra $D(W)$.
This construction was proposed by one of the authors in \cite{chang2014kitaev}. Here we will give a more detailed and comprehensive investigation, including the lattice model construction and ribbon operators.

\subsection{Weak Hopf quantum double model}

Consider a given 2$d$ closed oriented manifold $\Sigma$, a lattice on it  (a.k.a. a cellulation of $\Sigma$) is $C(\Sigma)=V(\Sigma)\cup E(\Sigma)\cup F(\Sigma)$, where $V(\Sigma)$, $ E(\Sigma)$ and $F(\Sigma)$ are sets of vertices, edges, and faces, respectively. 
The dual lattice of $C(\Sigma)$ is a lattice $\tilde{C}(\Sigma)$ for which
the vertices and faces of the original lattice are switched while the edge set remains unchanged.
A directed lattice is a lattice such that each edge $e \in E(\Sigma)$ is assigned a direction. The direction of the corresponding dual edge $\tilde{e}\in \tilde{E}(\Sigma)$ is defined by rotating the direction of $e$ counterclockwise by $\pi/2$. 
A site $s=(v,f)$ is a pair of vertex $v$ and an adjacent face $f$. Two sites $s,s'$ are called adjacent if they share a common vertex or face.

Let $W$ be a fixed weak Hopf algebra. Then the total Hilbert space is $\mathcal{H}_{\rm tot} = \otimes_{e\in E(\Sigma)}W$. The $W$ is a left $W$-module with respect to the following two actions $L_{\pm}$ and also a left $\hat{W}$-module with respect to the actions $T_{\pm}$:
\begin{align}
	&L_+^h |x\rangle =|h \triangleright  x\rangle=|hx\rangle,\label{eq:L1}\\
	&L_-^h|x\rangle=|x \triangleleft S^{-1}(h) \rangle=|xS^{-1}(h)\rangle,\\
	&T_+^{\varphi}|x\rangle=|\varphi \rightharpoonup x\rangle=| \sum_{(x)}\langle \varphi, x^{(2)}\rangle x^{(1)}\rangle,\\
	&T_-^{\varphi}|x\rangle=|x\leftharpoonup \hat{S}(\varphi) \rangle
	=| \sum_{(x)}\langle \hat{S}(\varphi), x^{(1)}\rangle x^{(2)}\rangle
	=| \sum_{(x)}\langle \varphi, S(x^{(1)})\rangle x^{(2)}\rangle,
\end{align}
where $h\in W$ and $\varphi \in \hat{W}$. Here, we have adopted Sweedler's arrow notations. 
Similarly, we have the right module actions:
\begin{align}
	&\tilde{L}_-^h |x\rangle =|x \triangleleft  h\rangle=|xh\rangle,\\
	&\tilde{L}_+^h|x\rangle=|S^{-1}(h) \triangleright x\rangle=|S^{-1}(h) x\rangle,\\
	&\tilde{T}_-^{\varphi}|x\rangle = |x \leftharpoonup \varphi \rangle  
	=|\sum_{(x)}\langle \varphi, x^{(1)}\rangle x^{(2)}\rangle,\\
	&\tilde{T}_+^{\varphi}|x \rangle=| \hat{S}(\varphi) \rightharpoonup x\rangle
	=| \sum_{(x)}\langle \hat{S}(\varphi), x^{(2)}\rangle x^{(1)}\rangle
	=| \sum_{(x)}\langle \varphi, S(x^{(2)})\rangle x^{(1)}\rangle. \label{eq:L2}
\end{align}
In this work, we will construct the lattice Hamiltonian based on the left module actions. Nevertheless, it will become clear later that the right actions are necessary for the construction of ribbon operators.

For a given directed edge $e$, we assign  $L_-$ and $L_+$  to its starting and ending vertices respectively.
Similarly, for the dual edge $\tilde{e}$, we assign  $T_+$ and $T_-$  to its starting and ending faces respectively.
Fix a site $s=(v,f)$, we order the edges around the vertex $v$ and the dual edges around the face $f$ counterclockwise with the origin $s$ (the operators act on edges from the left-hand side of edges).
We then introduce two kinds of local operators associated with $s$. The first one is attached to the vertex of $s$: for $h\in W$, 
\begin{equation}\label{eq:Av}
	A^h(s)
\big{|}	\begin{aligned}
		\begin{tikzpicture}
			\draw[-latex,black] (-0.5,0) -- (0,0);
			\draw[-latex,black] (0,0) -- (0,0.5); 
			\draw[-latex,black] (0,0) -- (0.5,0); 
			\draw[-latex,black] (0,-0.5) -- (0,0); 
			\draw[red,line width=1pt] (0,0) -- (0.25,-0.25);
			\draw [fill = black] (0,0) circle (1.2pt);
			\node[ line width=0.2pt, dashed, draw opacity=0.5] (a) at (0.7,0){$x_1$};
			\node[ line width=0.2pt, dashed, draw opacity=0.5] (a) at (-0.7,0){$x_3$};
			\node[ line width=0.2pt, dashed, draw opacity=0.5] (a) at (0,-0.7){$x_4$};
			\node[ line width=0.2pt, dashed, draw opacity=0.5] (a) at (0,0.7){$x_2$};
			\node[ line width=0.2pt, dashed, draw opacity=0.5] (a) at (0.4,-0.4){$f$};
			\node[ line width=0.2pt, dashed, draw opacity=0.5] (a) at (-0.2,0.2){$v$};
		\end{tikzpicture}
	\end{aligned}   \big{ \rangle}     
= \sum_{(h)}
\big{|}	\begin{aligned}
	\begin{tikzpicture}
		\draw[-latex,black] (-0.5,0) -- (0,0);
		\draw[-latex,black] (0,0) -- (0,0.5); 
		\draw[-latex,black] (0,0) -- (0.5,0); 
		\draw[-latex,black] (0,-0.5) -- (0,0); 
		\draw[red,line width=1pt] (0,0) -- (0.25,-0.25);
		\node[ line width=0.2pt, dashed, draw opacity=0.5] (a) at (1.2,0){$L_{-}^{h^{(1)}}x_1$};
		\node[ line width=0.2pt, dashed, draw opacity=0.5] (a) at (-1.2,0){$L_{+}^{h^{(3)}}x_3$};
		\node[ line width=0.2pt, dashed, draw opacity=0.5] (a) at (0,-0.8){$L_{+}^{h^{(4)}}x_4$};
		\node[ line width=0.2pt, dashed, draw opacity=0.5] (a) at (0,0.8){$L_{-}^{h^{(2)}}x_2$};
		\node[ line width=0.2pt, dashed, draw opacity=0.5] (a) at (0.4,-0.4){$f$};
		\node[ line width=0.2pt, dashed, draw opacity=0.5] (a) at (-0.2,+0.2){$v$};
	\end{tikzpicture}
\end{aligned}   \big{ \rangle} .
\end{equation}
If $h$ is cocommutative, the corresponding operator will
 not depend on $s$ but only on $v$, and thus we will denote it as $A_v^h$.
The second one is attached to the face of $s$: for $\varphi\in\hat{W}$, 
\begin{equation}\label{eq:Bf}
	B^{\varphi}(s)
	\big{|}	\begin{aligned}
		\begin{tikzpicture}
			\draw[-latex,black] (-0.5,0.5) -- (0.5,0.5);
			\draw[-latex,black] (-0.5,-0.5) -- (-0.5,0.5); 
			\draw[-latex,black] (0.5,-0.5) -- (0.5,0.5); 
			\draw[-latex,black] (-0.5,-0.5) -- (0.5,-0.5); 
			\draw [fill = black] (0,0) circle (1.2pt);
			\draw[red,line width=1pt] (0,0) -- (0.5,-0.5);
			\node[ line width=0.2pt, dashed, draw opacity=0.5] (a) at (0.75,0){$x_1$};
			\node[ line width=0.2pt, dashed, draw opacity=0.5] (a) at (-0.75,0){$x_3$};
			\node[ line width=0.2pt, dashed, draw opacity=0.5] (a) at (0,-0.7){$x_4$};
			\node[ line width=0.2pt, dashed, draw opacity=0.5] (a) at (0,0.7){$x_2$};
		\end{tikzpicture}
	\end{aligned}   \big{ \rangle}     
	= 
	\sum_{(\varphi)}
	\big{|}	\begin{aligned}
	\begin{tikzpicture}
		\draw[-latex,black] (-0.5,0.5) -- (0.5,0.5);
		\draw[-latex,black] (-0.5,-0.5) -- (-0.5,0.5); 
		\draw[-latex,black] (0.5,-0.5) -- (0.5,0.5); 
		\draw[-latex,black] (-0.5,-0.5) -- (0.5,-0.5); 
		\draw [fill = black] (0,0) circle (1.2pt);
		\draw[red,line width=1pt] (0,0) -- (0.5,-0.5);
		\node[ line width=0.2pt, dashed, draw opacity=0.5] (a) at (1.3,0){$T^{\varphi^{(1)}}_{-}x_1$};
		\node[ line width=0.2pt, dashed, draw opacity=0.5] (a) at (-1.3,0){$T^{\varphi^{(3)}}_{+}x_3$};
		\node[ line width=0.2pt, dashed, draw opacity=0.5] (a) at (0,-0.9){$T^{\varphi^{(4)}}_{-}x_4$};
		\node[ line width=0.2pt, dashed, draw opacity=0.5] (a) at (0,0.9){$T^{\varphi^{(2)}}_+x_2$};
	\end{tikzpicture}
\end{aligned}   \big{ \rangle} .
\end{equation}
If $\varphi$ is cocommutative, the corresponding operator will not depend on $s$ but only on $f$. We will denote it as $B_f^{\varphi}$.

\begin{proposition}\label{prop:doubledRep}
    For any site $s$, the operators $A^h(s)$ and $B^\varphi(s)$ satisfy the commutation relations
    \begin{align}
        A^h(&s)B^\varphi(s) = \sum_{(h),(\varphi)}B^{\varphi^{(2)}}(s)A^{h^{(2)}}(s)\langle\varphi^{(1)},S^{-1}(h^{(3)})\rangle \langle \varphi^{(3)},h^{(1)}\rangle, \label{eq:rep-comm-rel-1} \\
        B^{\varphi}(s)&A^{xh}(s) = B^{\varphi(x\rightharpoonup \varepsilon)}(s)A^h(s),\quad B^{\varphi}(s)A^{yh}(s) = B^{\varphi(\varepsilon \leftharpoonup y)}(s)A^h(s),  \label{eq:rep-comm-rel-2}
    \end{align}
    for  any $x\in W_L$, $y\in W_R$. Therefore, the map 
    \begin{equation}
        \rho:D(W)\to \operatorname{End}(\mathcal{H}(s)),\quad [\varphi\otimes h]\mapsto D^{\varphi\otimes h}(s)=B^\varphi(s)A^h(s)
    \end{equation}
    is an algebra homomorphism, that is, $\rho$ is a representation of $D(W)$ over $\mathcal{H}(s)=\otimes_{j\in \partial s} W$. 
\end{proposition}

\begin{proof}
Consider the following configuration 
\begin{equation}
	\begin{aligned}
\begin{tikzpicture}
	\draw[-latex,black] (-1,0) -- (0,0); 
	\draw[-latex,black] (0,0) -- (1,0); 
	\draw[-latex,black] (0,0) -- (0,1); 
	\draw[-latex,black] (0,-1) -- (0,0); 
	\draw[-latex,black] (0,1) -- (1,1);
	\draw[-latex,black] (1,0) -- (1,1);  
	\draw[line width=0.5pt, red] (0,0) -- (0.5,0.5);
	\draw [fill = black] (0,0) circle (1.2pt);
	\draw [fill = black] (0.5,0.5) circle (1.2pt);
	\node[ line width=0.2pt, dashed, draw opacity=0.5] (a) at (0.2,0.5){$s$};
	\node[ line width=0.2pt, dashed, draw opacity=0.5] (a) at (-0.2,-0.2){$v$};
	\node[ line width=0.2pt, dashed, draw opacity=0.5] (a) at (0.7,0.7){$f$};
	\node[ line width=0.2pt, dashed, draw opacity=0.5] (a) at (-1.2,0){$x_2$};
	\node[ line width=0.2pt, dashed, draw opacity=0.5] (a) at (0.3,-0.9){$x_3$};
	\node[ line width=0.2pt, dashed, draw opacity=0.5] (a) at (0.8,-0.3){$x_4$};
	\node[ line width=0.2pt, dashed, draw opacity=0.5] (a) at (1.3,0.5){$x_5$};
	\node[ line width=0.2pt, dashed, draw opacity=0.5] (a) at (0.5,1.2){$x_6$};
	\node[ line width=0.2pt, dashed, draw opacity=0.5] (a) at (-0.2,0.5){$x_1$};
\end{tikzpicture}
\end{aligned}
\quad 
\begin{aligned}
&A^h(s)=\sum_{(h)}L_-^{h^{(1)}}(j_1)\otimes L_+^{h^{(2)}}(j_2)\otimes L_+^{h^{(3)}}(j_3)\otimes L_-^{h^{(4)}}(j_4),\\
& B^{\varphi}(s)=\sum_{(\varphi)}T_{-}^{\varphi^{(1)}} (j_4) \otimes T_{-}^{\varphi^{(2)}} (j_5) \otimes T_{+}^{\varphi^{(3)}} (j_6) \otimes T_{+}^{\varphi^{(4)}} (j_1).
\end{aligned}
\end{equation}
To prove the first identity, on the one hand, we see that
\begin{align*}
    & \quad A^h(s)B^\varphi(s)|x_1, x_2, x_3,x_4,x_5,x_6\rangle \\
    & = \sum_{(h)}\sum_{(x_i)}\varphi\left(S(x_4^{(1)})S(x_5^{(1)})x_6^{(2)}x_1^{(2)}\right)\big{|}x_1^{(1)}S^{-1}(h^{(1)}),h^{(2)}x_2,h^{(3)}x_3,x_4^{(2)}S^{-1}(h^{(4)}),x_5^{(2)},x_6^{(1)}\big{\rangle}. 
\end{align*}
On the other hand, we have 
\begin{align*}
    &\quad \sum_{(h),(\varphi)}B^{\varphi^{(2)}}(s)A^{h^{(2)}}(s)\langle\varphi^{(1)},S^{-1}(h^{(3)})\rangle \langle \varphi^{(3)},h^{(1)}\rangle |x_1, x_2, x_3,x_4,x_5,x_6\rangle  \\
    & = \sum_{(h)}\sum_{(x_i)}\varphi\left(S^{-1}(h^{(8)})h^{(7)}S(x_4^{(1)})S(x_5^{(1)})x_6^{(2)}x_1^{(2)}S^{-1}(h^{(2)})h^{(1)}\right) \\
    & \quad  \big{|} x_1^{(1)}S^{-1}(h^{(3)}),h^{(4)}x_2,h^{(5)}x_3,x_4^{(2)}S^{-1}(h^{(6)}),x_5^{(2)},x_6^{(1)} \big{\rangle} \\
    & = \sum_{(h)}\sum_{(x_i)}\varphi\left(S^{-1}(\varepsilon_R(h^{(6)}))S(x_4^{(1)})S(x_5^{(1)})x_6^{(2)}x_1^{(2)}S^{-1}(\varepsilon_R(h^{(1)}))\right) \\
    & \quad  \big{|} S^{-1}(S(x_1^{(1)}))S^{-1}(h^{(2)}),h^{(3)}x_2,h^{(4)}x_3,x_4^{(2)}S^{-1}(h^{(5)}),x_5^{(2)},x_6^{(1)} \big{\rangle} \\ 
    & = \sum_{(h)}\sum_{(x_i)}\varphi\left(S(x_4^{(1)})S(x_5^{(1)})x_6^{(2)}x_1^{(2)}\right) \\
    & \quad  \big{|} x_1^{(1)}S^{-1}(h^{(2)}S^{-1}(\varepsilon_R(h^{(1)}))),h^{(3)}x_2,h^{(4)}x_3,x_4^{(2)}S^{-1}(h^{(5)}\varepsilon_R(h^{(6)})),x_5^{(2)},x_6^{(1)} \big{\rangle} \\ 
    & = \sum_{(h)}\sum_{(x_i)}\varphi\left(S(x_4^{(1)})S(x_5^{(1)})x_6^{(2)}x_1^{(2)}\right)\\
    &\quad  \big{|}x_1^{(1)}S^{-1}(h^{(1)}),h^{(2)}x_2,h^{(3)}x_3,x_4^{(2)}S^{-1}(h^{(4)}),x_5^{(2)},x_6^{(1)}\big{\rangle}. 
\end{align*}
Here in the third equality we used $\sum_{(w)}zS(w^{(1)})\otimes w^{(2)} = \sum_{(w)}S(w^{(1)})\otimes w^{(2)}z$ for any $z\in W_L$ and $w\in W$; the last equality follows since $w = \sum_{(w)}w^{(1)}\varepsilon_R(w^{(2)})$ and $\sum_{(w)}w^{(2)}S^{-1}(\varepsilon_R(w^{(1)})) = \sum_{(w)}S^{-1}(S(w^{(1)})w^{(2)}S(w^{(3)})) = S^{-1}(S(w)) = w$. To prove the second identity, one proceeds
\begin{align*}
    &\quad B^\varphi(s)A^{xh}(s) |x_1, x_2, x_3,x_4,x_5,x_6\rangle \\
    & = \sum_{(xh)} \sum_{(x_i)}\varphi\left(S(x_4^{(1)}S^{-1}((xh)^{(6)})S(x_5^{(1)})x_6^{(2)}x_1^{(2)}S^{-1}((xh)^{(1)})\right) \\
    & \quad \big{|} x_1^{(1)}S^{-1}((xh)^{(2)}),(xh)^{(3)}x_2,(xh)^{(4)}x_3,x_4^{(2)}S^{-1}((xh)^{(5)}),x_5^{(2)},x_6^{(1)}\big{\rangle} \\
    & = \sum_{(h)} \sum_{(x_i)}\varphi\left(S(x_4^{(1)}S^{-1}(h^{(6)}))S(x_5^{(1)})x_6^{(2)}x_1^{(2)}S^{-1}(h^{(1)})S^{-1}(x)\right) \\
    & \quad \big{|} x_1^{(1)}S^{-1}(h^{(2)}),h^{(3)}x_2,h^{(4)}x_3,x_4^{(2)}S^{-1}(h^{(5)}),x_5^{(2)},x_6^{(1)}\big{\rangle} \\
    & = \sum_{(h)} \sum_{(x_i)}(S^{-1}(x)\rightharpoonup\varphi)\left(S(x_4^{(1)}S^{-1}(h^{(6)}))S(x_5^{(1)})x_6^{(2)}x_1^{(2)}S^{-1}(h^{(1)})\right) \\
    & \quad \big{|} x_1^{(1)}S^{-1}(h^{(2)}),h^{(3)}x_2,h^{(4)}x_3,x_4^{(2)}S^{-1}(h^{(5)}),x_5^{(2)},x_6^{(1)}\big{\rangle} \\
    & = \sum_{(h)} B^{\varphi(x\rightharpoonup \varepsilon)}(s) \big{|} x_1S^{-1}(h^{(1)}),h^{(2)}x_2, h^{(3)}x_3,x_4S^{-1}(h^{(4)}),x_5,x_6 \big{\rangle} \\
    & = B^{\varphi(x\rightharpoonup \varepsilon)}(s)A^h(s) |x_1, x_2, x_3,x_4,x_5,x_6\rangle,
\end{align*}
since $\Delta(xh) = \sum_{(h)} xh^{(1)}\otimes h^{(2)}$, and $S^{-1}(x)\rightharpoonup \varphi = \varphi(S^{-1}(x)\rightharpoonup \varepsilon) = \varphi(x\rightharpoonup \varepsilon)$. The third identity is proved similarly. 
 The last assertion follows from \eqref{eq:rep-comm-rel-1} and \eqref{eq:rep-comm-rel-2}.
\end{proof}

\begin{proposition}
If $h$ and $\varphi$ are the respective Haar integrals of $W$ and $\hat{W}$, then $A_v^h$ and $B_f^{\varphi}$ are Hermitian projectors, i.e., $(A^h_v)^{\dagger}=A^h_v$, $(A^h_v)^2=(A^h_v)$,
$(B_f^{\varphi})^{\dagger}=B_f^{\varphi}$, $(B_f^{\varphi})^{2}=B_f^{\varphi}$, and  $[A^h_v,B_f^{\varphi}]=0$ for all $v,f$.
\end{proposition}

\begin{proof}

In \cite{chang2014kitaev}, Chang proved that for cocommutative  $h,\varphi$, $[A^h_v,B_f^{\varphi}]=0$ for all $v,f$.
We only need to show that they are Hermitian projectors.
Unlike $C^*$ Hopf algebra case, $S$ is now not equal to $S^{-1}$, $(L_{-}^g)^{\dagger}\neq L^{g^*}_{-}$.
Recall that the inner product is given by $\langle x,y\rangle=\varphi_{\What}(x^*y)$;
we have 
\begin{equation}
    \langle x, g\triangleright y\rangle = 
    \langle g\triangleright x,  y\rangle, \quad
    \langle x,  y\triangleleft S(g)\rangle 
    = \langle x \triangleleft S^{-1}(g),  y\rangle,
\end{equation}
where we have used $S(x)^*=S^{-1}(x^*)$.
Nevertheless, for Haar integral $S(h_W)=S^{-1}(h_W)=h_W$, we have $S(h_W^{(i)})=S^{-1}(h_W^{(i)})$, which further implies that 
$(A^{h_W}_v)^{\dagger}=A^{h_W^*}_v=A^{h_W}_v$.
Since $h^2=h$, $(A^h_v)^2=(A^h_v)$. 

The proof for the face operator is subtler. Using the fact that $\varphi_{\What}$ is a left and a right integral, it satisfies (the dual version of  \cite[Eqs.~(3.4a) and (3.4d)]{BOHM1998weak})
\begin{align}
    (y\leftharpoonup \psi)\rightharpoonup \varphi_{\What}=\hat{S}(\psi)(y\rightharpoonup \varphi_{\What}),\\
    \varphi_{\What}\leftharpoonup (\psi \rightharpoonup y)=(\varphi_{\What}\leftharpoonup y)\hat{S}(\psi). \label{eq:rightInt}
\end{align}
When acting on $x^*$ for both sides of Eq.~\eqref{eq:rightInt}, we obtain 
\begin{equation}\label{eq:Svarphi1}
    \operatorname{LHS}=\sum\psi(y^{(2)}) \varphi_{\What}^{(1)}(y^{(1)}) \varphi_{\What}^{(2)}(x^{*})=\langle x,\psi \rightharpoonup y\rangle,
\end{equation}
\begin{equation}
    \begin{aligned}
    \operatorname{RHS}= \sum   \varphi_{\What}^{(1)}(y) \varphi_{\What}^{(2)}((x^*)^{(1)})  \hat{S}(\psi)((x^*)^{(2)})=\sum\varphi_{\What}((x^*)^{(1)}y) \langle \psi, S((x^*)^{(2)})\rangle.
    \end{aligned} \label{eq:Svarphi}
\end{equation}
For Haar integral $\varphi_{\What}$, $\hat{S}(\varphi_{\What}^{(i)})=\hat{S}^{-1}(\varphi_{\What}^{(i)})$.
Take $\psi$ in Eqs.~\eqref{eq:Svarphi1} and \eqref{eq:Svarphi} as $\varphi_{\What}^{(i)}$, we obtain 
\begin{equation}
   \langle x, \varphi_{\What}^{(i)} \rightharpoonup y\rangle = \langle  (\varphi_{\What}^{(i)})^* \rightharpoonup x,y\rangle.
\end{equation}
Thus $(T_{+}^{\varphi_{\What}^{(i)}})^{\dagger}= T_{+}^{\varphi_{\What}^{*(i)}}=T_{+}^{\varphi_{\What}^{(i)}}$.
Similarly, we have $(T_{-}^{\varphi_{\What}^{(i)}})^{\dagger}= T_{-}^{\varphi_{\What}^{*(i)}}=T_{-}^{\varphi_{\What}^{(i)}}$.
Therefore,  $(B^{\varphi_{\What}}_f)^{\dagger}=B^{\varphi_{\What}}_f$. Since $\varphi_{\What}^2=\varphi_{\What}$, we see that $B_f^{\varphi_{\What}}$ is a projector.
\end{proof}

Denote $A_v=A_v^{h_W}$ and $B_f = B_f^{\varphi_{\hat{W}}}$, where $h_W$ and $\varphi_{\hat{W}}$ are Haar integrals of $W$ and $\hat{W}$ respectively. The local commutating projector Hamiltonian is given by
\begin{equation}
	H[D(W);C(\Sigma)]=\sum_{v\in V(\Sigma)}(I-A_v)+\sum_{f\in F(\Sigma)}(I-B_f),
\end{equation}
where $A_v$ imposes a local Gauss constraint and $B_f$ imposes a local flatness condition.
Since we have ${L}_{-}^h=S\comp L_{+}^h\comp S^{-1}$, $T_-^{\phi} =S^{-1}\comp T_+^{\phi}\comp S$, if the direction of $e$ is reversed, we just act the antipode on it. This means that the initial configuration of the lattice edge direction does not matter.
We will denote the quantum double model obtained in this way as $\mathsf{QD}(\What,W;\Sigma)$.
As we will see, its topological excitation is given by the representation category $\Rep(D(W))$, which is a braided fusion category.

\begin{remark}
	If we use the right-module structure, and order edges around the vertex and face clockwise, we will similarly obtain vertex and face operators $\tilde{A}^{h}(s)$ and $\tilde{B}^{\varphi}(s)$.
\end{remark}

\begin{remark}
    Since $W$ is finite dimensional, for each $\varphi\in \What$, there is a corresponding $g_{\varphi}\in W$ such that $\varphi(x)=\langle g_{\varphi},x\rangle$ for all $x\in W$.
\end{remark}

\begin{remark}
   For any given weak Hopf subalgebras $J\subset \What$ and $K\subset W$ for which there is a pairing $\langle \bullet,\bullet\rangle: J\otimes K \to \mathbb{C}$ induced by the canonical pairing of $\What$ and $W$, we can construct a lattice gauge theory model such that the topological excitation is given by the representation category of the generalized quantum double $J^{\rm cop}\Join K$ (the Hopf algebra case has be briefly discussed in \cite{jia2022boundary}).
\end{remark}

\vspace{1em}
\emph{Weak Hopf symmetry of the model.}\,\,\textemdash\, The vacuum sector of the weak Hopf quantum double model is given by
\begin{equation}
   \mathcal{V}_{\one}=\prod_v (I-A_v) \prod_f (I-B_f) \mathcal{H}_{\rm tot}.
\end{equation}
For any ground state $|\Omega\rangle\in \VV_{\one}$, we have
\begin{align*}
    B^\psi A^g|\Omega\rangle & = B^\psi A^gA^{h_W}B^{\varphi_{\hat{W}}}|\Omega\rangle  = B^\psi A^{\varepsilon_L(g)h_W}B^{\varphi_{\hat{W}}}|\Omega \rangle  \\
    &  = B^{\psi(\varepsilon_L(g)\rightharpoonup\varepsilon)\varphi_{\hat{W}}}A^{h_W}|\Omega\rangle = B^{(S^{-1}(\varepsilon_L(g))\rightharpoonup\psi)\varphi_{\hat{W}}}A^{h_W}|\Omega\rangle \\
    &=  B^{\hat{\varepsilon}_L((S^{-1}(\varepsilon_L(g))\rightharpoonup\psi))\varphi_{\hat{W}}}A^{h_W}|\Omega\rangle = B^{\hat{\varepsilon}_L((S^{-1}(\varepsilon_L(g))\rightharpoonup\psi))}A^{1_W}|\Omega\rangle.   
\end{align*}

\vspace{1em}
\emph{Generalized weak Hopf quantum double model and hierarchy structure.}\,\,\textemdash\,\,Consider two weak Hopf algebras $J,W$ equipped with a pairing $\langle \bullet, \bullet\rangle: J\otimes W \to \mathbb{C}$, we could construct a generalized quantum double $D(J^{\rm cop},W)$.
Using a very similar lattice construction, we could obtain a generalized weak Hopf quantum double model $\mathsf{QD}(J,W;\Sigma)$ on a lattice $\Sigma$.
The topological excitation is given by the representation category of the generalized quantum double $\Rep(D(J^{\rm cop},W))$.
It is interesting here that we have a hierarchy structure, namely, we could choose weak Hopf subalgebras $J'\subset J$ and $W'\subset W$ to build a quantum double model $\mathsf{QD}(J',W';\Sigma)$.
This theory is a sub-theory of the original theory in the sense that all topological excitation can be obtained from the original theory by restriction of the original weak Hopf algebras to their subalgebras.

\emph{Topological excitations.} --- 
Let us now provide an outline of the classification of the topological excitations for the weak Hopf quantum double model.
A more rigorous and comprehensive treatment will be given in our forthcoming work \cite{Jia2023classify}.

Before delving into the exploration of topological excitations within the weak Hopf algebraic quantum double model, it is pertinent to revisit the case where the algebraic structure is given by $W=\mathbb{C}[G]$, with $G$ representing a finite group \cite{Kitaev2003,lusztig1987leading,dijkgraaf1991quasi,gould1993quantum,Witherspoon1996the}. 
In this particular case, topological excitations can be classified using the notation $([g], \pi)$, where $[g]$ represents a conjugacy class of the group $G$, and $\pi$ denotes an irreducible representation of the centralizer $C_G(g)$. It is worth noting that there exists a $\mathbb{C}C_G(g)$-module, denoted as $M_\pi$, corresponding to the representation $\pi$. Consequently, a topological charge can be expressed as follows:
\begin{equation}\label{eq:anyon}
a_{[g],\pi}= \mathbb{C}[G] \otimes_{{\mathbb{C} C_G(g)}} M_{\pi}.
\end{equation}
Furthermore, it is essential to mention that the vacuum charge corresponds to the case where $g=e_G$ (the identity element of the group $G$) and $\pi=\mathds{1}$, indicating the trivial representation.
The antiparticle of the topological excitation described in Eq.~(\ref{eq:anyon}) is given by (notably, $C_G(g)=C_G(g^{-1})$):
\begin{equation}
	a_{[g^{-1}],\pi^{\dagger}}= \mathbb{C}[G] \otimes _{\mathbb{C} C_G(g^{-1})} M_{\pi^{\dagger}}.
\end{equation}
The conjugacy class $[g]$ associated with a topological charge is referred to as the magnetic charge, while the irreducible representation $\pi$ is known as the electric charge. In the case where $g=e_G$, $a_{[e_G],\pi}$ is characterized by a representation of $G$ and is referred to as a chargeon. When $\pi=\mathds{1}$ (the trivial representation), $a_{[g],\mathds{1}}$ is called a fluxion. Lastly, if both $g\neq e_G$ and $\pi\neq \mathds{1}$, $a_{[g],\pi}$ is referred to as a dyon.
The quantum dimension of a topological excitation is given by:
\begin{equation}
	\FPdim a_{[g],\pi} = |[g]| \dim \pi,
\end{equation}
where $|[g]|$ represents the size of the conjugacy class and $\dim \pi$ denotes the dimension of the irreducible representation. These topological excitations form a UMTC denoted as $\mathsf{Rep}(D(G))$, which represents the representation category of the quantum double of the finite group $G$.

In our previous work \cite{jia2022boundary}, we extensively discussed the topological excitations within the context of the Hopf quantum double model based on Refs.~\cite{gelaki2008nilpotent,burciu2012irreducible}. Notably, the insights and analysis presented in that study can be readily extended and applied to the weak Hopf quantum double model.
The topological excitations are classified by the irreducible representations of the quantum double $D(W)$ of the weak Hopf algebra $W$, and these topological excitations form a braided fusion category $\mathsf{Rep}(D(W))$.
To make it clearer, we need to introduce the notion of \emph{universal grading group} for a multi-fusion category \cite{gelaki2008nilpotent}.
A multi-fusion category $\EC$ is referred to as $G$-graded if it can be decomposed as follows:
\begin{equation}
\EC=\bigoplus_{g\in G} \EC_g,
\end{equation}
where $\EC_g$ represents full Abelian subcategories and the tensor product of objects from $\EC_g$ and $\EC_h$ maps to objects in $\EC_{gh}$ for all $g$ and $h$ in the finite group $G$. Here, $G$ is known as the grading group of $\EC$. When $G$ is maximal in the sense that any other grading can be obtained by taking a quotient group of $G$, it is referred to as the universal grading group, denoted as $G=U(\EC)$.
The Grothendieck ring $K_0(\EC)$ of $\EC$ is a ring with a multiplicative structure induced by the tensor products of $\EC$. Notice that
$K_0(\EC)$ is a multi-fusion ring when $\EC$ is a multi-fusion category.

For the weak Hopf quantum double model, we know that the topological excitations are classified by the irreducible representations of the quantum double $D(W)$. 
Since $\Rep(\What^{\rm cop})$ is a multi-fusion category, we have a corresponding universal grading group $G=U(\Rep(\What^{\rm cop}))$.
Then $\mathsf{Rep} (\hat{W}^{\rm cop}) =\oplus_{g\in G}	\mathsf{Rep} (\hat{W}^{\rm cop})_g$.
Suppose that $W_g$ is a weak Hopf subalgebra of $W$ such that $\mathsf{Rep}(\hat{W}_g) =\oplus_{x\in C_G(g)} \mathsf{Rep} (\hat{W}^{\rm cop}) x$. We denote by $\mathcal{I}_g={M_g}$ the set of all irreducible representations of $\hat{W}^{\rm cop} \Join W_g$ (where ``$\Join$'' denotes the bicrossed product) satisfying the condition that the character $\chi_{M_g}$, when restricted to the largest central weak Hopf subalgebra $Z(\hat{W}^{\rm cop})$ of $\hat{W}^{\rm cop}$, satisfies ${\chi_{M_g}}_{|{Z(\hat{W}^{\rm cop})}}=g \dim M_g$.
The topological excitations can be classified as follows:
\begin{equation}
a_{g, M_g}= W \otimes_{W_g} M_g,
\end{equation}
where $g\in G$ and $M_g\in \mathcal{I}_g$. This classification provides a complete description of the irreducible representations of the quantum double $D(W)$ \cite{Jia2023classify}.
The element $g\in G$ can be interpreted as a magnetic charge, while $M_g$ represents an electric charge. Specifically, when $g=e_G$, the topological excitation $a_{e_G,M_{e_G}}$ is referred to as a chargeon. If $M_g=\mathds{1}$, the topological excitation $a_{g,\mathds{1}}$ is called a fluxion. For the case where both $g\neq e_G$ and $M_g \neq \mathds{1}$, the topological excitation $a_{g, M_g}$ is known as a dyon.
The quantum dimension of a topological excitation is given by:
\begin{equation}
\FPdim a_{g, M_g} = \frac{|G|}{\dim W_g} \dim M_g,
\end{equation}
where $|G|$ represents the order of the group $G$, and $\dim W_g$ corresponds to the dimension of the subspace $W_g$. 

\subsection{Local algebra and ribbon operator} \label{subsec:ribbon-operator}

For a site $s$, the vertex and face operators will generate a local algebra
\begin{equation}
	\mathcal{A}(s):=\{D^{\varphi\otimes h}(s)=B^{\varphi}(s)A^h(s),\varphi\in \hat{W},h\in W\}.
\end{equation}
It is clear that $\mathcal{A}(s)\cong D(W)$ is a weak Hopf algebra. In this subsection, we will construct the ribbon operators on the quantum double model on closed surfaces based on $W$. These operators form an algebra isomorphic to the dual of $D(W)$. The ribbon operators on the model based on Hopf algebras have been studied in Refs.~\cite{chen2021ribbon,jia2022boundary}. 

We begin with the basic ingredients of ribbons. The building blocks of a ribbon are direct/dual triangles. A direct triangle is a triangle on the original lattice, consisting of a directed edge and two adjacent sites such that the edge connects the two sites, and a dual triangle is a direct triangle on the dual lattice. A direct (resp. dual) triangle is denoted as $\tau=(e,s_0,s_1)$ (resp. $\tilde{\tau} = (\tilde{e},s_0,s_1)$), where the ordered pair $(s_0,s_1)$ indicates the direction of the triangle. A (direct or dual) triangle is called a left-handed (resp. right-handed) triangle if the edge of the triangle is on the left-hand (resp. right-hand) side when one passes through the triangle along its positive direction. We call this feature the chirality of the triangle (this is called local orientation in Ref.~\cite{chen2021ribbon}). Left-handed (resp. right-handed) triangles are denoted as $\tau_L$ and $\tilde{\tau}_L$ (resp. $\tau_R$ and $\tilde{\tau}_R$). We use a superscript $+$ or $-$ to indicate that the direction of $\tau=(e,s_0,s_1)$ and $e$ coincide or not (the same for $\tilde{\tau}$):
\begin{align}
    &	\begin{aligned}
        \begin{tikzpicture}
            \draw[-latex,black] (1,0) -- (-1,0); 
            \node[ line width=0.2pt, dashed, draw opacity=0.5] (a) at (0,0.2){$x$};
            \draw[line width=0.5pt, red] (-1,0) -- (0,-1);
            \node[ line width=0.2pt, dashed, draw opacity=0.5] (a) at (-0.6,-0.7){$s_1$};
            \draw[line width=0.5pt, red] (1,0) -- (0,-1);
            \node[ line width=0.2pt, dashed, draw opacity=0.5] (a) at (0.6,-0.7){$s_0$};
            \draw[-stealth,gray, line width=3pt] (0.5,-0.4) -- (-0.5,-0.4); 
            \node[ line width=0.2pt, dashed, draw opacity=0.5] (a) at (0.1,-1.5){$\tau^+$};
        \end{tikzpicture}
    \end{aligned}
    \quad \quad \quad \quad \quad
    \begin{aligned}
        \begin{tikzpicture}
            \draw[-latex,black] (-1,0) -- (1,0); 
            \node[line width=0.2pt, dashed, draw opacity=0.5] (a) at (0,0.2){$x$};
            \draw[line width=0.5pt, red] (-1,0) -- (0,-1);
            \node[line width=0.2pt, dashed, draw opacity=0.5] (a) at (-0.6,-0.7){$s_1$};
            \draw[line width=0.5pt, red] (1,0) -- (0,-1);
            \node[ line width=0.2pt, dashed, draw opacity=0.5] (a) at (0.6,-0.7){$s_0$};
            \draw[-stealth,gray, line width=3pt] (0.5,-0.4) -- (-0.5,-0.4); 
            \node[ line width=0.2pt, dashed, draw opacity=0.5] (a) at (0.1,-1.5){$\tau^-$};
        \end{tikzpicture}
    \end{aligned}
\end{align}

A ribbon, denoted as $\rho$, consists of a sequence of consecutive triangles such that the ending site of a triangle is the starting site of the next triangle and that there is no self-overlapping. A ribbon is called direct (resp. dual) if it consists of only direct (resp. dual) triangles, otherwise it is called proper. A closed ribbon is one with the starting and ending sites being the same. There are two types of ribbons: type-A ribbons, denoted as $\rho_A$, are the ones consisting of left-handed direct triangles and right-handed dual triangles, while type-B ribbons, denoted as $\rho_B$, are the ones consisting of right-handed direct triangles and left-handed dual triangles. See Fig.~\ref{fig:ribbon_type} for an illustration.  In the course of constructing the ribbon operators, we need to deal with different cases separately. 

\begin{figure}
    \centering
    \includegraphics[scale=1]{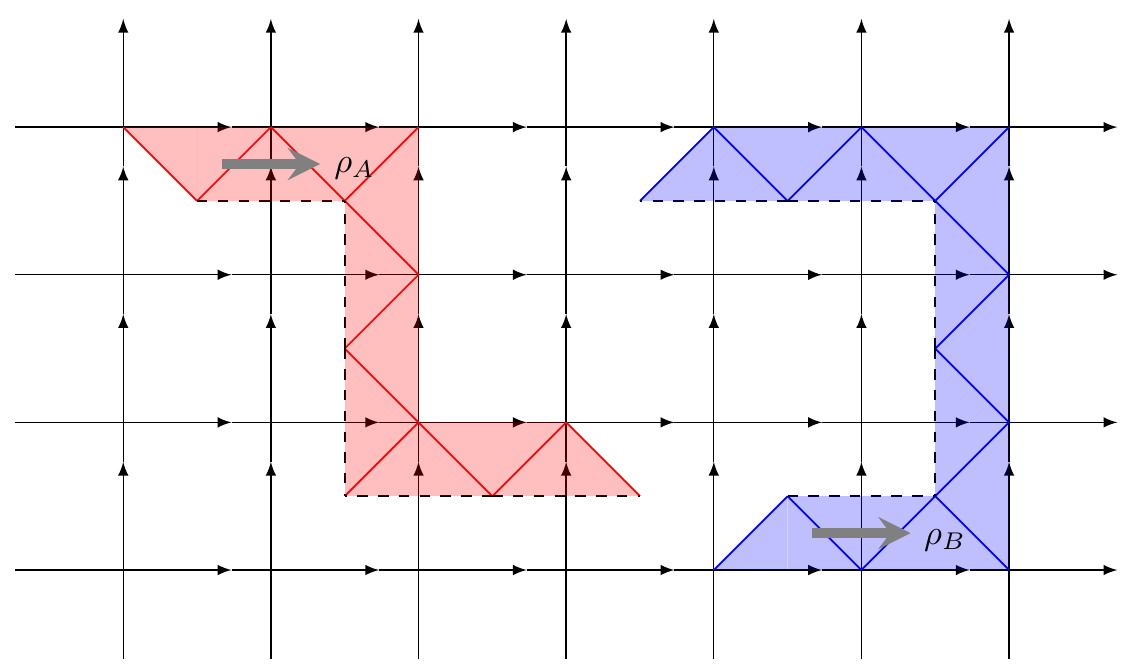}
    \caption{Ribbon types. A type-A ribbon $\rho_A$ (red) consists of left-handed direct triangles and right-handed dual triangles; a type-B ribbon $\rho_B$ (blue) consists of right-handed direct triangles and left-handed dual triangles.  \label{fig:ribbon_type}}
\end{figure}

First we define the triangle operators. We use the same convention as in Ref.~\cite{jia2022boundary} (note that this is slightly different from the one in Ref.~\cite{chen2021ribbon}). Specifically, the triangle operators are defined as follows: 
\begin{align}
    &F^{h,\varphi}( \tau_R^+)\vert x\rangle =  
    \varepsilon(h)  T_{-}^{\varphi} \vert x\rangle, \quad F^{h,\varphi}( \tau_R^-)\vert x\rangle = \varepsilon(h)  T_{+}^{\varphi} \vert x\rangle, \\
    &F^{h,\varphi}( \tau_L^-)\vert x\rangle =  
    \varepsilon(h)  \tilde{T}_{-}^{\varphi} \vert x\rangle, \quad F^{h,\varphi}( \tau_L^+)\vert x\rangle = 
    \varepsilon(h)  \tilde{T}_{+}^{\varphi} \vert x\rangle, \\
    &F^{h,\varphi}( \tilde{\tau}_R^+)\vert x\rangle =\hat{\varepsilon}(\varphi) L^{h}_{-}\vert x\rangle, \quad F^{h,\varphi}( \tilde{\tau}_R^-)\vert x\rangle =  
    \hat{\varepsilon}(\varphi) L^{h}_{+}\vert x\rangle, \\
    &F^{h,\varphi}( \tilde{\tau}_L^-)\vert x\rangle =\hat{\varepsilon}(\varphi)\tilde{L}_-^{h} \vert x\rangle, \quad F^{h,\varphi}( \tilde{\tau}_L^+)\vert x\rangle =   \hat{\varepsilon}(\varphi) \tilde{L}_{+}\vert x\rangle, 
\end{align}
for $x,\, h\in W$ and $\varphi\in\hat{W}$. We will see later that, by the above choice of convention, the vertex and face operators are special cases of ribbon operators.

Next we define the ribbon operators for a general ribbon. The strategy is the same as that in the Hopf algebra case \cite{chen2021ribbon}. The ribbon operators for type-B ribbons $\rho=\rho_B$ are parameterized by $h\otimes\varphi\in W^{\rm op}\otimes \hat{W}\simeq(\hat{W}^{\rm cop}\otimes W)^\vee$, denoted as $F^{h,\varphi}(\rho)$ or $F^{h,\varphi}_{\rho}$. It only acts non-trivially on edges contained in $\rho$. Since $\hat{W}^{\rm cop}\otimes W$ is an algebra, its dual has a coalgebra structure with comultiplication given by 
\begin{equation}
    \Delta(h\otimes \varphi) = \sum_k\sum_{(k),(h)}(h^{(1)}\otimes \hat{k})\otimes (S^{-1}(k^{(3)})h^{(2)}k^{(1)}\otimes \varphi(k^{(2)}\bullet )),
\end{equation}
where $\{k\}$ and $\{\hat{k}\}$ are dual bases of $W$ and $\hat{W}$. To define $F^{h,\varphi}(\rho)$, take a decomposition $\rho =  \rho_1\cup\rho_2$ such that both $\rho_1$ and $\rho_2$ have the same direction with $\rho$ and the terminal site of $\rho_1$ is the initial site of $\rho_2$. Then the ribbon operator $F^{h,\varphi}(\rho)$ is defined recursively by
\begin{equation} \label{eq:ribbon-operator}
    F^{h,\varphi}(\rho) = \sum_k\sum_{(k),(h)} F^{h^{(1)},\hat{k}}(\rho_1)F^{S^{-1}(k^{(3)})h^{(2)}k^{(1)},\varphi(k^{(2)}\bullet)}(\rho_2). 
\end{equation}
It is independent of the choice of the decomposition $\rho=\rho_1\cup \rho_2$ by the co-associativity of $(\hat{W}^{\rm cop}\otimes W)^\vee$. The ribbon operators for type-A ribbons are parameterized by $h\otimes \varphi\in W\otimes\hat{W}^{\rm op}\simeq (\hat{W}^{\rm cop}\otimes W)^{\vee,\rm op}$, and the construction is similar by the recursive formula \eqref{eq:ribbon-operator}.

Let $s=(v,f)$ be a site. There is a unique closed ribbon $\sigma_v$ of type-A surrounding $v$ which contains only dual triangles with one vertex $v$. Likewise, there is a unique closed ribbon $\sigma_f$ of type-B inside $f$ which contains only direct triangles with one vertex $\tilde{f}$, the vertex in the dual lattice corresponding to $f$. The directions of these closed ribbons are taken counterclockwise starting from the site $s$. Then one can see that $A^h(s) = F^{h,\varepsilon}(\sigma_v)$ and $B^{\varphi}(s) = F^{1_W,\varphi}$ for $h\in\operatorname{Cocom}(W)$ and $\varphi\in\operatorname{Cocom}(\hat{W})$. Thus, the vertex and face operators are special cases of ribbon operators, as expected. 

Note that $D(W)^\vee$ can be identified with a subspace of $W^{\rm op}\otimes \hat{W}\simeq(\hat{W}^{\rm cop}\otimes W)^\vee$ consisting of elements vanishing on the ideal $J$. With this identification, we have the following relations when the ribbon operators are parameterized by elements in $D(W)^\vee$. 

\begin{lemma} \label{lem:ribbon-local-ope}
    Suppose $\rho$ is a ribbon with initial site $s$. If $g\otimes \psi\in D(W)^\vee$, and $x\in W_L$, then 
    \begin{align}
        &F^{xg,\psi}(\rho)A^h(s)  = F^{g,\psi}(\rho)A^{S(x)h}(s),  \label{eq:flip-WL1} \\
        &F^{xg,\psi}(\rho)B^\varphi(s) = F^{g,\psi}(\rho)B^{\varphi\leftharpoonup x}(s). \label{eq:flip-WL}
    \end{align}
\end{lemma}

\begin{proof}
    See Appendix \ref{app:ribbon}. 
\end{proof}

\begin{lemma} \label{lem:comm-rel}
    Let $\rho=\rho_A$ or $\rho_B$ be an open ribbon with starting and ending sites $s_0,s_1$. Then the ribbon operator $F^{h,\varphi}(\rho)$ satisfies the following commutation relations with local operators: 
    \begin{itemize}
        \item[(1)] At the starting site $s_0$:
        \begin{align}
            A^g(s_0)F^{h,\varphi}(\rho_A)&=\sum_{(g)}F^{S^{-2}(g^{(1)})hS^{-1}(g^{(3)}),\varphi(S^{-1}(g^{(2)})\bullet)}(\rho_A)A^{g^{(4)}}(s_0),\label{eq:comm1}\\
            A^g(s_0)F^{h,\varphi}(\rho_B)&=\sum_{(g)}F^{S^{-2}(g^{(2)})hS^{-1}(g^{(4)}),\varphi(S^{-1}(g^{(3)})\bullet)}(\rho_B)A^{g^{(1)}}(s_0), \label{eq:comm2} \\
            B^{\psi}(s_0)F^{h,\varphi}(\rho_A)&=\sum_{(h)}F^{h^{(2)},\varphi}(\rho_A)B^{\psi(\bullet S^{-1}(h^{(1)}))}(s_0), \label{eq:comm3} \\
            B^{\psi}(s_0)F^{h,\varphi}(\rho_B)&=\sum_{(h)}F^{h^{(2)},\varphi}(\rho_B)B^{\psi(S(h^{(1)})\bullet)}(s_0).  \label{eq:comm4} 
        \end{align}
        \item[(2)] At the ending $s_1$: 
        \begin{align}
            A^g(s_1)F^{h,\varphi}(\rho_A)&=\sum_{(g)}F^{h,\varphi(\bullet S^{-2}(g^{(2)}))}(\rho_A)A^{g^{(1)}}(s_1), \label{eq:comm5} \\
            A^g(s_1)F^{h,\varphi}(\rho_B)&=\sum_{(g)}F^{h,\varphi(\bullet g^{(1)})}(\rho_B)A^{g^{(2)}}(s_1), \label{eq:comm6} \\
            B^{\psi}(s_1)F^{h,\varphi}(\rho_A)&=\sum_k\sum_{(k),(h)}\varphi(k^{(2)})F^{h^{(1)},\hat{k}}(\rho_A)B^{\psi(S^{-1}(k^{(3)})h^{(2)}k^{(1)} \bullet)}(s_1), \label{eq:comm7} \\
            B^{\psi}(s_1)F^{h,\varphi}(\rho_B)&=\sum_k\sum_{(k),(h)}\varphi(k^{(2)})F^{h^{(1)},\hat{k}}(\rho_B)B^{\psi(\bullet S^{-1}(k^{(3)})h^{(2)}k^{(1)})}(s_1).  \label{eq:comm8}  
        \end{align}
    \end{itemize}
\end{lemma}

\begin{proof}
    See Appendix \ref{app:ribbon}. 
\end{proof}

Now we show that the ribbon operators defined above are commuting with all stabilizer operators in the Hamiltonian except possibly those ones at the two ending sites of the underlying ribbons. 

\begin{proposition}\label{prop:comm-rib-local}
    Let $\rho$ be a ribbon. If $h\in\operatorname{Cocom}(W)$ and $\varphi\in\operatorname{Cocom}(\hat{W})$, then $A^h(s)$ and $B^{\varphi}(s)$ are commuting with $F^{g,\psi}(\rho)$ for all $s\neq \partial_0\rho,\partial_1\rho$ and $g\otimes \psi\in D(W)^\vee$:
    \begin{equation}
        A^h(s)F^{g,\psi}(\rho) = F^{g,\psi}(\rho) A^h(s),\quad B^\varphi(s)F^{g,\psi}(\rho) = F^{g,\psi}(\rho)B^\varphi(s). 
    \end{equation}
\end{proposition}

\begin{proof}
We prove the assertions for type-B ribbons. The proof for type-A ribbons is similar. Take the decomposition $\rho=\rho_1\cup\rho_2$ such that $\partial_1\rho_1=\partial_0\rho_2=s$. Then by Eqs.~\eqref{eq:comm2} and \eqref{eq:comm6}, we have 
\begin{align}
    &\quad A^h(s)F^{g,\psi}(\rho)  \nonumber \\
    & = \sum_{k,(k),(g)}A^h(s)F^{g^{(1)},\hat{k}}(\rho_1)F^{S^{-1}(k^{(3)})g^{(2)}k^{(1)},\psi(k^{(2)}\bullet)}(\rho_2)  \nonumber \\
    & = \sum_{k,(k),(g)}\sum_{(h)}F^{g^{(1)},\hat{k}(\bullet h^{(1)})}(\rho_1)A^{h^{(2)}}(s)F^{S^{-1}(k^{(3)})g^{(2)}k^{(1)},\psi(k^{(2)}\bullet)}(\rho_2)  \nonumber \\
    & = \sum_{k,(k),(g)}\sum_{(h)}F^{g^{(1)},\hat{k}(\bullet h^{(1)})}(\rho_1)F^{S^{-2}(h^{(3)})S^{-1}(k^{(3)})g^{(2)}k^{(1)}S^{-1}(h^{(5)}),\psi(k^{(2)}S^{-1}(h^{(4)})\bullet)}(\rho_2) A^{h^{(2)}}(s)  \nonumber \\
    & = \sum_{j,(j),(g)}\sum_{(h)}F^{g^{(1)},\hat{j}}(\rho_1)F^{S^{-2}(h^{(5)})S^{-1}(j^{(3)}h^{(3)})g^{(2)}j^{(1)}h^{(1)}S^{-1}(h^{(7)}),\psi(j^{(2)}h^{(2)}S^{-1}(h^{(6)})\bullet)}(\rho_2) A^{h^{(4)}}(s)  \nonumber \\
    & = \sum_{j,(j),(g)}\sum_{(h)}F^{g^{(1)},\hat{j}}(\rho_1)F^{S^{-2}(h^{(5)})S^{-1}(j^{(3)}h^{(3)})g^{(2)}j^{(1)}S^{-1}(\varepsilon_L(h^{(1)})),\psi(j^{(2)}h^{(2)}S^{-1}(h^{(6)})\bullet)}(\rho_2) A^{h^{(4)}}(s)  \nonumber \\
    & = \sum_{j,(j),(g)}\sum_{(h)}F^{g^{(1)},\hat{j}}(\rho_1)F^{S^{-2}(h^{(5)})S^{-1}(j^{(3)}h^{(3)})g^{(2)}j^{(1)},\psi(j^{(2)}\varepsilon_L(h^{(1)})h^{(2)}S^{-1}(h^{(6)})\bullet)}(\rho_2) A^{h^{(4)}}(s)  \nonumber \\
    & = \sum_{j,(j),(g)}\sum_{(h)}F^{g^{(1)},\hat{j}}(\rho_1)F^{S^{-2}(h^{(4)})S^{-1}(j^{(3)}h^{(2)})g^{(2)}j^{(1)},\psi(j^{(2)}S^{-1}(\varepsilon_L(h^{(1)}))\bullet)}(\rho_2) A^{h^{(3)}}(s)  \nonumber \\
    & = \sum_{j,(j),(g)}\sum_{(h)}F^{g^{(1)},\hat{j}}(\rho_1)F^{S^{-2}(h^{(4)})S^{-1}(j^{(3)}\varepsilon_L(h^{(1)})h^{(2)})g^{(2)}j^{(1)},\psi(j^{(2)}\bullet)}(\rho_2) A^{h^{(3)}}(s)  \nonumber \\
    & = \sum_{j,(j),(g)}\sum_{(h)}F^{g^{(1)},\hat{j}}(\rho_1)F^{S^{-2}(\varepsilon_L(h^{(1)}))S^{-1}(j^{(3)})g^{(2)}j^{(1)},\psi(j^{(2)}\bullet)}(\rho_2) A^{h^{(2)}}(s) \nonumber  \\
    & = \sum_{j,(j),(g)}\sum_{(h)}F^{g^{(1)},\hat{j}}(\rho_1)F^{S^{-1}(j^{(3)})g^{(2)}j^{(1)},\psi(j^{(2)}\bullet)}(\rho_2) A^{S^{-1}(\varepsilon_L(h^{(2)}))h^{(1)}}(s) \nonumber \\
    & = F^{g,\psi}(\rho) A^h(s). 
\end{align}
Here, in the fifth to ninth equalities, we used a cyclic rotation $h^{(i)}$ to $h^{(i+1)}$ since $h\in\operatorname{Cocom}(W)$, and the identity $\sum_{(x)}x^{(1)}y\otimes x^{(2)} = \sum_{(x)}x^{(1)}\otimes x^{(2)}S(y)$ for $y\in W_R$; the tenth equality follows from Eq.~\eqref{eq:flip-WL1}. 

Similarly, by Eqs.~\eqref{eq:comm4} and \eqref{eq:comm8}, we have 
\begin{align}
    &\quad  B^\varphi(s)F^{g,\psi}(\rho) \nonumber \\
    & = \sum_{k,(k),(g)} B^\varphi(s) F^{g^{(1)},\hat{k}}(\rho_1)F^{S^{-1}(k^{(3)})g^{(2)}k^{(1)},\psi(k^{(2)}\bullet)}(\rho_2) \nonumber  \\
    & = \sum_{k,(k),(g)}\sum_{j,(j)}\hat{k}(j^{(2)})  F^{g^{(1)},\hat{j}}(\rho_1) B^{\varphi(\bullet S^{-1}(j^{(3)})g^{(2)}j^{(1)})}(s) F^{S^{-1}(k^{(3)})g^{(3)}k^{(1)},\psi(k^{(2)}\bullet)}(\rho_2) \nonumber  \\
    & = \sum_{k,(k),(g)}\sum_{j,(j)}\hat{k}(j^{(2)})  F^{g^{(1)},\hat{j}}(\rho_1)  F^{S^{-1}(k^{(4)})g^{(4)}k^{(2)},\psi(k^{(3)}\bullet)}(\rho_2) \nonumber  B^{\varphi(S(k^{(1)})S(g^{(3)})k^{(5)}\bullet S^{-1}(j^{(3)})g^{(2)}j^{(1)})}(s)  \nonumber \\
    & = \sum_{(g)}\sum_{j,(j)}  F^{g^{(1)},\hat{j}}(\rho_1)  F^{S^{-1}(j^{(5)})g^{(4)}j^{(3)},\psi(j^{(4)}\bullet)}(\rho_2)B^{\varphi(S(j^{(2)})S(g^{(3)})j^{(6)}\bullet S^{-1}(j^{(7)})g^{(2)}j^{(1)})}(s) \nonumber  \\
    & = \sum_{(g)}\sum_{j,(j)}  F^{g^{(1)},\hat{j}}(\rho_1)  F^{S^{-1}(j^{(4)})g^{(4)}j^{(2)},\psi(j^{(3)}\bullet)}(\rho_2)B^{\varphi(S^{-1}(j^{(6)})g^{(2)}\varepsilon_L(j^{(1)})S(g^{(3)})j^{(5)}\bullet )}(s)  \nonumber \\
    & = \sum_{(g)}\sum_{j,(j)}  F^{g^{(1)},\hat{j}}(\rho_1)  F^{S^{-1}(j^{(4)})g^{(3)}\varepsilon_L(j^{(1)})j^{(2)},\psi(j^{(3)}\bullet)}(\rho_2)B^{\varphi(S^{-1}(j^{(6)})\varepsilon_L(g^{(2)})j^{(5)}\bullet )}(s)  \nonumber \\
    & = \sum_{(g)}\sum_{j,(j)}  F^{g^{(1)},\hat{j}}(\rho_1)  F^{S^{-1}(S(\varepsilon_L(g^{(2)}))j^{(3)})g^{(3)}j^{(1)},\psi(j^{(2)}\bullet)}(\rho_2)B^{\varphi\leftharpoonup S^{-1}(\varepsilon_R(j^{(4)}))}(s)  \nonumber \\
    & = \sum_{(g)}\sum_{j,(j)}  F^{g^{(1)},\hat{j}}(\rho_1)  F^{S^{-1}(\varepsilon_R(j^{(4)}))S^{-1}(j^{(3)})g^{(2)}j^{(1)},\psi(j^{(2)}\bullet)}(\rho_2)B^{\varphi}(s)  \nonumber \\
    & = \sum_{(g)}\sum_{j,(j)}  F^{g^{(1)},\hat{j}}(\rho_1)  F^{S^{-1}(j^{(3)})g^{(2)}j^{(1)},\psi(j^{(2)}\bullet)}(\rho_2)B^{\varphi}(s)  \nonumber \\
    & = F^{g,\psi}(\rho)B^\varphi(s). 
\end{align}
Here, the fifth equality follows from the fact that $\varphi\in\operatorname{Cocom}(\hat{W})$, the seventh from $\sum_{(y)}y^{(1)}\otimes xy^{(2)} = \sum_{(y)}S(x)y^{(1)}\otimes y^{(2)}$ for $x\in W_L$, and the eighth from Eq.~\eqref{eq:flip-WL}. 
\end{proof}

\emph{Ribbon operator and topological excitation.} ---
Let us now consider how to use the ribbon operators to create topological excitations.
We take the type-B ribbon as an example to illustrate the mechanism, and the type-A ribbon can be handled similarly.
Given a ribbon $\rho$, and ground state $|\Omega\rangle$, we define
\begin{equation}
    \Vcal_\rho = \{ |\Omega^{h,\varphi}_\rho\rangle = F^{h,\varphi}(\rho)|\Omega\rangle\,|\,h\otimes\varphi\in D(W)^\vee\}, 
\end{equation}
where we assume the ground state degeneracy is one for simplicity. 
For the general case, we can gather all the $|\Omega^{h,\varphi}_\rho\rangle$ corresponding to the ground states $|\Omega\rangle$ together and form such a set.
Denote the two ends of the ribbon as $s_i = \partial_i\rho$ ($i=0,1$). Then both the local operators $B^\varphi(s_0)A^h(s_0)$ and $B^\varphi(s_1)A^h(s_1)$ define representations of $D(W)$ on $\Vcal_\rho$. 
    This follows from Lemma \ref{lem:comm-rel} and $A^{h_W}|\Omega\rangle = |\Omega\rangle = B^{\varphi_{\hat{W}}}|\Omega\rangle$. 
This means that at two ends of the ribbon, there are two topological excitations created.

More precisely, by Proposition~\ref{prop:doubledRep}, we see that the total space $\mathcal{V}=\otimes_{e\in E(\Sigma)} W$ is a $D(W)$-module, and at each site $s$, there is a canonical representation of $D(W)$, $[\varphi\otimes h]\mapsto D^{\varphi\otimes h}(s)=B^\varphi(s)A^h(s)$.
A topological excitation is mathematically characterized by a representation $X$ of $D(W)$ (\emph{viz.}, $X$ is a $D(W)$-module).
If at a site $s$ there is a topological excitation $X$, we denote the corresponding $D(W)$-module as $\Vcal(s,X)$, which is a submodule of $\Vcal$.
By the Artin-Wedderburn theorem, it is clear that 
\begin{equation}
    \Vcal= \bigoplus_{X} N_X\Vcal(s,X),
\end{equation}
for a fixed site $s$, where $N_X$ is a multiplicity of $\Vcal(s,X)$ (which is an integer). For $n$ disjoint sites $s_1,\cdots,s_n$, we define  
\begin{equation}\label{eq:Vspace}
    \Vcal(s_1,X_1;\cdots;s_n,X_n)=\bigcap_{i=1}^n \Vcal(s_i,X_i).
\end{equation}
The vacuum sector of the weak Hopf quantum double model is 
\begin{equation}
    \mathcal{V}_{\mathds{1}}
    :=
    \{\,|\Omega\rangle\,|\,A^h(s)|\Omega\rangle=\varepsilon(h) |\Omega\rangle,   B^{\phi}(s)|\Omega\rangle =\hat{\varepsilon}(\phi) |\Omega\rangle, \forall\,h, \phi, \forall\,s\}.
\end{equation}
Using the notation of Eq.~\eqref{eq:Vspace}, the vacuum sector is 
\begin{equation}
    \Vcal_{\one}=\bigcap_{s\,:\, \text{site}} \Vcal(s,\one),
\end{equation}
\emph{viz.}, for all sites, the excitations are vacuum particle $\one$.
The case of particular interest for our discussion of ribbon operators is when there are two excitations
\begin{equation}
    \Vcal(s_0,X_0;s_1,X_1),
\end{equation}
where $s_0=\partial_0\rho$ and $s_1=\partial_1 \rho$ are two ends of a ribbon $\rho$.
For type-B ribbons, we have shown that the ribbon operator algebra $\mathcal{F}\cong D(W)^{\vee}$.
By the Artin-Wedderburn theorem for the quantum double $D(W)$ and treating $D(W)^{\vee}$ as a $D(W)|D(W)$-bimodule (the left and right actions are given by ``$\rightharpoonup$'' and ``$\leftharpoonup$'' respectively, see Eq.~\eqref{eq:poon} for the definition), we have the decomposition\,\footnote{This is also known as the Peter-Weyl theorem in the context of Lie group theory.}
\begin{equation}\label{eq:DWdecomp}
    D(W)^{\vee}=\bigoplus_{X\in \operatorname{Irr}(D(W))} X\otimes X^{\vee},
\end{equation}
where $\operatorname{Irr}(D(W))$ is the set of all equivalence classes of irreducible representations of $D(W)$.
Therefore $X\otimes X^{\vee}$ can be embedded into  $D(W)^{\vee}$. We denote the corresponding image in $D(W)^{\vee}$ as $\operatorname{Proj}_XD(W)^{\vee}$; it is clear that $D(W)^{\vee}=\oplus_{X\in \operatorname{Irr}(D(W))}\operatorname{Proj}_XD(W)^{\vee}$.
The ribbon operator that creates topological excitations $X$ and $X^{\vee}$ at two ends of the ribbon $\rho$ can be constructed by choosing elements $g\otimes \psi \in \operatorname{Proj}_XD(W)^{\vee}$, i.e., 
\begin{equation}
    F^{g,\psi}_{\rho}, g\otimes \psi \in \operatorname{Proj}_XD(W)^{\vee}.
\end{equation}
When acting on the spaces with two given excitations $X_1$ and $X_2$, this kind of ribbon operators will induce the fusion of anyons at two ends $\Vcal(s_1,X_0\otimes X;s_1, X_1\otimes X^{\vee})$.


\section{Algebraic theory of gapped boundaries and domain walls}
\label{sec:BdDomain}

The gapped boundary of the topological phases is a crucial topic in understanding the topological phase.
The holographic boundary-bulk duality \cite{kong2017boundary} asserts that all the bulk information can be recovered from the boundary by taking the center of the boundary phase.
In our previous work \cite{jia2022boundary}, we establish the algebraic theory of gapped boundary and domain wall for the general Hopf quantum double model which can recover many existing results for the finite group boundary \cite{bravyi1998quantum,Bombin2008family,Beigi2011the,Cong2017}.
Here we will generalize the results there to the weak Hopf quantum double model.

\subsection{Gapped boundary}
\label{sec:GBoundary}

For $2d$ topological phase with input unitary (multi-) fusion category $\EC$, the gapped boundary theory is determined by finite indecomposable module categories $\EM$ over $\EC$ \cite{Kitaev2012a,fuchs2013bicategories}.
The boundary topological excitations are given by the category of all $\EC$-module functors from $\EM$ to itself, $\mathsf{Fun}_{\EC}(\EM,\EM)$.
For a bulk weak Hopf phase $D(W)$, we need to find an algebraic structure that fits into this formalism.
Our main observation is that the gapped boundary is characterized by a $W$-comodule algebra $\mathfrak{A}$ or equivalently a $W$-module algebra $\mathfrak{M}$.
The mathematical discussion of these structures can be found in \cite{nill1998weak,bohm2000doi,nikshych2000duality,henker2011module}.

\begin{definition}
 Let $W$ be a weak Hopf algebra:
 
 (1) An algebra $\mathfrak{A}$ is called a right $W$-comodule algebra if there is a comodule map $\beta:\FA \to \FA\otimes W$ such that $\beta(xy)=\beta(x)\beta(y)$ and $\beta(1_{\FA})(x\otimes 1_W)=(\id_{\FA}\otimes \varepsilon_L)\comp \beta (x)$ for all $x,y\in \mathfrak{A}$. We will adopt the Sweedler's notation $\beta(x)=\sum_{(x)}x^{[0]}\otimes x^{[1]}$.\,\footnote{We use square brackets to denote the comodule structure to distinguish it from the comultiplication of weak Hopf algebra.} The left $W$-comodule algebra can be defined similarly.

 (2) An algebra $\FM$ is called a left $W$-module algebra if $\FM$ is a left $W$-module such that $h\triangleright (xy)=\sum_{(h)}(h^{(1)}\triangleright x) (h^{(2)}\triangleright y)$ and $h\triangleright 1_\FM=\varepsilon_L(h)\triangleright 1_\FM$ for all $h\in W$ and $x\in \FM$. The right $W$-module algebra can be defined similarly.
\end{definition}

\begin{table}[t]
\centering \small 
\begin{tabular} {|l|c|c|} 
\hline
   &Weak Hopf quantum double model & String-net model   \\ \hline
 Bulk  & Hopf algebra $W$ &  $\EC=\mathsf{Rep}(W)$  \\ \hline
 Bulk phase  & 
 $\ED=\mathsf{Rep}(D(W))$ & $\mathsf{Fun}_{\EC |\EC}(\EC,\EC) $  \\ \hline
 Boundary & $W$-comodule algebra $\mathfrak{A}$ & $\EC$-module category ${_{\mathfrak{A}}}\EM={_{\mathfrak{A}}}\mathsf{Mod}$ \\\hline
 Boundary phase & $\EB\simeq {_{\mathfrak{A}}}\mathsf{Mod}_{\mathfrak{A}}$ & $\mathsf{Fun}_{\EC}({_{\mathfrak{A}}}\EM,{_{\mathfrak{A}}}\EM)$ \\\hline
Boundary defect & ${_{\mathfrak{B}}}\mathsf{Mod}_{\mathfrak{A}}$ & $\mathsf{Fun}_{\EC}({_{\mathfrak{A}}}\EM,{_{\mathfrak{B}}}\EM)$ \\\hline
\end{tabular}
\caption{The dictionary between $W$-comodule algebra description of weak Hopf quantum double boundary and string-net boundary.\label{tab:bdTop1}}
\end{table}

Let us first consider how to describe the gapped boundary via $W$-comodule algebra.
In the following, we assume $\EC=\Rep(W)$.
For the gapped boundary determined by a $W$-comodule algebra $\mathfrak{A}$, its corresponding $\EC$-module is given by the category ${_{\FA}}\Mod$ of finite dimensional $\FA$-modules.
It is proved in \cite[Lemma 10.1.1]{henker2011module} that ${_{\FA}}\Mod$ is a $\EC$-module category.
First notice that ${_\mathfrak{A}}\mathsf{Mod}$ is finite semisimple, where the simple objects are just simple $\mathfrak{A}$-modules. Since $\dim \mathfrak{A} < \infty$, there are finite simple objects up to equivalence, thus ${_{\FA}}\Mod$ is finite.
To make ${_{\FA}}\Mod$ a $\EC$-module, we need to introduce a bifunctor $\otimes:\EC\times {_{\FA}}\Mod\to {_{\FA}}\Mod$.
To this end, notice that there is a left $\mathfrak{A}$-module structure over $X\otimes M$ with $X\in \EC$ and $M\in \AMod$. More explicitly, the structure map $\mu_{X\otimes M}:\mathfrak{A}\otimes (X\otimes M)\to X\otimes M$ is given by (in diagrammatic representation)
\begin{equation}
    \mu_{X\otimes M}= \begin{aligned}
        \begin{tikzpicture}
			\draw[black, line width=1.0pt] (-0.1,-0.8) arc (180:360:0.3);
			\draw[black, line width=1.0pt] (1,0) arc (180:90:0.5);
			\draw[black, line width=1.0pt] (-0.1,0) arc (180:90:0.6);
			\draw[black, line width=1.0pt] (-0.1,0)--(-0.1,-0.8);
			\braid[
				width=0.5cm,
				height=0.3cm,
				line width=1.0pt,
				style strands={1}{black},
				style strands={2}{black}] (Kevin)
				s_1^{-1} ;
	    	\draw[black, line width=1.0pt] (1.0,-0.8)--(1.0,-1.35);
	    	\draw[black, line width=1.0pt] (0.2,-1.1)--(0.2,-1.35);
	    	\draw[black, line width=1.0pt] (1.5,0.9)--(1.5,-1.35);
	    	 \draw[black, line width=1.0pt] (0.5,0)--(0.5,0.9);
	    	\node[ line width=0.2pt, dashed, draw opacity=0.5] (a) at (0.2,-1.6){$\mathfrak{A}$};
	    	\node[ line width=0.2pt, dashed, draw opacity=0.5] (a) at (1,-1.6){$X$};
	    	\node[ line width=0.2pt, dashed, draw opacity=0.5] (a) at (1.5,-1.6){$M$};
	    	\node[ line width=0.2pt, dashed, draw opacity=0.5] (a) at (0.55,1.1){$X$};
	    	\node[ line width=0.2pt, dashed, draw opacity=0.5] (a) at (1.5,1.1){$M$};
	     	\node[ line width=0.2pt, dashed, draw opacity=0.5] (a) at (0.7,-1){$\mathfrak{A}$};
			\node[ line width=0.2pt, dashed, draw opacity=0.5] (a) at (-0.4,-1){$W$};
			\node[ line width=0.2pt, dashed, draw opacity=0.5] (a) at (0.2,-1.1){$\bullet$};
			\node[ line width=0.2pt, dashed, draw opacity=0.5] (a) at (0.5,0.6){$\bullet$};
			\node[ line width=0.2pt, dashed, draw opacity=0.5] (a) at (1.5,0.5){$\bullet$};
			\end{tikzpicture}
    \end{aligned},
\end{equation}
\emph{viz.}, $\mu_{X\otimes M}=(\mu_X\otimes \mu_M)\comp(\id_W\otimes \tau_{\mathfrak{A},X} \otimes \id_{M}) \comp (\beta_{\mathfrak{A}} \otimes \id_X\otimes \id_M)$, where $\mu_X$ and $\mu_M$ are the $W$-module structure and $\mathfrak{A}$-module structure maps of $X$ and $M$ respectively, $\beta_{\mathfrak{A}}$ is the $W$-comodule structure map of $\mathfrak{A}$, and $\tau_{\mathfrak{A},X}$ is the swap map.
The tensor product is defined as the submodule
\begin{equation}
    X\otimes_{\FA} M:=\{ x\,|\,\beta(1)\,x=x\}=\sum_{(1)}1^{[1]}X\otimes 1^{[0]}M.
\end{equation}
The associator and identity morphisms can be defined naturally.

If $\FA$ and $\FB$ are Morita equivalent, then $\AMod\simeq \BMod$, hence they give the same gapped boundary.
Therefore, the gapped boundaries are classified by Morita equivalent classes of $W$-comodule algebras.
The boundary excitation can be regarded as a point defect between two boundaries of the same type.
Thus we can consider the more general defects between two gapped boundaries $\AMod$ and $\BMod$, where the defects are classified by $\EC$-module functors between $\AMod$ and $\BMod$.
An Eilenberg-Watts type theorem \cite{eilenberg1960abstract,watts1960intrinsic} shows that these boundary defects can equivalently be classified by the $\FB|\FA$-bimodules.

\begin{theorem}
    The gapped boundaries of a weak Hopf quantum double phase $D(W)$ are characterized by $W$-comodule algebras $\FA$.
    The boundary excitations are described by the category of $\FA|\FA$-bimodules.
    For two gapped boundaries determined by $\FA,\FB$, the point defects between them are classified by the $\FB|\FA$-bimodules. The results are summarized in Table \ref{tab:bdTop1}.
\end{theorem}

\begin{proof}
    This is proved in \cite[Lemma 10.1.2]{henker2011module}.
\end{proof}

For Hopf algebra $W=H$ case, we know that for any indecomposable module category $\EM$ over the representation category $\Rep(H)$ of the Hopf algebra, there is a corresponding $H$-comodule algebra $\mathfrak{A}$ such that $\EM \simeq \AMod$.
However, for weak Hopf algebra, this direction is still open.
Nevertheless, for any Morita equivalence of $W$-comodule algebra, we have a corresponding gapped boundary theory.
The main results of this formalism of gapped boundary are summarized in Table~\ref{tab:bdTop1}.
Similar to what we have done for Hopf quantum double phase, the gapped boundary can also be described by $W$-module algebra. We will not repeat the discussion we gave in \cite{jia2022boundary} here and just summarize the results in Table \ref{tab:bdTop2}.

\begin{table}[t]
\centering \small 
\begin{tabular} {|l|c|c|} 
\hline
   &Weak Hopf quantum double model & String-net model   \\ \hline
 Bulk  & Hopf algebra $W$ & UFC $\EC=\mathsf{Rep}(W)$  \\ \hline
 Bulk phase  & 
 $\ED=\mathsf{Rep}(D(W))$ & $\mathsf{Fun}_{\EC|\EC}(\EC,\EC) $  \\ \hline
 Boundary & $W$-module algebra $\mathfrak{M}$ & $\EC_{\FM}$ \\\hline
 Boundary phase & $\EB\simeq  {_{\FM}}\EC_{\FM}$ & $\mathsf{Fun}_{\EC}(\EC_{\FM},\EC_{\FM})$ \\\hline
Boundary defect & ${_{\FN}}\EC_{\FM}$ & $\mathsf{Fun}_{\EC}(\EC
_{\FM},\EC_{\FN})$ \\\hline
\end{tabular}
\caption{The dictionary between $W$-module algebra description of weak Hopf quantum double boundary and string-net boundary.\label{tab:bdTop2}}
\end{table}

\subsection{Gapped domain wall}

Consider the domain wall that separates two $2d$ topological phases with input unitary (multi-) fusion categories $\EC_1,\EC_2$, the domain wall theory is defined by a $\EC_1|\EC_2$-bimodule category $\EM$ \cite{Kitaev2012a,fuchs2013bicategories}.
The topological boundary excitations are given by the category of all $\EC_1|\EC_2$-bimodule functors from $\EM$ to itself, $\mathsf{Fun}_{\EC_1|\EC_2}(\EM,\EM)$.

Here in our case, two weak Hopf quantum double phases are determined by $\EC_i=\Rep(W_i)$, $i=1,2$.
To generalize our previous gapped domain wall theory \cite{jia2022boundary}, we introduce the concept of weak Hopf bicomodule algebra. For the weak Hopf case, we need a similar definition. This seems not to be systematically discussed before. In this part, we outline some results, and a more comprehensive and rigorous discussion of weak Hopf bicomodule algebra will be given elsewhere \cite{Jia2023bicomodule}.

\begin{definition}
    Let $W_1,W_2$ be two weak Hopf algebras. An algebra $\FA$ is called a $W_1|W_2$-bicomodule algebra if $\FA$ is a 
    left $W_1$-comodule algebra with comodule map $\alpha$ and a right $W_2$-comodule algebra with comodule map $\beta$ such that $(\id_{W_1}\otimes \beta)\comp \alpha =(\alpha \otimes \id_{W_2})\comp \beta$. A bicomodule algebra map is an algebra map intertwining both the left and right comodule maps.
\end{definition}

For a given bicomodule algebra $\FA$,  the category ${_{\mathfrak{A}}}\mathsf{Mod}$ of finite dimensional modules is a $\mathsf{Rep}(W_1)|\mathsf{Rep}(W_2)$-bimodule category.
The domain wall excitations are classified by the $\EC_1|\EC_2$-bimodule functors. 

\begin{lemma}
   For a $W_1|W_2$ bicomodule algebra $\FA$, the category $\AMod$ of left $\FA$-modules is a $\Rep(W_1)|\Rep(W_2)$-bimodule category.
\end{lemma}
    
\begin{proof}
From the discussion in Sec.~\ref{sec:GBoundary}, $X_1\otimes M$ is a left $\FA$-module for all $X_1\in \Rep(W_1)$ and $M\in \AMod$, and the tensor product is defined as the submodule $X_1\otimes_{\FA}M$. The associator is defined as the identity, and the unit isomorphism is given by $\sum_{i}x_i\otimes m_i \mapsto \sum_i \varepsilon(x_i)m_i$.
Thus $\AMod$ is a left $\Rep(W_1)$-module.

Similarly, $M\otimes X_2$ is a left $\FA$-module for $M\in \AMod$ and $X_2\in \Rep(W_2)$ (since $\FA$ is a right $W_2$-comodule), with the structure map given by (in diagrammatic representation)
\begin{equation}
    \mu_{M\otimes X_2}= \begin{aligned}
        \begin{tikzpicture}
			\draw[black, line width=1.0pt] (-0.1,-0.8) arc (180:360:0.3);
			\draw[black, line width=1.0pt] (1,0) arc (180:90:0.5);
			\draw[black, line width=1.0pt] (-0.1,0) arc (180:90:0.6);
			\draw[black, line width=1.0pt] (-0.1,0)--(-0.1,-0.8);
			\braid[
				width=0.5cm,
				height=0.3cm,
				line width=1.0pt,
				style strands={1}{black},
				style strands={2}{black}] (Kevin)
				s_1^{-1} ;
	    	\draw[black, line width=1.0pt] (1.0,-0.8)--(1.0,-1.35);
	    	\draw[black, line width=1.0pt] (0.2,-1.1)--(0.2,-1.35);
	    	\draw[black, line width=1.0pt] (1.5,0.9)--(1.5,-1.35);
	    	 \draw[black, line width=1.0pt] (0.5,0)--(0.5,0.9);
	    	\node[ line width=0.2pt, dashed, draw opacity=0.5] (a) at (0.2,-1.6){$\mathfrak{A}$};
	    	\node[ line width=0.2pt, dashed, draw opacity=0.5] (a) at (1,-1.6){$M$};
	    	\node[ line width=0.2pt, dashed, draw opacity=0.5] (a) at (1.5,-1.6){$X_2$};
	    	\node[ line width=0.2pt, dashed, draw opacity=0.5] (a) at (0.55,1.1){$M$};
	    	\node[ line width=0.2pt, dashed, draw opacity=0.5] (a) at (1.5,1.1){$X_2$};
	     	\node[ line width=0.2pt, dashed, draw opacity=0.5] (a) at (0.7,-1.2){$W_2$};
			\node[ line width=0.2pt, dashed, draw opacity=0.5] (a) at (-0.4,-1){$\FA$};
			\node[ line width=0.2pt, dashed, draw opacity=0.5] (a) at (0.2,-1.1){$\bullet$};
			\node[ line width=0.2pt, dashed, draw opacity=0.5] (a) at (0.5,0.6){$\bullet$};
			\node[ line width=0.2pt, dashed, draw opacity=0.5] (a) at (1.5,0.5){$\bullet$};
			\end{tikzpicture}
    \end{aligned}.
\end{equation}
The tensor product is defined as the submodule $M\otimes_{\FA} X_2$, the associator is the identity, and the unit isomorphism is given by $\sum_{i}m_i\otimes y_i \mapsto \sum_i m_i\varepsilon(y_i)$.
In this way, $\AMod$ is a right $\Rep(W_2)$-module.

For $X_i\in \mathsf{Rep}(W_i)$ and $M\in {_{\mathfrak{A}}}\mathsf{Mod}$,  $X_1\otimes M \otimes X_2$ is also an $\mathfrak{A}$-module.
In fact, it is easy to verify that the module structure map is given by (in diagrammatic representation)
\begin{equation}
    \mu_{X_1\otimes M \otimes X_2}= \begin{aligned}
        \begin{tikzpicture}
		\draw[black, line width=1.0pt] (0.5,0.5) arc (90:180:0.5);
		\draw[black, line width=1.0pt] (1,-1.4) arc (0:-180:0.5);
		\draw[black, line width=1.0pt] (1,0) arc (180:90:0.5);
		\draw[black, line width=1.0pt] (1.5,0)--(1.5,.8);
		\draw[black, line width=1.0pt] (2,0) arc (180:90:0.5);
	    \braid[width=0.5cm,
				height=0.3cm,
				line width=1.0pt,number of strands=4] (braid) a_1^{-1} a_3^{-1} a_2^{-1};
	    	\draw[black, line width=1.0pt] (0.5,0)--(.5,.8);
	    	\draw[black, line width=1.0pt] (0,0)--(0,-1.4);
	    	\draw[black, line width=1.0pt] (2.5,-2.2)--(2.5,0.8);
	    	\draw[black, line width=1.0pt] (0.5,-2.2)--(0.5,-1.3);
	    	\draw[black, line width=1.0pt] (2,-1.4)--(2,-2.2);
	    	\draw[black, line width=1.0pt] (1.5,-1.4)--(1.5,-2.2);
	    	\node[ line width=0.2pt, dashed, draw opacity=0.5] (a) at (0.5,-2.5){$\mathfrak{A}$};
	    	\node[ line width=0.2pt, dashed, draw opacity=0.5] (a) at (0,-2.1){$W_1$};
	    	\node[ line width=0.2pt, dashed, draw opacity=0.5] (a) at (1,-2.1){$W_2$};
	    	\node[ line width=0.2pt, dashed, draw opacity=0.5] (a) at (1.5,-2.5){$X_1$};
	    	\node[ line width=0.2pt, dashed, draw opacity=0.5] (a) at (2,-2.49){$M$};
	    	\node[ line width=0.2pt, dashed, draw opacity=0.5] (a) at (2.5,-2.5){$X_2$};
	    	\node[ line width=0.2pt, dashed, draw opacity=0.5] (a) at (0.55,1.1){$X_1$};
	    	\node[ line width=0.2pt, dashed, draw opacity=0.5] (a) at (1.5,1.12){$M$};
	    	\node[ line width=0.2pt, dashed, draw opacity=0.5] (a) at (2.5,1.1){$X_2$};
	     	\node[ line width=0.2pt, dashed, draw opacity=0.5] (a) at (0.7,-1){$\mathfrak{A}$};
	     	\node[ line width=0.2pt, dashed, draw opacity=0.5] (a) at (0.5,-1.9){$\bullet$};
	    	\node[ line width=0.2pt, dashed, draw opacity=0.5] (a) at (0.5,0.5){$\bullet$};
	    	\node[ line width=0.2pt, dashed, draw opacity=0.5] (a) at (-0.2,0.4){$\mu_{X_1}$};
	    	\node[ line width=0.2pt, dashed, draw opacity=0.5] (a) at (1.5,0.5){$\bullet$};
	    	\node[ line width=0.2pt, dashed, draw opacity=0.5] (a) at (1.85,0.4){$\mu_{M}$};
	    	\node[ line width=0.2pt, dashed, draw opacity=0.5] (a) at (2.5,0.5){$\bullet$};
	    	\node[ line width=0.2pt, dashed, draw opacity=0.5] (a) at (2.9,0.4){$\mu_{X_2}$};
	     	\node[ line width=0.2pt, dashed, draw opacity=0.5] (a) at (0.25,-1.5){$\beta_{\mathfrak{A}}$};
			\end{tikzpicture}
    \end{aligned}.
\end{equation}
The tensor product is defined as $X_1\otimes_{\FA}M\otimes_{\FA} X_2$. Thus the middle associator is also defined as the identity.
This proves that $\AMod$ is a $\Rep(W_1)|\Rep(W_2)$-bimodule category.
\end{proof}

\begin{proposition}
   The topological excitations of the gapped domain wall determined by a $W_1|W_2$-bicomodule algebra $\FA$ are equivalently described by the following categories: 
   \begin{enumerate}
    \item The category of $\mathsf{Rep}(W_1)|\mathsf{Rep}(W_2)$-bimodule functors $\mathsf{Fun}_{\mathsf{Rep}(W_1)|\mathsf{Rep}(W_2)}({_{\FA}}\mathsf{Mod},{_{\FA}}\mathsf{Mod})$. 
    \item The category of $\mathsf{Rep}(W_1\otimes W_2^{\rm cop})$-module functors $ \mathsf{Fun}_{\mathsf{Rep}(W_1\otimes W_2^{\rm cop})} ({_{\FA}}\mathsf{Mod},{_{\FA}}\mathsf{Mod})$. 
    \item The category of all $\FA|\FA$-bimodules ${_{\FA}}\mathsf{Mod}_{\FA}$.
   \end{enumerate}
\end{proposition}

Instead of proving the above proposition about topological excitations on the domain wall, we will prove a more general result about the point defects on the domain wall. The topological excitations can be regarded as a special case of point defects.

\begin{proposition}
   The point defects between two gapped domain walls determined by $W_1|W_2$-bicomodule algebras $\FA,\FB$ are equivalently described by the following categories: 
   \begin{enumerate}
    \item The category of $\mathsf{Rep}(W_1)|\mathsf{Rep}(W_2)$-bimodule functors $\mathsf{Fun}_{\mathsf{Rep}(W_1)|\mathsf{Rep}(W_2)}({_{\FA}}\mathsf{Mod},{_{\FB}}\mathsf{Mod})$. 
    \item The category of $\mathsf{Rep}(W_1\otimes W_2^{\rm cop})$-module functors $\mathsf{Fun}_{\mathsf{Rep}(W_1\otimes W_2^{\rm cop})} ({_{\FA}}\mathsf{Mod},{_{\FB}}\mathsf{Mod})$. 
    \item The category of all $\FB|\FA$-bimodules ${_{\FB}}\mathsf{Mod}_{\FA}$.
   \end{enumerate}
\end{proposition}

\begin{proof}
    1 $\Leftrightarrow$ 2: This is clear from the facts that, $\Rep(W^{\rm cop})=\Rep(W)^{\otimes \rm op}$, a $\EC_1|\EC_2$-bimodule category is equivalent to a left $\EC_1\boxtimes\EC_2^{\rm \otimes \rm op}$-module category, and $\Rep(W_1\otimes W_2^{\rm cop})=\Rep(W_1)\boxtimes \Rep(W_2^{\rm cop})$.

     2 $\Leftrightarrow$ 3: This is ensured by Eilenberg-Watts theorem \cite{eilenberg1960abstract,watts1960intrinsic}.
\end{proof}

Notice that using the folding trick, a gapped domain wall can be transformed into a gapped boundary. And the gapped boundary can be regarded as a special gapped domain wall that separates the quantum double phase from the trivial phase. Our established theory of boundary and domain wall matches well with this correspondence.


\section{Hamiltonian realization of gapped boundaries and domain walls}
\label{sec:bdHam}

In this section, let us give the explicit lattice Hamiltonian of the gapped boundary and domain wall of the weak Hopf quantum double model.
We want to stress that various existing boundary Hamiltonian models of the quantum double  model are incomplete \cite{Bombin2008family,Beigi2011the,Cong2017,jia2022boundary}, namely, they cannot cover all possible boundaries.
From our discussion in the previous section, we see that to cover all possible gapped boundaries of the quantum double model $D(W)$, we need to build the boundary Hamiltonian for all $W$-comodule algebras $\FA$ such that the boundary excitation realizes the bimodule category ${_{\FA}}\Mod_{\FA}$.
For the domain wall, a similar issue exists.
This part will give a complete solution to these problems in the most general settings.

\subsection{Gapped boundary Hamiltonian}

We have argued that all gapped boundaries of quantum double model $D(W)$ are characterized by $W$-comodule algebras $\FA$.
We now construct the corresponding lattice model for a given $\FA$.
Notice that our construction based on the generalized quantum double in \cite{jia2022boundary} can be straightforwardly generalized to the weak Hopf quantum double model.
Here, we will present a different but more intuitive construction.

Without loss of generality, we only consider the surface $\Sigma$ with one boundary $\partial \Sigma$.
To each bulk edge $e$, we still have $\mathcal{H}_e=W$, while for boundary edge $e_b$, we set $\mathcal{H}_{e_b}=\FA$.
Notice that here the boundary edge is directed: if the boundary face $f$ is on the left-hand side of $e_b$, $\FA$ should be chosen as a left $W$-comodule algebra, otherwise, it should be chosen as a right $W$-comodule algebra. 
This is because we need to introduce a pairing between $x\in \FA$ and $\varphi\in \What$. When $f$ in on the left-hand side of $e_b$, $x\leftharpoonup \hat{S}(\phi)$ is well defined only when $\FA$ is a left $W$-comodule algebra. In this case, $x^{[1]}\in W$, and we have $x\leftharpoonup \hat{S}(\phi)=\sum \varphi(S(x^{[1]}))x^{[0]}$.
Similarly, when $f$ is on the right-hand side of $e_b$, $\FA$ must be a right $W$-comodule algebra.
In the following, we will assume that the boundary direction is chosen so that the bulk is always on the right-hand side of the boundary, thus we assign all boundary edges the right $W$-comodule algebra $\FA$.

To build the boundary model, we need the notion of \emph{symmetric separability idempotent} of an algebra $\FA$ (see  \cite{aguiar2000note,koppen2020defects}).
By definition, it is an element $\lambda=\sum_{<\lambda>}\lambda^{<1>}\otimes \lambda^{<2>}\in \FA\otimes \FA$ such that 
\begin{enumerate}
    \item $\sum_{<\lambda>}x\lambda^{<1>}\otimes \lambda^{<2>} =\sum_{<\lambda>}  \lambda^{<1>}\otimes \lambda^{<2>} x$ for all $x\in \FA$; 
    \item $\sum_{<\lambda>}\lambda^{<1>} \lambda^{<2>} =1$; 
    \item  $\sum_{<\lambda>}\lambda^{<1>}\otimes \lambda^{<2>}=\sum_{<\lambda>}\lambda^{<2>}\otimes \lambda^{<1>}$.
\end{enumerate}
It can be proved that $\lambda$ is an idempotent of the enveloping algebra $\FA\otimes \FA^{\rm op}$.
If $K$ is a weak Hopf algebra with Haar integral $h_K$, it is easy to check that $\lambda=\sum_{(h_K)} h_K^{(1)}\otimes S(h_K^{(2)})$ is a symmetric separability idempotent.
The existence and uniqueness of the symmetric separability idempotent for a finite-dimensional semisimple algebra over an algebraically closed field of characteristic zero are proved in \cite[Corollary 3.1]{aguiar2000note}.
It is also proved that \cite[Proposition 19]{koppen2020defects}
\begin{equation}\label{eq:lambda}
    \sum \lambda^{<1>[0]}\otimes \lambda^{<2>[0]} \otimes \lambda^{<1>[1]}\lambda^{<2>[1]}=\sum \lambda^{<1>}\otimes \lambda^{<2>}\otimes 1, 
\end{equation}
that is, $\sum \beta(\lambda^{<1>})_{13}\beta(\lambda^{<2>})_{23} = \lambda_{12}$ in $\mathfrak{A}\otimes \mathfrak{A}\otimes W$.

For each boundary edge $e_b$, we propose the following edge operator  $E^{z\otimes w}_{e_b}$ for $z\otimes w \in \FA\otimes \FA$:
\begin{equation}
    E^{z\otimes w}_{e_b} |x\rangle =|zxw\rangle.
\end{equation}
For general $\alpha \in \FA\otimes \FA$, using linear expansion and the above expression we can define $E^{\alpha}_{e_b}$.
It is clear that
\begin{equation}\label{eq:WHATN}
    E^{z\otimes w}_{e_b} E^{g\otimes k}_{e_b}=E^{zg\otimes k w}_{e_b}.
\end{equation}
Thus it forms a representation of the enveloping algebra $\FA\otimes \FA^{\rm op}$ of $\mathfrak{A}$.

Let $\lambda$ be a symmetric separability idempotent. Then
\begin{equation}
E^{\lambda}_{e_b}|x\rangle=\sum_{<\lambda>}
|\lambda^{<1>} x \lambda^{<2>} \rangle
\end{equation}
is the boundary local stabilizer that we need to construct the boundary lattice model.
From the definition of $\lambda$, we see that $(E^{\lambda}_{e_b})^{2}=E^{\lambda}_{e_b}$.

Another local operator we need is the boundary vertex operator.
Since we assume that the direction of all boundary edges is chosen so that the bulk is on the right-hand side of the boundary edge, there are only two configurations depending on the bulk edges connecting to the boundary vertex.
More precisely, we define the following operator for $z\otimes w\in \FA \otimes \FA^{\rm op}$:
\begin{align}
     \begin{aligned}
    \begin{tikzpicture}
				\draw[-latex,black,line width = 1.6pt] (0,0) -- (0,1);
    	        \draw[red,line width = 1pt] (0,1) -- (0.5,0.5); 
    			\draw[-latex,black,line width = 1.6pt] (0,1) -- (0,2); 
				\draw[-latex,black] (1,1) -- (0,1); 
				\node[ line width=0.2pt, dashed, draw opacity=0.5] (a) at (-0.4,0.5){$x$};
    		\node[ line width=0.2pt, dashed, draw opacity=0.5] (a) at (-0.4,1.5){$y$};
            \node[ line width=0.2pt, dashed, draw opacity=0.5] (a) at (0.5,1.3){$h$};
            \node[ line width=0.2pt, dashed, draw opacity=0.5] (a) at (0.7,0.5){$s_b$}; 
	\end{tikzpicture}  \end{aligned} & \quad \quad  A^{z\otimes w}(s_b)|x,y,h\rangle=\sum_{[w]} |zx,yw^{[0]},S^{-1}(w^{[1]})h\rangle, \label{eq:bdd-sta-A1}\\
       \begin{aligned}
    \begin{tikzpicture}
				\draw[-latex,black,line width = 1.6pt] (0,0) -- (0,1);
    	        \draw[red,line width = 1pt] (0,1) -- (0.5,1.5); 
    			\draw[-latex,black,line width = 1.6pt] (0,1) -- (0,2); 
				\draw[-latex,black] (1,1) -- (0,1); 
				\node[ line width=0.2pt, dashed, draw opacity=0.5] (a) at (-0.4,0.5){$x$};
    		\node[ line width=0.2pt, dashed, draw opacity=0.5] (a) at (-0.4,1.5){$y$};
            \node[ line width=0.2pt, dashed, draw opacity=0.5] (a) at (0.7,1.4){$s_b$};
            \node[ line width=0.2pt, dashed, draw opacity=0.5] (a) at (0.5,0.75){$h$};
	\end{tikzpicture} \end{aligned} & \quad \quad A^{z\otimes w}(s_b)|x,y,h\rangle=\sum_{[z]} | z^{[0]}x, yw, z^{[1]}h\rangle,  \label{eq:bdd-sta-A2}\\
       \begin{aligned}
    \begin{tikzpicture}
				\draw[-latex,black,line width = 1.6pt] (0,0) -- (0,1); 
        	    \draw[red,line width = 1pt] (0,1) -- (0.5,0.5); 
    			\draw[-latex,black,line width = 1.6pt] (0,1) -- (0,2); 
				\draw[-latex,black] (0,1) -- (1,1); 
				\node[ line width=0.2pt, dashed, draw opacity=0.5] (a) at (-0.4,0.5){$x$};
    		\node[ line width=0.2pt, dashed, draw opacity=0.5] (a) at (-0.4,1.5){$y$};
            \node[ line width=0.2pt, dashed, draw opacity=0.5] (a) at (0.5,1.3){$h$};
            \node[ line width=0.2pt, dashed, draw opacity=0.5] (a) at (0.7,0.5){$s_b$};
	\end{tikzpicture}\end{aligned} & \quad\quad A^{z\otimes w} (s_b)|x,y,h\rangle=\sum_{[w]}|zx, yw^{[0]},hw^{[1]}\rangle, \label{eq:bdd-sta-A3} \\
      \begin{aligned} \begin{tikzpicture}
				\draw[-latex,black,line width = 1.6pt] (0,0) -- (0,1); 
        	    \draw[red,line width = 1pt] (0,1) -- (0.5,1.5); 
    			\draw[-latex,black,line width = 1.6pt] (0,1) -- (0,2); 
				\draw[-latex,black] (0,1) -- (1,1); 
				\node[ line width=0.2pt, dashed, draw opacity=0.5] (a) at (-0.4,0.5){$x$};
    		\node[ line width=0.2pt, dashed, draw opacity=0.5] (a) at (-0.4,1.5){$y$};
            \node[ line width=0.2pt, dashed, draw opacity=0.5] (a) at (0.5,0.75){$h$};
            \node[ line width=0.2pt, dashed, draw opacity=0.5] (a) at (0.7,1.4){$s_b$};
	\end{tikzpicture}  \end{aligned}  &\quad \quad A^{z\otimes w} (s_b)|x,y,h\rangle=\sum_{[z]}|z^{[0]}x, yw,hS^{-1}(z^{[1]})\rangle. \label{eq:bdd-sta-A4} 
\end{align}
The convention here follows from Eqs.~\eqref{eq:L1}-\eqref{eq:L2}.
It is clear that the $A_{v_b}^{z\otimes w}$'s form a representation of $\FA\otimes \FA^{\rm op}$.
The local stabilizers are defined as $A^{\lambda}(s_b)$ with $\lambda$ the symmetric separability idempotent, and $B_{f_b}=B^{\varphi_{\hat{W}}}(s_b)$.
The boundary Hamiltonian is thus of the form
\begin{equation}
  H[C(\partial \Sigma)]=-\sum_{s_b}   A^{\lambda}({s_b})-\sum_{f_b}B_{f_b}.
\end{equation}
Notice that for the gapped boundary determined by a weak Hopf subalgebra $K\subseteq W$, $K$ is naturally a $W$-comodule algebra.
Using the connection between symmetric separability idempotent $\lambda$ and Haar integral $h$, i.e., $\lambda=\sum_{(h)}h^{(1)}\otimes S(h^{(2)})$, we can obtain the counterpart of the model given in our previous work \cite{jia2022boundary} for Hopf algebra quantum double boundary. 

\begin{remark}
    Compare this with the construction given in \cite{Bombin2008family,Cong2017}, where they try to introduce the local boundary edge term and local boundary face and vertex terms to model the gapped boundary.
    Our model here is simpler and more natural in our framework since the $\FA|\FA$-bimodule structure is encoded in a way just as that for the bulk quantum double on a given site.
    The philosophy is also different. In their construction, they are searching for a local algebra $\mathcal{A}$ such that the boundary excitation is just the category of the representation category $\Rep(\mathcal{A})$.
    In the algebraic framework we established in the last section, we see that the boundary excitation is just the $\FA|\FA$-bimodule category ${_\FA}\mathsf{Mod}_\FA$. Our construction fits better into our framework.
\end{remark}

\begin{remark}
    For the group quantum double boundary, the construction given by Beigi et al.~\cite{Beigi2011the} can be recovered in our framework.
    Notice that in that case, the boundary is determined by a subgroup $K\subset G$.
    $\FA=\mathbb{C}[K]$ is a $W=\mathbb{C}[G]$-comodule algebra.
    The local stabilizer $A^K=\frac{1}{|K|} \sum_{k\in K}A^k$.
    Using the fact that symmetric separability idempotent $\lambda=\frac{1}{|K|}\sum_{k\in K} k\otimes S(k)$, it is clear that $A^K=A^{\lambda}_{v_b}$.
\end{remark}

\begin{remark}
    Here we choose the comodule algebra formalism of gapped boundary. For the module algebra description, we have a similar construction. The only thing that needs to be changed is that in the module algebra case, the boundary local stabilizer is built from symmetric separability idempotents.
\end{remark}

\subsection{Gapped domain wall Hamiltonian}

As we have pointed out before, a domain wall can be understood equivalently as a boundary via the folding trick.
Nevertheless, it is worth spelling out the explicit data for the construction of a gapped domain wall.
The domain wall between $W_1$ and $W_2$ is equivalent to a boundary of $W_1\otimes W_2^{\rm cop}$.
Thus the domain wall is given by a $W_1\otimes W_2^{\rm cop}$-comodule algebra $\FA$.
Choose the symmetric separability idempotent $\lambda$, we can construct a corresponding $\FA$-boundary of $W_1\otimes W_2^{\rm cop}$ phase.
See from domain wall settings, we have a left $W_1$- right $W_2$- comodule algebra $\FA$.
We construct the domain wall vertex operator  $A^{z\otimes w}(s_{d,i})$
as follows (for clarity, we have changed the tensor product order when necessary):
\begin{align}
     \begin{aligned}
    \begin{tikzpicture}
				\draw[-latex,black,line width = 1.6pt] (0,0) -- (0,1);
    	        \draw[red,line width = 1pt] (0,1) -- (0.5,0.5); 
             \draw[red,line width = 1pt] (0,1) -- (-0.5,0.5); 
    			\draw[-latex,black,line width = 1.6pt] (0,1) -- (0,2); 
				\draw[-latex,black] (1,1) -- (0,1); 
                    \draw[-latex,black] (0,1) -- (-1,1); 
				\node[ line width=0.2pt, dashed, draw opacity=0.5] (a) at (-0.2,0.5){$x$};
    		\node[ line width=0.2pt, dashed, draw opacity=0.5] (a) at (-0.2,1.5){$y$};
            \node[ line width=0.2pt, dashed, draw opacity=0.5] (a) at (0.5,1.2){$h$};
            \node[ line width=0.2pt, dashed, draw opacity=0.5] (a) at (-0.5,1.2){$k$};
            \node[ line width=0.2pt, dashed, draw opacity=0.5] (a) at (0.85,0.4){$s_{d,2}$}; 
            \node[ line width=0.2pt, dashed, draw opacity=0.5] (a) at (-0.8,0.4){$s_{d,1}$}; 
	\end{tikzpicture}  \end{aligned} & \quad \quad  A^{z\otimes w}(s_{d,i})|x,k,y,h\rangle=\sum_{[w]} |zx,kw^{[-1]},yw^{[0]},S^{-1}(w^{[1]})h\rangle, \label{eq:bdd-sta-D1}\\
       \begin{aligned}
    \begin{tikzpicture}
				\draw[-latex,black,line width = 1.6pt] (0,0) -- (0,1);
    	        \draw[red,line width = 1pt] (0,1) -- (0.5,1.5); 
             \draw[red,line width = 1pt] (0,1) -- (-0.5,1.5); 
    			\draw[-latex,black,line width = 1.6pt] (0,1) -- (0,2); 
				\draw[-latex,black] (1,1) -- (0,1); 
                    \draw[-latex,black] (0,1) -- (-1,1); 
				\node[ line width=0.2pt, dashed, draw opacity=0.5] (a) at (-0.2,0.5){$x$};
    		\node[ line width=0.2pt, dashed, draw opacity=0.5] (a) at (-0.2,1.5){$y$};
            \node[ line width=0.2pt, dashed, draw opacity=0.5] (a) at (0.85,1.4){$s_{d,2}$};
            \node[ line width=0.2pt, dashed, draw opacity=0.5] (a) at (-0.8,1.4){$s_{d,1}$};
            \node[ line width=0.2pt, dashed, draw opacity=0.5] (a) at (0.5,0.8){$h$};
            \node[ line width=0.2pt, dashed, draw opacity=0.5] (a) at (-0.5,0.8){$k$};
	\end{tikzpicture} \end{aligned} & \quad \quad A^{z\otimes w}(s_{d,i})|x,k,y,h\rangle=\sum_{[z]} | z^{[0]}x, kS^{-1}(z^{[-1]}), yw,z^{[1]}h\rangle,  \label{eq:bdd-sta-D2}\\
       \begin{aligned}
    \begin{tikzpicture}
				\draw[-latex,black,line width = 1.6pt] (0,0) -- (0,1); 
        	    \draw[red,line width = 1pt] (0,1) -- (0.5,0.5); 
                \draw[red,line width = 1pt] (0,1) -- (-0.5,0.5); 
    			\draw[-latex,black,line width = 1.6pt] (0,1) -- (0,2); 
				\draw[-latex,black] (0,1) -- (1,1); 
                    \draw[-latex,black] (-1,1) -- (0,1); 
				\node[ line width=0.2pt, dashed, draw opacity=0.5] (a) at (-0.2,0.5){$x$};
    		\node[ line width=0.2pt, dashed, draw opacity=0.5] (a) at (-0.2,1.5){$y$};
            \node[ line width=0.2pt, dashed, draw opacity=0.5] (a) at (0.5,1.2){$h$};
            \node[ line width=0.2pt, dashed, draw opacity=0.5] (a) at (-0.5,1.2){$k$};
            \node[ line width=0.2pt, dashed, draw opacity=0.5] (a) at (0.85,0.4){$s_{d,2}$}; 
            \node[ line width=0.2pt, dashed, draw opacity=0.5] (a) at (-0.8,0.4){$s_{d,1}$};
	\end{tikzpicture}\end{aligned} & \quad\quad A^{z\otimes w} (s_{d,i})|x,k,y,h\rangle=\sum_{[w]}|zx, S^{-1}(w^{[-1]})k, yw^{[0]},hw^{[1]}\rangle, \label{eq:bdd-sta-D3} \\
      \begin{aligned} \begin{tikzpicture}
				\draw[-latex,black,line width = 1.6pt] (0,0) -- (0,1); 
        	    \draw[red,line width = 1pt] (0,1) -- (0.5,1.5); 
                    \draw[red,line width = 1pt] (0,1) -- (-0.5,1.5); 
    			\draw[-latex,black,line width = 1.6pt] (0,1) -- (0,2); 
				\draw[-latex,black] (0,1) -- (1,1); 
                    \draw[-latex,black] (-1,1) -- (0,1);
				\node[ line width=0.2pt, dashed, draw opacity=0.5] (a) at (-0.2,0.5){$x$};
    		\node[ line width=0.2pt, dashed, draw opacity=0.5] (a) at (-0.2,1.5){$y$};
            \node[ line width=0.2pt, dashed, draw opacity=0.5] (a) at (0.5,0.75){$h$};
            \node[ line width=0.2pt, dashed, draw opacity=0.5] (a) at (-0.5,0.8){$k$};
            \node[ line width=0.2pt, dashed, draw opacity=0.5] (a) at (0.85,1.4){$s_{d,2}$};
            \node[ line width=0.2pt, dashed, draw opacity=0.5] (a) at (-0.8,1.4){$s_{d,1}$};
	\end{tikzpicture}  \end{aligned}  &\quad \quad A^{z\otimes w} (s_{d,i})|x,k,y,h\rangle=\sum_{[z]}|z^{[0]}x,z^{[-1]}k, yw,hS^{-1}(z^{[1]})\rangle, \label{eq:bdd-sta-D4} \\ \begin{aligned}
    \begin{tikzpicture}
				\draw[-latex,black,line width = 1.6pt] (0,0) -- (0,1);
    	        \draw[red,line width = 1pt] (0,1) -- (0.5,0.5); 
             \draw[red,line width = 1pt] (0,1) -- (-0.5,0.5); 
    			\draw[-latex,black,line width = 1.6pt] (0,1) -- (0,2); 
				\draw[-latex,black] (0,1) -- (1,1); 
                    \draw[-latex,black] (0,1) -- (-1,1); 
				\node[ line width=0.2pt, dashed, draw opacity=0.5] (a) at (-0.2,0.5){$x$};
    		\node[ line width=0.2pt, dashed, draw opacity=0.5] (a) at (-0.2,1.5){$y$};
            \node[ line width=0.2pt, dashed, draw opacity=0.5] (a) at (0.5,1.2){$h$};
            \node[ line width=0.2pt, dashed, draw opacity=0.5] (a) at (-0.5,1.2){$k$};
            \node[ line width=0.2pt, dashed, draw opacity=0.5] (a) at (0.85,0.4){$s_{d,2}$}; 
            \node[ line width=0.2pt, dashed, draw opacity=0.5] (a) at (-0.8,0.4){$s_{d,1}$}; 
	\end{tikzpicture}  \end{aligned} & \quad \quad  A^{z\otimes w}(s_{d,i})|x,k,y,h\rangle=\sum_{[w]} |zx,kw^{[-1]},yw^{[0]},hw^{[1]}\rangle, \label{eq:bdd-sta-D5}\\
       \begin{aligned}
    \begin{tikzpicture}
				\draw[-latex,black,line width = 1.6pt] (0,0) -- (0,1);
    	        \draw[red,line width = 1pt] (0,1) -- (0.5,1.5); 
             \draw[red,line width = 1pt] (0,1) -- (-0.5,1.5); 
    			\draw[-latex,black,line width = 1.6pt] (0,1) -- (0,2); 
				\draw[-latex,black] (0,1) -- (1,1); 
                    \draw[-latex,black] (0,1) -- (-1,1); 
				\node[ line width=0.2pt, dashed, draw opacity=0.5] (a) at (-0.2,0.5){$x$};
    		\node[ line width=0.2pt, dashed, draw opacity=0.5] (a) at (-0.2,1.5){$y$};
            \node[ line width=0.2pt, dashed, draw opacity=0.5] (a) at (0.85,1.4){$s_{d,2}$};
            \node[ line width=0.2pt, dashed, draw opacity=0.5] (a) at (-0.8,1.4){$s_{d,1}$};
            \node[ line width=0.2pt, dashed, draw opacity=0.5] (a) at (0.5,0.8){$h$};
            \node[ line width=0.2pt, dashed, draw opacity=0.5] (a) at (-0.5,0.8){$k$};
	\end{tikzpicture} \end{aligned} & \quad \quad A^{z\otimes w}(s_{d,i})|x,k,y,h\rangle=\sum_{[z]} |z^{[0]}x, kS^{-1}(z^{[-1]}), yw,hS^{-1}(z^{[1]})\rangle,  \label{eq:bdd-sta-D6}\\
       \begin{aligned}
    \begin{tikzpicture}
				\draw[-latex,black,line width = 1.6pt] (0,0) -- (0,1); 
        	    \draw[red,line width = 1pt] (0,1) -- (0.5,0.5); 
                \draw[red,line width = 1pt] (0,1) -- (-0.5,0.5); 
    			\draw[-latex,black,line width = 1.6pt] (0,1) -- (0,2); 
				\draw[-latex,black] (1,1) -- (0,1); 
                    \draw[-latex,black] (-1,1) -- (0,1); 
				\node[ line width=0.2pt, dashed, draw opacity=0.5] (a) at (-0.2,0.5){$x$};
    		\node[ line width=0.2pt, dashed, draw opacity=0.5] (a) at (-0.2,1.5){$y$};
            \node[ line width=0.2pt, dashed, draw opacity=0.5] (a) at (0.5,1.2){$h$};
            \node[ line width=0.2pt, dashed, draw opacity=0.5] (a) at (-0.5,1.2){$k$};
            \node[ line width=0.2pt, dashed, draw opacity=0.5] (a) at (0.85,0.4){$s_{d,2}$}; 
            \node[ line width=0.2pt, dashed, draw opacity=0.5] (a) at (-0.8,0.4){$s_{d,1}$};
	\end{tikzpicture}\end{aligned} & \quad\quad A^{z\otimes w} (s_{d,i})|x,k,y,h\rangle=\sum_{[w]}|zx, S^{-1}(w^{[-1]})k, yw^{[0]},S^{-1}(w^{[1]})h\rangle, \label{eq:bdd-sta-D7} \\
      \begin{aligned} \begin{tikzpicture}
				\draw[-latex,black,line width = 1.6pt] (0,0) -- (0,1); 
        	    \draw[red,line width = 1pt] (0,1) -- (0.5,1.5); 
                    \draw[red,line width = 1pt] (0,1) -- (-0.5,1.5); 
    			\draw[-latex,black,line width = 1.6pt] (0,1) -- (0,2); 
				\draw[-latex,black] (1,1) -- (0,1); 
                    \draw[-latex,black] (-1,1) -- (0,1);
				\node[ line width=0.2pt, dashed, draw opacity=0.5] (a) at (-0.2,0.5){$x$};
    		\node[ line width=0.2pt, dashed, draw opacity=0.5] (a) at (-0.2,1.5){$y$};
            \node[ line width=0.2pt, dashed, draw opacity=0.5] (a) at (0.5,0.75){$h$};
            \node[ line width=0.2pt, dashed, draw opacity=0.5] (a) at (-0.5,0.8){$k$};
            \node[ line width=0.2pt, dashed, draw opacity=0.5] (a) at (0.85,1.4){$s_{d,2}$};
            \node[ line width=0.2pt, dashed, draw opacity=0.5] (a) at (-0.8,1.4){$s_{d,1}$};
	\end{tikzpicture}  \end{aligned}  &\quad \quad A^{z\otimes w} (s_{d,i})|x,k,y,h\rangle=\sum_{[z]}|z^{[0]}x,z^{[-1]}k, yw,z^{[1]}h\rangle. \label{eq:bdd-sta-D8} 
\end{align}
Notice that in the above expressions, we assume that $k\in W_1$, $h\in W_2$ and $x,y,z,w\in \FA$.
The local stabilizers are defined as  $A^{\lambda}(s_{d,i})$ and $B_{f_{d,i}}$ ($i=1,2$), and the domain wall Hamiltonian is 
\begin{equation}
    H[C(\Sigma_d)]=-\sum_{i=1,2}\sum_{s_{d,i}} A^{\lambda}(s_{d_i})-\sum_{i=1,2}\sum_{f_{d,i}}B_{f_{d,i}}.
\end{equation}
If one of the two bulks is set as a trivial phase, the domain wall reduces to a boundary.


\section{Weak Hopf tensor network states}
\label{sec:HopfTN}
In this section, we shall introduce the weak Hopf tensor network based on pairings and solve the weak Hopf quantum double model with weak Hopf tensor network states.
Our discussion here inherits the philosophy of the construction in \cite{Buerschaper2013a,girelli2021semidual,jia2022boundary}, that is, introducing face and edge labels and pairing the face labels with edge labels. Although the construction is similar, the proof is more complicated in this situation.
There are also alternative approaches to the topological tensor network which is mainly based on the vertex freedom \cite{schuch2010peps,lootens2021matrix,molnar2022matrix}. If we consider the dual lattice model, our tensor network will share some similarities with these vertex-based constructions.

\subsection{Weak Hopf tensor networks}

Let us first introduce the weak Hopf tensor network and then in the sequel subsections, we show how to use this tool to solve the weak Hopf quantum double model. The construction is similar to that of \cite{Buerschaper2013a,jia2022boundary}, however, as we will see, the calculation and proof are more complicated in this case. The results here can also be generalized to the semidual case \cite{girelli2021semidual}.

\vspace{1em}
\emph{Local rank-2 tensor from the pairing of weak Hopf algebras.}\,\,\textemdash\, Our construction of the weak Hopf tensor network is based on pairings between weak Hopf algebras.
Let $J,W$ be two weak Hopf algebras equipped with a pairing $\langle \bullet, \bullet \rangle: J\otimes W\to \mathbb{C}$.
The basic building block of the weak Hopf tensor network is the rank-2  tensor $\Psi(x,\phi)=\langle \phi,x\rangle$ with $x\in W, \phi\in J$, which is represented by  
	\begin{equation}
		\langle \phi, x\rangle=	\begin{aligned}
	    	\begin{tikzpicture}
	    	\fill[red, opacity=0.8] (-0.4,0) rectangle (0.4,0.4); 
			\draw[-latex,black] (-1,0) -- (1,0);
			\draw[-latex,blue] (-0.9,0.7) arc(180:270:0.5);
			\draw[-latex,blue] (0.4,0.2) arc(-90:0:0.5);
			\node[ line width=0.2pt, dashed, draw opacity=0.5] (a) at (0,0.7){$\phi$};
			\node[ line width=0.2pt, dashed, draw opacity=0.5] (a) at (1.3,0){$x$};
					\end{tikzpicture}
				\end{aligned}
				=
				\begin{aligned}
	    	\begin{tikzpicture}
	    	\fill[red, opacity=0.8] (-0.4,0) rectangle (0.4,0.4); 
			\draw[-latex,black] (1,0) -- (-1,0);
			\draw[-latex,blue] (-0.9,0.7) arc(180:270:0.5);
			\draw[-latex,blue] (0.4,0.2) arc(-90:0:0.5);
			\node[ line width=0.2pt, dashed, draw opacity=0.5] (a) at (0,0.7){$\phi$};
			\node[ line width=0.2pt, dashed, draw opacity=0.5] (a) at (1.6,0){$S(x)$};
					\end{tikzpicture}
				\end{aligned},
			\end{equation}
   
		\begin{equation}
		\langle \phi, x\rangle=		\begin{aligned}
	    	\begin{tikzpicture}
	    	\fill[red, opacity=0.8] (-0.4,0) rectangle (0.4,0.4); 
			\draw[-latex,black] (-1,0) -- (1,0);
			\draw[-latex,blue] (-0.9,0.7) arc(180:270:0.5);
			\draw[-latex,blue] (0.4,0.2) arc(-90:0:0.5);
			\node[ line width=0.2pt, dashed, draw opacity=0.5] (a) at (0,0.7){$\phi$};
			\node[ line width=0.2pt, dashed, draw opacity=0.5] (a) at (1.3,0){$x$};
					\end{tikzpicture}
				\end{aligned}
				=
			\begin{aligned}
	    	\begin{tikzpicture}
	    	\fill[red, opacity=0.8] (-0.4,0) rectangle (0.4,0.4); 
			\draw[-latex,black] (-1,0) -- (1,0);
			\draw[-latex,blue] (-0.4,0.2) arc(-90:-180:0.5);
			\draw[-latex,blue] (0.9,0.7) arc(0:-90:0.5);
			\node[ line width=0.2pt, dashed, draw opacity=0.5] (a) at (0,0.7){$S(\phi)$};
			\node[ line width=0.2pt, dashed, draw opacity=0.5] (a) at (1.3,0){$x$};
					\end{tikzpicture}
				\end{aligned}.
			\end{equation}
The black edge is labeled with $x\in W$ and the reversing of the edge is realized by applying the antipode of $W$ on $x$. 
Similarly, the blue circular edge is labeled with $\phi\in J$ and the reversing of the edge is realized by applying the antipode of $J$ on $\phi$.
The Hopf tensor network representation is a diagrammatic representation of the pairing. For arbitrary given rank-2 tensors with their corresponding label, the evaluation is determined by pairing, and this process will be called Hopf trace.

\vspace{1em}
\emph{Gluing Local rank-2 tensors.}\,\,\textemdash\, Using the above local rank-2 tensor, we can construct a topological tensor network by introducing the gluing process of these local tensors.
Consider two weak Hopf algebras $J$ and $I$	which both have their pairings with weak Hopf algebra $W$, we can contract the rank-2 tensors to obtain rank-3 tensors. 
There are two types of basic gluing processes. The first one is parallel gluing, for which the contraction is determined by the coproduct of $W$.
For example, let $\phi \in J$, $\psi \in I$ and $x \in W$, we define 
			\begin{equation}\label{eq:pglue}
			\begin{aligned}
	    	\begin{tikzpicture}
	    	\fill[red, opacity=0.8] (-0.4,-0.4) rectangle (0.4,0.4); 
			\draw[-latex,black] (-1,0) -- (1,0);
			\draw[-latex,blue] (-0.9,0.7) arc(180:270:0.5);
			\draw[-latex,blue] (0.4,0.2) arc(-90:0:0.5);
			\draw[-latex,blue] (-0.4,-0.2) arc(90:180:0.5);
			\draw[-latex,blue] (0.9,-0.7) arc(0:90:0.5);
			\node[ line width=0.2pt, dashed, draw opacity=0.5] (a) at (0,0.7){$\phi$};
			\node[ line width=0.2pt, dashed, draw opacity=0.5] (a) at (0,-0.7){$\psi$};
			\node[ line width=0.2pt, dashed, draw opacity=0.5] (a) at (1.3,0){$x$};
					\end{tikzpicture}
				\end{aligned}
				:=\sum_{(x)}
			\begin{aligned}
	    	\begin{tikzpicture}
	    	\fill[red, opacity=0.8] (-0.4,0) rectangle (0.4,0.4); 
			\draw[-latex,black] (-1,0) -- (1,0);
		\draw[-latex,blue] (-0.9,0.7) arc(180:270:0.5);
			\draw[-latex,blue] (0.4,0.2) arc(-90:0:0.5);
			\node[ line width=0.2pt, dashed, draw opacity=0.5] (a) at (0,0.7){$\phi$};
			\node[ line width=0.2pt, dashed, draw opacity=0.5] (a) at (1.3,0.2){$x^{(2)}$};
			 \fill[red, opacity=0.8] (-0.4,-0.6) rectangle (0.4,-0.2); 
			\draw[-latex,black] (-1,-0.2) -- (1,-0.2);
			\draw[-latex,blue] (-0.4,-0.4) arc(90:180:0.5);
			\draw[-latex,blue] (0.9,-0.9) arc(0:90:0.5);
			\node[ line width=0.2pt, dashed, draw opacity=0.5] (a) at (0,-0.9){$\psi$};
			\node[ line width=0.2pt, dashed, draw opacity=0.5] (a) at (1.3,-0.4){$x^{(1)}$};
					\end{tikzpicture}
				\end{aligned}
				= \sum_{(x)}\langle \psi, S(x^{(1)}) \rangle \langle \phi, x^{(2)}\rangle.
			\end{equation}
The second one is vertical gluing, for which the contraction is determined by the coproduct of $J$. For example, let $\phi \in J$ and $x,y\in W$, we have
		\begin{equation}
		\begin{aligned}
	    	\begin{tikzpicture}
	    	\fill[red, opacity=0.8] (-0.4,0) rectangle (0.4,0.4); 
			\draw[-latex,black] (-1,0) -- (1.1,0);
			\draw[-latex,blue] (-0.9,0.7) arc(180:270:0.5);
			\draw[-latex,blue] (0.4,0.2) arc(-90:0:0.5);
			\node[ line width=0.2pt, dashed, draw opacity=0.5] (a) at (0,1.1){$\phi$};
			\node[ line width=0.2pt, dashed, draw opacity=0.5] (a) at (0.8,-0.4){$x$};
			\fill[red, opacity=0.8] (0.7,0.7) rectangle (1.1,1.5); 
			\draw[-latex,black] (1.1,0) -- (1.1,2);
			\draw[-latex,blue] (0.9,1.5) arc(0:90:0.5);			\node[ line width=0.2pt, dashed, draw opacity=0.5] (a) at (1.4,1.8){$y$};
			\end{tikzpicture}
				\end{aligned}
				:=
				\sum_{(\phi)}
			\begin{aligned}
	    	\begin{tikzpicture}
	    	\fill[red, opacity=0.8] (-0.4,0) rectangle (0.4,0.4); 
	    	\fill[red, opacity=0.8] (1.3,1.0) rectangle (1.7,1.8); 
			\draw[-latex,black] (-1,0) -- (1,0);
			\draw[-latex,black] (1.7,0.4) -- (1.7,2.4);
			\draw[-latex,blue] (-0.9,0.7) arc(180:270:0.5);
			\draw[-latex,blue] (1.5,0.4) -- (1.5,1.0);
			\draw[-latex,blue] (0.4,0.2) -- (1,0.2);
			\draw[-latex,blue] (1.5,1.8) arc(0:90:0.5);
			\node[ line width=0.2pt, dashed, draw opacity=0.5] (a) at (0,0.7){$\phi^{(1)}$};
			\node[ line width=0.2pt, dashed, draw opacity=0.5] (a) at (0.9,1.4){$\phi^{(2)}$};
			\node[ line width=0.2pt, dashed, draw opacity=0.5] (a) at (1.3,0){$x$};
			\node[ line width=0.2pt, dashed, draw opacity=0.5] (a) at (2.0,2.0){$y$};
					\end{tikzpicture}
				\end{aligned}=	\sum_{(\phi)}\langle \phi^{(1)},x\rangle \langle \phi^{(2)} ,y\rangle. \label{eq:vglue}
			\end{equation}

\begin{figure}
    \centering
    \includegraphics[scale=0.6]{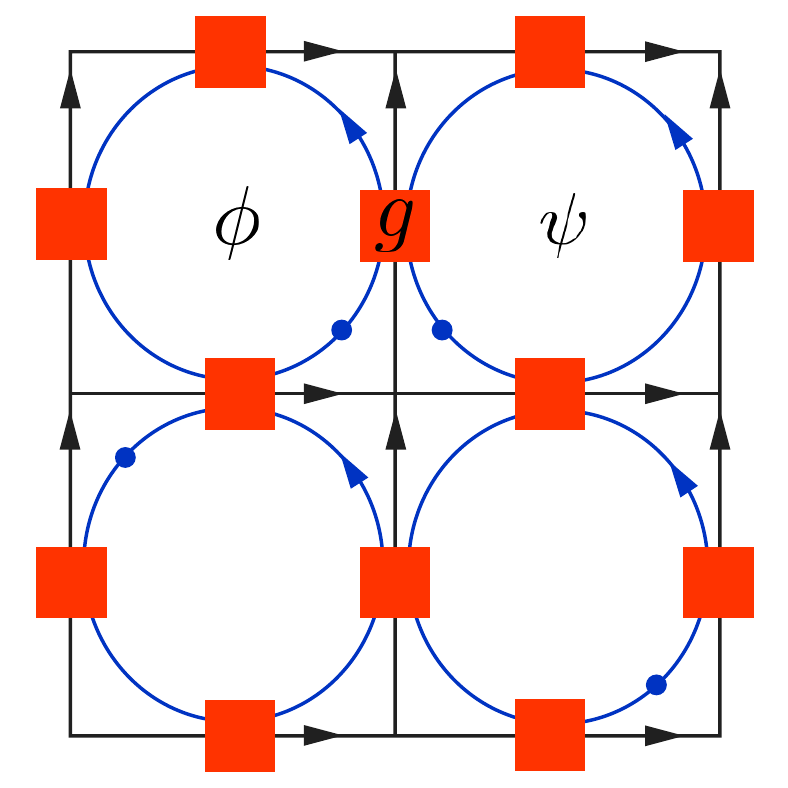}
    \caption{Example of the weak Hopf tensor network on a square lattice. The black lattice represents the physical lattice, the blue circle in the face $f$ is labeled with a face element $\phi_f \in J_f$, and the red box on the edge is labeled with an edge element $g\in W$. There is a pairing between $J$ and $W$.}
    \label{fig:WTN}
\end{figure}
   
\vspace{1em}
\emph{Weak Hopf tensor network on a lattice.}\,\,\textemdash\,\, With the above preparation, we are now ready to introduce the weak Hopf tensor network on a lattice $\Gamma$.
We assign a weak Hopf algebra $W_e$ to each edge $e\in E(\Gamma)$ and a weak Hopf algebra $J_f$ to each face $f\in F(\Gamma)$.
For an edge $e$, two face algebras $J_{f_L}$ and $J_{f_R}$ neighboring $e$ have their corresponding pairings with $W_e$.
There is one more extra data necessary. When dealing with a face, the starting site determines the starting point of the ordered pairing.
See Fig.~\ref{fig:WTN}, where the starting point is drawn as a blue dot. The evaluation rule is given as follows:
\begin{equation}
      		\begin{aligned}
	    	\begin{tikzpicture}
	    	\fill[red, opacity=0.8] (-0.4,0) rectangle (0.4,0.4); 
			\draw[-latex,black] (-1,0) -- (1.1,0);
                \draw[-latex,black] (-1,0) -- (-1,2.2);
			\draw[-latex,blue] (-0.9,0.7) arc(180:270:0.45);
			\draw[-latex,blue] (0.4,0.2) arc(-90:0:0.45);
			\node[ line width=0.2pt, dashed, draw opacity=0.5] (a) at (0,1.1){$\phi$};
			\node[ line width=0.2pt, dashed, draw opacity=0.5] (a) at (0.1,-0.4){$w$};
   	        \node[ line width=0.2pt, dashed, draw opacity=0.5] (a) at (0.1,2.4){$y$};
			\fill[red, opacity=0.8] (0.7,0.7) rectangle (1.1,1.5);
                \fill[red, opacity=0.8] (-1,0.7) rectangle (-0.6,1.5);
                \fill[red, opacity=0.8] (-0.4,1.8) rectangle (0.4,2.2);
			\draw[-latex,black] (1.1,0) -- (1.1,2.2);
                \draw[-latex,black] (-1,2.2) -- (1.1,2.2);
			\draw[-latex,blue] (0.9,1.5) arc(0:90:0.45);
                \draw[-latex,blue] (-0.4,2.0) arc(90:180:0.45);	
                \node[ line width=0.2pt, dashed, draw opacity=0.5] (a) at (1.4,1){$x$};
                 \node[ line width=0.2pt, dashed, draw opacity=0.5] (a) at (-1.2,1){$z$};
                \draw[fill=blue] (0.7,0.3) circle (.5ex);
			\end{tikzpicture}
			\end{aligned}	
   =\sum_{(\phi)}\begin{aligned}
       \langle \phi^{(1)},x\rangle \langle \phi^{(2)},S(y)\rangle\langle \phi^{(3)},S(z)\rangle\langle \phi^{(4)},w\rangle.
   \end{aligned}
\end{equation}
In this way, for a lattice $\Gamma$ with all edges labeled with $g_e$ and faces labeled with $\phi_f$, we can evaluate it and obtain a tensor
	\begin{equation}
	  \Psi_{\Gamma}(\{x_e\}_{e\in E(\Gamma)}, \{ \phi_f\}_{f\in F(\Gamma)})= \mathrm{ttr}_{\Gamma} (\{x_e\}_{e\in E_{\Gamma}}, \{ \phi_f\}_{f\in F_{\Gamma}}),
	\end{equation} 
where $\operatorname{ttr}_{\Gamma}$ is called the weak Hopf trace over the lattice $\Gamma$.

\vspace{1em}
\emph{Hierarchy.} \textemdash\, Similar to the Hopf algebra case, if we choose a pair of weak Hopf algebras $J\subset \What$ and $K\subset W$, we obtain the solution of the Hamiltonian constructed from $J^{\rm cop}\Join K$.

\subsection{Ground state of weak Hopf quantum double model}

Let us now define the quantum states that are determined by a weak Hopf tensor network.
Consider a lattice $C(\Sigma)$ of a $2d$ surface $\Sigma$ with face set $F$ and edge set $E$. We assign to each face $f$ a weak Hopf algebra $J_f$ and each edge $e$ a weak Hopf algebra $W_e$. 
Assume that for each face, $J_f$ has the corresponding pairings with the edge Hopf algebras for edges $e \in \partial f$.
If we set the edge with values $x_e \in W_e$ and face with values $\phi_f\in J_f$, then we obtain a weak Hopf tensor network 
	\begin{equation}
	    (\otimes_{e\in E} x_e) \otimes (\otimes_{f\in F} \phi_f) \mapsto    \Psi_{C(\Sigma)}(\{x_e\},\{\phi_f\})=  \mathrm{ttr}_{C(\Sigma)} (\{x_e\},\{\phi_f\}).
	\end{equation}
	The corresponding weak Hopf tensor network states are defined as 
	\begin{equation}\label{eq:TNS}
	    |\Psi_{C(\Sigma)} (\{x_e\},\{ \phi_f\})\rangle = \sum_{(x_e)} \mathrm{ttr}_{C(\Sigma)} (\{x_e^{(2)}\},\{\phi_f\}) \otimes_{e\in E}|x_e^{(1)}\rangle.
	\end{equation}

\emph{Ground state of weak Hopf quantum double model.} \!\textemdash\,
Let us now give the explicit expression of the ground state of the weak Hopf quantum double model using the weak Hopf tensor network state.
In this case, all edge weak Hopf algebras are set as $W$, and all face weak Hopf algebras are set as $\What$.
We assign Haar integrals to edges and faces, then from Eq.~\eqref{eq:TNS} we obtain a quantum state, which is to be proved the ground state of the weak Hopf quantum double model.
The depiction of local tensor networks for face and vertex are given in Fig.~\ref{fig:FaceVertex}.

\begin{theorem}\label{thm:GS-closed-surface}
    Consider a given lattice $C(\Sigma)$ and the corresponding quantum double model $H[C(\Sigma)]$. If we assign each edge with Haar integral $h_W$ and each face with Haar integral $\varphi_{\What}$, then the corresponding weak Hopf tensor network state 
    \begin{equation}\label{eq:QD-GS}
        |\Psi_{C(\Sigma)}(\{h_e=h_W\}_e,\{\phi_f=\varphi_{\What}\}_f)\rangle
    \end{equation}
    is the ground state of $H[C(\Sigma)]$.
\end{theorem}

\begin{figure}
    \centering
    \includegraphics[scale=0.8]{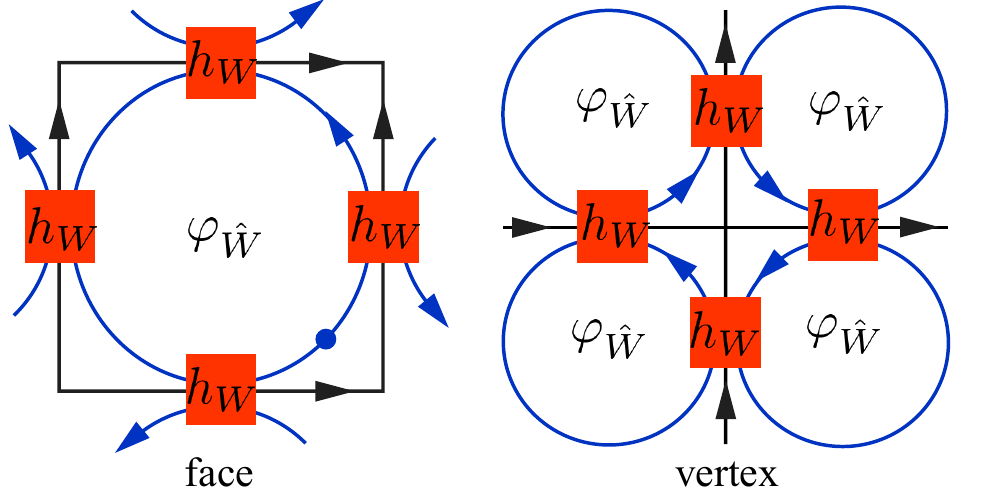}
    \caption{The depiction of face and vertex of the weak Hopf tensor network representation of the ground state of the weak Hopf quantum double model.}
    \label{fig:FaceVertex}
\end{figure}

\begin{proof}
    Without loss of generality, consider the vertex $v$ in Fig.~\ref{fig:FaceVertex}. Since $A_v$ only acts non-trivially on all edges attaching to $v$, it suffices to consider its behavior around $v$. Label the edges $e_1,e_2,e_3,e_4$ counterclockwise starting from the rightmost edge, and denote $h_j=h_W,\, \varphi_j=\varphi_{\hat{W}}$ for $j=1,2,3,4$. Near this vertex, the state $|\Psi_{C(\Sigma)}(\{h_e=h_W\}_e,\{\phi_f=\varphi_{\What}\}_f)\rangle$ is given by 
    \begin{equation}
    \begin{aligned}
        |\Psi_{v}\rangle &= \sum_{(h_j)}\varphi_1(h_1^{(3)}S(h_2^{(2)}))\varphi_2(h_2^{(3)}h_3^{(3)})\varphi_3(S(h_3^{(2)})h_4^{(3)})\\
        &\quad \quad\times \varphi_4(S(h_4^{(2)})S(h_1^{(2)}))|h_1^{(1)},h_2^{(1)},h_3^{(1)},h_4^{(1)}\rangle. 
    \end{aligned}
    \end{equation}
    For $h=h_W$, using $S(h^{(i)}) = S^{-1}(h^{(i)})$, one computes 
    \begin{align*} 
        A_v|\Psi_{v}\rangle
        & = \sum_{(h_j)}\sum_{(h)}\varphi_1(h_1^{(3)}S(h_2^{(2)}))\varphi_2(h_2^{(3)}h_3^{(3)})\varphi_3(S(h_3^{(2)})h_4^{(3)})\varphi_4(S(h_4^{(2)})S(h_1^{(2)})) \\
        &\quad \quad |h_1^{(1)}S^{-1}(h^{(1)}),h_2^{(1)}S^{-1}(h^{(2)}), h^{(3)}h_3^{(1)}, h^{(4)}h_4^{(1)}\rangle \\
        & = \sum_{(h_j)}\sum_{(h)}\varphi_1(h_1^{(3)}h^{(2)}S(h_2^{(2)}))\varphi_2(h_2^{(3)}h_3^{(3)})\varphi_3(S(h_3^{(2)})h_4^{(3)}) \\
        &\quad \quad \times \varphi_4(S(h_4^{(2)})S(h_1^{(2)}h^{(1)}))|h_1^{(1)},h_2^{(1)}S^{-1}(h^{(3)}), h^{(4)}h_3^{(1)}, h^{(5)}h_4^{(1)}\rangle \\
        & = \sum_{(h_j)}\sum_{(h)}\varphi_1(h_1^{(3)}\varepsilon_L(h^{(2)})S(h_2^{(2)}))\varphi_2(h_2^{(3)}h^{(3)}h_3^{(3)})\varphi_3(S(h_3^{(2)})h_4^{(3)}) \\
        &\quad \quad \times \varphi_4(S(h_4^{(2)})S(h_1^{(2)}h^{(1)}))|h_1^{(1)},h_2^{(1)}, h^{(4)}h_3^{(1)}, h^{(5)}h_4^{(1)}\rangle \\
        & = \sum_{(h_j)}\sum_{(h)}\varphi_1(h_1^{(3)}S(h_2^{(2)}))\varphi_2(h_2^{(3)}\varepsilon_L(h^{(2)})h^{(3)}h_3^{(3)})\varphi_3(S(h_3^{(2)})h_4^{(3)}) \\
        &\quad \quad \times \varphi_4(S(h_4^{(2)})S(h_1^{(2)}h^{(1)}))|h_1^{(1)},h_2^{(1)}, h^{(4)}h_3^{(1)}, h^{(5)}h_4^{(1)}\rangle \\
        & = \sum_{(h_j)}\sum_{(h)}\varphi_1(h_1^{(3)}S(h_2^{(2)}))\varphi_2(h_2^{(3)}\varepsilon_L(h^{(2)})h_3^{(3)})\varphi_3(S(S(h^{(3)})h_3^{(2)})h_4^{(3)}) \\
        &\quad \quad \times \varphi_4(S(h_4^{(2)})S(h_1^{(2)}h^{(1)})) |h_1^{(1)},h_2^{(1)}, h_3^{(1)}, h^{(4)}h_4^{(1)}\rangle \\
        & = \sum_{(h_j)}\sum_{(h)}\varphi_1(h_1^{(3)}S(h_2^{(2)}))\varphi_2(h_2^{(3)}h_3^{(3)})\varphi_3(S(S(\varepsilon_L(h^{(2)})h^{(3)})h_3^{(2)})h_4^{(3)}) \\
        &\quad \quad \times \varphi_4(S(h_4^{(2)})S(h_1^{(2)}h^{(1)})) |h_1^{(1)},h_2^{(1)}, h_3^{(1)}, h^{(4)}h_4^{(1)}\rangle \\
        & = \sum_{(h_j)}\sum_{(h)}\varphi_1(h_1^{(3)}S(h_2^{(2)}))\varphi_2(h_2^{(3)}h_3^{(3)})\varphi_3(S(h_3^{(2)})\varepsilon_L(h^{(2)})h_4^{(3)}) \\
        &\quad \quad \times \varphi_4(S(S(h^{(3)})h_4^{(2)})S(h_1^{(2)}h^{(1)})) |h_1^{(1)},h_2^{(1)}, h_3^{(1)}, h_4^{(1)}\rangle \\
        & = \sum_{(h_j)}\sum_{(h)}\varphi_1(h_1^{(3)}S(h_2^{(2)}))\varphi_2(h_2^{(3)}h_3^{(3)})\varphi_3(S(h_3^{(2)})h_4^{(3)}) \\
        &\quad \quad \times \varphi_4(S(S(h^{(3)})S(\varepsilon_L(h^{(2)}))h_4^{(2)})S(h_1^{(2)}h^{(1)})) |h_1^{(1)},h_2^{(1)}, h_3^{(1)}, h_4^{(1)}\rangle \\
        &= \sum_{(h_j)}\varphi_1(h_1^{(3)}S(h_2^{(2)}))\varphi_2(h_2^{(3)}h_3^{(3)})\varphi_3(S(h_3^{(2)})h_4^{(3)})\\
        &\quad \quad\times \varphi_4(S(h_4^{(2)})S(h_1^{(2)}))|h_1^{(1)},h_2^{(1)},h_3^{(1)},h_4^{(1)}\rangle. 
    \end{align*}
    Here, the fourth, sixth and eighth  equalities follow from $\sum_{(y)}xS(y^{(1)})\otimes y^{(2)} = \sum_{(y)} S(y^{(1)})\otimes y^{(2)}x$ and $\sum_{(y)} y^{(1)}\otimes xy^{(2)} = \sum_{(y)}S(x)y^{(1)}\otimes y^{(2)}$ for $x\in W_L$. We show how to deduce the second equality, and the remaining equalities are obtained similarly. Indeed, the second equality is from the following computation: 
    \begin{align*}
        &\quad \sum_{(h_1)}\Delta(h_1^{(2)})\otimes |h_1^{(1)}S^{-1}(h^{(1)})\rangle  \\
        & = \sum_{(h_1)} \sum_{(h^{(1)})} \Delta(h_1^{(2)})\otimes |S^{-1}(h^{(2)}S^{-1}(\varepsilon_R(h^{(1)}))S(h_1^{(1)}))\rangle \\
        & = \sum_{(h_1)}\sum_{(h^{(1)})} \Delta(h_1^{(2)}S^{-1}(\varepsilon_R(h^{(1)}))) \otimes |h_1^{(1)}S^{-1}(h^{(2)})\rangle \\
        & = \sum_{(h_1)}\sum_{(h^{(1)})} \Delta(h_1^{(2)}S^{-1}(h^{(2)})h^{(1)})\otimes  |h_1^{(1)}S^{-1}(h^{(3)})\rangle  \\ 
        & = \sum_{(h_1)}\sum_{(h^{(1)})} \Delta(h_1^{(2)}\varepsilon_R(S^{-1}(h^{(2)}))h^{(1)}) \otimes |h_1^{(1)}\rangle  \\
        & = \sum_{(h_1)} \Delta(h_1^{(2)}h^{(1)}) \otimes |h_1^{(1)}\rangle, 
    \end{align*}
    where the fourth equality comes from $h_1S^{-1}(h^{(2)}) = h_1\varepsilon_R(S^{-1}(h^{(2)}))$. 
    
    Next, consider the face in Fig.~\ref{fig:FaceVertex}. In the same spirit, the state $|\Psi_{C(\Sigma)}(\{h_e=h_W\}_e,\{\phi_f=\varphi_{\What}\}_f)\rangle$ is given, near this face, by
    \begin{equation}
    \begin{aligned}
        |\Psi_f\rangle & = \sum_{(h_j)} \varphi_1(S(h_1^{(2)}))\varphi_2(h_2^{(3)})\varphi_3(h_3^{(3)})\varphi_4(S(h_4^{(2)}))\\
        & \quad \quad \times \varphi(h_1^{(3)}S(h_2^{(2)})S(h_3^{(2)})h_4^{(3)})|h_1^{(1)},h_2^{(1)},h_3^{(1)},h_4^{(1)} \rangle 
    \end{aligned}
    \end{equation}
    where $\varphi=\varphi_j=\varphi_{\hat{W}}$, $j=1,2,3,4$. For $\varphi'=\varphi_{\hat{W}}$, one computes 
    \begin{align*}
        B_f|\Psi_f\rangle & = \sum_{(h_j)} \varphi(h_1^{(4)}S(h_2^{(3)})S(h_3^{(3)})h_4^{(4)})\varphi'(S(h_1^{(1)})h_2^{(2)}h_3^{(2)}S(h_4^{(1)}))\\
        & \quad \quad \times \varphi_1(S(h_1^{(3)}))\varphi_2(h_2^{(4)})\varphi_3(h_3^{(4)})\varphi_4(S(h_4^{(3)}))|h_1^{(2)},h_2^{(1)},h_3^{(1)},h_4^{(2)} \rangle  \\
        & = \sum_{(h_j)} \varphi(h_1^{(1)}S(h_2^{(3)})S(h_3^{(3)})h_4^{(1)})\varphi'(h_1^{(2)}S(h_2^{(2)})S(h_3^{(2)})h_4^{(2)})\\
        & \quad \quad \times \varphi_1(S(h_1^{(4)}))\varphi_2(h_2^{(4)})\varphi_3(h_3^{(4)})\varphi_4(S(h_4^{(4)}))|h_1^{(3)},h_2^{(1)},h_3^{(1)},h_4^{(3)} \rangle  \\
        & = \sum_{(h_j)} \varphi^2(h_1^{(1)}S(h_2^{(2)})S(h_3^{(2)})h_4^{(1)})\\
        & \quad \quad \times \varphi_1(S(h_1^{(3)}))\varphi_2(h_2^{(3)})\varphi_3(h_3^{(3)})\varphi_4(S(h_4^{(3)}))|h_1^{(2)},h_2^{(1)},h_3^{(1)},h_4^{(2)} \rangle  \\
        & = \sum_{(h_j)} \varphi(h_1^{(3)}S(h_2^{(2)})S(h_3^{(2)})h_4^{(3)})\\
        & \quad \quad \times \varphi_1(S(h_1^{(2)}))\varphi_2(h_2^{(3)})\varphi_3(h_3^{(3)})\varphi_4(S(h_4^{(2)}))|h_1^{(1)},h_2^{(1)},h_3^{(1)},h_4^{(1)} \rangle,
    \end{align*}
    by $\varphi^2=\varphi$. This completes the proof. 
\end{proof}

\vspace{1em}
\emph{Ground states of the boundary.} \textemdash\,
We now solve the boundary ground state.
Let $(C(\Sigma\setminus\partial\Sigma),C(\partial\Sigma))$ be a lattice on a $2d$ surface $\Sigma$ with boundary $\partial\Sigma$. Consider the boundary model determined by a right $W$-comodule algebra $\mathfrak{A}$. 
The Hopf tensor network states of the boundary model are defined as 
\begin{equation}
\begin{aligned}
        &\quad |\Psi_{C(\Sigma\setminus\partial\Sigma),C(\partial\Sigma)}(\{g_e\},\{\varphi_f\},\{x_{e_b}\},\{\phi_{f_b}\})\rangle \\
        & = \sum_{(g_e),(x_{e_b})} \operatorname{ttr}_{C(\Sigma\setminus\partial\Sigma),C(\partial\Sigma)}(\{g_e^{(2)}\},\{\varphi_f\},\{x_{e_b}^{[1]}\},\{\phi_{f_b}\})\\
        & \quad \quad \quad \quad \quad \otimes_{e\in E(\Sigma\setminus\partial\Sigma)} g_e^{(1)}\otimes_{e_b\in E(\partial \Sigma)}x_{e_b}^{[0]}, 
\end{aligned}
\end{equation}
where $g_e\in W,\,\varphi_f\in\hat{W},\,x_{e_b}\in\mathfrak{A}$, and $\phi_{f_b}\in\hat{W}$. 

For simplicity, let us consider the case $\mathfrak{A} = K$ where $K$ is a weak Hopf subalgebra of $W$. In this case, the symmetric separability idempotent is $\lambda = \sum_{(h)}h^{(1)}\otimes S(h^{(2)})$ with $h=h_K$ the Haar integral of $K$. If we assign each bulk edge with Haar integral $h_W$, each face with Haar integral $\varphi_{\What}$, and each boundary edge $h_K$, then the corresponding weak Hopf tensor network state 
\begin{equation}\label{eq:gs-bdd}
    |\Psi_{C(\Sigma\setminus\partial\Sigma),C(\partial\Sigma)}(\{g_e=h_W\}_e,\{\varphi_f=\varphi_{\What}\}_f,\{x_{e_b}=h_K\}_{e_b},\{\phi_{f_b} = \varphi_{\What}\}_{f_b})\rangle
\end{equation}
is the ground state of the quantum double model with boundary $H[C(\Sigma\setminus\partial\Sigma),C(\partial\Sigma)]$. To see this, consider the boundary vertex $v_b$ in Eq.~\eqref{eq:bdd-sta-A1}, over which the state \eqref{eq:gs-bdd} is given by 
\begin{equation}
    |\Psi_{v_b}\rangle = \sum \varphi_1(h_W^{(2)}S(h_1^{(2)}))\varphi_2(S(h_W^{(3)})S(h_2^{(2)}))|h_1^{(1)},h_W^{(1)},h_2^{(1)}\rangle, 
\end{equation}
with $h_1=h_2 = h_K$. It suffices to show that $|\Psi_{v_b}\rangle$ is preserved by the boundary stabilizer $A^\lambda_{v_b}$. In the following computation, we used repeatedly the fact that $S(h^{(i)}) = S^{-1}(h^{(i)})$ for Haar integral $h$.  One computes 
\begin{align*}
    A^\lambda_{v_b}|\Psi_{v_b}\rangle & = \sum \varphi_1(h_W^{(2)}S(h_1^{(2)}))\varphi_2(S(h_W^{(3)})S(h_2^{(2)}))|h^{(1)}h_1^{(1)},h^{(2)}h_W^{(1)},h_2^{(1)}S(h^{(3)})\rangle\\
    & = \sum \varphi_1(h_W^{(2)}S(S(h^{(1)})h_1^{(2)}))\varphi_2(S(h_W^{(3)})S(h_2^{(2)}))|h_1^{(1)},h^{(2)}h_W^{(1)},h_2^{(1)}S(h^{(3)})\rangle\\
    & = \sum \varphi_1(S(h^{(3)})h_W^{(2)}S(h_1^{(2)})h^{(1)})\varphi_2(S(S(h^{(2)})h_W^{(3)})S(h_2^{(2)}))|h_1^{(1)},h_W^{(1)},h_2^{(1)}S(h^{(4)})\rangle\\
    & = \sum \varphi_1(h_W^{(2)}S(h_1^{(2)}))\varphi_2(S(S(h^{(1)})h_W^{(3)})S(h_2^{(2)}))|h_1^{(1)},h_W^{(1)},h_2^{(1)}S(h^{(2)})\rangle\\
    & = \sum \varphi_1(h_W^{(2)}S(h_1^{(2)}))\varphi_2(S(h_W^{(3)})h^{(1)}S(h_2^{(2)}h^{(2)})|h_1^{(1)},h_W^{(1)},h_2^{(1)}\rangle\\
    & = \sum \varphi_1(h_W^{(2)}S(h_1^{(2)}))\varphi_2(S(h_W^{(3)})S(h_2^{(2)}))|h_1^{(1)},h_W^{(1)},h_2^{(1)}\rangle. 
\end{align*}
Here, the second equality holds as 
\begin{align*}
    \sum |h_1^{(2)}\rangle\otimes |h^{(1)}h_1^{(1)}\rangle & = \sum |h_1^{(2)}\rangle\otimes |h^{(1)}S(S(\varepsilon_R(h^{(2)})))h_1^{(1)}\rangle \\
    & = \sum |S(\varepsilon_R(h^{(2)}))h_1^{(2)}\rangle \otimes |h^{(1)}h_1^{(1)}\rangle \\ 
    & = \sum |S(h^{(3)})h^{(2)}h_1^{(2)}\rangle \otimes |h^{(1)}h_1^{(1)}\rangle \\ 
    & = \sum |S(h^{(2)})h_1^{(2)}\rangle \otimes |\varepsilon_L(h^{(1)})h_1^{(1)}\rangle \\
    & = \sum |S(h^{(2)})S(h_1^{(2)})\rangle \otimes |\varepsilon_L(h^{(1)})S(h_1^{(1)})\rangle \\ 
    & = \sum |S(h^{(2)})S(h_1^{(2)}\varepsilon_L(h^{(1)}))\rangle \otimes |S(h_1^{(1)})\rangle \\ 
    & = \sum |S(h^{(1)})h_1^{(2)} \rangle \otimes |h_1^{(1)}\rangle. 
\end{align*}
The third and fifth equalities follow from similar computations. The trick is similar to the proof of Theorem \ref{thm:GS-closed-surface}.

\vspace{1em}
\emph{Ground states of the domain wall.} \textemdash\, Similarly, we can solve the domain wall ground state. For a $W_1|W_2$-bicomodule algebra $\mathfrak{A}$, denote by $\alpha(x) = \sum_{(x)}x^{[-1]}\otimes x^{[0]}\in W_1\otimes\mathfrak{A}$ and by $\beta(x)=\sum_{(x)}x^{[0]}\otimes x^{[1]}\in \mathfrak{A}\otimes W_2$, $x\in \mathfrak{A}$, respectively the left $W_1$- and right $W_2$- coactions as usual. So $(\id_{W_1}\otimes \beta)\comp\alpha(x)=\sum_{(x)}x^{[-1]}\otimes x^{[0]}\otimes x^{[1]}$ is well defined. For a domain wall determined by $\mathfrak{A}$, the Hopf tensor network states are defined as 
\begin{equation*}
\begin{aligned}
        &\quad |\Psi_{C(\Sigma_1),C(\Sigma_d),C(\Sigma_2)}(\{g_{e_1}\},\{\varphi_{f_1}\},\{\phi_{f_{d,1}}\},\{x_{e_d}\},\{\phi_{f_{d,2}}\},\{g_{e_2}\},\{\varphi_{f_2}\})\rangle \\
        & = \sum_{(g_{e_i}),(x_{e_b})} \operatorname{ttr}_{C(\Sigma_1),C(\Sigma_d),C(\Sigma_2)}(\{g_{e_1}^{(2)}\},\{\varphi_{f_1}\},\{x_{e_d}^{[-1]}\},\{\phi_{f_{d,1}}\},\{x_{e_d}^{[1]}\},\{\phi_{f_{d,2}}\},\{g_{e_2}^{(2)}\},\{\varphi_{f_2}\})\rangle \\
        & \quad \quad \quad \quad \quad \otimes_{e_1\in E(\Sigma_1)}g_e^{(1)}\otimes_{e_d\in E(\Sigma_d)} x_{e_d}^{[0]} \otimes_{e_2\in E( \Sigma_2)}g_{e_2}^{(1)}, 
\end{aligned}
\end{equation*}
where $g_{e_1}\in W_1,\,\varphi_{f_1},\,\phi_{f_{d,1}}\in\hat{W}_1,\, g_{e_2}\in W_2,\,\varphi_{f_2},\,\phi_{f_{d,2}}\in\hat{W}_2$, and $x_{e_d}\in\mathfrak{A}$. 

Consider the case that $\mathfrak{A}=K$ is a weak Hopf subalgebra of $W_1\otimes W_2^{\rm cop}$. Then one can use a similar computation as in the boundary to show that the ground state of the domain wall is given by the weak Hopf tensor network state 
\begin{equation*}
    |\Psi_{C(\Sigma_1),C(\Sigma_d),C(\Sigma_2)}(\{h_{W_1}\}_{e_1},\{\varphi_{\hat{W}_1}\}_{f_1},\{h_K\}_{e_d},\{\varphi_{\hat{W}_1}\}_{f_{d,1}},\{\varphi_{\hat{W}_2}\}_{f_{d,2}},\{h_{W_2}\}_{e_2},\{\varphi_{\hat{W}_2}\}_{f_2})\rangle. 
\end{equation*}

\subsection{Closed ribbon operator and ground state degeneracy}

We have shown that the ground state of the quantum double model can be described as a weak Hopf tensor network state. However, the ground state space is degenerate for general closed surfaces. We now address the question of generating additional ground states from the expression in Eq.~\eqref{eq:QD-GS}.
This can be achieved via the action of a closed ribbon operator.

Let us take the type-B ribbon as an example. As discussed before, we have the decomposition of ribbon operator algebra $D(W)^{\vee}$ as in Eq.~\eqref{eq:DWdecomp}.
Consider the projection $\operatorname{Proj}_X D(W)^{\vee}\cong X \otimes X^{\vee}$.
For a closed ribbon $\rho$ with $\partial_0 \rho=\partial_1 \rho=s$, the particle $X$ and $X^{\vee}$ fuse to vacuum $\one$ at site $s$.
Consider the projection of $X\otimes X^{\vee} \to \one$, this can also be embedded into $D(W)$ via
\begin{equation}
    \one \hookrightarrow X \otimes X^{\vee} \hookrightarrow D(W)^{\vee}.
\end{equation}
We denote the corresponding image as $\operatorname{Proj}_{\one, X} D(W)^{\vee}$.
For closed ribbon $\rho$ and $g\otimes \psi \in \operatorname{Proj}_{\one, X} D(W)^{\vee}$ we have the ribbon operator $F^{g,\psi}_{\rho}$.
Now for a closed surface $\Sigma$ with the cellulation $C(\Sigma)$ and weak Hopf quantum double model $H[D(W),C(\Sigma)]$, using the expression in Eq.~\eqref{eq:QD-GS}, we obtain a ground state $|\Omega\rangle$.
Then for a closed ribbon $\rho$, $F_{\rho}^{g,\psi} |\Omega\rangle$ is also a ground state.
Depending on the topology of the surface $\Sigma$, there may be non-homotopical closed ribbons, for these ribbons, applying the ribbon operator above, we will obtain  different ground states in the vacuum sector $\Vcal_{\one}$.

\section{Topological duality of weak Hopf quantum double phase}
\label{sec:duality}

One crucial property of the weak Hopf quantum double model is that they have EM duality. Mathematically, this stems from the fact that weak Hopf algebras are self-dual, that is, when $W$ is a weak Hopf algebra, its dual algebra $\What$ is also a weak Hopf algebra. And the representation categories of $W$ and $\What$ are categorical Morita equivalent.
This EM duality is important to understand the matter phases and phase transitions. It is proved that the EM duality is closely related to the well-known Kramers-Wannier duality of Ising model \cite{freed2022topological,aasen2016topological,aasen2020topological}.

We know that the Abelian group quantum double model has EM duality, under which the gauge charges and gauge fluxes are exchanged in the dual theory.
The corresponding dual group is $\hat{G}$ which consists of all irreducible representations of $G$.
Since $G\cong \hat{G}$, the topological phases for the two theories are equivalent.
However, when $G$ is non-Abelian, the EM duality is broken, since the irreducible representations of $G$ do not form a group.
One approach to remedy this problem is to generalize the EM duality as a partial EM duality \cite{hu2020electric}, i.e., consider a normal subgroup $N\subset G$, the $N$-charge and $N$-flux are exchanged under the partial EM duality. 
Another way is to generalize the notion of group symmetry to the quantum group symmetry (like the Hopf algebras \cite{buerschaper2013electric} and the weak Hopf algebras discussed in this work).
In this case, the input data of the dual model is the dual Hopf algebra.
Two corresponding quantum double phases are equivalent.
The EM duality of the group quantum double model in the presence of a gapped boundary and domain wall is also discussed previously \cite{wang2020electric,Jia2022electric,Kitaev2012a}. Under the bulk EM duality, the boundary types are permuted accordingly. 

In this part, we consider the most general case where the EM duality can be realized, \emph{viz.}, the weak Hopf quantum double model without and with the gapped boundary or domain wall.
Besides EM duality, we will also define the general concept of duality which connects two equivalent quantum double phases.

\subsection{Topological duality for the $2d$ bulk}

The duality is a correspondence between two physically equivalent theories. Before we discuss the duality, let us first consider the general map between two anomaly-free topological phases\,\footnote{The phases that have lattice realizations in the same spatial dimension \cite{kong2014braided}.}.
Assume that the topological excitations of two phases are given by UMTCs $\EuScript{P}_i$, $i=1,2$, and their lattice realizations are $H[C_i(\Sigma)]$, $i=1,2$ (the Hamiltonians are built from two different cellulations of the same surface $\Sigma$).
At excitation level, the general duality map $F:\EuScript{P}_1\to \EuScript{P}_2$ should satisfy the following natural conditions:
\begin{enumerate}
    \item The vacuum should be mapped to vacuum $r_F: F(\one)\cong \one$.
    \item It preserves the fusion relation $a_{X,Y,F}: F(X\otimes Y)\cong F(X)\otimes F(Y)$. This means that fusion multiplicity is preserved under the duality map.
    \item It preserves the axioms of fusion relations. This means that the $F$-symbol is preserved under the duality map.
    \item Since the direct sum is indeed equivalent to the quantum superposition in the fusion/splitting channel, the duality map should also preserve the direct sum structure. 
    That is, $F$ should map zero objects to zero objects and $F(X\oplus Y)\cong F(X)\oplus F(Y)$, and the axioms of direct sum structure is also preserved.
    \item When acting on the fusion/splitting space, it is a linear map.
    \item It preserves the braiding structure $F(c_{X,Y})\simeq c_{F(X),F(Y)}$. This means that the duality map preserves the mutual statistics of particles.
\end{enumerate}
At the lattice level, there should be a map $U_F:\mathcal{H}_1\to \mathcal{H}_2$ which maps local terms in the Hamiltonian in a way consistent with the map of topological charge. 
A duality between two theories means that the map is bijective, namely, $F$ is an equivalence between two UMTCs and $U_F$ is a Fourier transform between two lattice models.

Based on the above discussion, for quantum double phases, a general map $\Phi$ between two weak Hopf quantum double models consists of the following data: (i) A braided monoidal functor $F: \Rep(D(W))\to \Rep(D(W'))$; (ii) A Fourier transform $U_F$ between lattice models $H[D(W);C(\Sigma)]$ and $H[D(W');C(\Sigma')]$. We simply denote it as
\begin{equation}
    \Phi:\mathsf{QD}(D(W);\Sigma) \to \mathsf{QD}(D(W');\Sigma').
\end{equation}
The braided monoidal functor means that the physical data, such as fusion and braiding, are preserved by the transformation. 
The map between two lattice models is a lattice realization of the map between two topological phases.
In this framework, we can define the duality between quantum double phases as follows:

\begin{definition}
A duality between two quantum double phases is an equivalence
    \begin{equation}
    \Phi:\mathsf{QD}(D(W);\Sigma) \to \mathsf{QD}(D(W');\Sigma'),
\end{equation}
which consists of a braided monoidal equivalence $F:\Rep(D(W))\to \Rep(D(W'))$ and an invertible map $U_F$ between two lattice realizations.
\end{definition}

The EM duality is a special case that relates the quantum double model constructed from $W$ and that constructed from its dual $\What$:
\begin{equation}
    \Phi_{\rm EM}:\mathsf{QD}(D(W);\Sigma) \to \mathsf{QD}(D(\hat{W});\Sigma').
\end{equation}
This is guaranteed by the fact that there is a braided monoidal equivalence between $\Rep(D(W))$ and $\Rep(D(\hat{W}))$.
And we can also map the lattice model $H(D(W),\Sigma)$ to $H(D(\What),\Sigma')$ by sending the direct lattice $\Sigma$ to its dual lattice $\Sigma'=\tilde{\Sigma}$, changing the edge space via Fourier transform between $W$ and $\What$, and exchanging the vertex and face operators with each other (thus the charge and flux are exchanged).
Since $W^{\vee\vee}\cong W$, we see that this map is invertible.

For the more general duality of the weak Hopf quantum double model, we need a notion of categorical Morita equivalence: two monoidal categories $\EC$ and $\ED$ are called categorical Morita equivalent\,\footnote{Also called weakly Morita equivalent by some authors \cite{muger2003subfactorsI}.} if their Drinfeld centers are braided monoidal equivalent $\mathcal{Z}(\EC)\simeq \mathcal{Z}(\ED)$.
For two weak Hopf gauge symmetries $W,W'$, their representation categories are categorical Morita equivalent if and only if $\mathcal{Z}(\Rep(W))\simeq \mathcal{Z}(\Rep(W'))$. 
From the discussion in Sec.~\ref{sec:pre}, we know that $\mathcal{Z}(\Rep(W))$ is equivalent to the Yetter-Drinfeld module category ${_W}\mathsf{YD}^W$ as a braided monoidal category, and ${_W}\mathsf{YD}^W$ is equivalent to $\Rep(D(W))$ as a braided monoidal category. Therefore, we have:
    two quantum double phases with their respective weak Hopf gauge symmetry $W,W'$ are equivalent (i.e., there is a duality between them)
    if and only if $\Rep(W)$ and $\Rep(W')$ are categorical Morita equivalent, \emph{viz.}, $\mathcal{Z}(\Rep(W))\simeq\mathcal{Z}(\Rep(W'))$.

\begin{remark}
    The above theorem also holds for general $2d$ topological order.
    There is a duality between two topological orders with UFC symmetry $\EC$ and $\ED$ if and only if $\mathcal{Z}(\EC)\simeq \mathcal{Z}(\ED)$. 
    One way to understand this duality is to treat both $\EC$ and $\ED$ as boundary theories of a $3d$ topological order $\EP$.
    Then using the boundary-bulk correspondence, the boundary anyons of one boundary can be dragged into the bulk, then it can be condensed into another boundary, which realizes the duality between two boundary topological orders.
\end{remark}

\begin{remark}[Duality of Levin-Wen string-net model]
The Levin-Wen string-net model is a more general lattice realization of non-chiral anomaly-free topological phase \cite{Levin2005}.
The input data of the model is a UFC $\EC$, and the topological excitation is given by the Drinfeld center $\mathcal{Z}(\EC)$.
The general theory of duality in this case can be understood via the module category over $\EC$.
It is proved that \cite{etingof2010fusion,etingof2016tensor} for any $\EC$-module category $\EM$, the category of $\EC$-module functors $\EC_{\EM}^{\vee}=\Fun(\EM,\EM)^{\otimes \rm op}$ is categorical Morita equivalent to $\EC$.
Thus there is a duality
\begin{equation}
    \Phi: \mathsf{SN}(\EC, H_{\rm SN}[\EC,C(\Sigma)])\to  \mathsf{SN}(\EC_{\EM}^{\vee}, H_{\rm SN}[\EC_{\EM}^{\vee},C(\Sigma)]);
\end{equation}
then for different $\EM,\EM'$, there is also a duality between $\EC_{\EM}^{\vee}$ and $\EC_{\EM'}^{\vee}$ theories. Notice that in the string-net model, we cannot define the EM duality, since the Hilbert space is not large enough. To remedy this, we need to introduce the extended string-net model by extending the Hilbert space to support the dyonic excitations \cite{Hu2018full,buerschaper2013electric}.
\end{remark}

For weak Hopf symmetry, since there are more structures in $W$ than $\Rep(W)$ that we can use (actually, via the Tannaka–Krein duality, by introducing a fiber functor to $\Rep(W)$, we can recover $W$ from the fiber functor). Thus there are more equivalent ways to explicitly express the duality between quantum double models. 
Let us take Hopf algebra and its twist deformation as an example to illustrate this.

\vspace{1em}
\emph{Example: the duality of Hopf quantum double model.} \textemdash\,  Consider the Hopf quantum double model \cite{Buerschaper2013a,jia2022boundary}, the duality is equivalent to the twist deformation of quantum double\,\footnote{Notice that this is different from the notion of twisted quantum double, which is a variation of the quantum double induced by a $3$-cocycle $\alpha:W \otimes W \otimes W\to U(1)$.}. Namely,
for any two Hopf gauge symmetries $W,W'$, their corresponding quantum double phases are equivalent (i.e., there is a duality between them) if and only if there exists a Drinfeld twist $T$ such that $D(W)\cong D(W')^{T}$ as Hopf algebras.
In fact, it is proved \cite[Proposition 5.14.4]{etingof2016tensor} that the representation categories of two Hopf algebras are equivalent if and only two Hopf algebras are twist equivalent. Apply this result to quantum double and using the definition of duality, we arrive at the conclusion.

The theorem tells us that the duality between two Hopf quantum double models is fully characterized by a twist.
By definition a twist of a Hopf algebra $W$ is an invertible element $T \in W\otimes W$ such that 
\begin{equation}
   ( \Delta\otimes \id)(T)(T\otimes 1)=(\id \otimes \Delta)(T)(1\otimes T).
\end{equation}
The twist can be normalized to ensure that $(\varepsilon\otimes \id)(T)=(\id\otimes \varepsilon)(T)=1$. We will adopt shorthand notation $T=T^{<1>}\otimes T^{<2>}$ and $T^{-1}=T^{-<1>}\otimes T^{-<2>}$ (similar to the Einstein summation notation, we omit the summation symbol).
We denote $Q_T=S(T^{<1>})T^{<2>}$ and $Q_T^{-1}=T^{-<1>}S(T^{-<2>})$.
The twisted Hopf algebra $W^T$ is an algebra with the same algebra structure and the same counit as that of $W$, and the comultiplication and antipode are given by
\begin{equation}
    \Delta^{T}(g)=T^{-1}\Delta(g) T, \quad S^{T}(g)=Q^{-1}_T S(g) Q_T,\quad g\in W.
\end{equation}
If $W$ is quasitriangular with $R$-matrix $R$, then $W^T$ is quasitriangular with $R$-matrix $R^T=T_{21}^{-1} R T$.
The Haar integral remains unchanged under the twist deformation \cite{aljadeff2002twisting}.

\vspace{1em}
\emph{Correspondence of boundaries under duality.} \textemdash\,
The boundary theory is determined by a Lagrangian algebra \cite{Kong2014}, and it can be proved that the duality maps Lagrangian algebra to Lagrangian algebra.
To be precise, if $L$ is a Lagrangian algebra of $\Rep(D(W))$, then $F(L)$ is a Lagrangian algebra of $\Rep(D(W'))$, where $F:\Rep(D(W)) \to \Rep(D(W'))$ is a duality, namely, $F$ is a braided monoidal equivalence. 
Since $F(\oplus_i X_i)=\oplus_iF(X_i)$, the quantum dimension of $\dim F(L)=\dim L$. All other structures are preserved by $F$. 
This implies that, for the surface with boundaries, the bulk theory is mapped to the dual theory, and the boundary theory is mapped to the boundary theory of the dual theory.

\subsection{Duality for the boundary and domain wall}

For the $1d$ boundary and domain wall of the quantum double phase, the duality still means that two different lattice realizations give the same (monoidal equivalent) boundary phase. Notice that the bulk is now fixed, thus it is different from the boundary correspondence induced by the duality of the bulk phase.

\begin{definition}
   1. A duality between two gapped boundaries is an equivalence 
    \begin{equation}
        \Phi: \mathsf{QD}_{\rm bd}(\mathfrak{A};\partial \Sigma)\to \mathsf{QD}_{\rm bd}(\mathfrak{A}';\partial \Sigma'),
    \end{equation}
    which consists of (i) a monoidal equivalence between two boundary phases $F:{_{\FA}}\Mod_{\FA} \to {_{\FA'}}\Mod_{\FA'}$; and (ii) a lattice Fourier transform $U_F$.

    2. For domain wall, a similar definition exists. We can also use the folding trick to transform the domain wall into a boundary and invoke the definition of duality for the boundary to give that to the domain wall.
\end{definition}

Therefore, we need firstly to investigate the monoidal equivalent $F:{_{\FA}}\Mod_{\FA} \to {_{\FA'}}\Mod_{\FA'}$.
First notice that ${_{\FA}}\Mod_{\FA} \simeq {_{\FA \otimes \FA^{\rm op}}}\Mod$. Then using the Eilenberg-Watts theorem \cite{eilenberg1960abstract,watts1960intrinsic}, $F: {_{\FA \otimes \FA^{\rm op}}}\Mod \to {_{\FA' \otimes {\FA'}^{\rm op}}}\Mod$ is determined by a ${\FA' \otimes {\FA'}^{\rm op}}|{\FA \otimes \FA^{\rm op}}$-bimodule $M$.
Notice that Morita equivalence of $\FA$ and $\FA'$ is a stronger condition, in this case ${_{\FA}}\Mod \simeq {_{\FA'}}\Mod$.
In a more general TQFT formalism of boundary, this boundary duality between two $\EC$-module categories $\EM$ and $\EN$ is equivalent to $\Fun_{\EC}(\EM,\EM)\simeq \Fun_{\EC}(\EN,\EN)$. But the Morita equivalence means that $\EM\simeq \EN$.
We conjecture that the duality for the boundary can also be characterized by twist deformation of two $W$-comodule algebras $\FA_1$ and $\FA_2$.
Namely, ${_{\FA_1}}\Mod_{\FA_1} \simeq {_{\FA_1 \otimes \FA_1^{\rm op}}}\Mod$ is monoidal equivalent to ${_{\FA_2}}\Mod_{\FA_2} \simeq {_{\FA_2 \otimes \FA_2^{\rm op}}}\Mod$ if and only if $\FA_1$ and $\FA_2$ are twist equivalent.
We will leave this for our future study.

We would like to stress that the concept of duality given above can be relaxed to \emph{weak duality}.
By which we mean, any two $nd$ topological order $\EP_1$ and $\EP_2$ are weakly dual to each other if they can be embedded into gapped boundaries for an $(n+1)d$ topological order.
This more general definition of duality allows us to subsume more existing dualities in one unified framework. For example,
the well-known Kramers-Wannier duality between the low-temperature and high-temperature phases of the Ising model can be understood in this way \cite{freed2022topological,aasen2016topological,aasen2020topological}.

\section{Conclusion and discussion}
\label{sec:conclusion}
In this paper, we investigated the weak Hopf symmetry and weak Hopf quantum double model in detail.
We established the relationship between the weak Hopf quotient algebras and sub-algebras in $W$ and $\What$.
Using this correspondence, we presented a theory of weak Hopf symmetry for vacuum states and its breaking during anyon condensation.
We constructed the weak Hopf quantum double model in $2d$ lattice, which has weak Hopf symmetry.
The local stabilizer, ribbon operators, and topological excitations were discussed.
The algebraic theory of gapped boundary and gapped domain wall was established based on comodule algebras and bicomodule algebras.
The lattice construction based on symmetric separability idempotent was given.
To solve the weak Hopf quantum double model, we introduced the weak Hopf tensor network. Then using this representation, we presented the exact solution of the model.
The theory of duality between different weak Hopf quantum double models was also established.

Our work is necessarily technical due to the complication of dealing with weak Hopf algebras. 
While some progress has been made, there is much left to be done, and we outline some potentially interesting directions below.

\emph{Bulk and boundary twist defects.} \textemdash\,
In our work, we mainly focused on $1d$ defects, such as boundaries and domain walls. 
We briefly discuss the algebraic classification of the boundary defects and domain wall twist defects via the bimodule category.
A lattice construction of these point defects is an interesting and crucial topic.
Typical examples are Abelian toric codes, wherein the domain wall twist defects have the same statistics as Ising anyons \cite{Bombin2010,Jia2022electric}.
And the boundary defects of the Abelian toric code are Majorana fermionic and parafermionic zero modes \cite{cong2017defects}.
This phenomenon of emerging non-Abelian anyon in the Abelian phase is crucial for both theoretical investigation and applications in topological quantum computing.
For general Hopf and weak Hopf quantum double the structures of these twist defects are more complicated and largely remain open.

\emph{Entanglement perspective.} \textemdash \,
We know that the quantum double models have intrinsic topological order, thus their ground states have long-range entanglement.
An entanglement renormalization approach based on our weak Hopf tensor network representations will be helpful for us to understand the weak Hopf quantum double phase.
On the other hand, entanglement entropy is proven to be sensitive to the existence of topological defects, and understanding the entanglement properties in the presence of a gapped boundary, domain wall, and twist defects is a crucial topic.
We left these for our future research.

\emph{Symmetry-enriched quantum double model.} \textemdash\,
The quantum double model has the EM duality (see Sec.~\ref{sec:duality}), and this EM duality could induce a global $\mathbb{Z}_2$ symmetry that exchanges the electric and magnetic charge of the model.
The quantum double can be enriched by this EM duality symmetry.
Systematic construction of the symmetry-enriched quantum double model for arbitrary group $G$ (or more generally, Hopf algebra and weak Hopf algebra) is still an open problem. 

\emph{Higher dimensional extension.} \textemdash\,
While we have a relatively complete understanding of the $2d$ topological phase, the algebraic theory and Hamiltonian theory behind the higher dimensional topological orders are still unclear.
Recently, The $3d$ $\mathbb{C}[\mathbb{Z}_2]$ model has been investigated from aspects \cite{Levin2005,Hamma2005string,kong2020defects,delcamp2021tensornet}. The $3d$ finite  group model and its twisted generalization are discussed in \cite{Moradi2015universal,Wan2015twisted,wang2018gapped}. However, both the higher dimensional case and the more complicated Hopf algebra case are largely unexplored.

\subsection*{Acknowledgments}

We are very grateful to the anonymous referees for careful reading of our paper and for their valuable comments. 
S.T. would like to thank his advisor Uli Walther for his constant encouragement and helpful discussions. Z.J. and D.K. are supported by the National Research Foundation and the Ministry of Education in Singapore through the Tier 3 MOE2012-T3-1-009 Grant: Random numbers from quantum processes. S.T. is partially supported by NSF grant DMS-2100288 and by Simons Foundation Collaboration Grant for Mathematicians \#580839. L.C. is supported by NSFC Grant No.12171249. 


\subsection*{Declarations}
\paragraph{Conflict of interest:} All authors certify that there are no conflicts of interest for this work.

\paragraph{Data Availability Statement:} Data sharing is not applicable to this article as no datasets were generated or analyzed during the current study.

\appendix

\section{Some technical details about quantum double} \label{app:QD}

In this appendix, we will give a detailed proof that the definition of the quantum double $D(W)$ of a weak Hopf algebra $W$ is well defined. This has been shown in Ref. \cite{nikshych2003invariants} for an equivalent form of $D(W)$. Nevertheless, it is still worth proving the assertion for the form we used in this work.

Let $J$ be the linear span in $\hat{W}^{\rm cop}\otimes W$ of the elements 
\begin{align}
	\varphi \otimes xh -  \varphi (x\rightharpoonup \varepsilon)\otimes h, \quad x\in W_L,\\
	\varphi \otimes yh -  \varphi ( \varepsilon \leftharpoonup y) \otimes h, \quad y\in W_R. 
\end{align}
Then $J$ is a two-sided ideal of $\hat{W}^{\rm cop} \otimes W$. To prove this claim, clearly, $J$ is an additive subgroup of $\hat{W}^{\rm cop}\otimes W$. For any $\psi\otimes g\in \hat{W}^{\rm cop}\otimes W$ and $x\in W_L$, one has 
    \begin{align*}
        (\psi\otimes g)(\varphi (x\rightharpoonup \varepsilon)\otimes h) 
        & = \sum_{(g),(\varphi)}\psi \varphi^{(2)}\otimes g^{(2)}h \langle \varphi^{(1)}, S^{-1}(g^{(3)}) \rangle \langle \varphi^{(3)}(x \rightharpoonup \varepsilon),g^{(1)} \rangle \\
        & = \sum_{(g),(\varphi)}\psi \varphi^{(2)}\otimes \varepsilon(g^{(2)}x)g^{(3)}h \langle \varphi^{(1)}, S^{-1}(g^{(4)}) \rangle \langle \varphi^{(3)}, g^{(1)} \rangle  \\
        & = \sum_{(g),(\varphi)}\psi \varphi^{(2)}\otimes g^{(2)}xh \langle \varphi^{(1)}, S^{-1}(g^{(3)}) \rangle \langle \varphi^{(3)}, g^{(1)} \rangle  \\
        & = (\psi\otimes g)(\varphi \otimes xh), \\
        (\varphi\otimes xh)(\psi\otimes g) & = \sum_{(h),(\psi)}\varphi\psi^{(2)}\otimes h^{(2)} g \langle \psi^{(1)}, S^{-1}(h^{(3)}) \rangle \langle \psi^{(3)}, xh^{(1)} \rangle \\
        & = \sum_{(h),(\psi)}\varphi \psi^{(2)} \langle \psi^{(3)}, x \rangle \otimes h^{(2)} g \langle \psi^{(1)}, S^{-1}(h^{(3)}) \rangle  \langle \psi^{(4)},h^{(1)} \rangle \\
        & = \sum_{(h),(\psi)}\varphi (x\rightharpoonup \varepsilon) \psi^{(2)}\otimes h^{(2)} g \langle \psi^{(1)}, S^{-1}(h^{(3)}) \rangle  \langle \psi^{(3)},h^{(1)} \rangle \\
        & = (\varphi(x\rightharpoonup x)\otimes h)(\psi\otimes g).
    \end{align*}
    Similarly, one has 
    \begin{align*}
         (\psi\otimes g)(\varphi\otimes yh) & = (\psi\otimes g)(\varphi(\varepsilon \leftharpoonup y)\otimes h), \\
        (\varphi\otimes yh)(\psi\otimes g) & = (\varphi(\varepsilon \leftharpoonup y)\otimes h)(\psi\otimes g).
    \end{align*}
    The above computation shows that for $x\in J$, we have $(\psi\otimes g)x = x(\psi\otimes g) = 0$, so $J$ is a two-sided ideal. 

The quotient algebra $D(W):=(\hat{W}^{\rm cop} \otimes W)/J$ is the quantum double of $W$. We will verify in the following that $D(W)$ is a weak Hopf algebra, whose weak Hopf algebra structure is given in Definition \ref{def:quantum-double}.

  (i) $D(W)$ is an algebra. Associativity: for $\varphi,\psi,\theta\in \hat{W}$ and $g,h,k\in W$, 
    \begin{align*}
        &\quad ([\varphi\otimes h][\psi\otimes g])[\theta\otimes k] \\
        &= \sum_{(h),(\psi)}   [\varphi \psi^{(2)} \otimes  h^{(2)} g][\theta\otimes k]  \langle \psi^{(1)}, S^{-1}(h^{(3)}) \rangle \langle \psi^{(3)}, h^{(1)}\rangle \\
        & = \sum_{(h),(\psi)} \sum_{(g),(\theta)} [(\varphi \psi^{(2)})\theta^{(2)}\otimes (h^{(3)}g^{(2)})k] \langle \theta^{(1)},S^{-1}(h^{(4)}g^{(3)})\rangle \\
        & \quad \quad \times \langle \theta^{(3)},h^{(2)}g^{(1)} \rangle  \langle \psi^{(1)},S^{-1}(h^{(5)}) \rangle \langle \psi^{(3)},h^{(1)} \rangle  \\
        & = \sum_{(h),(\psi)} \sum_{(g),(\theta)} [\varphi (\psi^{(2)}\theta^{(3)})\otimes h^{(3)}(g^{(2)}k)] \langle \theta^{(1)},S^{-1}(g^{(3)})\rangle \langle\theta^{(2)},S^{-1}(h^{(4)})\rangle \\
        & \quad \quad \times \langle \theta^{(4)},h^{(2)} \rangle \langle \theta^{(5)},g^{(1)} \rangle  \langle \psi^{(1)},S^{-1}(h^{(5)}) \rangle \langle \psi^{(3)},h^{(1)} \rangle \\
        & = \sum_{(g),(\theta)} \sum_{(h),(\psi)} [\varphi (\psi^{(2)}\theta^{(3)})\otimes h^{(2)}(g^{(2)}k)] \langle \psi^{(1)}\theta^{(2)},S^{-1}(h^{(3)})\rangle  \langle \psi^{(3)}\theta^{(4)},h^{(1)} \rangle   \\
        & \quad \quad \times  \langle \theta^{(5)},g^{(1)} \rangle  \langle \theta^{(1)},S^{-1}(g^{(3)})\rangle   \\
        & = \sum_{(g),(\theta)}[\varphi\otimes h][\psi \theta^{(2)}\otimes g^{(2)} k ]\langle \theta^{(1)},S^{-1}(g^{(3)})\rangle  \langle \theta^{(3)},g^{(1)} \rangle \\
        & = [\varphi\otimes h]([\psi\otimes g][\theta\otimes k]).
    \end{align*}
    Unit property: since $\Delta(1_W)\in W_R\otimes W_L$, 
    \begin{align*}
        [\varepsilon\otimes 1_W][\varphi\otimes h] & = \sum_{(1_W)}\sum_{(\varphi)}[\varphi^{(2)}\otimes 1_W^{(2)}h]\langle \varphi^{(1)}, S^{-1}(1_W^{(3)}) \rangle \langle \varphi^{(3)}, 1_W^{(1)} \rangle \\
        & = \sum_{(1_W)}\sum_{(1'_W)}\sum_{(\varphi)}[\varphi^{(2)}\otimes \langle \varphi^{(3)}, 1_W^{(1)} \rangle 1_W^{(2)}1_W'^{(1)} \langle \varphi^{(1)}, S^{-1}(1_W'^{(2)}) \rangle h] \\
        & = \sum_{(1_W)}\sum_{(1'_W)}\sum_{(\varphi)}[\varphi^{(2)}(1_W^{(2)} \rightharpoonup \varepsilon)\langle \varphi^{(3)}, 1_W^{(1)} \rangle (\varepsilon \leftharpoonup 1_W'^{(1)}) \langle \varphi^{(1)}, S^{-1}(1_W'^{(2)}) \rangle \otimes h] \\
        & = \sum_{(\varphi)} [\varphi^{(2)}\varepsilon_R(\varphi^{(3)})S^{-1}(\varepsilon_R(\varphi^{(1)}))\otimes h] = [\varphi\otimes h]. 
    \end{align*} 
    Similarly, $[\varphi\otimes h][\varepsilon\otimes 1_W] = [\varphi\otimes h]$.

 (ii) $\Delta,~\varepsilon$ and $S$ are well defined. First, $\Delta$ is well defined: 
    \begin{align*}
        \Delta([\varphi\otimes xh]) & = \sum_{(h),(\varphi)}[\varphi^{(2)}\otimes xh^{(1)}]\otimes[\varphi^{(1)}\otimes h^{(2)}] \\ 
        & = \sum_{(h),(\varphi)}[\varphi^{(2)}(x\rightharpoonup \varepsilon)\otimes h^{(1)}]\otimes[\varphi^{(1)}\otimes h^{(2)}] \\ 
        & = \Delta([\varphi(x\rightharpoonup \varepsilon)\otimes h]), \\
        \Delta([\varphi\otimes yh]) & = \sum_{(h),(\varphi)} [\varphi^{(2)}\otimes h^{(1)}]\otimes [\varphi^{(1)}\otimes yh^{(2)}] \\
        & = \sum_{(h),(\varphi)} [\varphi^{(2)}\otimes h^{(1)}]\otimes [\varphi^{(1)}(\varepsilon \leftharpoonup y)\otimes h^{(2)}] \\
        & = \Delta([\varphi(\varepsilon \leftharpoonup y)\otimes h]). 
    \end{align*}
    Next, $\varepsilon$ is well defined:
    \begin{align*}
        \varepsilon([\varphi(x\rightharpoonup \varepsilon)\otimes h]) & = \langle \varphi(x\rightharpoonup \varepsilon),\varepsilon_R(S^{-1}(h)) \rangle \\
        & = \sum_{(1_W)}\langle \varphi, 1_W^{(1)} \rangle \langle \varepsilon,\varepsilon_R(S^{-1}(h))1_W^{(2)}x \rangle \\
        & = \sum_{(1_W)}\langle \varphi, 1_W^{(1)} \rangle \langle \varepsilon,1_W^{(2)}S^{-1}(\varepsilon_L(h))S^{-1}(x) \rangle  \\
        & = \sum_{(1_W)}\langle \varphi, 1_W^{(1)} \rangle \langle \varepsilon,1_W^{(2)}S^{-1}(\varepsilon_L(xh)) \rangle \\
        & = \sum_{(1_W)}\langle \varphi, 1_W^{(1)} \rangle \langle \varepsilon,\varepsilon_R(S^{-1}(xh))1_W^{(2)} \rangle \\
        & = \langle \varphi,\varepsilon_R(S^{-1}(xh)) \rangle = \varepsilon([\varphi\otimes xh]),\\
        \varepsilon([\varphi(\varepsilon \leftharpoonup y)\otimes h]) & = \langle \varphi(\varepsilon\leftharpoonup y),\varepsilon_R(S^{-1}(h)) \rangle \\
        & = \sum_{(1_W)}\langle \varphi,1_W^{(1)}\rangle \langle \varepsilon,y\varepsilon_R(S^{-1}(h))1_W^{(2)} \rangle  \\
        & = \langle \varphi, yS^{-1}(\varepsilon_L(h))\rangle = \langle \varphi,S^{-1}(\varepsilon_L(h)S(y))\rangle  \\
        & = \langle \varphi,S^{-1}(\varepsilon_L(yh))\rangle = \varepsilon([\varphi\otimes yh]),
    \end{align*}
    where we used 
    \begin{align*}
        x\varepsilon_L(h) &= \sum_{(h)}xh^{(1)}S(h^{(2)}) = \sum_{(xh)}(xh)^{(1)}S((xh)^{(2)}) = \varepsilon_L(xh), \\
        \varepsilon_L(h)S(&y) = \sum_{(h)}h^{(1)}S(yh^{(2)}) =\sum_{(yh)}(yh)^{(1)}S((yh)^{(2)}) = \varepsilon_L(yh). 
    \end{align*}
    Lastly, $S$ is well defined: 
    \begin{align*}
        S([\varphi\otimes xh]) & = \sum_{(\varphi),(h)} [\hat{S}^{-1}(\varphi^{(2)})\otimes S(h^{(2)})] \langle \varphi^{(1)},h^{(3)}\rangle \langle \varphi^{(3)}, S^{-1}(xh^{(1)}) \rangle  \\ 
        & = \sum_{(\varphi),(h)} [\hat{S}^{-1}(\varphi^{(2)})\otimes S(h^{(2)})] \langle \varphi^{(1)},h^{(3)}\rangle \langle \varphi^{(3)}(x\rightharpoonup \varepsilon), S^{-1}(h^{(1)})  \rangle  \\ 
        & = S([\varphi(x\rightharpoonup \varepsilon)\otimes h]), \\
        S([\varphi\otimes yh]) & = \sum_{(\varphi),(h)} [\hat{S}^{-1}(\varphi^{(2)})\otimes S(h^{(2)})] \langle \varphi^{(1)},yh^{(3)}\rangle \langle \varphi^{(3)}, S^{-1}(h^{(1)}) \rangle  \\ 
        & = \sum_{(\varphi),(h)} [\hat{S}^{-1}(\varphi^{(2)})\otimes S(h^{(2)})] \langle \varphi^{(1)}(\varepsilon\leftharpoonup y),h^{(3)}\rangle \langle \varphi^{(3)}, S^{-1}(h^{(1)}) \rangle  \\ 
        & = S([\varphi(\varepsilon \leftharpoonup y)\otimes h]). 
    \end{align*}
    
    (iii) $\Delta$ is coassociative and multiplicative: 
    \begin{align*}
        (\Delta \otimes \id)\comp \Delta([\varphi\otimes h]) & = \sum_{(\varphi), (h)} (\Delta \otimes \id)([\varphi^{(2)} \otimes h^{(1)}] \otimes [\varphi^{(1)} \otimes h^{(2)}]) \\ 
        & = \sum_{(\varphi), (h)} [\varphi^{(3)} \otimes h^{(1)}] \otimes [\varphi^{(2)} \otimes h^{(2)}]\otimes [\varphi^{(1)} \otimes h^{(3)}] \\ 
        & = \sum_{(\varphi), (h)} (\id \otimes \Delta)([\varphi^{(2)} \otimes h^{(1)}] \otimes [\varphi^{(1)} \otimes h^{(2)}]) \\
        & = (\id\otimes\Delta)\comp\Delta([\varphi\otimes h]), \\
        \Delta([\varphi\otimes h])\Delta([\psi\otimes g]) & = \sum_{(\varphi),(h)}\sum_{(\psi),(g)} [\varphi^{(2)}\otimes h^{(1)}][\psi^{(2)}\otimes g^{(1)}]\otimes [\varphi^{(1)}\otimes h^{(2)}][\psi^{(1)}\otimes g^{(2)}] \\
        & = \sum_{(\varphi),(h)}\sum_{(\psi),(g)}[\varphi^{(2)}\psi^{(5)}\otimes h^{(2)}g^{(1)}]\langle \psi^{(4)},S^{-1}(h^{(3)})\rangle \langle \psi^{(6)},h^{(1)} \rangle \\
        & \quad \quad \quad \otimes [\varphi^{(1)}\psi^{(2)}\otimes h^{(5)}g^{(2)}]\langle \psi^{(1)},S^{-1}(h^{(6)})\rangle \langle\psi^{(3)},h^{(4)} \rangle \\
        & = \sum_{(\varphi),(h)}\sum_{(\psi),(g)} [\varphi^{(2)}\psi^{(4)}\otimes h^{(2)}g^{(1)}]\langle \psi^{(1)},S^{-1}(h^{(5)})\rangle  \langle \psi^{(5)},h^{(1)} \rangle  \\
        & \quad \quad \quad \otimes [\varphi^{(1)}\psi^{(2)}\otimes h^{(4)}g^{(2)}] \langle \psi^{(3)},S^{-1}(\varepsilon_L(h^{(3)}))\rangle \\
        & = \sum_{(\varphi),(h)}\sum_{(\psi),(g)} [\varphi^{(2)}\psi^{(3)}\otimes h^{(2)}g^{(1)}]\langle \psi^{(1)},S^{-1}(h^{(5)})\rangle  \langle \psi^{(4)},h^{(1)} \rangle  \\
        & \quad \quad \quad \otimes [\varphi^{(1)}(S^{-1}(\varepsilon_L(h^{(3)}))\rightharpoonup\psi^{(2)})\otimes h^{(4)}g^{(2)}] \\
        & = \sum_{(\varphi),(h)}\sum_{(\psi),(g)} [\varphi^{(2)}\psi^{(3)}\otimes h^{(2)}g^{(1)}]\langle \psi^{(1)},S^{-1}(h^{(5)})\rangle  \langle \psi^{(4)},h^{(1)} \rangle  \\
        & \quad \quad \quad \otimes [\varphi^{(1)}\psi^{(2)}(\varepsilon_L(h^{(3)})\rightharpoonup \varepsilon)\otimes h^{(4)}g^{(2)}] \\
        & = \sum_{(\varphi),(h)}\sum_{(\psi),(g)} [\varphi^{(2)}\psi^{(3)}\otimes h^{(2)}g^{(1)}]\langle \psi^{(1)},S^{-1}(h^{(5)})\rangle  \langle \psi^{(4)},h^{(1)} \rangle  \\
        & \quad \quad \quad \otimes [\varphi^{(1)}\psi^{(2)}\otimes \varepsilon_L(h^{(3)})h^{(4)}g^{(2)}] \\
        & = \sum_{(\varphi),(h)}\sum_{(\psi),(g)} [\varphi^{(2)}\psi^{(3)}\otimes h^{(2)}g^{(1)}]\langle \psi^{(1)},S^{-1}(h^{(4)})\rangle  \langle \psi^{(4)},h^{(1)} \rangle  \\
        & \quad \quad \quad \otimes [\varphi^{(1)}\psi^{(2)}\otimes h^{(3)}g^{(2)}] \\
        & = \Delta([\varphi\otimes h][\psi\otimes g]). 
    \end{align*} 
    
   (iv) The counit property holds: 
    \begin{align*}
        (\varepsilon\otimes \id)\Delta([\varphi\otimes h]) & =  \sum_{(\varphi), (h)} \langle \varphi^{(2)},\varepsilon_R(S^{-1}(h^{(1)})) \rangle  [\varphi^{(1)} \otimes h^{(2)}]\\
        & = \sum_{(h)}\,[S^{-1}(\varepsilon_L(h^{(1)}))\rightharpoonup \varphi)\otimes h^{(2)}]  \\
        & = \sum_{(h)}\,[\varphi(\varepsilon_L(h^{(1)})\rightharpoonup \varepsilon)\otimes h^{(2)}] \\
        & = \sum_{(h)}\,[\varphi\otimes \varepsilon_L(h^{(1)})h^{(2)}] = [\varphi\otimes h], \\
        (\id\otimes\,\varepsilon)\Delta([\varphi\otimes h]) & = \sum_{(\varphi), (h)} [\varphi^{(2)} \otimes h^{(1)}] \langle\varphi^{(1)},\varepsilon_R(S^{-1}(h^{(2)})) \rangle \\ 
        & = \sum_{(h)} [(\varphi\leftharpoonup \varepsilon_R(S^{-1}(h^{(2)}))) \otimes h^{(1)}] \\
        & = \sum_{(h)} [\varphi(\varepsilon\leftharpoonup \varepsilon_R(S^{-1}(h^{(2)}))) \otimes h^{(1)}] \\
        & = \sum_{(h)} [\varphi \otimes  \varepsilon_R(S^{-1}(h^{(2)}))h^{(1)}]  = [\varphi\otimes
         h]. 
    \end{align*}

   (v) Weak multiplicativity of the counit: 
    \begin{align*}
        &\quad \varepsilon([\varphi\otimes h][\psi\otimes g][\theta\otimes k]) \\
        & = \sum_{(h),(\psi)} \sum_{(g),(\theta)} \langle \hat{S}^{-1}(\varepsilon_R(\varphi \psi^{(2)}\theta^{(2)})),h^{(3)}g^{(2)}k \rangle  \langle \theta^{(1)},S^{-1}(h^{(4)}g^{(3)})\rangle \\
        & \quad \quad \times \langle \theta^{(3)},h^{(2)}g^{(1)} \rangle  \langle \psi^{(1)},S^{-1}(h^{(5)}) \rangle \langle \psi^{(3)},h^{(1)} \rangle  \\
        & = \sum_{(h),(\psi)} \sum_{(g),(\theta)}\sum_{(\varepsilon)} \langle \varepsilon^{(1)} \hat{S}^{-1}(\varepsilon_R(\varphi \psi^{(2)}\theta^{(2)})),h^{(3)}g^{(2)} \rangle \langle \varepsilon^{(2)},k \rangle  \langle S^{-1}(\theta^{(1)}),h^{(4)}g^{(3)}\rangle \\
        & \quad \quad \times \langle \theta^{(3)},h^{(2)}g^{(1)} \rangle  \langle \psi^{(1)},S^{-1}(h^{(5)}) \rangle \langle \psi^{(3)},h^{(1)} \rangle  \\
        & = \sum_{(h),(\psi)} \sum_{(\theta)}\sum_{(\varepsilon)} \langle \theta^{(3)}\varepsilon^{(1)} \hat{S}^{-1}(\varepsilon_R(\varphi \psi^{(2)}\theta^{(2)}))\hat{S}^{-1}(\theta^{(1)}),h^{(2)}g \rangle \langle \varepsilon^{(2)},k \rangle   \\
        & \quad \quad \times   \langle \psi^{(1)},S^{-1}(h^{(3)}) \rangle \langle \psi^{(3)},h^{(1)} \rangle  \\
        & = \sum_{(h),(\psi)} \sum_{(\varepsilon)} \langle \hat{S}^{-1}(\varepsilon_L(\theta\varepsilon^{(1)}))\hat{S}^{-1}(\varepsilon_R(\varphi\psi^{(2)})),h^{(2)}g \rangle \langle \varepsilon^{(2)},k \rangle   \\
        & \quad \quad \times   \langle \psi^{(1)},S^{-1}(h^{(3)}) \rangle \langle \psi^{(3)},h^{(1)} \rangle  \\
        & = \sum_{(h),(\psi)} \sum_{(g)} \sum_{(\varepsilon),(\varepsilon')}\langle \hat{S}^{-1}(\varepsilon_L(\theta\varepsilon^{(1)}))\varepsilon'^{(2)},g^{(1)} \rangle \langle \hat{S}^{-1}(\varepsilon_R(\varphi\psi^{(2)})),h^{(2)}g^{(2)} \rangle    \\
        & \quad \quad \times  \langle \varepsilon^{(2)},k \rangle \langle \psi^{(1)},S^{-1}(h^{(3)}) \rangle \langle \psi^{(3)}\varepsilon'^{(1)},h^{(1)} \rangle  \\
        & = \sum_{(h),(\psi)} \sum_{(\theta),(g)} \sum_{(\varepsilon)}\langle \theta^{(3)}\varepsilon^{(1)}\hat{S}^{-1}(\varepsilon_R(\psi^{(4)}\theta^{(2)}))\hat{S}^{-1}(\theta^{(1)}),g^{(1)} \rangle \langle \hat{S}^{-1}(\varepsilon_R(\varphi\psi^{(2)})),h^{(2)}g^{(2)} \rangle    \\
        & \quad \quad \times  \langle \varepsilon^{(2)},k \rangle \langle \psi^{(1)},S^{-1}(h^{(3)}) \rangle \langle \psi^{(3)},h^{(1)} \rangle  \\
        & = \sum_{(h),(\psi)} \sum_{(\theta),(g)} \langle \hat{S}^{-1}(\varepsilon_R(\varphi\psi^{(2)})),h^{(2)}g^{(4)} \rangle   \langle \psi^{(1)},S^{-1}(h^{(3)}) \rangle \langle \psi^{(3)},h^{(1)} \rangle  \\
        & \quad \quad \times   \langle \hat{S}^{-1}(\varepsilon_R(\psi^{(4)}\theta^{(2)})),g^{(2)}k \rangle  \langle \theta^{(1)},S^{-1}(g^{(3)}) \rangle \langle \theta^{(3)},g^{(1)}\rangle \\
        & = \sum_{(\psi),(g)}\varepsilon([\varphi\otimes h][\psi^{(1)}\otimes g^{(2)}])\varepsilon([\psi^{(2)}\otimes g^{(1)}][\theta\otimes k]). 
    \end{align*}
    Similarly, one can show the other identity for weak multiplicativity of the counit. 

    (vi) Weak comultiplicativity of the unit: 
    \begin{align*}
        &\quad (\Delta([\varepsilon\otimes 1_W])\otimes[\varepsilon\otimes 1_W])([\varepsilon\otimes 1_W]\otimes\Delta([\varepsilon\otimes 1_W])) \\
        & = \sum_{(1_W),(\varepsilon)}\sum_{(1'_W),(\varepsilon')}[\varepsilon^{(2)}\otimes 1_W^{(1)}]\otimes [\varepsilon^{(1)}\otimes 1_W^{(2)}][\varepsilon'^{(2)}\otimes 1_W'^{(1)}]\otimes [\varepsilon'^{(1)}\otimes 1_W'^{(2)}] \\
        & = \sum_{(1_W),(\varepsilon)}\sum_{(1'_W),(\varepsilon')}[\varepsilon^{(2)}\otimes 1_W^{(1)}]\otimes [\varepsilon^{(1)}\varepsilon'^{(2)}\otimes 1_W^{(2)}1_W'^{(1)}]\otimes  [\varepsilon'^{(1)}\otimes 1_W'^{(2)}] \\ 
        & = \sum_{(1_W),(\varepsilon)}[\varepsilon^{(3)}\otimes 1_W^{(1)}]\otimes [\varepsilon^{(2)}\otimes 1_W^{(2)}]\otimes [\varepsilon^{(1)},1_W^{(3)}] \\
        & = ((\Delta\otimes \id)\comp\Delta)([\varepsilon\otimes 1_W]). 
    \end{align*}
    Similarly, one can show the other identity for weak comultiplicativity of the unit. 

    (vii) The above computation shows that $D(W)$ is a weak bialgebra. In order to prove it is a weak Hopf algebra, let us compute the left counit:
    \begin{align}
        &\quad \varepsilon_L([\varphi\otimes h]) \nonumber  \\
        & = \sum_{(1_W),(\varepsilon)}\varepsilon([\varepsilon^{(2)}\otimes 1_W^{(1)}][\varphi\otimes h])[\varepsilon^{(1)}\otimes 1_W^{(2)}] \nonumber  \\
        &= \sum_{(1_W),(\varepsilon)} \sum_{(\varphi)}\langle \varepsilon^{(2)}\varphi^{(2)},\varepsilon_R(S^{-1}(1_W^{(2)}h)) \rangle \langle \varphi^{(1)},S^{-1}(1_W^{(3)}) \rangle \langle \varphi^{(3)},1_W^{(1)} \rangle  [\varepsilon^{(1)}\otimes 1_W^{(4)}]  \nonumber  \\
        &= \sum_{(1_W),(\varepsilon)} \sum_{(\varphi)}\sum_{(1_W')}\langle \varepsilon^{(2)},1_W'^{(1)}\rangle \langle \varphi^{(2)},\varepsilon_R(S^{-1}(1_W^{(2)}h)) 1_W'^{(2)} \rangle  \nonumber  \\
        & \quad \quad \times \langle \varphi^{(1)},S^{-1}(1_W^{(3)}) \rangle \langle \varphi^{(3)},1_W^{(1)} \rangle  [\varepsilon^{(1)}\otimes 1_W^{(4)}] \nonumber  \\
        & = \sum_{(1_W),(\varepsilon)} \sum_{(1_W')}\langle \varepsilon^{(2)},1_W'^{(1)}\rangle \langle \varphi,S^{-1}(1_W^{(3)})\varepsilon_R(S^{-1}(1_W^{(2)}h)) 1_W'^{(2)} 1_W^{(1)}\rangle [\varepsilon^{(1)}\otimes 1_W^{(4)}]  \nonumber  \\
        & = \sum_{(1_W),(\varepsilon)} \sum_{(1_W')}\langle \varepsilon^{(2)},1_W'^{(1)}\rangle \langle \varphi,S^{-1}(\varepsilon_R(1_W^{(3)}))S^{-1}(\varepsilon_L(h))S^{-1}(1_W^{(2)}) 1_W'^{(2)} 1_W^{(1)}\rangle [\varepsilon^{(1)}\otimes 1_W^{(4)}]  \nonumber  \\
        & = \sum_{(1_W)} \sum_{(1_W')} \langle \varphi,S^{-1}(\varepsilon_L(h))S^{-1}(S(1_W'^{(2)})1_W^{(2)})  1_W^{(1)}\rangle [\varepsilon\otimes S(1_W'^{(1)})1_W^{(3)}]  \nonumber  \\
        & = \sum_{(1_W)}  \langle \varphi,S^{-1}(\varepsilon_L(h))S^{-1}(1_W^{(2)})  1_W^{(1)}\rangle [\varepsilon\otimes 1_W^{(3)}] \nonumber   \\
        & =\sum_{(1_W)} \langle \varphi,S^{-1}(\varepsilon_L(h))S^{-1}(1_W^{(1)}) \rangle [\varepsilon\otimes 1_W^{(2)}]. \label{eq:left-counit-double}
    \end{align}
    Therefore, we have 
    \begin{align*}
        &\quad \sum_{(\varphi),(h)}[\varphi^{(2)}\otimes h^{(1)}]S([\varphi^{(1)}\otimes h^{(2)}]) \\
        & = \sum_{(\varphi),(h)}[\varphi^{(6)}\hat{S}^{-1}(\varphi^{(3)})\otimes h^{(2)}S(h^{(5)})] \langle \hat{S}^{-1}(\varphi^{(4)}),S^{-1}(h^{(3)}) \rangle\\
        & \quad \quad \times \langle \hat{S}^{-1}(\varphi^{(2)}),h^{(1)} \rangle  \langle \varphi^{(1)},h^{(6)} \rangle \langle \varphi^{(5)}, S^{-1}(h^{(4)}) \rangle \\
        & = \sum_{(\varphi),(h)}[\varphi^{(4)}\hat{S}^{-1}(\varphi^{(2)})\otimes h^{(2)}S(h^{(4)})] \langle \hat{S}^{-1}(\varphi^{(3)}),S^{-1}(\varepsilon_L(h^{(3)})) \rangle \\
        & \quad \quad \times \langle \varphi^{(1)},h^{(5)}S^{-1}(h^{(1)}) \rangle \\
        & = \sum_{(\varphi),(h)}\sum_{(1_W)}[\varphi^{(4)}\hat{S}^{-1}(\varphi^{(2)})\otimes 1_W^{(1)}\varepsilon_L(h^{(2)})] \langle \hat{S}^{-1}(\varphi^{(3)}),S^{-1}(1_W^{(2)}) \rangle \\
        & \quad \quad \times \langle \varphi^{(1)},S^{-1}(h^{(1)}S(h^{(3)})) \rangle \\
        & = \sum_{(\varphi)}\sum_{(1_W),(1_W')}[\varphi^{(4)}\hat{S}^{-1}(\varphi^{(2)})\otimes 1_W^{(1)}1_W'^{(2)}] \langle \hat{S}^{-1}(\varphi^{(3)}),S^{-1}(1_W^{(2)}) \rangle \\
        & \quad\quad \times \langle \varphi^{(1)},S^{-1}(1_W'^{(1)}\varepsilon_L(h)) \rangle \\
        & = \sum_{(\varphi)}\sum_{(1_W),(1_W')}[\varphi^{(3)}\hat{S}^{-1}(\varphi^{(2)})(\varepsilon\leftharpoonup S^{-1}(1_W^{(2)}))\otimes 1_W^{(1)}1_W'^{(2)}]  \\
        & \quad\quad \times \langle \varphi^{(1)},S^{-1}(1_W'^{(1)}\varepsilon_L(h)) \rangle \\
        & = \sum_{(\varphi)}\sum_{(1_W')}[\varphi^{(3)}\hat{S}^{-1}(\varphi^{(2)})\otimes 1_W'^{(2)}] \langle \varphi^{(1)},S^{-1}(1_W'^{(1)}\varepsilon_L(h)) \rangle \\
        & =\sum_{(1_W')} \langle \varphi,S^{-1}(\varepsilon_L(h))S^{-1}(1_W'^{(1)}) \rangle [\varepsilon\otimes 1_W'^{(2)}] = \varepsilon_L([\varphi\otimes h]).
    \end{align*}
    Similarly, one can show that 
    \begin{align}
        \varepsilon_R([\varphi\otimes h]) &= \sum_{(\varepsilon)}\langle \varepsilon^{(1)}\hat{S}^{-1}(\varepsilon_R(\varphi)),h\rangle [\varepsilon^{(2)}\otimes 1_W], \label{eq:right-counit-double}\\
        \text{and}\quad \varepsilon_R([\varphi\otimes &h])=\sum_{(\varphi),(h)}S([\varphi^{(2)}\otimes h^{(1)}])[\varphi^{(1)}\otimes h^{(2)}].
    \end{align}
    Finally, we compute 
    \begin{align*}
        & \quad \mu\comp (\id\otimes \mu)\comp(S\otimes \id\otimes S)\comp(\id\otimes\Delta)\comp\Delta([\varphi\otimes h]) \\
        & = \mu\comp (S\otimes \varepsilon_L)\comp\Delta([\varphi\otimes h]) \\
        & = \sum_{(\varphi),(h)} S([\varphi^{(2)}\otimes h^{(1)}])\varepsilon_L([\varphi^{(1)}\otimes h^{(2)}]) \\
        & = \sum_{(\varphi),(h)} \sum_{(1_W)} [\hat{S}^{-1}(\varphi^{(3)})\otimes S(h^{(2)})]\langle \varphi^{(2)},h^{(3)} \rangle \langle \varphi^{(4)},S^{-1}(h^{(1)}) \rangle \\
        & \quad \quad \times \langle \varphi^{(1)},S^{-1}(\varepsilon_L(h^{(4)}))S^{-1}(1_W^{(1)}) \rangle [\varepsilon\otimes 1_W^{(2)}] \\
        & = \sum_{(\varphi),(h)} \sum_{(1_W)} [\hat{S}^{-1}(\varphi^{(2)})\otimes S(h^{(2)})1_W^{(2)}]\langle \varphi^{(3)},S^{-1}(h^{(1)}) \rangle \langle \varphi^{(1)},S^{-1}(1_W^{(1)})h^{(3)} \rangle \\
        & = \sum_{(\varphi),(h)} [\hat{S}^{-1}(\varphi^{(2)})\otimes S(h^{(2)})]\langle \varphi^{(3)},S^{-1}(h^{(1)}) \rangle \langle \varphi^{(1)},h^{(3)} \rangle= S([\varphi\otimes h]),
    \end{align*}
    as desired. Hence we finish the proof that $D(W)$ is a weak Hopf algebra.

The following result has been shown in \cite{majid2000foundations} for Hopf algebras and in \cite{nikshych2003invariants} for weak Hopf algebras (note that they used a different (yet equivalent) form of the quantum double). We include proof for the form of the quantum double we used here for completeness. 

\begin{proposition}\label{prop:quasitriangular-QD}
    Let $W$ be a weak Hopf algebra. Then the quantum double $D(W)$ in Definition \ref{def:quantum-double} has a quasitriangular structure given by 
    \begin{equation}
        R = \sum_i [\varepsilon\otimes x_i]\otimes[x^i\otimes 1_W], \quad \tilde{R} = \sum_j[\varepsilon\otimes S(x_j)]\otimes[x^j\otimes 1_W],
    \end{equation}
    where $\{x_i\}$ and $\{x^i\}$ are dual bases of $W$ and $\hat{W}$. 
\end{proposition}

\begin{proof}
    With the identifications $W\simeq (\varepsilon \otimes W)/J$ and $\hat{W}^{\rm cop}\simeq (\hat{W}^{\rm cop}\otimes 1_W)/J$, the identities $(\id \otimes \Delta)(R)=R_{13}R_{12}$ and $(\Delta \otimes \id)(R)=R_{13}R_{23}$ are equivalent to 
    \begin{align}
        &\sum_i\sum_{(x^i)}x_i\otimes x^{i,(2)} \otimes x^{i,(1)} = \sum_{i,j} x_ix_j\otimes x^j\otimes x^i, \label{eq:R1-eva} \\
        &\sum_i\sum_{(x_i)}x_i^{(1)}\otimes x_i^{(2)}\otimes x^i  = \sum_{i,j} x_i\otimes x_j \otimes x^ix^j. 
    \end{align}
    We prove the first one and the second one is similarly proved. To this end, by evaluating both sides on an element $\varphi\in \hat{W}$ in the first factor, we have 
    \begin{equation}
        \begin{aligned}
        \text{LHS} & = \sum_i\sum_{(x^i)}\langle \varphi, x_i \rangle x^{i,(2)}\otimes x^{i,(1)} = \sum_i \Delta^{\rm op}(\langle \varphi,x_i \rangle x^i) = \Delta^{\rm op}(\varphi), \\
        \text{RHS} & = \sum_{i,j} \langle \varphi,x_ix_j\rangle x^j\otimes x^i = \sum_{i,j}\sum_{(\varphi)}\langle \varphi^{(2)},x_j\rangle x^j \otimes \langle \varphi^{(1)},x_i \rangle x^i =\Delta^{\rm op}(\varphi),  
    \end{aligned}
    \end{equation}
    where we have used $\varphi = \sum_i\langle\varphi,x_i\rangle x^i$. By evaluating both sides on general elements $a\otimes b\in W\otimes W$ in the second and third factors also gives the same result, hence Eq.~\eqref{eq:R1-eva} holds. Now we show $\Delta^{\rm op}([\varphi\otimes h])R=R\Delta([\varphi\otimes h])$ for all $[\varphi\otimes h]\in D(W)$. In fact, using $h = \sum_i\langle x^i,h\rangle x_i$ and $\varphi = \sum_i\langle\varphi,x_i\rangle x^i$, we get
    \begin{align}
        &\quad R(\Delta[\varphi\otimes h]) \nonumber \\
        & = \sum_{i,(x_i)}\sum_{(h),(\varphi)} [\varphi^{(3)}\otimes x_i^{(2)}h^{(1)}]\otimes [x^i\varphi^{(1)}\otimes h^{(2)}]\langle\varphi^{(2)},S^{-1}(x_i^{(3)}) \rangle \langle\varphi^{(4)},x_i^{(1)} \rangle \nonumber \\
        & = \sum_{i,(x_i)}\sum_{(h),(\varphi)}\sum_j\, [\varphi^{(3)}\otimes \langle x^j,x_i^{(2)}\rangle x_jh^{(1)}]\otimes [x^i\varphi^{(1)}\otimes h^{(2)}]\langle\varphi^{(2)},S^{-1}(x_i^{(3)}) \rangle \langle\varphi^{(4)},x_i^{(1)} \rangle \nonumber \\
        & = \sum_{(h),(\varphi)} \sum_j\,[\varphi^{(3)}\otimes x_jh^{(1)}]\otimes [\varphi^{(4)}x^jS^{-1}(\varphi^{(2)})\varphi^{(1)}\otimes h^{(2)}]\nonumber \\
        & =  \sum_{(h),(\varphi)}\sum_{(1_W)}\sum_j\,[\varphi^{(2)}\otimes x_jh^{(1)}]\otimes [\varphi^{(3)}x^j(\varepsilon \leftharpoonup 1_W^{(1)})\langle S^{-1}(\varepsilon_R(\varphi^{(1)})),1_W^{(2)} \rangle \otimes h^{(2)}]\nonumber \\
        & =  \sum_{(h),(\varphi)}\sum_{(1_W)}\sum_j\,[\varphi^{(2)}\otimes x_jh^{(1)}]\otimes [\varphi^{(3)}x^j \otimes\langle S^{-1}(\varepsilon_R(\varphi^{(1)})),1_W^{(2)} \rangle 1_W^{(1)}h^{(2)}]\nonumber \\
        & =  \sum_{(h),(\varphi)}\sum_{(1_W),(\varepsilon)}\sum_j\,[\varphi^{(1)}\varepsilon^{(2)}\otimes x_jh^{(1)}]\otimes [\varphi^{(2)}x^j \otimes\langle S^{-1}(\varepsilon^{(1)}),1_W^{(2)} \rangle 1_W^{(1)}h^{(2)}] \\
        & =  \sum_{(h),(\varphi)}\sum_{(\varepsilon)}\sum_j\,[\varphi^{(1)}\varepsilon^{(2)}\otimes x_jh^{(1)}]\otimes [\varphi^{(2)}x^j \otimes\langle S^{-1}(\varepsilon^{(1)}),\varepsilon_L(h^{(3)}) \rangle h^{(2)}]\nonumber \\
        & =  \sum_{(h),(\varphi)}\sum_j\,[\varphi^{(1)}(\varepsilon\leftharpoonup S^{-1}(\varepsilon_L(h^{(3)})))\otimes x_jh^{(1)}]\otimes [\varphi^{(2)}x^j \otimes h^{(2)}]\nonumber \\
        & =  \sum_{(h),(\varphi)}\sum_j\sum_i\,[\varphi^{(1)}\otimes h^{(4)}\langle x^i,S^{-1}(h^{(3)}) x_jh^{(1)}\rangle x_i]\otimes [\varphi^{(2)}x^j \otimes h^{(2)}]\nonumber \\
        & =  \sum_{(h),(\varphi)}\sum_{i,(x^i)}\sum_j\,[\varphi^{(1)}\otimes h^{(4)}x_i]\otimes [\varphi^{(2)}\langle x^{i,(2)},x_j\rangle x^j \otimes h^{(2)}]\langle x^{i,(1)},S^{-1}(h^{(3)})\rangle \langle x^{i,(3)},h^{(1)}\rangle\nonumber \\
        & = \sum_{(h),(\varphi)} \sum_{i,(x^i)}[\varphi^{(1)}\otimes h^{(4)}x_i]\otimes [\varphi^{(2)}x^{i,(2)}\otimes h^{(2)}] \langle x^{i,(1)}, S^{-1}(h^{(3)}) \rangle \langle x^{i,(3)}, h^{(1)}\rangle \nonumber \\
        & = \Delta^{\rm op}([\varphi\otimes h])R, \nonumber
    \end{align}
    where in the sixth equality, we used $\sum_{(\varphi)}\varepsilon_R(\varphi^{(1)})\otimes \varphi^{(2)} = \sum_{(\varepsilon)}\varepsilon^{(1)}\otimes\varphi\varepsilon^{(2)}$; in the seventh equality, we used $\sum_{(h)}h^{(1)}\otimes \varepsilon_L(h^{(2)}) = \sum_{(1_W)}1_W^{(1)}h\otimes 1_W^{(2)}$. It remains to show that the element $\tilde{R}$ satisfies $\tilde{R}R = \Delta(1_{D(W)})$ and $R\tilde{R} = \Delta^{\rm op}(1_{D(W)})$. The first identity is equivalent to 
    \begin{equation}
        \begin{aligned}
            &\quad\sum_{i,j} \,[\varepsilon\otimes S(x_i)x_j]\otimes [x^ix^j\otimes 1_W] \\
            &= \sum_{(\varepsilon),(1_W)}\sum_{(\varepsilon'),(1'_W)} [\varepsilon\otimes \langle \varepsilon^{(2)},1_W'^{(2)}\rangle 1_W'^{(1)}1_W^{(1)} ]
            \otimes  [\langle \varepsilon'^{(2)}, 1_W^{(2)} \rangle \varepsilon^{(1)}\varepsilon'^{(1)}\otimes 1_W],
        \end{aligned}
    \end{equation}
    which is equivalent to 
    \begin{equation}
        \quad\sum_{i,j} \,x^ix^j\otimes S(x_i)x_j= \sum_{(\varepsilon),(1_W)}\sum_{(\varepsilon'),(1'_W)}\langle \varepsilon'^{(2)}, 1_W^{(2)} \rangle \varepsilon^{(1)}\varepsilon'^{(1)} \otimes  \langle \varepsilon^{(2)},1_W'^{(2)}\rangle 1_W'^{(1)}1_W^{(1)}   \label{eq:equi-equa}
    \end{equation}
    when regarded as an equality in $\hat{W}^{\rm cop}\otimes W$. To show this equality, evaluating both sides for a general $h\in W$ in the first factor, one has 
    \begin{align}
        \text{RHS} & = \sum_{(\varepsilon),(1_W)}\sum_{(\varepsilon'),(1'_W)}\langle \varepsilon'^{(2)}, 1_W^{(2)} \rangle \langle \varepsilon^{(1)}\varepsilon'^{(1)}, h \rangle  \langle \varepsilon^{(2)},1_W'^{(2)}\rangle 1_W'^{(1)}1_W^{(1)} \nonumber  \\
        & = \sum_{(1'_W)}    1_W'^{(1)}\langle \varepsilon,h^{(1)}1_W'^{(2)}\rangle \sum_{(1_W)}1_W^{(1)}\langle \varepsilon', h^{(2)}1_W^{(2)} \rangle  \\
        & = \varepsilon_R(h^{(1)})\varepsilon_R(h^{(2)}) = \sum_{(h)} S(h^{(1)})h^{(2)} = \text{LHS}. \nonumber 
    \end{align}
    When evaluating both sides for a general element $\varphi\in\hat{W}$ in the second factor, the results are still the same. Thus Eq.~\eqref{eq:equi-equa} holds. The second identity is proved similarly. 
\end{proof}


\section{Proofs of Lemma \ref{lem:ribbon-local-ope} and Lemma \ref{lem:comm-rel}} \label{app:ribbon}

\subsection*{Proof of Lemma \ref{lem:ribbon-local-ope}.} 

First note that the dual space $D(W)^\vee$ can be identified with a subspace of $(\hat{W}^{\rm cop}\otimes W)^\vee\simeq W^{\rm op}\otimes \hat{W}$ consisting of elements $g\otimes\psi\in W^{\rm op}\otimes \hat{W}$ which vanish on the ideal $J\subset \hat{W}^{\rm cop}\otimes W$, i.e., $g\otimes \psi\in D(W)^\vee$ if and only if 
\begin{equation} \label{eq:aaa} 
    \varphi(g)\psi(xh)  = \varphi(x\rightharpoonup \varepsilon)(g)\psi(h), \quad
    \varphi(g)\psi(yh)  = \varphi(\varepsilon\leftharpoonup y)(g)\psi(h),  
\end{equation}
for any $x\in W_L,\, y\in W_R$ and $h\in W,\, \varphi\in\hat{W}$. 

First consider the case $\rho = \tau_R^+$: 
\begin{equation*}
    \begin{tikzpicture}
        \draw[line width = 0.5pt, red] (0.75,0.75) -- (0,0);
        \draw[line width = 0.5pt, red] (0.75,0.75) -- (1.5,0);
        \draw[-latex,black] (0,0) -- (1.5,0);
        \draw[-latex,black] (0,1.5) -- (0,0);
        \draw[-latex,black] (1.5,0) -- (1.5,1.5);
        \draw[-latex,black] (1.5,1.5) -- (0,1.5);
        \draw[-latex,black] (0,0) -- (-1.5,0);
        \draw[-latex,black] (0,-1.5) -- (0,0);
        \draw[-stealth,gray, line width=2pt] (0.5,0.25) -- (1,0.25); 
        \node[ line width=0.2pt, dashed, draw opacity=0.5] (a) at (-0.3,0.75){$x_4$};
        \node[ line width=0.2pt, dashed, draw opacity=0.5] (a) at (1.8,0.75){$x_2$};
        \node[ line width=0.2pt, dashed, draw opacity=0.5] (a) at (0.8,1.75){$x_3$};
        \node[ line width=0.2pt, dashed, draw opacity=0.5] (a) at (0.8,-0.25){$x_1$};
        \node[ line width=0.2pt, dashed, draw opacity=0.5] (a) at (0.25,0.58){$s$};
        \node[ line width=0.2pt, dashed, draw opacity=0.5] (a) at (-0.75,-0.2){$x_5$};
        \node[ line width=0.2pt, dashed, draw opacity=0.5] (a) at (0.3,-0.85){$x_6$};
        \node[ line width=0.2pt, dashed, draw opacity=0.5] (a) at (5,0){$F^{g,\psi}(\tau_R^+)|x_1\rangle = \varepsilon(g)T_-^\psi|x_1\rangle$};
    \end{tikzpicture}
\end{equation*}
In this case, one has 
\begin{align*}
    &\quad F^{xg,\psi}(\rho)A^h(s)|x_4,x_5,x_6,x_1\rangle \\
    & = \sum_{(h)}\sum_{(x_1)} F^{xg,\psi}(\rho)|h^{(1)}x_4,x_5S^{-1}(h^{(2)}),h^{(3)}x_6,x_1S^{-1}(h^{(4)})\rangle \\
    & = \sum_{(h)} \sum_{(x_1)}\varepsilon(xg)\psi(S(x_1^{(1)}S^{-1}(h^{(5)})))|h^{(1)}x_4,x_5S^{-1}(h^{(2)}),h^{(3)}x_6,x_1^{(2)}S^{-1}(h^{(4)})\rangle \\
    & = \sum_{(h)} \sum_{(x_1)}\varepsilon(S(x)g)\psi(S(x_1^{(1)}S^{-1}(h^{(5)})))|h^{(1)}x_4,x_5S^{-1}(h^{(2)}),h^{(3)}x_6,x_1^{(2)}S^{-1}(h^{(4)})\rangle \\
    & = \sum_{(h)} \sum_{(x_1)}\varepsilon(g)\psi(S(x_1^{(1)}S^{-1}(h^{(5)})x))|h^{(1)}x_4,x_5S^{-1}(h^{(2)}),h^{(3)}x_6,x_1^{(2)}S^{-1}(h^{(4)})\rangle \\
    & = \sum_{(h)} F^{g,\psi}(\rho)|h^{(1)}x_4,x_5S^{-1}(h^{(2)}),h^{(3)}x_6,x_1S^{-1}(S(x)h^{(4)})\rangle \\
    & = F^{g,\psi}(\rho)A^{S(x)h}(s)|x_4,x_5,x_6,x_1\rangle,
\end{align*}
and
\begin{align*}
    &\quad F^{g,\psi}(\rho)B^{\varphi\leftharpoonup x}(s) |x_1,x_2,x_3,x_4\rangle \\
    & = \sum_{(x_i)}F^{g,\psi}(\rho) \varphi(xS(x_1^{(1)})S(x_2^{(1)})S(x_3^{(1)})S(x_4^{(1)}))|x_1^{(2)},x_2^{(2)},x_3^{(2)},x_4^{(2)} \rangle \\
    & = \sum_{(x_i)}F^{g,\psi}(\rho) \varphi(S(x_1^{(1)})S(x_2^{(1)})S(x_3^{(1)})S(x_4^{(1)}))|x_1^{(2)}x,x_2^{(2)},x_3^{(2)},x_4^{(2)} \rangle \\
    & = \sum_{(x_i)}\varepsilon(g)\psi(S(x_1^{(2)}x)) \varphi(S(x_1^{(1)})S(x_2^{(1)})S(x_3^{(1)})S(x_4^{(1)}))|x_1^{(3)},x_2^{(2)},x_3^{(2)},x_4^{(2)} \rangle \\
    & = \sum_{(x_i)}\varepsilon(S(x)g)\psi(S(x_1^{(2)})) \varphi(S(x_1^{(1)})S(x_2^{(1)})S(x_3^{(1)})S(x_4^{(1)}))|x_1^{(3)},x_2^{(2)},x_3^{(2)},x_4^{(2)} \rangle \\
    & = \sum_{(x_i)}\varepsilon(xg)\psi(S(x_1^{(2)})) \varphi(S(x_1^{(1)})S(x_2^{(1)})S(x_3^{(1)})S(x_4^{(1)}))|x_1^{(3)},x_2^{(2)},x_3^{(2)},x_4^{(2)} \rangle \\
    & = \sum_{(x_i)}F^{xg,\psi}(\rho) \varphi(S(x_1^{(1)})S(x_2^{(1)})S(x_3^{(1)})S(x_4^{(1)}))|x_1^{(2)},x_2^{(2)},x_3^{(2)},x_4^{(2)} \rangle \\
    & = F^{xg,\psi}(\rho)B^\varphi(s)|x_1,x_2,x_3,x_4\rangle.
\end{align*}
The general cases follow from the recursive formula of ribbon operators:
\begin{align*}
    F^{xg,\varphi}(\rho)A^h(s) & = \sum F^{xg^{(1)},\hat{k}}(\rho_1)F^{S^{-1}(k^{(3)})g^{(2)}k^{(1)},\psi(k^{(2)}\bullet)}(\rho_2)A^h(s) \\ 
    & = \sum F^{xg^{(1)},\hat{k}}(\rho_1)A^h(s)F^{S^{-1}(k^{(3)})g^{(2)}k^{(1)},\psi(k^{(2)}\bullet)}(\rho_2) \\ 
    & = \sum F^{g^{(1)},\hat{k}}(\rho_1)A^{S(x)h}(s)F^{S^{-1}(k^{(3)})g^{(2)}k^{(1)},\psi(k^{(2)}\bullet)}(\rho_2) \\ 
    & = \sum F^{g^{(1)},\hat{k}}(\rho_1)F^{S^{-1}(k^{(3)})g^{(2)}k^{(1)},\varphi(k^{(2)}\bullet)}(\rho_2) A^{S(x)h}(s)\\ 
    & = F^{g,\psi}(\rho)A^{S(x)h}(s). 
\end{align*}
The other identity is verified similarly. Thus this finishes the proof of Lemma \ref{lem:ribbon-local-ope}.  \qed

\subsection*{Proof of Lemma \ref{lem:comm-rel}.}

Although we are in the context of weak Hopf algebras, the proof of Lemma \ref{lem:comm-rel} is essentially the same as that in Ref.~\cite{chen2021ribbon}. Following the technique there, we will only prove the identities \eqref{eq:comm4} and \eqref{eq:comm5} for an example to show how to deal with the computation in our setting. 

$\bullet$ \eqref{eq:comm4} for short ribbons:
    \begin{equation*}
        \begin{tikzpicture}
            \draw[line width = 0.5pt, red] (1,1) -- (0,0);
            \draw[line width = 0.5pt, red] (1,1) -- (2,0);
            \draw[line width = 0.5pt, red] (2,0) -- (3,1);
            \draw[line width = 0.5pt, dashed, black] (1,1) -- (3,1);
            \draw[-latex,black] (0,0) -- (2,0);
            \draw[-latex,black] (0,0) -- (0,2);
            \draw[-latex,black] (2,0) -- (2,2);
            \draw[-latex,black] (0,2) -- (2,2);
            \draw[-stealth,gray, line width=3pt] (1.2,0.53) -- (1.85,0.53); 
            \node[ line width=0.2pt, dashed, draw opacity=0.5] (a) at (-0.3,1){$x_4$};
            \node[ line width=0.2pt, dashed, draw opacity=0.5] (a) at (2.3,1.3){$x_2$};
            \node[ line width=0.2pt, dashed, draw opacity=0.5] (a) at (1,2.25){$x_3$};
            \node[ line width=0.2pt, dashed, draw opacity=0.5] (a) at (1.2,-0.3){$x_1$};
            \node[ line width=0.2pt, dashed, draw opacity=0.5] (a) at (0.35,0.65){$s_0$};
            \node[ line width=0.2pt, dashed, draw opacity=0.5] (a) at (-2,1){$\rho= \tau_R\cup\tilde{\tau}_L:$};
        \end{tikzpicture}
    \end{equation*}
    \begin{align*}
        &\quad B^\psi(s_0)F^{h,\varphi}(\rho)|x_1,x_2,x_3,x_4\rangle \\
        & = \sum_{k,(k),(h)}B^\psi(s_0)F^{h^{(1)},\hat{k}}(\tau_R)F^{S^{-1}(k^{(3)})h^{(2)}k^{(1)},\varphi(k^{(2)}\bullet)}({\tilde{\tau}_L})|x_1,x_2,x_3,x_4\rangle \\
        & = \sum_{k,(k),(h)}B^\psi(s_0)F^{h^{(1)},\hat{k}}(\tau_R)\varepsilon(\varphi(k^{(2)})\bullet)|x_1,x_2S^{-1}(k^{(3)})h^{(2)}k^{(1)},x_3,x_4\rangle \\
        & = \sum_{k,(k),(h)}\sum_{(x_1)}B^\psi(s_0)\varepsilon(h^{(1)})\varepsilon(\varphi(k^{(2)}\bullet))\hat{k}(S(x_1^{(1)})) |x_1^{(2)},x_2S^{-1}(k^{(3)})h^{(2)}k^{(1)},x_3,x_4\rangle \\
        & = \sum_{k,(k)}\sum_{(x_1)}B^\psi(s_0)\varphi(k^{(2)})\hat{k}(S(x_1^{(1)}))|x_1^{(2)},x_2S^{-1}(k^{(3)})hk^{(1)},x_3,x_4\rangle \\
        & = \sum_{(x_1)}B^\psi(s_0)\varphi(S(x_1^{(2)}))|x_1^{(4)},x_2x_1^{(1)}hS(x_1^{(3)}),x_3,x_4\rangle \\
        & = \sum_{(h),(x_i)}\varphi(S(x_1^{(3)}))\psi(S(x_1^{(6)})S(x_2^{(1)}x_1^{(1)}h^{(1)}S(x_1^{(5)}))x_3^{(2)}x_4^{(2)})\\
        & \quad \quad |x_1^{(7)},x_2^{(2)}x_1^{(2)}h^{(2)}S(x_1^{(4)}),x_3^{(1)},x_4^{(1)}\rangle \\
        & = \sum_{(h),(x_i)}\varphi(S(x_1^{(3)}))\psi(S(\varepsilon_R(x_1^{(5)}))S(h^{(1)})S(x_1^{(1)})S(x_2^{(1)})x_3^{(2)}x_4^{(2)})\\
        & \quad \quad |x_1^{(6)},x_2^{(2)}x_1^{(2)}h^{(2)}S(x_1^{(4)}),x_3^{(1)},x_4^{(1)}\rangle \\
        & = \sum_{(h),(x_i)}\varphi(S(x_1^{(3)}))\psi(S(h^{(1)})S(x_1^{(1)})S(x_2^{(1)})x_3^{(2)}x_4^{(2)})\\
        & \quad \quad |x_1^{(6)},x_2^{(2)}x_1^{(2)}h^{(2)}S(\varepsilon_R(x_1^{(5)}))S(x_1^{(4)}),x_3^{(1)},x_4^{(1)}\rangle \\
        & = \sum_{(h),(x_i)}\varphi(S(x_1^{(3)}))\psi(S(h^{(1)})S(x_1^{(1)})S(x_2^{(1)})x_3^{(2)}x_4^{(2)})\\
        & \quad \quad |x_1^{(5)},x_2^{(2)}x_1^{(2)}h^{(2)}S(x_1^{(4)}),x_3^{(1)},x_4^{(1)}\rangle \\
        & = \sum_{(h),(x_i)}\sum_{k,(k)}\varepsilon(h^{(2)})\varphi(k^{(2)})\hat{k}(S(x_1^{(2)}))\psi(S(h^{(1)})S(x_1^{(1)})S(x_2^{(1)})x_3^{(2)}x_4^{(2)})\\
        & \quad \quad |x_1^{(3)},x_2^{(2)}S^{-1}(k^{(3)})h^{(3)}k^{(1)},x_3^{(1)},x_4^{(1)}\rangle \\
        & = \sum_{(h),(x_i)}\sum_{k,(k)}F^{h^{(2)},\hat{k}}(\tau_R)\varphi(k^{(2)})\psi(S(h^{(1)})S(x_1^{(1)})S(x_2^{(1)})x_3^{(2)}x_4^{(2)})\\
        & \quad \quad |x_1^{(2)},x_2^{(2)}S^{-1}(k^{(3)})h^{(3)}k^{(1)},x_3^{(1)},x_4^{(1)}\rangle \\
        & = \sum_{(h),(x_i)}\sum_{k,(k)}F^{h^{(2)},\hat{k}}(\tau_R)F^{S^{-1}(k^{(3)})h^{(3)}k^{(1)},\varphi(k^{(2)}\bullet)}(\tilde{\tau}_L)\\
        & \quad \quad \psi(S(h^{(1)})S(x_1^{(1)})S(x_2^{(1)})x_3^{(2)}x_4^{(2)})|x_1^{(2)},x_2^{(2)},x_3^{(1)},x_4^{(1)}\rangle \\
        & = \sum_{(h)}F^{h^{(2)},\varphi}(\rho)B^{\psi(S(h^{(1)}\bullet )}(s_0)|x_1,x_2,x_3,x_4\rangle. 
    \end{align*}

    $\bullet$ \eqref{eq:comm4} for long ribbons:
    \begin{equation*}
        \begin{tikzpicture}
            \draw[line width = 0.5pt, red] (1,1) -- (0,0);
            \draw[line width = 0.5pt, red] (1,1) -- (2,0);
            \draw[line width = 0.5pt, red] (2,0) -- (3,1);
            \draw[line width = 0.5pt, dashed, black] (1,1) -- (3,1);
            \draw[line width = 0.5pt, black] (2,0) -- (4,0);
            \draw[line width = 0.5pt, dashed, black] (3,1) -- (4,1);
            \draw[-latex,black] (0,0) -- (2,0);
            \draw[-latex,black] (0,0) -- (0,2);
            \draw[-latex,black] (2,0) -- (2,2);
            \draw[-latex,black] (0,2) -- (2,2);
            \draw[-stealth,gray, line width=3pt] (1.2,0.53) -- (1.85,0.53); 
            \node[ line width=0.2pt, dashed, draw opacity=0.5] (a) at (-0.3,1){$x_4$};
            \node[ line width=0.2pt, dashed, draw opacity=0.5] (a) at (2.3,1.3){$x_2$};
            \node[ line width=0.2pt, dashed, draw opacity=0.5] (a) at (1,2.25){$x_3$};
            \node[ line width=0.2pt, dashed, draw opacity=0.5] (a) at (1.2,-0.3){$x_1$};
            \node[ line width=0.2pt, dashed, draw opacity=0.5] (a) at (0.35,0.65){$s_0$};
            \node[ line width=0.2pt, dashed, draw opacity=0.5] (a) at (3.3,0.5){$\cdots$};
            \node[ line width=0.2pt, dashed, draw opacity=0.5] (a) at (-2,1.3){$\rho= \rho_1\cup\rho_2$};
            \node[ line width=0.2pt, dashed, draw opacity=0.5] (a) at (-2,0.7){$\rho_1= \tau_R$};
        \end{tikzpicture}
    \end{equation*}
    \begin{align*}
        &\quad B^\psi(s_0)F^{h,\varphi}(\rho) \\
        & = \sum_k\sum_{(k),(h)} B^\psi(s_0)F^{h^{(1)},\hat{k}}(\rho_1)F^{S^{-1}(k^{(3)})h^{(2)}k^{(1)},\varphi(k^{(2)}\bullet)}(\rho_2) \\
        & = \sum_k\sum_{(k),(h)} F^{h^{(2)},\hat{k}}(\rho_1)B^{\psi(S(h^{(1)})\bullet )}(s_0)F^{S^{-1}(k^{(3)})h^{(3)}k^{(1)},\varphi(k^{(2)}\bullet)}(\rho_2) \\
        & = \sum_{(h)}\sum_k\sum_{(k),(h^{(2)})} F^{h^{(2)},\hat{k}}(\rho_1)F^{S^{-1}(k^{(3)})h^{(3)}k^{(1)},\varphi(k^{(2)}\bullet)}(\rho_2)B^{\psi(S(h^{(1)})\bullet )}(s_0) \\
        & = \sum_{(h)} F^{h^{(2)},\varphi}(\rho)B^{\psi(S(h^{(1)})\bullet )}(s_0). 
    \end{align*}

$\bullet$ \eqref{eq:comm5} for short ribbons:   
\begin{equation*}
    \begin{tikzpicture}
        \draw[line width = 0.5pt, red] (1,1) -- (0,0);
        \draw[line width = 0.5pt, red] (1,1) -- (2,0);
        \draw[-latex,black] (0,0) -- (2,0);
        \draw[-latex,black] (0,1.5) -- (0,0);
        \draw[-latex,black] (0,0) -- (-2,0);
        \draw[-latex,black] (0,-1.5) -- (0,0);
        \draw[-stealth,gray, line width=3pt] (1.4,0.35) -- (0.6,0.35); 
        \node[ line width=0.2pt, dashed, draw opacity=0.5] (a) at (-0.3,1){$x_1$};
        \node[ line width=0.2pt, dashed, draw opacity=0.5] (a) at (-1.2,-0.3){$x_2$};
        \node[ line width=0.2pt, dashed, draw opacity=0.5] (a) at (0.3,-1){$x_3$};
        \node[ line width=0.2pt, dashed, draw opacity=0.5] (a) at (1.2,-0.3){$x_4$};
        \node[ line width=0.2pt, dashed, draw opacity=0.5] (a) at (0.3,0.6){$s_1$};
        \node[ line width=0.2pt, dashed, draw opacity=0.5] (a) at (1.8,0.58){$s_0$};
        \node[ line width=0.2pt, dashed, draw opacity=0.5] (a) at (-3,0){$\rho= \tau_L:$};
    \end{tikzpicture}
\end{equation*}
\begin{align*}
    &\quad A^g(s_1)F^{h,\varphi}(\tau_L)|x_1,x_2,x_3,x_4\rangle \\
    & = \sum_{(x_4)} A^g(s_1)\varepsilon(h)\varphi(x_4^{(1)})|x_1,x_2,x_3,x_4^{(2)}\rangle \\
    & = \sum_{(x_4)}\sum_{(g)} \varepsilon(h)\varphi(x_4^{(1)}) |g^{(1)}x_1,x_2S^{-1}(g^{(2)}),g^{(3)}x_3,x_4^{(2)}S^{-1}(g^{(4)})\rangle \\ 
    & = \sum_{(x_4)}\sum_{(g)} \varepsilon(h)\varphi(x_4^{(1)}) |g^{(1)}x_1,x_2S^{-1}(g^{(2)}),g^{(3)}x_3,x_4^{(2)}S^{-1}(\varepsilon_R(g^{(5)}))S^{-1}(g^{(4)})\rangle \\ 
    & = \sum_{(x_4)}\sum_{(g)} \varepsilon(h)\varphi(x_4^{(1)}S^{-2}(\varepsilon_R(g^{(5)}))) |g^{(1)}x_1,x_2S^{-1}(g^{(2)}),g^{(3)}x_3,x_4^{(2)}S^{-1}(g^{(4)})\rangle \\
    & = \sum_{(x_4)}\sum_{(g)} \varepsilon(h)\varphi(x_4^{(1)}S^{-1}(g^{(5)})S^{-2}(g^{(6)})) |g^{(1)}x_1,x_2S^{-1}(g^{(2)}),g^{(3)}x_3,x_4^{(2)}S^{-1}(g^{(4)})\rangle \\
    & = \sum_{(g)}F^{h,\varphi(\bullet S^{-2}(g^{(5)}))}(\tau_L)|g^{(1)}x_1,x_2S^{-1}(g^{(2)}),g^{(3)}x_3,x_4S^{-1}(g^{(4)})\rangle \\
    & = \sum_{(g)}F^{h,\varphi(\bullet S^{-2}(g^{(2)}))}(\tau_L)A^{g^{(1)}}(s_1)|x_1,x_2,x_3,x_4\rangle. 
\end{align*}

     $\bullet$ \eqref{eq:comm5} for long ribbons:
    \begin{equation*}
        \begin{tikzpicture}
            \draw[line width = 0.5pt, red] (1,1) -- (0,0);
            \draw[line width = 0.5pt, red] (1,1) -- (2,0);
            \draw[line width = 0.5pt, red] (2,0) -- (3,1);
            \draw[line width = 0.5pt, black] (2,0) -- (4,0);
            \draw[line width = 0.5pt, dashed, black] (3,1) -- (4,1);
            \draw[-latex,black] (0,0) -- (2,0);
            \draw[-latex,black] (0,1.5) -- (0,0);
            \draw[-latex,black] (0,0) -- (-2,0);
            \draw[-latex,black] (0,-1.5) -- (0,0);
            \draw[-stealth,gray, line width=2.5pt] (1.4,0.25) -- (0.6,0.25); 
            \draw[-stealth,gray, line width=2.5pt] (2.4,0.6) -- (1.6,0.6); 
            \node[ line width=0.2pt, dashed, draw opacity=0.5] (a) at (-0.3,1){$x_1$};
            \node[ line width=0.2pt, dashed, draw opacity=0.5] (a) at (-1.2,-0.3){$x_2$};
            \node[ line width=0.2pt, dashed, draw opacity=0.5] (a) at (0.3,-1){$x_3$};
            \node[ line width=0.2pt, dashed, draw opacity=0.5] (a) at (1.2,-0.3){$x_4$};
            \node[ line width=0.2pt, dashed, draw opacity=0.5] (a) at (0.3,0.6){$s_1$};
            \node[ line width=0.2pt, dashed, draw opacity=0.5] (a) at (3.3,0.45){$\cdots$};
            \node[ line width=0.2pt, dashed, draw opacity=0.5] (a) at (1,0.55){$\tau_L$};
            \node[ line width=0.2pt, dashed, draw opacity=0.5] (a) at (-3.5,0.3){$\rho=\rho_1\cup\rho_2$};
            \node[ line width=0.2pt, dashed, draw opacity=0.5] (a) at (-3.5,-0.3){$\rho_2=\tau_L$};
            \draw[-latex, dashed, black] (1,1) -- (3,1);
        \end{tikzpicture}
    \end{equation*} 
    \begin{align*}
        &\quad A^g(s_1)F^{h,\varphi}(\rho) \\
        & = \sum_k\sum_{(k),(h)} A^g(s_1)F^{h^{(1)},\hat{k}}(\rho_1)F^{S^{-1}(k^{(3)})h^{(2)}k^{(1)},\varphi(k^{(2)}\bullet)}(\rho_2) \\
        & = \sum_k\sum_{(k),(h)} F^{h^{(1)},\hat{k}}(\rho_1)A^g(s_1)F^{S^{-1}(k^{(3)})h^{(2)}k^{(1)},\varphi(k^{(2)}\bullet)}(\rho_2)   \\
        & = \sum_{(g)}\sum_k\sum_{(k),(h)} F^{h^{(1)},\hat{k}}(\rho_1)F^{S^{-1}(k^{(3)})h^{(2)}k^{(1)},\varphi(k^{(2)}\bullet S^{-2}(g^{(2)}))}(\rho_2)A^{g^{(1)}}(s_1) \\
        & = \sum_{(g)}F^{h,\varphi(\bullet S^{-2}(g^{(2)}))}(\rho_A)A^{g^{(1)}}(s_1). 
    \end{align*}

The other identities can be proved with the same technique. \qed


\begin{thebibliography}{103}%
	\makeatletter
	\providecommand \@ifxundefined [1]{%
		\@ifx{#1\undefined}
	}%
	\providecommand \@ifnum [1]{%
		\ifnum #1\expandafter \@firstoftwo
		\else \expandafter \@secondoftwo
		\fi
	}%
	\providecommand \@ifx [1]{%
		\ifx #1\expandafter \@firstoftwo
		\else \expandafter \@secondoftwo
		\fi
	}%
	\providecommand \natexlab [1]{#1}%
	\providecommand \enquote  [1]{``#1''}%
	\providecommand \bibnamefont  [1]{#1}%
	\providecommand \bibfnamefont [1]{#1}%
	\providecommand \citenamefont [1]{#1}%
	\providecommand \href@noop [0]{\@secondoftwo}%
	\providecommand \href [0]{\begingroup \@sanitize@url \@href}%
	\providecommand \@href[1]{\@@startlink{#1}\@@href}%
	\providecommand \@@href[1]{\endgroup#1\@@endlink}%
	\providecommand \@sanitize@url [0]{\catcode `\\12\catcode `\$12\catcode
		`\&12\catcode `\#12\catcode `\^12\catcode `\_12\catcode `\%12\relax}%
	\providecommand \@@startlink[1]{}%
	\providecommand \@@endlink[0]{}%
	\providecommand \url  [0]{\begingroup\@sanitize@url \@url }%
	\providecommand \@url [1]{\endgroup\@href {#1}{\urlprefix }}%
	\providecommand \urlprefix  [0]{URL }%
	\providecommand \Eprint [0]{\href }%
	\providecommand \doibase [0]{http://dx.doi.org/}%
	\providecommand \selectlanguage [0]{\@gobble}%
	\providecommand \bibinfo  [0]{\@secondoftwo}%
	\providecommand \bibfield  [0]{\@secondoftwo}%
	\providecommand \translation [1]{[#1]}%
	\providecommand \BibitemOpen [0]{}%
	\providecommand \bibitemStop [0]{}%
	\providecommand \bibitemNoStop [0]{.\EOS\space}%
	\providecommand \EOS [0]{\spacefactor3000\relax}%
	\providecommand \BibitemShut  [1]{\csname bibitem#1\endcsname}%
	\let\auto@bib@innerbib\@empty
	\bibitem [{\citenamefont {Wen}(2004)}]{Wen2004}%
	\BibitemOpen
	\bibfield  {author} {\bibinfo {author} {\bibfnamefont {X.-G.}\ \bibnamefont
			{Wen}},\ }\href
	{http://www.oxfordscholarship.com/view/10.1093/acprof:oso/9780199227259.001.0001/acprof-9780199227259}
	{\emph {\bibinfo {title} {Quantum field theory of many-body systems: from the
				origin of sound to an origin of light and electrons}}}\ (\bibinfo
	{publisher} {Oxford University Press on Demand},\ \bibinfo {year}
	{2004})\BibitemShut {NoStop}%
	\bibitem [{\citenamefont {Levin}\ and\ \citenamefont
		{Wen}(2005{\natexlab{a}})}]{Levin200photons}%
	\BibitemOpen
	\bibfield  {author} {\bibinfo {author} {\bibfnamefont {M.}~\bibnamefont
			{Levin}}\ and\ \bibinfo {author} {\bibfnamefont {X.-G.}\ \bibnamefont
			{Wen}},\ }\bibfield  {title} {\enquote {\bibinfo {title} {Colloquium: Photons
				and electrons as emergent phenomena},}\ }\href {\doibase
		10.1103/RevModPhys.77.871} {\bibfield  {journal} {\bibinfo  {journal} {Rev.
				Mod. Phys.}\ }\textbf {\bibinfo {volume} {77}},\ \bibinfo {pages} {871}
		(\bibinfo {year} {2005}{\natexlab{a}})},\ \Eprint
	{http://arxiv.org/abs/cond-mat/0407140} {arXiv:cond-mat/0407140
		[cond-mat.str-el]} \BibitemShut {NoStop}%
	\bibitem [{\citenamefont {Nayak}\ \emph {et~al.}(2008)\citenamefont {Nayak},
		\citenamefont {Simon}, \citenamefont {Stern}, \citenamefont {Freedman},\ and\
		\citenamefont {Das~Sarma}}]{Nayak2008}%
	\BibitemOpen
	\bibfield  {author} {\bibinfo {author} {\bibfnamefont {C.}~\bibnamefont
			{Nayak}}, \bibinfo {author} {\bibfnamefont {S.~H.}\ \bibnamefont {Simon}},
		\bibinfo {author} {\bibfnamefont {A.}~\bibnamefont {Stern}}, \bibinfo
		{author} {\bibfnamefont {M.}~\bibnamefont {Freedman}}, \ and\ \bibinfo
		{author} {\bibfnamefont {S.}~\bibnamefont {Das~Sarma}},\ }\bibfield  {title}
	{\enquote {\bibinfo {title} {Non-{A}belian anyons and topological quantum
				computation},}\ }\href {\doibase 10.1103/RevModPhys.80.1083} {\bibfield
		{journal} {\bibinfo  {journal} {Rev. Mod. Phys.}\ }\textbf {\bibinfo {volume}
			{80}},\ \bibinfo {pages} {1083} (\bibinfo {year} {2008})},\ \Eprint
	{http://arxiv.org/abs/0707.1889} {arXiv:0707.1889 [cond-mat.str-el]}
	\BibitemShut {NoStop}%
	\bibitem [{\citenamefont {Chiu}\ \emph {et~al.}(2016)\citenamefont {Chiu},
		\citenamefont {Teo}, \citenamefont {Schnyder},\ and\ \citenamefont
		{Ryu}}]{Chiu2016classification}%
	\BibitemOpen
	\bibfield  {author} {\bibinfo {author} {\bibfnamefont {C.-K.}\ \bibnamefont
			{Chiu}}, \bibinfo {author} {\bibfnamefont {J.~C.~Y.}\ \bibnamefont {Teo}},
		\bibinfo {author} {\bibfnamefont {A.~P.}\ \bibnamefont {Schnyder}}, \ and\
		\bibinfo {author} {\bibfnamefont {S.}~\bibnamefont {Ryu}},\ }\bibfield
	{title} {\enquote {\bibinfo {title} {Classification of topological quantum
				matter with symmetries},}\ }\href {\doibase 10.1103/RevModPhys.88.035005}
	{\bibfield  {journal} {\bibinfo  {journal} {Rev. Mod. Phys.}\ }\textbf
		{\bibinfo {volume} {88}},\ \bibinfo {pages} {035005} (\bibinfo {year}
		{2016})},\ \Eprint {http://arxiv.org/abs/1505.03535} {arXiv:1505.03535
		[cond-mat.str-el]} \BibitemShut {NoStop}%
	\bibitem [{\citenamefont {Witten}(2016)}]{Witten2016fermion}%
	\BibitemOpen
	\bibfield  {author} {\bibinfo {author} {\bibfnamefont {E.}~\bibnamefont
			{Witten}},\ }\bibfield  {title} {\enquote {\bibinfo {title} {Fermion path
				integrals and topological phases},}\ }\href {\doibase
		10.1103/RevModPhys.88.035001} {\bibfield  {journal} {\bibinfo  {journal}
			{Rev. Mod. Phys.}\ }\textbf {\bibinfo {volume} {88}},\ \bibinfo {pages}
		{035001} (\bibinfo {year} {2016})},\ \Eprint
	{http://arxiv.org/abs/1508.04715} {arXiv:1508.04715 [cond-mat.str-el]}
	\BibitemShut {NoStop}%
	\bibitem [{\citenamefont {Zhou}\ \emph {et~al.}(2017)\citenamefont {Zhou},
		\citenamefont {Kanoda},\ and\ \citenamefont {Ng}}]{Zhou2017quantum}%
	\BibitemOpen
	\bibfield  {author} {\bibinfo {author} {\bibfnamefont {Y.}~\bibnamefont
			{Zhou}}, \bibinfo {author} {\bibfnamefont {K.}~\bibnamefont {Kanoda}}, \ and\
		\bibinfo {author} {\bibfnamefont {T.-K.}\ \bibnamefont {Ng}},\ }\bibfield
	{title} {\enquote {\bibinfo {title} {Quantum spin liquid states},}\ }\href
	{\doibase 10.1103/RevModPhys.89.025003} {\bibfield  {journal} {\bibinfo
			{journal} {Rev. Mod. Phys.}\ }\textbf {\bibinfo {volume} {89}},\ \bibinfo
		{pages} {025003} (\bibinfo {year} {2017})},\ \Eprint
	{http://arxiv.org/abs/1607.03228} {arXiv:1607.03228 [cond-mat.str-el]}
	\BibitemShut {NoStop}%
	\bibitem [{\citenamefont {Wen}(2017)}]{Wen2017zoo}%
	\BibitemOpen
	\bibfield  {author} {\bibinfo {author} {\bibfnamefont {X.-G.}\ \bibnamefont
			{Wen}},\ }\bibfield  {title} {\enquote {\bibinfo {title} {Colloquium: Zoo of
				quantum-topological phases of matter},}\ }\href {\doibase
		10.1103/RevModPhys.89.041004} {\bibfield  {journal} {\bibinfo  {journal}
			{Rev. Mod. Phys.}\ }\textbf {\bibinfo {volume} {89}},\ \bibinfo {pages}
		{041004} (\bibinfo {year} {2017})},\ \Eprint
	{http://arxiv.org/abs/1610.03911} {arXiv:1610.03911 [cond-mat.str-el]}
	\BibitemShut {NoStop}%
	\bibitem [{\citenamefont {Gaiotto}\ \emph {et~al.}(2015)\citenamefont
		{Gaiotto}, \citenamefont {Kapustin}, \citenamefont {Seiberg},\ and\
		\citenamefont {Willett}}]{gaiotto2015generalized}%
	\BibitemOpen
	\bibfield  {author} {\bibinfo {author} {\bibfnamefont {D.}~\bibnamefont
			{Gaiotto}}, \bibinfo {author} {\bibfnamefont {A.}~\bibnamefont {Kapustin}},
		\bibinfo {author} {\bibfnamefont {N.}~\bibnamefont {Seiberg}}, \ and\
		\bibinfo {author} {\bibfnamefont {B.}~\bibnamefont {Willett}},\ }\bibfield
	{title} {\enquote {\bibinfo {title} {Generalized global symmetries},}\
	}\href@noop {} {\bibfield  {journal} {\bibinfo  {journal} {Journal of High
				Energy Physics}\ }\textbf {\bibinfo {volume} {2015}},\ \bibinfo {pages} {1}
		(\bibinfo {year} {2015})},\ \Eprint {http://arxiv.org/abs/1412.5148}
	{arXiv:1412.5148 [hep-th]} \BibitemShut {NoStop}%
	\bibitem [{\citenamefont {Bhardwaj}\ and\ \citenamefont
		{Tachikawa}(2018)}]{bhardwaj2018finite}%
	\BibitemOpen
	\bibfield  {author} {\bibinfo {author} {\bibfnamefont {L.}~\bibnamefont
			{Bhardwaj}}\ and\ \bibinfo {author} {\bibfnamefont {Y.}~\bibnamefont
			{Tachikawa}},\ }\bibfield  {title} {\enquote {\bibinfo {title} {On finite
				symmetries and their gauging in two dimensions},}\ }\href@noop {} {\bibfield
		{journal} {\bibinfo  {journal} {Journal of High Energy Physics}\ }\textbf
		{\bibinfo {volume} {2018}},\ \bibinfo {pages} {1} (\bibinfo {year} {2018})},\
	\Eprint {http://arxiv.org/abs/1704.02330} {arXiv:1704.02330 [hep-th]}
	\BibitemShut {NoStop}%
	\bibitem [{\citenamefont {Bhardwaj}\ \emph {et~al.}(2022)\citenamefont
		{Bhardwaj}, \citenamefont {Sch{\"a}fer-Nameki},\ and\ \citenamefont
		{Wu}}]{bhardwaj2022universal}%
	\BibitemOpen
	\bibfield  {author} {\bibinfo {author} {\bibfnamefont {L.}~\bibnamefont
			{Bhardwaj}}, \bibinfo {author} {\bibfnamefont {S.}~\bibnamefont
			{Sch{\"a}fer-Nameki}}, \ and\ \bibinfo {author} {\bibfnamefont
			{J.}~\bibnamefont {Wu}},\ }\bibfield  {title} {\enquote {\bibinfo {title}
			{Universal non-invertible symmetries},}\ }\href@noop {} {\bibfield  {journal}
		{\bibinfo  {journal} {Fortschritte der Physik}\ }\textbf {\bibinfo {volume}
			{70}},\ \bibinfo {pages} {2200143} (\bibinfo {year} {2022})},\ \Eprint
	{http://arxiv.org/abs/2208.05973} {arXiv:2208.05973 [hep-th]} \BibitemShut
	{NoStop}%
	\bibitem [{\citenamefont {Bartsch}\ \emph {et~al.}(2022)\citenamefont
		{Bartsch}, \citenamefont {Bullimore}, \citenamefont {Ferrari},\ and\
		\citenamefont {Pearson}}]{bartsch2022non}%
	\BibitemOpen
	\bibfield  {author} {\bibinfo {author} {\bibfnamefont {T.}~\bibnamefont
			{Bartsch}}, \bibinfo {author} {\bibfnamefont {M.}~\bibnamefont {Bullimore}},
		\bibinfo {author} {\bibfnamefont {A.~E.}\ \bibnamefont {Ferrari}}, \ and\
		\bibinfo {author} {\bibfnamefont {J.}~\bibnamefont {Pearson}},\ }\bibfield
	{title} {\enquote {\bibinfo {title} {Non-invertible symmetries and higher
				representation theory {I}},}\ }\href {https://arxiv.org/abs/2208.05993}
	{\bibfield  {journal} {\bibinfo  {journal} {arXiv preprint arXiv:2208.05993}\
		} (\bibinfo {year} {2022})}\BibitemShut {NoStop}%
	\bibitem [{\citenamefont {Levin}\ and\ \citenamefont
		{Wen}(2005{\natexlab{b}})}]{Levin2005}%
	\BibitemOpen
	\bibfield  {author} {\bibinfo {author} {\bibfnamefont {M.~A.}\ \bibnamefont
			{Levin}}\ and\ \bibinfo {author} {\bibfnamefont {X.-G.}\ \bibnamefont
			{Wen}},\ }\bibfield  {title} {\enquote {\bibinfo {title} {String-net
				condensation: A physical mechanism for topological phases},}\ }\href
	{\doibase 10.1103/PhysRevB.71.045110} {\bibfield  {journal} {\bibinfo
			{journal} {Phys. Rev. B}\ }\textbf {\bibinfo {volume} {71}},\ \bibinfo
		{pages} {045110} (\bibinfo {year} {2005}{\natexlab{b}})},\ \Eprint
	{http://arxiv.org/abs/cond-mat/0404617} {arXiv:cond-mat/0404617
		[cond-mat.str-el]} \BibitemShut {NoStop}%
	\bibitem [{\citenamefont {Kong}\ and\ \citenamefont
		{Wen}(2014)}]{kong2014braided}%
	\BibitemOpen
	\bibfield  {author} {\bibinfo {author} {\bibfnamefont {L.}~\bibnamefont
			{Kong}}\ and\ \bibinfo {author} {\bibfnamefont {X.-G.}\ \bibnamefont {Wen}},\
	}\bibfield  {title} {\enquote {\bibinfo {title} {Braided fusion categories,
				gravitational anomalies, and the mathematical framework for topological
				orders in any dimensions},}\ }\href {https://arxiv.org/abs/1405.5858}
	{\bibfield  {journal} {\bibinfo  {journal} {arXiv preprint arXiv:1405.5858}\
		} (\bibinfo {year} {2014})}\BibitemShut {NoStop}%
	\bibitem [{\citenamefont {Johnson-Freyd}(2022)}]{johnson2022classification}%
	\BibitemOpen
	\bibfield  {author} {\bibinfo {author} {\bibfnamefont {T.}~\bibnamefont
			{Johnson-Freyd}},\ }\bibfield  {title} {\enquote {\bibinfo {title} {On the
				classification of topological orders},}\ }\href
	{https://link.springer.com/article/10.1007/s00220-022-04380-3} {\bibfield
		{journal} {\bibinfo  {journal} {Communications in Mathematical Physics}\ ,\
			\bibinfo {pages} {1}} (\bibinfo {year} {2022})},\ \Eprint
	{http://arxiv.org/abs/2003.06663} {arXiv:2003.06663 [math.CT]} \BibitemShut
	{NoStop}%
	\bibitem [{\citenamefont {Tsui}\ \emph {et~al.}(1982)\citenamefont {Tsui},
		\citenamefont {Stormer},\ and\ \citenamefont {Gossard}}]{Tsui1982}%
	\BibitemOpen
	\bibfield  {author} {\bibinfo {author} {\bibfnamefont {D.~C.}\ \bibnamefont
			{Tsui}}, \bibinfo {author} {\bibfnamefont {H.~L.}\ \bibnamefont {Stormer}}, \
		and\ \bibinfo {author} {\bibfnamefont {A.~C.}\ \bibnamefont {Gossard}},\
	}\bibfield  {title} {\enquote {\bibinfo {title} {Two-dimensional
				magnetotransport in the extreme quantum limit},}\ }\href {\doibase
		10.1103/PhysRevLett.48.1559} {\bibfield  {journal} {\bibinfo  {journal}
			{Phys. Rev. Lett.}\ }\textbf {\bibinfo {volume} {48}},\ \bibinfo {pages}
		{1559} (\bibinfo {year} {1982})}\BibitemShut {NoStop}%
	\bibitem [{\citenamefont {Fr\"{o}hlich}\ \emph {et~al.}(2007)\citenamefont
		{Fr\"{o}hlich}, \citenamefont {Fuchs}, \citenamefont {Runkel},\ and\
		\citenamefont {Schweigert}}]{FROHLICH2007duality}%
	\BibitemOpen
	\bibfield  {author} {\bibinfo {author} {\bibfnamefont {J.}~\bibnamefont
			{Fr\"{o}hlich}}, \bibinfo {author} {\bibfnamefont {J.}~\bibnamefont {Fuchs}},
		\bibinfo {author} {\bibfnamefont {I.}~\bibnamefont {Runkel}}, \ and\ \bibinfo
		{author} {\bibfnamefont {C.}~\bibnamefont {Schweigert}},\ }\bibfield  {title}
	{\enquote {\bibinfo {title} {Duality and defects in rational conformal field
				theory},}\ }\href {\doibase https://doi.org/10.1016/j.nuclphysb.2006.11.017}
	{\bibfield  {journal} {\bibinfo  {journal} {Nuclear Physics B}\ }\textbf
		{\bibinfo {volume} {763}},\ \bibinfo {pages} {354} (\bibinfo {year}
		{2007})},\ \Eprint {http://arxiv.org/abs/hep-th/0607247}
	{arXiv:hep-th/0607247 [hep-th]} \BibitemShut {NoStop}%
	\bibitem [{\citenamefont {Thorngren}\ and\ \citenamefont
		{Wang}(2019)}]{thorngren2019fusion}%
	\BibitemOpen
	\bibfield  {author} {\bibinfo {author} {\bibfnamefont {R.}~\bibnamefont
			{Thorngren}}\ and\ \bibinfo {author} {\bibfnamefont {Y.}~\bibnamefont
			{Wang}},\ }\bibfield  {title} {\enquote {\bibinfo {title} {Fusion category
				symmetry {I}: anomaly in-flow and gapped phases},}\ }\href@noop {} {\
		(\bibinfo {year} {2019})},\ \Eprint {http://arxiv.org/abs/1912.02817}
	{arXiv:1912.02817 [hep-th]} \BibitemShut {NoStop}%
	\bibitem [{\citenamefont {Thorngren}\ and\ \citenamefont
		{Wang}(2021)}]{thorngren2021fusion}%
	\BibitemOpen
	\bibfield  {author} {\bibinfo {author} {\bibfnamefont {R.}~\bibnamefont
			{Thorngren}}\ and\ \bibinfo {author} {\bibfnamefont {Y.}~\bibnamefont
			{Wang}},\ }\bibfield  {title} {\enquote {\bibinfo {title} {Fusion category
				symmetry {II}: categoriosities at $c= 1$ and beyond},}\ }\href@noop {} {\
		(\bibinfo {year} {2021})},\ \Eprint {http://arxiv.org/abs/2106.12577}
	{arXiv:2106.12577 [hep-th]} \BibitemShut {NoStop}%
	\bibitem [{\citenamefont {Inamura}(2021)}]{inamura2021topological}%
	\BibitemOpen
	\bibfield  {author} {\bibinfo {author} {\bibfnamefont {K.}~\bibnamefont
			{Inamura}},\ }\bibfield  {title} {\enquote {\bibinfo {title} {Topological
				field theories and symmetry protected topological phases with fusion category
				symmetries},}\ }\href
	{https://link.springer.com/article/10.1007/JHEP05(2021)204} {\bibfield
		{journal} {\bibinfo  {journal} {Journal of High Energy Physics}\ }\textbf
		{\bibinfo {volume} {2021}},\ \bibinfo {pages} {1} (\bibinfo {year} {2021})},\
	\Eprint {http://arxiv.org/abs/2103.15588} {arXiv:2103.15588
		[cond-mat.str-el]} \BibitemShut {NoStop}%
	\bibitem [{\citenamefont {Inamura}(2022)}]{inamura2022lattice}%
	\BibitemOpen
	\bibfield  {author} {\bibinfo {author} {\bibfnamefont {K.}~\bibnamefont
			{Inamura}},\ }\bibfield  {title} {\enquote {\bibinfo {title} {On lattice
				models of gapped phases with fusion category symmetries},}\ }\href
	{https://link.springer.com/article/10.1007/JHEP03(2022)036} {\bibfield
		{journal} {\bibinfo  {journal} {Journal of High Energy Physics}\ }\textbf
		{\bibinfo {volume} {2022}},\ \bibinfo {pages} {1} (\bibinfo {year} {2022})},\
	\Eprint {http://arxiv.org/abs/2110.12882} {arXiv:2110.12882
		[cond-mat.str-el]} \BibitemShut {NoStop}%
	\bibitem [{\citenamefont {Feiguin}\ \emph {et~al.}(2007)\citenamefont
		{Feiguin}, \citenamefont {Trebst}, \citenamefont {Ludwig}, \citenamefont
		{Troyer}, \citenamefont {Kitaev}, \citenamefont {Wang},\ and\ \citenamefont
		{Freedman}}]{Feiguin2007interacting}%
	\BibitemOpen
	\bibfield  {author} {\bibinfo {author} {\bibfnamefont {A.}~\bibnamefont
			{Feiguin}}, \bibinfo {author} {\bibfnamefont {S.}~\bibnamefont {Trebst}},
		\bibinfo {author} {\bibfnamefont {A.~W.~W.}\ \bibnamefont {Ludwig}}, \bibinfo
		{author} {\bibfnamefont {M.}~\bibnamefont {Troyer}}, \bibinfo {author}
		{\bibfnamefont {A.}~\bibnamefont {Kitaev}}, \bibinfo {author} {\bibfnamefont
			{Z.}~\bibnamefont {Wang}}, \ and\ \bibinfo {author} {\bibfnamefont {M.~H.}\
			\bibnamefont {Freedman}},\ }\bibfield  {title} {\enquote {\bibinfo {title}
			{Interacting anyons in topological quantum liquids: The golden chain},}\
	}\href {\doibase 10.1103/PhysRevLett.98.160409} {\bibfield  {journal}
		{\bibinfo  {journal} {Phys. Rev. Lett.}\ }\textbf {\bibinfo {volume} {98}},\
		\bibinfo {pages} {160409} (\bibinfo {year} {2007})},\ \Eprint
	{http://arxiv.org/abs/cond-mat/0612341} {arXiv:cond-mat/0612341
		[cond-mat.str-el]} \BibitemShut {NoStop}%
	\bibitem [{\citenamefont {Gils}\ \emph {et~al.}(2013)\citenamefont {Gils},
		\citenamefont {Ardonne}, \citenamefont {Trebst}, \citenamefont {Huse},
		\citenamefont {Ludwig}, \citenamefont {Troyer},\ and\ \citenamefont
		{Wang}}]{Gils2013anyonic}%
	\BibitemOpen
	\bibfield  {author} {\bibinfo {author} {\bibfnamefont {C.}~\bibnamefont
			{Gils}}, \bibinfo {author} {\bibfnamefont {E.}~\bibnamefont {Ardonne}},
		\bibinfo {author} {\bibfnamefont {S.}~\bibnamefont {Trebst}}, \bibinfo
		{author} {\bibfnamefont {D.~A.}\ \bibnamefont {Huse}}, \bibinfo {author}
		{\bibfnamefont {A.~W.~W.}\ \bibnamefont {Ludwig}}, \bibinfo {author}
		{\bibfnamefont {M.}~\bibnamefont {Troyer}}, \ and\ \bibinfo {author}
		{\bibfnamefont {Z.}~\bibnamefont {Wang}},\ }\bibfield  {title} {\enquote
		{\bibinfo {title} {Anyonic quantum spin chains: Spin-1 generalizations and
				topological stability},}\ }\href {\doibase 10.1103/PhysRevB.87.235120}
	{\bibfield  {journal} {\bibinfo  {journal} {Phys. Rev. B}\ }\textbf {\bibinfo
			{volume} {87}},\ \bibinfo {pages} {235120} (\bibinfo {year} {2013})},\
	\Eprint {http://arxiv.org/abs/1303.4290} {arXiv:1303.4290 [cond-mat.str-el]}
	\BibitemShut {NoStop}%
	\bibitem [{\citenamefont {Buican}\ and\ \citenamefont
		{Gromov}(2017)}]{buican2017anyonic}%
	\BibitemOpen
	\bibfield  {author} {\bibinfo {author} {\bibfnamefont {M.}~\bibnamefont
			{Buican}}\ and\ \bibinfo {author} {\bibfnamefont {A.}~\bibnamefont
			{Gromov}},\ }\bibfield  {title} {\enquote {\bibinfo {title} {Anyonic chains,
				topological defects, and conformal field theory},}\ }\href
	{https://link.springer.com/article/10.1007/s00220-017-2995-6} {\bibfield
		{journal} {\bibinfo  {journal} {Communications in Mathematical Physics}\
		}\textbf {\bibinfo {volume} {356}},\ \bibinfo {pages} {1017} (\bibinfo {year}
		{2017})},\ \Eprint {http://arxiv.org/abs/1701.02800} {arXiv:1701.02800
		[hep-th]} \BibitemShut {NoStop}%
	\bibitem [{\citenamefont {Kitaev}(2003)}]{Kitaev2003}%
	\BibitemOpen
	\bibfield  {author} {\bibinfo {author} {\bibfnamefont {A.}~\bibnamefont
			{Kitaev}},\ }\bibfield  {title} {\enquote {\bibinfo {title} {Fault-tolerant
				quantum computation by anyons},}\ }\href {\doibase
		https://doi.org/10.1016/S0003-4916(02)00018-0} {\bibfield  {journal}
		{\bibinfo  {journal} {Annals of Physics}\ }\textbf {\bibinfo {volume}
			{303}},\ \bibinfo {pages} {2 } (\bibinfo {year} {2003})},\ \Eprint
	{http://arxiv.org/abs/quant-ph/9707021} {arXiv:quant-ph/9707021 [quant-ph]}
	\BibitemShut {NoStop}%
	\bibitem [{\citenamefont {Majid}(2000)}]{majid2000foundations}%
	\BibitemOpen
	\bibfield  {author} {\bibinfo {author} {\bibfnamefont {S.}~\bibnamefont
			{Majid}},\ }\href
	{https://www.cambridge.org/core/books/foundations-of-quantum-group-theory/BDBBAB645399E72AA1A01BDECAFC7E8C#}
	{\emph {\bibinfo {title} {Foundations of quantum group theory}}}\ (\bibinfo
	{publisher} {Cambridge university press, Cambridge},\ \bibinfo {year}
	{2000})\ pp.\ \bibinfo {pages} {x+607}\BibitemShut {NoStop}%
	\bibitem [{\citenamefont {Propitius}\ and\ \citenamefont
		{Bais}(1995)}]{Propitius1995}%
	\BibitemOpen
	\bibfield  {author} {\bibinfo {author} {\bibfnamefont {M.~d.~W.}\
			\bibnamefont {Propitius}}\ and\ \bibinfo {author} {\bibfnamefont {F.~A.}\
			\bibnamefont {Bais}},\ }\bibfield  {title} {\enquote {\bibinfo {title}
			{Discrete gauge theories},}\ }\href@noop {} {\  (\bibinfo {year} {1995})},\
	\Eprint {http://arxiv.org/abs/hep-th/9511201} {arXiv:hep-th/9511201 [hep-th]}
	\BibitemShut {NoStop}%
	\bibitem [{\citenamefont {{Alexander Bais}}\ \emph {et~al.}(1992)\citenamefont
		{{Alexander Bais}}, \citenamefont {{van Driel}},\ and\ \citenamefont {{de
				Wild Propitius}}}]{bais1992quantum}%
	\BibitemOpen
	\bibfield  {author} {\bibinfo {author} {\bibfnamefont {F.}~\bibnamefont
			{{Alexander Bais}}}, \bibinfo {author} {\bibfnamefont {P.}~\bibnamefont {{van
					Driel}}}, \ and\ \bibinfo {author} {\bibfnamefont {M.}~\bibnamefont {{de Wild
					Propitius}}},\ }\bibfield  {title} {\enquote {\bibinfo {title} {Quantum
				symmetries in discrete gauge theories},}\ }\href {\doibase
		https://doi.org/10.1016/0370-2693(92)90773-W} {\bibfield  {journal} {\bibinfo
			{journal} {Physics Letters B}\ }\textbf {\bibinfo {volume} {280}},\ \bibinfo
		{pages} {63} (\bibinfo {year} {1992})},\ \Eprint
	{http://arxiv.org/abs/hep-th/9203046} {arXiv:hep-th/9203046 [hep-th]}
	\BibitemShut {NoStop}%
	\bibitem [{\citenamefont {Bais}\ \emph {et~al.}(2003)\citenamefont {Bais},
		\citenamefont {Schroers},\ and\ \citenamefont {Slingerland}}]{bais2003hopf}%
	\BibitemOpen
	\bibfield  {author} {\bibinfo {author} {\bibfnamefont {A.~F.}\ \bibnamefont
			{Bais}}, \bibinfo {author} {\bibfnamefont {B.~J.}\ \bibnamefont {Schroers}},
		\ and\ \bibinfo {author} {\bibfnamefont {J.~K.}\ \bibnamefont
			{Slingerland}},\ }\bibfield  {title} {\enquote {\bibinfo {title} {Hopf
				symmetry breaking and confinement in (2+1)-dimensional gauge theory},}\
	}\href {https://doi.org/10.1088/1126-6708/2003/05/068} {\bibfield  {journal}
		{\bibinfo  {journal} {Journal of High Energy Physics}\ }\textbf {\bibinfo
			{volume} {2003}},\ \bibinfo {pages} {068} (\bibinfo {year} {2003})},\ \Eprint
	{http://arxiv.org/abs/hep-th/0205114} {arXiv:hep-th/0205114 [hep-th]}
	\BibitemShut {NoStop}%
	\bibitem [{\citenamefont {Bais}\ \emph {et~al.}(2002)\citenamefont {Bais},
		\citenamefont {Muller},\ and\ \citenamefont {Schroers}}]{Bais2002quantum}%
	\BibitemOpen
	\bibfield  {author} {\bibinfo {author} {\bibfnamefont {F.}~\bibnamefont
			{Bais}}, \bibinfo {author} {\bibfnamefont {N.}~\bibnamefont {Muller}}, \ and\
		\bibinfo {author} {\bibfnamefont {B.}~\bibnamefont {Schroers}},\ }\bibfield
	{title} {\enquote {\bibinfo {title} {Quantum group symmetry and particle
				scattering in (2+1)-dimensional quantum gravity},}\ }\href {\doibase
		https://doi.org/10.1016/S0550-3213(02)00572-2} {\bibfield  {journal}
		{\bibinfo  {journal} {Nuclear Physics B}\ }\textbf {\bibinfo {volume}
			{640}},\ \bibinfo {pages} {3} (\bibinfo {year} {2002})},\ \Eprint
	{http://arxiv.org/abs/hep-th/0205021} {arXiv:hep-th/0205021 [hep-th]}
	\BibitemShut {NoStop}%
	\bibitem [{\citenamefont {Delcamp}\ \emph {et~al.}(2017)\citenamefont
		{Delcamp}, \citenamefont {Dittrich},\ and\ \citenamefont
		{Riello}}]{delcamp2017fusion}%
	\BibitemOpen
	\bibfield  {author} {\bibinfo {author} {\bibfnamefont {C.}~\bibnamefont
			{Delcamp}}, \bibinfo {author} {\bibfnamefont {B.}~\bibnamefont {Dittrich}}, \
		and\ \bibinfo {author} {\bibfnamefont {A.}~\bibnamefont {Riello}},\
	}\bibfield  {title} {\enquote {\bibinfo {title} {Fusion basis for lattice
				gauge theory and loop quantum gravity},}\ }\href
	{https://link.springer.com/article/10.1007/JHEP02(2017)061} {\bibfield
		{journal} {\bibinfo  {journal} {Journal of High Energy Physics}\ }\textbf
		{\bibinfo {volume} {2017}},\ \bibinfo {pages} {1} (\bibinfo {year} {2017})},\
	\Eprint {http://arxiv.org/abs/1607.08881} {arXiv:1607.08881 [hep-th]}
	\BibitemShut {NoStop}%
	\bibitem [{\citenamefont {Fuchs}(1995)}]{fuchs1995affine}%
	\BibitemOpen
	\bibfield  {author} {\bibinfo {author} {\bibfnamefont {J.}~\bibnamefont
			{Fuchs}},\ }\href@noop {} {\emph {\bibinfo {title} {Affine Lie algebras and
				quantum groups: An Introduction, with applications in conformal field
				theory}}}\ (\bibinfo  {publisher} {Cambridge university press},\ \bibinfo
	{year} {1995})\BibitemShut {NoStop}%
	\bibitem [{\citenamefont {Meusburger}\ and\ \citenamefont
		{Wise}(2021)}]{meusburger2021hopf}%
	\BibitemOpen
	\bibfield  {author} {\bibinfo {author} {\bibfnamefont {C.}~\bibnamefont
			{Meusburger}}\ and\ \bibinfo {author} {\bibfnamefont {D.~K.}\ \bibnamefont
			{Wise}},\ }\bibfield  {title} {\enquote {\bibinfo {title} {Hopf algebra gauge
				theory on a ribbon graph},}\ }\href
	{https://www.worldscientific.com/doi/abs/10.1142/S0129055X21500161}
	{\bibfield  {journal} {\bibinfo  {journal} {Reviews in Mathematical Physics}\
			,\ \bibinfo {pages} {2150016}} (\bibinfo {year} {2021})},\ \Eprint
	{http://arxiv.org/abs/1512.03966} {arXiv:1512.03966 [math.QA]} \BibitemShut
	{NoStop}%
	\bibitem [{\citenamefont {Meusburger}(2017)}]{meusburger2017kitaev}%
	\BibitemOpen
	\bibfield  {author} {\bibinfo {author} {\bibfnamefont {C.}~\bibnamefont
			{Meusburger}},\ }\bibfield  {title} {\enquote {\bibinfo {title} {Kitaev
				lattice models as a {H}opf algebra gauge theory},}\ }\href
	{https://link.springer.com/article/10.1007%2Fs00220-017-2860-7} {\bibfield
		{journal} {\bibinfo  {journal} {Communications in Mathematical Physics}\
		}\textbf {\bibinfo {volume} {353}},\ \bibinfo {pages} {413} (\bibinfo {year}
		{2017})},\ \Eprint {http://arxiv.org/abs/1607.01144} {arXiv:1607.01144
		[math.QA]} \BibitemShut {NoStop}%
	\bibitem [{\citenamefont {Slingerland}\ and\ \citenamefont
		{Bais}(2001)}]{Slingerland2001quantum}%
	\BibitemOpen
	\bibfield  {author} {\bibinfo {author} {\bibfnamefont {J.}~\bibnamefont
			{Slingerland}}\ and\ \bibinfo {author} {\bibfnamefont {F.}~\bibnamefont
			{Bais}},\ }\bibfield  {title} {\enquote {\bibinfo {title} {Quantum groups and
				non-{A}belian braiding in quantum {H}all systems},}\ }\href {\doibase
		https://doi.org/10.1016/S0550-3213(01)00308-X} {\bibfield  {journal}
		{\bibinfo  {journal} {Nuclear Physics B}\ }\textbf {\bibinfo {volume}
			{612}},\ \bibinfo {pages} {229} (\bibinfo {year} {2001})},\ \Eprint
	{http://arxiv.org/abs/cond-mat/0104035} {arXiv:cond-mat/0104035
		[cond-mat.mes-hall]} \BibitemShut {NoStop}%
	\bibitem [{\citenamefont {Buerschaper}\ \emph
		{et~al.}(2013{\natexlab{a}})\citenamefont {Buerschaper}, \citenamefont
		{Mombelli}, \citenamefont {Christandl},\ and\ \citenamefont
		{Aguado}}]{Buerschaper2013a}%
	\BibitemOpen
	\bibfield  {author} {\bibinfo {author} {\bibfnamefont {O.}~\bibnamefont
			{Buerschaper}}, \bibinfo {author} {\bibfnamefont {J.~M.}\ \bibnamefont
			{Mombelli}}, \bibinfo {author} {\bibfnamefont {M.}~\bibnamefont
			{Christandl}}, \ and\ \bibinfo {author} {\bibfnamefont {M.}~\bibnamefont
			{Aguado}},\ }\bibfield  {title} {\enquote {\bibinfo {title} {A hierarchy of
				topological tensor network states},}\ }\href {\doibase 10.1063/1.4773316}
	{\bibfield  {journal} {\bibinfo  {journal} {Journal of Mathematical Physics}\
		}\textbf {\bibinfo {volume} {54}},\ \bibinfo {pages} {012201} (\bibinfo
		{year} {2013}{\natexlab{a}})},\ \Eprint {http://arxiv.org/abs/1007.5283}
	{arXiv:1007.5283 [cond-mat.str-el]} \BibitemShut {NoStop}%
	\bibitem [{\citenamefont {Buerschaper}\ \emph
		{et~al.}(2013{\natexlab{b}})\citenamefont {Buerschaper}, \citenamefont
		{Christandl}, \citenamefont {Kong},\ and\ \citenamefont
		{Aguado}}]{buerschaper2013electric}%
	\BibitemOpen
	\bibfield  {author} {\bibinfo {author} {\bibfnamefont {O.}~\bibnamefont
			{Buerschaper}}, \bibinfo {author} {\bibfnamefont {M.}~\bibnamefont
			{Christandl}}, \bibinfo {author} {\bibfnamefont {L.}~\bibnamefont {Kong}}, \
		and\ \bibinfo {author} {\bibfnamefont {M.}~\bibnamefont {Aguado}},\
	}\bibfield  {title} {\enquote {\bibinfo {title} {Electric--magnetic duality
				of lattice systems with topological order},}\ }\href
	{https://www.sciencedirect.com/science/article/abs/pii/S0550321313004367?via%3Dihub}
	{\bibfield  {journal} {\bibinfo  {journal} {Nuclear Physics B}\ }\textbf
		{\bibinfo {volume} {876}},\ \bibinfo {pages} {619} (\bibinfo {year}
		{2013}{\natexlab{b}})},\ \Eprint {http://arxiv.org/abs/1006.5823}
	{arXiv:1006.5823 [cond-mat.str-el]} \BibitemShut {NoStop}%
	\bibitem [{\citenamefont {Koppen}(2020)}]{koppen2020defects}%
	\BibitemOpen
	\bibfield  {author} {\bibinfo {author} {\bibfnamefont {V.}~\bibnamefont
			{Koppen}},\ }\bibfield  {title} {\enquote {\bibinfo {title} {Defects in
				{K}itaev models and bicomodule algebras},}\ }\href
	{https://arxiv.org/abs/2001.10578} {\bibfield  {journal} {\bibinfo  {journal}
			{arXiv preprint arXiv:2001.10578}\ } (\bibinfo {year} {2020})}\BibitemShut
	{NoStop}%
	\bibitem [{\citenamefont {Girelli}\ \emph {et~al.}(2021)\citenamefont
		{Girelli}, \citenamefont {Osei},\ and\ \citenamefont
		{Osumanu}}]{girelli2021semidual}%
	\BibitemOpen
	\bibfield  {author} {\bibinfo {author} {\bibfnamefont {F.}~\bibnamefont
			{Girelli}}, \bibinfo {author} {\bibfnamefont {P.~K.}\ \bibnamefont {Osei}}, \
		and\ \bibinfo {author} {\bibfnamefont {A.}~\bibnamefont {Osumanu}},\
	}\bibfield  {title} {\enquote {\bibinfo {title} {Semidual {K}itaev lattice
				model and tensor network representation},}\ }\href
	{https://link.springer.com/article/10.1007/JHEP09(2021)210} {\bibfield
		{journal} {\bibinfo  {journal} {Journal of High Energy Physics}\ }\textbf
		{\bibinfo {volume} {2021}},\ \bibinfo {pages} {1} (\bibinfo {year} {2021})},\
	\Eprint {http://arxiv.org/abs/1709.00522} {arXiv:1709.00522 [math.QA]}
	\BibitemShut {NoStop}%
	\bibitem [{\citenamefont {Vo{\ss}}(2021)}]{voss2021defects}%
	\BibitemOpen
	\bibfield  {author} {\bibinfo {author} {\bibfnamefont {T.}~\bibnamefont
			{Vo{\ss}}},\ }\emph {\bibinfo {title} {Defects and symmetries in {H}opf
			algebra lattice models}},\ \href@noop {} {Ph.D. thesis},\ \bibinfo  {school}
	{Friedrich-Alexander-Universit{\"a}t Erlangen-N{\"u}rnberg (FAU)} (\bibinfo
	{year} {2021})\BibitemShut {NoStop}%
	\bibitem [{\citenamefont {Yan}\ \emph {et~al.}(2022)\citenamefont {Yan},
		\citenamefont {Chen},\ and\ \citenamefont {Cui}}]{chen2021ribbon}%
	\BibitemOpen
	\bibfield  {author} {\bibinfo {author} {\bibfnamefont {B.}~\bibnamefont
			{Yan}}, \bibinfo {author} {\bibfnamefont {P.}~\bibnamefont {Chen}}, \ and\
		\bibinfo {author} {\bibfnamefont {S.}~\bibnamefont {Cui}},\ }\bibfield
	{title} {\enquote {\bibinfo {title} {Ribbon operators in the generalized
				{K}itaev quantum double model based on {H}opf algebras},}\ }\href
	{https://iopscience.iop.org/article/10.1088/1751-8121/ac552c/meta} {\bibfield
		{journal} {\bibinfo  {journal} {Journal of Physics A: Mathematical and
				Theoretical}\ } (\bibinfo {year} {2022})},\ \Eprint
	{http://arxiv.org/abs/2105.08202} {arXiv:2105.08202 [cond-mat.str-el]}
	\BibitemShut {NoStop}%
	\bibitem [{\citenamefont {Jia}\ \emph {et~al.}(2022)\citenamefont {Jia},
		\citenamefont {Kaszlikowski},\ and\ \citenamefont {Tan}}]{jia2022boundary}%
	\BibitemOpen
	\bibfield  {author} {\bibinfo {author} {\bibfnamefont {Z.}~\bibnamefont
			{Jia}}, \bibinfo {author} {\bibfnamefont {D.}~\bibnamefont {Kaszlikowski}}, \
		and\ \bibinfo {author} {\bibfnamefont {S.}~\bibnamefont {Tan}},\ }\bibfield
	{title} {\enquote {\bibinfo {title} {Boundary and domain wall theories of
				$2d$ generalized quantum double model},}\ }\href
	{https://arxiv.org/abs/2207.03970} {\bibfield  {journal} {\bibinfo  {journal}
			{arXiv preprint arXiv:2207.03970}\ } (\bibinfo {year} {2022})}\BibitemShut
	{NoStop}%
	\bibitem [{\citenamefont {Terhal}(2015)}]{Terhal2015quantum}%
	\BibitemOpen
	\bibfield  {author} {\bibinfo {author} {\bibfnamefont {B.~M.}\ \bibnamefont
			{Terhal}},\ }\bibfield  {title} {\enquote {\bibinfo {title} {Quantum error
				correction for quantum memories},}\ }\href {\doibase
		10.1103/RevModPhys.87.307} {\bibfield  {journal} {\bibinfo  {journal} {Rev.
				Mod. Phys.}\ }\textbf {\bibinfo {volume} {87}},\ \bibinfo {pages} {307}
		(\bibinfo {year} {2015})},\ \Eprint {http://arxiv.org/abs/1302.3428}
	{arXiv:1302.3428 [quant-ph]} \BibitemShut {NoStop}%
	\bibitem [{\citenamefont {Kitaev}\ and\ \citenamefont
		{Preskill}(2006)}]{Kitaev2005topological}%
	\BibitemOpen
	\bibfield  {author} {\bibinfo {author} {\bibfnamefont {A.}~\bibnamefont
			{Kitaev}}\ and\ \bibinfo {author} {\bibfnamefont {J.}~\bibnamefont
			{Preskill}},\ }\bibfield  {title} {\enquote {\bibinfo {title} {Topological
				entanglement entropy},}\ }\href {\doibase 10.1103/PhysRevLett.96.110404}
	{\bibfield  {journal} {\bibinfo  {journal} {Phys. Rev. Lett.}\ }\textbf
		{\bibinfo {volume} {96}},\ \bibinfo {pages} {110404} (\bibinfo {year}
		{2006})},\ \Eprint {http://arxiv.org/abs/hep-th/0510092}
	{arXiv:hep-th/0510092 [hep-th]} \BibitemShut {NoStop}%
	\bibitem [{\citenamefont {Bravyi}\ and\ \citenamefont
		{Kitaev}(1998)}]{bravyi1998quantum}%
	\BibitemOpen
	\bibfield  {author} {\bibinfo {author} {\bibfnamefont {S.~B.}\ \bibnamefont
			{Bravyi}}\ and\ \bibinfo {author} {\bibfnamefont {A.~Y.}\ \bibnamefont
			{Kitaev}},\ }\bibfield  {title} {\enquote {\bibinfo {title} {Quantum codes on
				a lattice with boundary},}\ }\href@noop {} {\  (\bibinfo {year} {1998})},\
	\Eprint {http://arxiv.org/abs/quant-ph/9811052} {arXiv:quant-ph/9811052
		[quant-ph]} \BibitemShut {NoStop}%
	\bibitem [{\citenamefont {Bombin}\ and\ \citenamefont
		{Martin-Delgado}(2008)}]{Bombin2008family}%
	\BibitemOpen
	\bibfield  {author} {\bibinfo {author} {\bibfnamefont {H.}~\bibnamefont
			{Bombin}}\ and\ \bibinfo {author} {\bibfnamefont {M.~A.}\ \bibnamefont
			{Martin-Delgado}},\ }\bibfield  {title} {\enquote {\bibinfo {title} {Family
				of non-{A}belian {K}itaev models on a lattice: Topological condensation and
				confinement},}\ }\href {\doibase 10.1103/PhysRevB.78.115421} {\bibfield
		{journal} {\bibinfo  {journal} {Phys. Rev. B}\ }\textbf {\bibinfo {volume}
			{78}},\ \bibinfo {pages} {115421} (\bibinfo {year} {2008})},\ \Eprint
	{http://arxiv.org/abs/0712.0190} {arXiv:0712.0190 [cond-mat.str-el]}
	\BibitemShut {NoStop}%
	\bibitem [{\citenamefont {Beigi}\ \emph {et~al.}(2011)\citenamefont {Beigi},
		\citenamefont {Shor},\ and\ \citenamefont {Whalen}}]{Beigi2011the}%
	\BibitemOpen
	\bibfield  {author} {\bibinfo {author} {\bibfnamefont {S.}~\bibnamefont
			{Beigi}}, \bibinfo {author} {\bibfnamefont {P.~W.}\ \bibnamefont {Shor}}, \
		and\ \bibinfo {author} {\bibfnamefont {D.}~\bibnamefont {Whalen}},\
	}\bibfield  {title} {\enquote {\bibinfo {title} {The quantum double model
				with boundary: Condensations and symmetries},}\ }\href {\doibase
		10.1007/s00220-011-1294-x} {\bibfield  {journal} {\bibinfo  {journal}
			{Communications in Mathematical Physics}\ }\textbf {\bibinfo {volume}
			{306}},\ \bibinfo {pages} {663} (\bibinfo {year} {2011})},\ \Eprint
	{http://arxiv.org/abs/1006.5479} {arXiv:1006.5479 [quant-ph]} \BibitemShut
	{NoStop}%
	\bibitem [{\citenamefont {Kitaev}\ and\ \citenamefont
		{Kong}(2012)}]{Kitaev2012a}%
	\BibitemOpen
	\bibfield  {author} {\bibinfo {author} {\bibfnamefont {A.}~\bibnamefont
			{Kitaev}}\ and\ \bibinfo {author} {\bibfnamefont {L.}~\bibnamefont {Kong}},\
	}\bibfield  {title} {\enquote {\bibinfo {title} {Models for gapped boundaries
				and domain walls},}\ }\href {\doibase 10.1007/s00220-012-1500-5} {\bibfield
		{journal} {\bibinfo  {journal} {Communications in Mathematical Physics}\
		}\textbf {\bibinfo {volume} {313}},\ \bibinfo {pages} {351} (\bibinfo {year}
		{2012})},\ \Eprint {http://arxiv.org/abs/1104.5047} {arXiv:1104.5047
		[cond-mat.str-el]} \BibitemShut {NoStop}%
	\bibitem [{\citenamefont {Cong}\ \emph
		{et~al.}(2017{\natexlab{a}})\citenamefont {Cong}, \citenamefont {Cheng},\
		and\ \citenamefont {Wang}}]{Cong2017}%
	\BibitemOpen
	\bibfield  {author} {\bibinfo {author} {\bibfnamefont {I.}~\bibnamefont
			{Cong}}, \bibinfo {author} {\bibfnamefont {M.}~\bibnamefont {Cheng}}, \ and\
		\bibinfo {author} {\bibfnamefont {Z.}~\bibnamefont {Wang}},\ }\bibfield
	{title} {\enquote {\bibinfo {title} {Hamiltonian and algebraic theories of
				gapped boundaries in topological phases of matter},}\ }\href {\doibase
		10.1007/s00220-017-2960-4} {\bibfield  {journal} {\bibinfo  {journal}
			{Communications in Mathematical Physics}\ }\textbf {\bibinfo {volume}
			{355}},\ \bibinfo {pages} {645} (\bibinfo {year} {2017}{\natexlab{a}})},\
	\Eprint {http://arxiv.org/abs/1707.04564} {arXiv:1707.04564
		[cond-mat.str-el]} \BibitemShut {NoStop}%
	\bibitem [{\citenamefont {Kong}(2014)}]{Kong2014}%
	\BibitemOpen
	\bibfield  {author} {\bibinfo {author} {\bibfnamefont {L.}~\bibnamefont
			{Kong}},\ }\bibfield  {title} {\enquote {\bibinfo {title} {Anyon condensation
				and tensor categories},}\ }\href {\doibase
		https://doi.org/10.1016/j.nuclphysb.2014.07.003} {\bibfield  {journal}
		{\bibinfo  {journal} {Nuclear Physics B}\ }\textbf {\bibinfo {volume}
			{886}},\ \bibinfo {pages} {436 } (\bibinfo {year} {2014})},\ \Eprint
	{http://arxiv.org/abs/1307.8244} {arXiv:1307.8244 [cond-mat.str-el]}
	\BibitemShut {NoStop}%
	\bibitem [{\citenamefont {Buerschaper}\ and\ \citenamefont
		{Aguado}(2009)}]{Buerschaper2009mapping}%
	\BibitemOpen
	\bibfield  {author} {\bibinfo {author} {\bibfnamefont {O.}~\bibnamefont
			{Buerschaper}}\ and\ \bibinfo {author} {\bibfnamefont {M.}~\bibnamefont
			{Aguado}},\ }\bibfield  {title} {\enquote {\bibinfo {title} {Mapping
				{K}itaev's quantum double lattice models to {L}evin and {W}en's string-net
				models},}\ }\href {\doibase 10.1103/PhysRevB.80.155136} {\bibfield  {journal}
		{\bibinfo  {journal} {Phys. Rev. B}\ }\textbf {\bibinfo {volume} {80}},\
		\bibinfo {pages} {155136} (\bibinfo {year} {2009})},\ \Eprint
	{http://arxiv.org/abs/0907.2670} {arXiv:0907.2670 [cond-mat.str-el]}
	\BibitemShut {NoStop}%
	\bibitem [{\citenamefont {Wang}\ \emph {et~al.}(2020)\citenamefont {Wang},
		\citenamefont {Li}, \citenamefont {Hu},\ and\ \citenamefont
		{Wan}}]{wang2020electric}%
	\BibitemOpen
	\bibfield  {author} {\bibinfo {author} {\bibfnamefont {H.}~\bibnamefont
			{Wang}}, \bibinfo {author} {\bibfnamefont {Y.}~\bibnamefont {Li}}, \bibinfo
		{author} {\bibfnamefont {Y.}~\bibnamefont {Hu}}, \ and\ \bibinfo {author}
		{\bibfnamefont {Y.}~\bibnamefont {Wan}},\ }\bibfield  {title} {\enquote
		{\bibinfo {title} {Electric-magnetic duality in the quantum double models of
				topological orders with gapped boundaries},}\ }\href
	{https://link.springer.com/article/10.1007%2FJHEP02%282020%29030} {\bibfield
		{journal} {\bibinfo  {journal} {Journal of High Energy Physics}\ }\textbf
		{\bibinfo {volume} {2020}},\ \bibinfo {pages} {1} (\bibinfo {year} {2020})},\
	\Eprint {http://arxiv.org/abs/1910.13441} {arXiv:1910.13441
		[cond-mat.str-el]} \BibitemShut {NoStop}%
	\bibitem [{\citenamefont {Hu}\ and\ \citenamefont
		{Wan}(2020)}]{hu2020electric}%
	\BibitemOpen
	\bibfield  {author} {\bibinfo {author} {\bibfnamefont {Y.}~\bibnamefont
			{Hu}}\ and\ \bibinfo {author} {\bibfnamefont {Y.}~\bibnamefont {Wan}},\
	}\bibfield  {title} {\enquote {\bibinfo {title} {Electric-magnetic duality in
				twisted quantum double model of topological orders},}\ }\href
	{https://link.springer.com/article/10.1007/JHEP11(2020)170} {\bibfield
		{journal} {\bibinfo  {journal} {Journal of High Energy Physics}\ }\textbf
		{\bibinfo {volume} {2020}},\ \bibinfo {pages} {1} (\bibinfo {year} {2020})},\
	\Eprint {http://arxiv.org/abs/2007.15636} {arXiv:2007.15636
		[cond-mat.str-el]} \BibitemShut {NoStop}%
	\bibitem [{\citenamefont {Jia}\ and\ \citenamefont
		{Kaszlikowski}(2022)}]{Jia2022electric}%
	\BibitemOpen
	\bibfield  {author} {\bibinfo {author} {\bibfnamefont {Z.}~\bibnamefont
			{Jia}}\ and\ \bibinfo {author} {\bibfnamefont {D.}~\bibnamefont
			{Kaszlikowski}},\ }\href {\doibase 10.48550/ARXIV.2201.12361} {\enquote
		{\bibinfo {title} {Electric-magnetic duality of $\mathbb{Z}_2$ symmetry
				enriched cyclic {A}belian lattice gauge theory},}\ } (\bibinfo {year}
	{2022}),\ \Eprint {http://arxiv.org/abs/2201.12361} {arXiv:2201.12361
		[quant-ph]} \BibitemShut {NoStop}%
	\bibitem [{\citenamefont {Bullivant}\ \emph {et~al.}(2017)\citenamefont
		{Bullivant}, \citenamefont {Hu},\ and\ \citenamefont
		{Wan}}]{Bullivant2017twisted}%
	\BibitemOpen
	\bibfield  {author} {\bibinfo {author} {\bibfnamefont {A.}~\bibnamefont
			{Bullivant}}, \bibinfo {author} {\bibfnamefont {Y.}~\bibnamefont {Hu}}, \
		and\ \bibinfo {author} {\bibfnamefont {Y.}~\bibnamefont {Wan}},\ }\bibfield
	{title} {\enquote {\bibinfo {title} {Twisted quantum double model of
				topological order with boundaries},}\ }\href {\doibase
		10.1103/PhysRevB.96.165138} {\bibfield  {journal} {\bibinfo  {journal} {Phys.
				Rev. B}\ }\textbf {\bibinfo {volume} {96}},\ \bibinfo {pages} {165138}
		(\bibinfo {year} {2017})},\ \Eprint {http://arxiv.org/abs/1706.03611}
	{arXiv:1706.03611 [cond-mat.str-el]} \BibitemShut {NoStop}%
	\bibitem [{\citenamefont {Moradi}\ and\ \citenamefont
		{Wen}(2015)}]{Moradi2015universal}%
	\BibitemOpen
	\bibfield  {author} {\bibinfo {author} {\bibfnamefont {H.}~\bibnamefont
			{Moradi}}\ and\ \bibinfo {author} {\bibfnamefont {X.-G.}\ \bibnamefont
			{Wen}},\ }\bibfield  {title} {\enquote {\bibinfo {title} {Universal
				topological data for gapped quantum liquids in three dimensions and fusion
				algebra for non-{A}belian string excitations},}\ }\href {\doibase
		10.1103/PhysRevB.91.075114} {\bibfield  {journal} {\bibinfo  {journal} {Phys.
				Rev. B}\ }\textbf {\bibinfo {volume} {91}},\ \bibinfo {pages} {075114}
		(\bibinfo {year} {2015})},\ \Eprint {http://arxiv.org/abs/1404.4618}
	{arXiv:1404.4618 [cond-mat.str-el]} \BibitemShut {NoStop}%
	\bibitem [{\citenamefont {Wan}\ \emph {et~al.}(2015)\citenamefont {Wan},
		\citenamefont {Wang},\ and\ \citenamefont {He}}]{Wan2015twisted}%
	\BibitemOpen
	\bibfield  {author} {\bibinfo {author} {\bibfnamefont {Y.}~\bibnamefont
			{Wan}}, \bibinfo {author} {\bibfnamefont {J.~C.}\ \bibnamefont {Wang}}, \
		and\ \bibinfo {author} {\bibfnamefont {H.}~\bibnamefont {He}},\ }\bibfield
	{title} {\enquote {\bibinfo {title} {Twisted gauge theory model of
				topological phases in three dimensions},}\ }\href {\doibase
		10.1103/PhysRevB.92.045101} {\bibfield  {journal} {\bibinfo  {journal} {Phys.
				Rev. B}\ }\textbf {\bibinfo {volume} {92}},\ \bibinfo {pages} {045101}
		(\bibinfo {year} {2015})},\ \Eprint {http://arxiv.org/abs/1409.3216}
	{arXiv:1409.3216 [cond-mat.str-el]} \BibitemShut {NoStop}%
	\bibitem [{\citenamefont {Wang}\ \emph {et~al.}(2018)\citenamefont {Wang},
		\citenamefont {Li}, \citenamefont {Hu},\ and\ \citenamefont
		{Wan}}]{wang2018gapped}%
	\BibitemOpen
	\bibfield  {author} {\bibinfo {author} {\bibfnamefont {H.}~\bibnamefont
			{Wang}}, \bibinfo {author} {\bibfnamefont {Y.}~\bibnamefont {Li}}, \bibinfo
		{author} {\bibfnamefont {Y.}~\bibnamefont {Hu}}, \ and\ \bibinfo {author}
		{\bibfnamefont {Y.}~\bibnamefont {Wan}},\ }\bibfield  {title} {\enquote
		{\bibinfo {title} {Gapped boundary theory of the twisted gauge theory model
				of three-dimensional topological orders},}\ }\href
	{https://link.springer.com/article/10.1007/JHEP10(2018)114} {\bibfield
		{journal} {\bibinfo  {journal} {Journal of High Energy Physics}\ }\textbf
		{\bibinfo {volume} {2018}},\ \bibinfo {pages} {1} (\bibinfo {year} {2018})},\
	\Eprint {http://arxiv.org/abs/1807.11083} {arXiv:1807.11083
		[cond-mat.str-el]} \BibitemShut {NoStop}%
	\bibitem [{\citenamefont {Hamma}\ \emph {et~al.}(2005)\citenamefont {Hamma},
		\citenamefont {Zanardi},\ and\ \citenamefont {Wen}}]{Hamma2005string}%
	\BibitemOpen
	\bibfield  {author} {\bibinfo {author} {\bibfnamefont {A.}~\bibnamefont
			{Hamma}}, \bibinfo {author} {\bibfnamefont {P.}~\bibnamefont {Zanardi}}, \
		and\ \bibinfo {author} {\bibfnamefont {X.-G.}\ \bibnamefont {Wen}},\
	}\bibfield  {title} {\enquote {\bibinfo {title} {String and membrane
				condensation on three-dimensional lattices},}\ }\href {\doibase
		10.1103/PhysRevB.72.035307} {\bibfield  {journal} {\bibinfo  {journal} {Phys.
				Rev. B}\ }\textbf {\bibinfo {volume} {72}},\ \bibinfo {pages} {035307}
		(\bibinfo {year} {2005})},\ \Eprint {http://arxiv.org/abs/cond-mat/0411752}
	{arXiv:cond-mat/0411752 [cond-mat.str-el]} \BibitemShut {NoStop}%
	\bibitem [{\citenamefont {Kong}\ \emph {et~al.}(2020)\citenamefont {Kong},
		\citenamefont {Tian},\ and\ \citenamefont {Zhang}}]{kong2020defects}%
	\BibitemOpen
	\bibfield  {author} {\bibinfo {author} {\bibfnamefont {L.}~\bibnamefont
			{Kong}}, \bibinfo {author} {\bibfnamefont {Y.}~\bibnamefont {Tian}}, \ and\
		\bibinfo {author} {\bibfnamefont {Z.-H.}\ \bibnamefont {Zhang}},\ }\bibfield
	{title} {\enquote {\bibinfo {title} {Defects in the 3-dimensional toric code
				model form a braided fusion 2-category},}\ }\href
	{https://link.springer.com/article/10.1007/JHEP12(2020)078} {\bibfield
		{journal} {\bibinfo  {journal} {Journal of High Energy Physics}\ }\textbf
		{\bibinfo {volume} {2020}},\ \bibinfo {pages} {1} (\bibinfo {year} {2020})},\
	\Eprint {http://arxiv.org/abs/2009.06564} {arXiv:2009.06564
		[cond-mat.str-el]} \BibitemShut {NoStop}%
	\bibitem [{\citenamefont {Delcamp}\ and\ \citenamefont
		{Schuch}(2021)}]{delcamp2021tensornet}%
	\BibitemOpen
	\bibfield  {author} {\bibinfo {author} {\bibfnamefont {C.}~\bibnamefont
			{Delcamp}}\ and\ \bibinfo {author} {\bibfnamefont {N.}~\bibnamefont
			{Schuch}},\ }\bibfield  {title} {\enquote {\bibinfo {title} {On tensor
				network representations of the $(3+1)d$ toric code},}\ }\href
	{https://quantum-journal.org/papers/q-2021-12-16-604/} {\bibfield  {journal}
		{\bibinfo  {journal} {Quantum}\ }\textbf {\bibinfo {volume} {5}},\ \bibinfo
		{pages} {604} (\bibinfo {year} {2021})},\ \Eprint
	{http://arxiv.org/abs/2012.15631} {arXiv:2012.15631 [cond-mat.str-el]}
	\BibitemShut {NoStop}%
	\bibitem [{\citenamefont {Delcamp}(2022)}]{delcamp2022tensor}%
	\BibitemOpen
	\bibfield  {author} {\bibinfo {author} {\bibfnamefont {C.}~\bibnamefont
			{Delcamp}},\ }\bibfield  {title} {\enquote {\bibinfo {title} {Tensor network
				approach to electromagnetic duality in $(3+1)d$ topological gauge models},}\
	}\href {https://link.springer.com/article/10.1007/JHEP08(2022)149} {\bibfield
		{journal} {\bibinfo  {journal} {Journal of High Energy Physics}\ }\textbf
		{\bibinfo {volume} {2022}},\ \bibinfo {pages} {1} (\bibinfo {year} {2022})},\
	\Eprint {http://arxiv.org/abs/2112.08324} {arXiv:2112.08324
		[cond-mat.str-el]} \BibitemShut {NoStop}%
	\bibitem [{\citenamefont {Chang}(2014)}]{chang2014kitaev}%
	\BibitemOpen
	\bibfield  {author} {\bibinfo {author} {\bibfnamefont {L.}~\bibnamefont
			{Chang}},\ }\bibfield  {title} {\enquote {\bibinfo {title} {Kitaev models
				based on unitary quantum groupoids},}\ }\href
	{https://aip.scitation.org/doi/abs/10.1063/1.4869326} {\bibfield  {journal}
		{\bibinfo  {journal} {Journal of Mathematical Physics}\ }\textbf {\bibinfo
			{volume} {55}},\ \bibinfo {pages} {041703} (\bibinfo {year} {2014})},\
	\Eprint {http://arxiv.org/abs/1309.4181} {arXiv:1309.4181 [math.QA]}
	\BibitemShut {NoStop}%
	\bibitem [{\citenamefont {B\"{o}hm}\ \emph {et~al.}(1999)\citenamefont
		{B\"{o}hm}, \citenamefont {Nill},\ and\ \citenamefont
		{Szlach\'{a}nyi}}]{BOHM1998weak}%
	\BibitemOpen
	\bibfield  {author} {\bibinfo {author} {\bibfnamefont {G.}~\bibnamefont
			{B\"{o}hm}}, \bibinfo {author} {\bibfnamefont {F.}~\bibnamefont {Nill}}, \
		and\ \bibinfo {author} {\bibfnamefont {K.}~\bibnamefont {Szlach\'{a}nyi}},\
	}\bibfield  {title} {\enquote {\bibinfo {title} {Weak {H}opf algebras: I.
				{I}ntegral theory and ${C}^*$-structure},}\ }\href {\doibase
		https://doi.org/10.1006/jabr.1999.7984} {\bibfield  {journal} {\bibinfo
			{journal} {Journal of Algebra}\ }\textbf {\bibinfo {volume} {221}},\ \bibinfo
		{pages} {385 } (\bibinfo {year} {1999})},\ \Eprint
	{http://arxiv.org/abs/math/9805116} {arXiv:math/9805116 [math.QA]}
	\BibitemShut {NoStop}%
	\bibitem [{\citenamefont {Ostrik}(2003)}]{ostrik2003module}%
	\BibitemOpen
	\bibfield  {author} {\bibinfo {author} {\bibfnamefont {V.}~\bibnamefont
			{Ostrik}},\ }\bibfield  {title} {\enquote {\bibinfo {title} {Module
				categories, weak {H}opf algebras and modular invariants},}\ }\href
	{https://link.springer.com/content/pdf/10.1007%2Fs00031-003-0515-6.pdf}
	{\bibfield  {journal} {\bibinfo  {journal} {Transformation Groups}\ }\textbf
		{\bibinfo {volume} {8}},\ \bibinfo {pages} {177} (\bibinfo {year} {2003})},\
	\Eprint {http://arxiv.org/abs/math/0111139} {arXiv:math/0111139 [math.QA]}
	\BibitemShut {NoStop}%
	\bibitem [{\citenamefont {B{\"o}hm}\ and\ \citenamefont
		{Szlach{\'o}nyi}(1996)}]{bohm1996coassociative}%
	\BibitemOpen
	\bibfield  {author} {\bibinfo {author} {\bibfnamefont {G.}~\bibnamefont
			{B{\"o}hm}}\ and\ \bibinfo {author} {\bibfnamefont {K.}~\bibnamefont
			{Szlach{\'o}nyi}},\ }\bibfield  {title} {\enquote {\bibinfo {title} {A
				coassociative ${C}^*$-quantum group with nonintegral dimensions},}\ }\href
	{https://link.springer.com/article/10.1007/BF01815526} {\bibfield  {journal}
		{\bibinfo  {journal} {Letters in Mathematical Physics}\ }\textbf {\bibinfo
			{volume} {38}},\ \bibinfo {pages} {437} (\bibinfo {year} {1996})},\ \Eprint
	{http://arxiv.org/abs/q-alg/9509008} {arXiv:q-alg/9509008 [math.QA]}
	\BibitemShut {NoStop}%
	\bibitem [{\citenamefont {Nill}(1998)}]{nill1998axioms}%
	\BibitemOpen
	\bibfield  {author} {\bibinfo {author} {\bibfnamefont {F.}~\bibnamefont
			{Nill}},\ }\bibfield  {title} {\enquote {\bibinfo {title} {Axioms for weak
				bialgebras},}\ }\href@noop {} {\bibfield  {journal} {\bibinfo  {journal}
			{arXiv preprint math/9805104}\ } (\bibinfo {year} {1998})}\BibitemShut
	{NoStop}%
	\bibitem [{\citenamefont {Nikshych}\ \emph {et~al.}(2003)\citenamefont
		{Nikshych}, \citenamefont {Turaev},\ and\ \citenamefont
		{Vainerman}}]{nikshych2003invariants}%
	\BibitemOpen
	\bibfield  {author} {\bibinfo {author} {\bibfnamefont {D.}~\bibnamefont
			{Nikshych}}, \bibinfo {author} {\bibfnamefont {V.}~\bibnamefont {Turaev}}, \
		and\ \bibinfo {author} {\bibfnamefont {L.}~\bibnamefont {Vainerman}},\
	}\bibfield  {title} {\enquote {\bibinfo {title} {Invariants of knots and
				3-manifolds from quantum groupoids},}\ }\href
	{https://www.sciencedirect.com/science/article/pii/S016686410200055X}
	{\bibfield  {journal} {\bibinfo  {journal} {Topology and its Applications}\
		}\textbf {\bibinfo {volume} {127}},\ \bibinfo {pages} {91} (\bibinfo {year}
		{2003})},\ \Eprint {http://arxiv.org/abs/math/0006078} {arXiv:math/0006078
		[math.QA]} \BibitemShut {NoStop}%
	\bibitem [{\citenamefont {Drinfel'd}(1988)}]{drinfel1988quantum}%
	\BibitemOpen
	\bibfield  {author} {\bibinfo {author} {\bibfnamefont {V.~G.}\ \bibnamefont
			{Drinfel'd}},\ }\bibfield  {title} {\enquote {\bibinfo {title} {Quantum
				groups},}\ }\href {https://link.springer.com/article/10.1007/BF01247086}
	{\bibfield  {journal} {\bibinfo  {journal} {Journal of Soviet mathematics}\
		}\textbf {\bibinfo {volume} {41}},\ \bibinfo {pages} {898} (\bibinfo {year}
		{1988})}\BibitemShut {NoStop}%
	\bibitem [{\citenamefont {Majid}(1990)}]{majid1990physics}%
	\BibitemOpen
	\bibfield  {author} {\bibinfo {author} {\bibfnamefont {S.}~\bibnamefont
			{Majid}},\ }\bibfield  {title} {\enquote {\bibinfo {title} {Physics for
				algebraists: Non-commutative and non-cocommutative {H}opf algebras by a
				bicrossproduct construction},}\ }\href
	{https://www.sciencedirect.com/science/article/pii/002186939090099A}
	{\bibfield  {journal} {\bibinfo  {journal} {Journal of Algebra}\ }\textbf
		{\bibinfo {volume} {130}},\ \bibinfo {pages} {17} (\bibinfo {year}
		{1990})}\BibitemShut {NoStop}%
	\bibitem [{\citenamefont {Majid}(1994)}]{majid1994some}%
	\BibitemOpen
	\bibfield  {author} {\bibinfo {author} {\bibfnamefont {S.}~\bibnamefont
			{Majid}},\ }\bibfield  {title} {\enquote {\bibinfo {title} {Some remarks on
				the quantum double},}\ }\href
	{https://link.springer.com/article/10.1007/BF01690458} {\bibfield  {journal}
		{\bibinfo  {journal} {Czechoslovak Journal of Physics}\ }\textbf {\bibinfo
			{volume} {44}},\ \bibinfo {pages} {1059} (\bibinfo {year} {1994})},\ \Eprint
	{http://arxiv.org/abs/hep-th/9409056} {arXiv:hep-th/9409056 [hep-th]}
	\BibitemShut {NoStop}%
	\bibitem [{\citenamefont {Nenciu}(2002)}]{nenciu2002center}%
	\BibitemOpen
	\bibfield  {author} {\bibinfo {author} {\bibfnamefont {A.}~\bibnamefont
			{Nenciu}},\ }\bibfield  {title} {\enquote {\bibinfo {title} {The center
				construction for weak {H}opf algebras},}\ }\href@noop {} {\bibfield
		{journal} {\bibinfo  {journal} {Tsukuba Journal of Mathematics}\ }\textbf
		{\bibinfo {volume} {26}},\ \bibinfo {pages} {189} (\bibinfo {year}
		{2002})}\BibitemShut {NoStop}%
	\bibitem [{\citenamefont {Wei}\ \emph {et~al.}(2022)\citenamefont {Wei},
		\citenamefont {Jia}, \citenamefont {Kaszlikowski},\ and\ \citenamefont
		{Tan}}]{wei2022antilinear}%
	\BibitemOpen
	\bibfield  {author} {\bibinfo {author} {\bibfnamefont {L.}~\bibnamefont
			{Wei}}, \bibinfo {author} {\bibfnamefont {Z.}~\bibnamefont {Jia}}, \bibinfo
		{author} {\bibfnamefont {D.}~\bibnamefont {Kaszlikowski}}, \ and\ \bibinfo
		{author} {\bibfnamefont {S.}~\bibnamefont {Tan}},\ }\bibfield  {title}
	{\enquote {\bibinfo {title} {Antilinear superoperator and quantum geometric
				invariance for higher-dimensional quantum systems},}\ }\href
	{https://arxiv.org/abs/2202.10989} {\bibfield  {journal} {\bibinfo  {journal}
			{arXiv preprint arXiv:2202.10989}\ } (\bibinfo {year} {2022})}\BibitemShut
	{NoStop}%
	\bibitem [{\citenamefont {B{\"o}hm}\ \emph {et~al.}(2011)\citenamefont
		{B{\"o}hm}, \citenamefont {Caenepeel},\ and\ \citenamefont
		{Janssen}}]{bohm2011weak}%
	\BibitemOpen
	\bibfield  {author} {\bibinfo {author} {\bibfnamefont {G.}~\bibnamefont
			{B{\"o}hm}}, \bibinfo {author} {\bibfnamefont {S.}~\bibnamefont {Caenepeel}},
		\ and\ \bibinfo {author} {\bibfnamefont {K.}~\bibnamefont {Janssen}},\
	}\bibfield  {title} {\enquote {\bibinfo {title} {Weak bialgebras and monoidal
				categories},}\ }\href
	{https://www.tandfonline.com/doi/abs/10.1080/00927872.2011.616438?journalCode=lagb20}
	{\bibfield  {journal} {\bibinfo  {journal} {Communications in Algebra}\
		}\textbf {\bibinfo {volume} {39}},\ \bibinfo {pages} {4584} (\bibinfo {year}
		{2011})},\ \Eprint {http://arxiv.org/abs/1103.2261} {arXiv:1103.2261
		[math.QA]} \BibitemShut {NoStop}%
	\bibitem [{\citenamefont {Jia}\ and\ \citenamefont
		{Tan}(tion{\natexlab{a}})}]{Jia2023classify}%
	\BibitemOpen
	\bibfield  {author} {\bibinfo {author} {\bibfnamefont {Z.}~\bibnamefont
			{Jia}}\ and\ \bibinfo {author} {\bibfnamefont {S.}~\bibnamefont {Tan}},\
	}\href@noop {} {\enquote {\bibinfo {title} {Classifying the topological
				excitations of {H}opf and weak {H}opf lattice gauge theory},}\ } (\bibinfo
	{year} {in preparation}{\natexlab{a}})\BibitemShut {NoStop}%
	\bibitem [{\citenamefont {Lusztig}(1987)}]{lusztig1987leading}%
	\BibitemOpen
	\bibfield  {author} {\bibinfo {author} {\bibfnamefont {G.}~\bibnamefont
			{Lusztig}},\ }\bibfield  {title} {\enquote {\bibinfo {title} {Leading
				coefficients of character values of hecke algebras},}\ }in\ \href@noop {}
	{\emph {\bibinfo {booktitle} {Proc. Symp. Pure Math}}},\ Vol.~\bibinfo
	{volume} {47}\ (\bibinfo {year} {1987})\ pp.\ \bibinfo {pages}
	{235--262}\BibitemShut {NoStop}%
	\bibitem [{\citenamefont {Dijkgraaf}\ \emph {et~al.}(1991)\citenamefont
		{Dijkgraaf}, \citenamefont {Pasquier},\ and\ \citenamefont
		{Roche}}]{dijkgraaf1991quasi}%
	\BibitemOpen
	\bibfield  {author} {\bibinfo {author} {\bibfnamefont {R.}~\bibnamefont
			{Dijkgraaf}}, \bibinfo {author} {\bibfnamefont {V.}~\bibnamefont {Pasquier}},
		\ and\ \bibinfo {author} {\bibfnamefont {P.}~\bibnamefont {Roche}},\
	}\bibfield  {title} {\enquote {\bibinfo {title} {Quasi {H}opf algebras, group
				cohomology and orbifold models},}\ }\href
	{https://www.sciencedirect.com/science/article/abs/pii/092056329190123V}
	{\bibfield  {journal} {\bibinfo  {journal} {Nuclear Physics B-Proceedings
				Supplements}\ }\textbf {\bibinfo {volume} {18}},\ \bibinfo {pages} {60}
		(\bibinfo {year} {1991})}\BibitemShut {NoStop}%
	\bibitem [{\citenamefont {Gould}(1993)}]{gould1993quantum}%
	\BibitemOpen
	\bibfield  {author} {\bibinfo {author} {\bibfnamefont {M.}~\bibnamefont
			{Gould}},\ }\bibfield  {title} {\enquote {\bibinfo {title} {Quantum double
				finite group algebras and their representations},}\ }\href
	{https://www.cambridge.org/core/journals/bulletin-of-the-australian-mathematical-society/article/quantum-double-finite-group-algebras-and-their-representations/D98C7E96811FB1C6E0CE05D250D7F527}
	{\bibfield  {journal} {\bibinfo  {journal} {Bulletin of the Australian
				Mathematical Society}\ }\textbf {\bibinfo {volume} {48}},\ \bibinfo {pages}
		{275} (\bibinfo {year} {1993})}\BibitemShut {NoStop}%
	\bibitem [{\citenamefont {Witherspoon}(1996)}]{Witherspoon1996the}%
	\BibitemOpen
	\bibfield  {author} {\bibinfo {author} {\bibfnamefont {S.}~\bibnamefont
			{Witherspoon}},\ }\bibfield  {title} {\enquote {\bibinfo {title} {The
				representation ring of the quantum double of a finite group},}\ }\href
	{\doibase https://doi.org/10.1006/jabr.1996.0014} {\bibfield  {journal}
		{\bibinfo  {journal} {Journal of Algebra}\ }\textbf {\bibinfo {volume}
			{179}},\ \bibinfo {pages} {305} (\bibinfo {year} {1996})}\BibitemShut
	{NoStop}%
	\bibitem [{\citenamefont {Gelaki}\ and\ \citenamefont
		{Nikshych}(2008)}]{gelaki2008nilpotent}%
	\BibitemOpen
	\bibfield  {author} {\bibinfo {author} {\bibfnamefont {S.}~\bibnamefont
			{Gelaki}}\ and\ \bibinfo {author} {\bibfnamefont {D.}~\bibnamefont
			{Nikshych}},\ }\bibfield  {title} {\enquote {\bibinfo {title} {Nilpotent
				fusion categories},}\ }\href
	{https://www.sciencedirect.com/science/article/pii/S0001870807002290}
	{\bibfield  {journal} {\bibinfo  {journal} {Advances in Mathematics}\
		}\textbf {\bibinfo {volume} {217}},\ \bibinfo {pages} {1053} (\bibinfo {year}
		{2008})},\ \Eprint {http://arxiv.org/abs/math/0610726} {arXiv:math/0610726
		[math.QA]} \BibitemShut {NoStop}%
	\bibitem [{\citenamefont {Burciu}(2012)}]{burciu2012irreducible}%
	\BibitemOpen
	\bibfield  {author} {\bibinfo {author} {\bibfnamefont {S.}~\bibnamefont
			{Burciu}},\ }\bibfield  {title} {\enquote {\bibinfo {title} {On the
				irreducible representations of generalized quantum doubles},}\ }\href@noop {}
	{\  (\bibinfo {year} {2012})},\ \Eprint {http://arxiv.org/abs/1202.4315}
	{arXiv:1202.4315 [math.QA]} \BibitemShut {NoStop}%
	\bibitem [{\citenamefont {Kong}\ \emph {et~al.}(2017)\citenamefont {Kong},
		\citenamefont {Wen},\ and\ \citenamefont {Zheng}}]{kong2017boundary}%
	\BibitemOpen
	\bibfield  {author} {\bibinfo {author} {\bibfnamefont {L.}~\bibnamefont
			{Kong}}, \bibinfo {author} {\bibfnamefont {X.-G.}\ \bibnamefont {Wen}}, \
		and\ \bibinfo {author} {\bibfnamefont {H.}~\bibnamefont {Zheng}},\ }\bibfield
	{title} {\enquote {\bibinfo {title} {Boundary-bulk relation in topological
				orders},}\ }\href
	{https://www.sciencedirect.com/science/article/pii/S0550321317302183}
	{\bibfield  {journal} {\bibinfo  {journal} {Nuclear Physics B}\ }\textbf
		{\bibinfo {volume} {922}},\ \bibinfo {pages} {62} (\bibinfo {year} {2017})},\
	\Eprint {http://arxiv.org/abs/1702.00673} {arXiv:1702.00673
		[cond-mat.str-el]} \BibitemShut {NoStop}%
	\bibitem [{\citenamefont {Fuchs}\ \emph {et~al.}(2013)\citenamefont {Fuchs},
		\citenamefont {Schweigert},\ and\ \citenamefont
		{Valentino}}]{fuchs2013bicategories}%
	\BibitemOpen
	\bibfield  {author} {\bibinfo {author} {\bibfnamefont {J.}~\bibnamefont
			{Fuchs}}, \bibinfo {author} {\bibfnamefont {C.}~\bibnamefont {Schweigert}}, \
		and\ \bibinfo {author} {\bibfnamefont {A.}~\bibnamefont {Valentino}},\
	}\bibfield  {title} {\enquote {\bibinfo {title} {Bicategories for boundary
				conditions and for surface defects in 3-d {TFT}},}\ }\href
	{https://link.springer.com/article/10.1007/s00220-013-1723-0} {\bibfield
		{journal} {\bibinfo  {journal} {Communications in Mathematical Physics}\
		}\textbf {\bibinfo {volume} {321}},\ \bibinfo {pages} {543} (\bibinfo {year}
		{2013})},\ \Eprint {http://arxiv.org/abs/1203.4568} {arXiv:1203.4568
		[hep-th]} \BibitemShut {NoStop}%
	\bibitem [{\citenamefont {Nill}\ \emph {et~al.}(1998)\citenamefont {Nill},
		\citenamefont {Szlachanyi},\ and\ \citenamefont {Wiesbrock}}]{nill1998weak}%
	\BibitemOpen
	\bibfield  {author} {\bibinfo {author} {\bibfnamefont {F.}~\bibnamefont
			{Nill}}, \bibinfo {author} {\bibfnamefont {K.}~\bibnamefont {Szlachanyi}}, \
		and\ \bibinfo {author} {\bibfnamefont {H.-W.}\ \bibnamefont {Wiesbrock}},\
	}\bibfield  {title} {\enquote {\bibinfo {title} {Weak {H}opf algebras and
				reducible {J}ones inclusions of depth 2. {I}: From crossed products to
				{J}ones towers},}\ }\href {https://arxiv.org/abs/math/9806130} {\bibfield
		{journal} {\bibinfo  {journal} {arXiv preprint math/9806130}\ } (\bibinfo
		{year} {1998})}\BibitemShut {NoStop}%
	\bibitem [{\citenamefont {B{\"o}hm}(2000)}]{bohm2000doi}%
	\BibitemOpen
	\bibfield  {author} {\bibinfo {author} {\bibfnamefont {G.}~\bibnamefont
			{B{\"o}hm}},\ }\bibfield  {title} {\enquote {\bibinfo {title} {Doi-{H}opf
				modules over weak {H}opf algebras},}\ }\href
	{https://www.tandfonline.com/doi/abs/10.1080/00927870008827113} {\bibfield
		{journal} {\bibinfo  {journal} {Communications in Algebra}\ }\textbf
		{\bibinfo {volume} {28}},\ \bibinfo {pages} {4687} (\bibinfo {year}
		{2000})},\ \Eprint {http://arxiv.org/abs/math/9905027} {arXiv:math/9905027
		[math.QA]} \BibitemShut {NoStop}%
	\bibitem [{\citenamefont {Nikshych}(2000)}]{nikshych2000duality}%
	\BibitemOpen
	\bibfield  {author} {\bibinfo {author} {\bibfnamefont {D.}~\bibnamefont
			{Nikshych}},\ }\bibfield  {title} {\enquote {\bibinfo {title} {A duality
				theorem for quantum groupoids},}\ }\href@noop {} {\bibfield  {journal}
		{\bibinfo  {journal} {Contemporary Mathematics}\ }\textbf {\bibinfo {volume}
			{267}},\ \bibinfo {pages} {237} (\bibinfo {year} {2000})},\ \Eprint
	{http://arxiv.org/abs/math/9912226} {arXiv:math/9912226 [math.QA]}
	\BibitemShut {NoStop}%
	\bibitem [{\citenamefont {Henker}(2011)}]{henker2011module}%
	\BibitemOpen
	\bibfield  {author} {\bibinfo {author} {\bibfnamefont {H.}~\bibnamefont
			{Henker}},\ }\emph {\bibinfo {title} {Module categories over quasi-Hopf
			algebras and weak Hopf algebras and the projectivity of Hopf modules}},\
	\href@noop {} {Ph.D. thesis},\ \bibinfo  {school} {LMU} (\bibinfo {year}
	{2011})\BibitemShut {NoStop}%
	\bibitem [{\citenamefont {Eilenberg}(1960)}]{eilenberg1960abstract}%
	\BibitemOpen
	\bibfield  {author} {\bibinfo {author} {\bibfnamefont {S.}~\bibnamefont
			{Eilenberg}},\ }\bibfield  {title} {\enquote {\bibinfo {title} {Abstract
				description of some basic functors},}\ }\href@noop {} {\bibfield  {journal}
		{\bibinfo  {journal} {J. Indian Math. Soc}\ }\textbf {\bibinfo {volume}
			{24}},\ \bibinfo {pages} {231} (\bibinfo {year} {1960})}\BibitemShut
	{NoStop}%
	\bibitem [{\citenamefont {Watts}(1960)}]{watts1960intrinsic}%
	\BibitemOpen
	\bibfield  {author} {\bibinfo {author} {\bibfnamefont {C.~E.}\ \bibnamefont
			{Watts}},\ }\bibfield  {title} {\enquote {\bibinfo {title} {Intrinsic
				characterizations of some additive functors},}\ }\href@noop {} {\bibfield
		{journal} {\bibinfo  {journal} {Proceedings of the American Mathematical
				Society}\ }\textbf {\bibinfo {volume} {11}},\ \bibinfo {pages} {5} (\bibinfo
		{year} {1960})}\BibitemShut {NoStop}%
	\bibitem [{\citenamefont {Jia}\ and\ \citenamefont
		{Tan}(tion{\natexlab{b}})}]{Jia2023bicomodule}%
	\BibitemOpen
	\bibfield  {author} {\bibinfo {author} {\bibfnamefont {Z.}~\bibnamefont
			{Jia}}\ and\ \bibinfo {author} {\bibfnamefont {S.}~\bibnamefont {Tan}},\
	}\href@noop {} {\enquote {\bibinfo {title} {On bimodule category over
				representation category of weak {H}opf algebra},}\ } (\bibinfo {year} {in
		preparation}{\natexlab{b}})\BibitemShut {NoStop}%
	\bibitem [{\citenamefont {Aguiar}(2000)}]{aguiar2000note}%
	\BibitemOpen
	\bibfield  {author} {\bibinfo {author} {\bibfnamefont {M.}~\bibnamefont
			{Aguiar}},\ }\bibfield  {title} {\enquote {\bibinfo {title} {A note on
				strongly separable algebras},}\ }\href@noop {} {\bibfield  {journal}
		{\bibinfo  {journal} {Bol. Acad. Nac. Cienc.(C{\'o}rdoba)}\ }\textbf
		{\bibinfo {volume} {65}},\ \bibinfo {pages} {51} (\bibinfo {year}
		{2000})}\BibitemShut {NoStop}%
	\bibitem [{\citenamefont {Schuch}\ \emph {et~al.}(2010)\citenamefont {Schuch},
		\citenamefont {Cirac},\ and\ \citenamefont
		{P{\'e}rez-Garc{\'\i}a}}]{schuch2010peps}%
	\BibitemOpen
	\bibfield  {author} {\bibinfo {author} {\bibfnamefont {N.}~\bibnamefont
			{Schuch}}, \bibinfo {author} {\bibfnamefont {I.}~\bibnamefont {Cirac}}, \
		and\ \bibinfo {author} {\bibfnamefont {D.}~\bibnamefont
			{P{\'e}rez-Garc{\'\i}a}},\ }\bibfield  {title} {\enquote {\bibinfo {title}
			{Peps as ground states: Degeneracy and topology},}\ }\href
	{https://www.sciencedirect.com/science/article/abs/pii/S0003491610000990}
	{\bibfield  {journal} {\bibinfo  {journal} {Annals of Physics}\ }\textbf
		{\bibinfo {volume} {325}},\ \bibinfo {pages} {2153} (\bibinfo {year}
		{2010})},\ \Eprint {http://arxiv.org/abs/1001.3807} {arXiv:1001.3807
		[quant-ph]} \BibitemShut {NoStop}%
	\bibitem [{\citenamefont {Lootens}\ \emph {et~al.}(2021)\citenamefont
		{Lootens}, \citenamefont {Fuchs}, \citenamefont {Haegeman}, \citenamefont
		{Schweigert},\ and\ \citenamefont {Verstraete}}]{lootens2021matrix}%
	\BibitemOpen
	\bibfield  {author} {\bibinfo {author} {\bibfnamefont {L.}~\bibnamefont
			{Lootens}}, \bibinfo {author} {\bibfnamefont {J.}~\bibnamefont {Fuchs}},
		\bibinfo {author} {\bibfnamefont {J.}~\bibnamefont {Haegeman}}, \bibinfo
		{author} {\bibfnamefont {C.}~\bibnamefont {Schweigert}}, \ and\ \bibinfo
		{author} {\bibfnamefont {F.}~\bibnamefont {Verstraete}},\ }\bibfield  {title}
	{\enquote {\bibinfo {title} {Matrix product operator symmetries and
				intertwiners in string-nets with domain walls},}\ }\href
	{https://scipost.org/10.21468/SciPostPhys.10.3.053} {\bibfield  {journal}
		{\bibinfo  {journal} {SciPost Physics}\ }\textbf {\bibinfo {volume} {10}},\
		\bibinfo {pages} {053} (\bibinfo {year} {2021})},\ \Eprint
	{http://arxiv.org/abs/2008.11187} {arXiv:2008.11187 [quant-ph]} \BibitemShut
	{NoStop}%
	\bibitem [{\citenamefont {Molnar}\ \emph {et~al.}(2022)\citenamefont {Molnar},
		\citenamefont {de~Alarc{\'o}n}, \citenamefont {Garre-Rubio}, \citenamefont
		{Schuch}, \citenamefont {Cirac},\ and\ \citenamefont
		{P{\'e}rez-Garc{\'\i}a}}]{molnar2022matrix}%
	\BibitemOpen
	\bibfield  {author} {\bibinfo {author} {\bibfnamefont {A.}~\bibnamefont
			{Molnar}}, \bibinfo {author} {\bibfnamefont {A.~R.}\ \bibnamefont
			{de~Alarc{\'o}n}}, \bibinfo {author} {\bibfnamefont {J.}~\bibnamefont
			{Garre-Rubio}}, \bibinfo {author} {\bibfnamefont {N.}~\bibnamefont {Schuch}},
		\bibinfo {author} {\bibfnamefont {J.~I.}\ \bibnamefont {Cirac}}, \ and\
		\bibinfo {author} {\bibfnamefont {D.}~\bibnamefont {P{\'e}rez-Garc{\'\i}a}},\
	}\bibfield  {title} {\enquote {\bibinfo {title} {Matrix product operator
				algebras {I}: representations of weak {H}opf algebras and projected entangled
				pair states},}\ }\href {https://arxiv.org/abs/2204.05940} {\bibfield
		{journal} {\bibinfo  {journal} {arXiv preprint arXiv:2204.05940}\ } (\bibinfo
		{year} {2022})}\BibitemShut {NoStop}%
	\bibitem [{\citenamefont {Freed}\ and\ \citenamefont
		{Teleman}(2022)}]{freed2022topological}%
	\BibitemOpen
	\bibfield  {author} {\bibinfo {author} {\bibfnamefont {D.~S.}\ \bibnamefont
			{Freed}}\ and\ \bibinfo {author} {\bibfnamefont {C.}~\bibnamefont
			{Teleman}},\ }\bibfield  {title} {\enquote {\bibinfo {title} {Topological
				dualities in the {I}sing model},}\ }\href
	{https://msp.org/gt/2022/26-5/gt-v26-n5-p01-s.pdf} {\bibfield  {journal}
		{\bibinfo  {journal} {Geometry \& Topology}\ }\textbf {\bibinfo {volume}
			{26}},\ \bibinfo {pages} {1907} (\bibinfo {year} {2022})},\ \Eprint
	{http://arxiv.org/abs/1806.00008} {arXiv:1806.00008 [math.AT]} \BibitemShut
	{NoStop}%
	\bibitem [{\citenamefont {Aasen}\ \emph {et~al.}(2016)\citenamefont {Aasen},
		\citenamefont {Mong},\ and\ \citenamefont {Fendley}}]{aasen2016topological}%
	\BibitemOpen
	\bibfield  {author} {\bibinfo {author} {\bibfnamefont {D.}~\bibnamefont
			{Aasen}}, \bibinfo {author} {\bibfnamefont {R.~S.}\ \bibnamefont {Mong}}, \
		and\ \bibinfo {author} {\bibfnamefont {P.}~\bibnamefont {Fendley}},\
	}\bibfield  {title} {\enquote {\bibinfo {title} {Topological defects on the
				lattice: I. {T}he {I}sing model},}\ }\href
	{https://iopscience.iop.org/article/10.1088/1751-8113/49/35/354001}
	{\bibfield  {journal} {\bibinfo  {journal} {Journal of Physics A:
				Mathematical and Theoretical}\ }\textbf {\bibinfo {volume} {49}},\ \bibinfo
		{pages} {354001} (\bibinfo {year} {2016})},\ \Eprint
	{http://arxiv.org/abs/1601.07185} {arXiv:1601.07185 [cond-mat.stat-mech]}
	\BibitemShut {NoStop}%
	\bibitem [{\citenamefont {Aasen}\ \emph {et~al.}(2020)\citenamefont {Aasen},
		\citenamefont {Fendley},\ and\ \citenamefont {Mong}}]{aasen2020topological}%
	\BibitemOpen
	\bibfield  {author} {\bibinfo {author} {\bibfnamefont {D.}~\bibnamefont
			{Aasen}}, \bibinfo {author} {\bibfnamefont {P.}~\bibnamefont {Fendley}}, \
		and\ \bibinfo {author} {\bibfnamefont {R.~S.}\ \bibnamefont {Mong}},\
	}\bibfield  {title} {\enquote {\bibinfo {title} {Topological defects on the
				lattice: dualities and degeneracies},}\ }\href
	{https://arxiv.org/abs/2008.08598} {\bibfield  {journal} {\bibinfo  {journal}
			{arXiv preprint arXiv:2008.08598}\ } (\bibinfo {year} {2020})}\BibitemShut
	{NoStop}%
	\bibitem [{\citenamefont {M{\"u}ger}(2003)}]{muger2003subfactorsI}%
	\BibitemOpen
	\bibfield  {author} {\bibinfo {author} {\bibfnamefont {M.}~\bibnamefont
			{M{\"u}ger}},\ }\bibfield  {title} {\enquote {\bibinfo {title} {From
				subfactors to categories and topology {I}: {F}robenius algebras in and
				{M}orita equivalence of tensor categories},}\ }\href@noop {} {\bibfield
		{journal} {\bibinfo  {journal} {Journal of Pure and Applied Algebra}\
		}\textbf {\bibinfo {volume} {180}},\ \bibinfo {pages} {81} (\bibinfo {year}
		{2003})}\BibitemShut {NoStop}%
	\bibitem [{\citenamefont {Etingof}\ \emph {et~al.}(2010)\citenamefont
		{Etingof}, \citenamefont {Nikshych},\ and\ \citenamefont
		{Ostrik}}]{etingof2010fusion}%
	\BibitemOpen
	\bibfield  {author} {\bibinfo {author} {\bibfnamefont {P.}~\bibnamefont
			{Etingof}}, \bibinfo {author} {\bibfnamefont {D.}~\bibnamefont {Nikshych}}, \
		and\ \bibinfo {author} {\bibfnamefont {V.}~\bibnamefont {Ostrik}},\
	}\bibfield  {title} {\enquote {\bibinfo {title} {Fusion categories and
				homotopy theory},}\ }\href
	{https://www.ems-ph.org/journals/show_abstract.php?issn=1663-487X&vol=1&iss=3&rank=1}
	{\bibfield  {journal} {\bibinfo  {journal} {Quantum topology}\ }\textbf
		{\bibinfo {volume} {1}},\ \bibinfo {pages} {209} (\bibinfo {year} {2010})},\
	\Eprint {http://arxiv.org/abs/0909.3140} {arXiv:0909.3140 [math.QA]}
	\BibitemShut {NoStop}%
	\bibitem [{\citenamefont {Etingof}\ \emph {et~al.}(2016)\citenamefont
		{Etingof}, \citenamefont {Gelaki}, \citenamefont {Nikshych},\ and\
		\citenamefont {Ostrik}}]{etingof2016tensor}%
	\BibitemOpen
	\bibfield  {author} {\bibinfo {author} {\bibfnamefont {P.}~\bibnamefont
			{Etingof}}, \bibinfo {author} {\bibfnamefont {S.}~\bibnamefont {Gelaki}},
		\bibinfo {author} {\bibfnamefont {D.}~\bibnamefont {Nikshych}}, \ and\
		\bibinfo {author} {\bibfnamefont {V.}~\bibnamefont {Ostrik}},\ }\href
	{https://bookstore.ams.org/surv-205} {\emph {\bibinfo {title} {Tensor
				categories}}},\ Vol.\ \bibinfo {volume} {205}\ (\bibinfo  {publisher}
	{American Mathematical Soc.},\ \bibinfo {year} {2016})\BibitemShut {NoStop}%
	\bibitem [{\citenamefont {Hu}\ \emph {et~al.}(2018)\citenamefont {Hu},
		\citenamefont {Geer},\ and\ \citenamefont {Wu}}]{Hu2018full}%
	\BibitemOpen
	\bibfield  {author} {\bibinfo {author} {\bibfnamefont {Y.}~\bibnamefont
			{Hu}}, \bibinfo {author} {\bibfnamefont {N.}~\bibnamefont {Geer}}, \ and\
		\bibinfo {author} {\bibfnamefont {Y.-S.}\ \bibnamefont {Wu}},\ }\bibfield
	{title} {\enquote {\bibinfo {title} {Full dyon excitation spectrum in
				extended {L}evin-{W}en models},}\ }\href {\doibase
		10.1103/PhysRevB.97.195154} {\bibfield  {journal} {\bibinfo  {journal} {Phys.
				Rev. B}\ }\textbf {\bibinfo {volume} {97}},\ \bibinfo {pages} {195154}
		(\bibinfo {year} {2018})},\ \Eprint {http://arxiv.org/abs/1502.03433}
	{arXiv:1502.03433 [cond-mat.str-el]} \BibitemShut {NoStop}%
	\bibitem [{\citenamefont {Aljadeff}\ \emph {et~al.}(2002)\citenamefont
		{Aljadeff}, \citenamefont {Etingof}, \citenamefont {Gelaki},\ and\
		\citenamefont {Nikshych}}]{aljadeff2002twisting}%
	\BibitemOpen
	\bibfield  {author} {\bibinfo {author} {\bibfnamefont {E.}~\bibnamefont
			{Aljadeff}}, \bibinfo {author} {\bibfnamefont {P.}~\bibnamefont {Etingof}},
		\bibinfo {author} {\bibfnamefont {S.}~\bibnamefont {Gelaki}}, \ and\ \bibinfo
		{author} {\bibfnamefont {D.}~\bibnamefont {Nikshych}},\ }\bibfield  {title}
	{\enquote {\bibinfo {title} {On twisting of finite-dimensional {H}opf
				algebras},}\ }\href
	{https://www.sciencedirect.com/science/article/pii/S0021869302000923}
	{\bibfield  {journal} {\bibinfo  {journal} {Journal of Algebra}\ }\textbf
		{\bibinfo {volume} {256}},\ \bibinfo {pages} {484} (\bibinfo {year}
		{2002})},\ \Eprint {http://arxiv.org/abs/math/0107167} {arXiv:math/0107167
		[math.QA]} \BibitemShut {NoStop}%
	\bibitem [{\citenamefont {Bombin}(2010)}]{Bombin2010}%
	\BibitemOpen
	\bibfield  {author} {\bibinfo {author} {\bibfnamefont {H.}~\bibnamefont
			{Bombin}},\ }\bibfield  {title} {\enquote {\bibinfo {title} {Topological
				order with a twist: Ising anyons from an {A}belian model},}\ }\href {\doibase
		10.1103/PhysRevLett.105.030403} {\bibfield  {journal} {\bibinfo  {journal}
			{Phys. Rev. Lett.}\ }\textbf {\bibinfo {volume} {105}},\ \bibinfo {pages}
		{030403} (\bibinfo {year} {2010})},\ \Eprint {http://arxiv.org/abs/1004.1838}
	{arXiv:1004.1838 [cond-mat.str-el]} \BibitemShut {NoStop}%
	\bibitem [{\citenamefont {Cong}\ \emph
		{et~al.}(2017{\natexlab{b}})\citenamefont {Cong}, \citenamefont {Cheng},\
		and\ \citenamefont {Wang}}]{cong2017defects}%
	\BibitemOpen
	\bibfield  {author} {\bibinfo {author} {\bibfnamefont {I.}~\bibnamefont
			{Cong}}, \bibinfo {author} {\bibfnamefont {M.}~\bibnamefont {Cheng}}, \ and\
		\bibinfo {author} {\bibfnamefont {Z.}~\bibnamefont {Wang}},\ }\bibfield
	{title} {\enquote {\bibinfo {title} {Defects between gapped boundaries in
				two-dimensional topological phases of matter},}\ }\href {\doibase
		10.1103/PhysRevB.96.195129} {\bibfield  {journal} {\bibinfo  {journal} {Phys.
				Rev. B}\ }\textbf {\bibinfo {volume} {96}},\ \bibinfo {pages} {195129}
		(\bibinfo {year} {2017}{\natexlab{b}})},\ \Eprint
	{http://arxiv.org/abs/1703.03564} {arXiv:1703.03564 [cond-mat.str-el]}
	\BibitemShut {NoStop}%
\end{thebibliography}

%

\end{document}